%% file: main.tex
\documentclass[pdftex,10pt,b5paper,twoside]{book} 
\usepackage[lmargin=25mm,rmargin=25mm,tmargin=27mm,bmargin=30mm]{geometry}

\input{docsetup.tex}

\begin{document}

\pagenumbering{Roman} 				
\setcounter{page}{1}
\title{Master's Thesis: Excitation Spectrum of a Weakly Interacting Spin-Orbit Coupled Bose-Einstein Condensate}

\author{Kristian Mæland}
\date{\today  \flushleft\endgraf\bigskip\bigskip\bigskip\bigskip\bigskip\bigskip\bigskip\bigskip \noindent Master of Science in Physics \endgraf \noindent Submission date: May 2020 \endgraf \noindent Supervisor: Asle Sudbø \bigskip\bigskip\endgraf \noindent Norwegian University of Science and Technology \endgraf \noindent Department of Physics}

\maketitle

\cleardoublepage
\input{summary.tex}

\input{oppsummering.tex}

\input{preface.tex}		
\input{list.tex}			
\cleardoublepage
\pagestyle{fancy}
\fancyhf{}
\renewcommand{\chaptermark}[1]{\markboth{\chaptername\ \thechapter.\ #1}{}}
\renewcommand{\sectionmark}[1]{\markright{\thesection\ #1}}
\renewcommand{\headrulewidth}{0.1ex}
\renewcommand{\footrulewidth}{0.1ex}
\fancyfoot[LE,RO]{\thepage}
\fancyhead[LE]{\leftmark}
\fancyhead[RO]{\rightmark}
\fancypagestyle{plain}{\fancyhf{}\fancyfoot[LE,RO]{\thepage}\renewcommand{\headrulewidth}{0ex}}
\pagenumbering{arabic} 				
\setcounter{page}{1}
\input{chapter1.tex}		
\input{chapter2.tex}
\input{chapter3.tex}
\input{chapter4.tex}
\input{chapter5.tex}
\input{chapter6.tex}

\input{references.tex}	
\input{appendix.tex}		

\end{document}

%% file: docsetup.tex
\usepackage{setspace}
\usepackage[utf8]{inputenc}
\usepackage[english]{babel}
\usepackage{graphicx}
\usepackage{multicol}
\usepackage{amsmath,amsfonts,amsthm,amssymb}
\usepackage{mathtools}
\usepackage{lipsum}
\usepackage{caption}
\usepackage{booktabs} 
\usepackage{subcaption} 
\usepackage{physics}
\usepackage{gensymb} 
\usepackage{anyfontsize} 
\usepackage{mathrsfs}
\usepackage{color}
\usepackage[Lenny]{fncychap}
\usepackage{hyperref}
\hypersetup{colorlinks,allcolors=black}
\usepackage[font=small,labelfont=bf]{caption}
\usepackage{fancyhdr}
\usepackage{times}
\usepackage{pdfpages}
\usepackage{cite} 
\usepackage{float}
\restylefloat{figure}

\usepackage{fancyhdr}
\renewcommand{\chaptermark}[1]{\markboth{\chaptername\ \thechapter.\ #1}{}}
\renewcommand{\sectionmark}[1]{\markright{\thesection\ #1}}
\renewcommand{\headrulewidth}{0.1ex}
\renewcommand{\footrulewidth}{0.1ex}
\fancypagestyle{plain}{\fancyhf{}\fancyfoot[LE,RO]{\thepage}\renewcommand{\headrulewidth}{0ex}}

\newtheorem{theorem}{Theorem}

%% file: summary.tex
\clearpage
\pagenumbering{roman} 				
\setcounter{page}{1}

\pagestyle{fancy}
\fancyhf{}
\renewcommand{\chaptermark}[1]{\markboth{\chaptername\ \thechapter.\ #1}{}}
\renewcommand{\sectionmark}[1]{\markright{\thesection\ #1}}
\renewcommand{\headrulewidth}{0.1ex}
\renewcommand{\footrulewidth}{0.1ex}
\fancyfoot[LE,RO]{\thepage}
\fancypagestyle{plain}{\fancyhf{}\fancyfoot[LE,RO]{\thepage}\renewcommand{\headrulewidth}{0ex}}

\section*{\Huge Summary}
\addcontentsline{toc}{chapter}{Summary}	
$\\[0.5cm]$

\noindent A weakly interacting, spin-orbit coupled, two-component, ultracold Bose gas bound to a Bravais lattice is studied. Motivated by recent experimental advances in the field of synthetically spin-orbit coupled, ultracold, neutral atomic gases showing Bose-Einstein condensation, an analytic framework with which to describe such systems in the superfluid regime is presented. This is applied to a Rashba spin-orbit-coupled Bose gas in a two-dimensional optical lattice. The exotic nature of Bose-Einstein condensation in the presence of spin-orbit coupling is an interesting study by itself. Additionally, when the optical lattice is introduced, the system provides a highly controllable experimental testing ground for numerous condensed matter physics phenomena. Five phases of the system are considered, and their excitation spectra, critical superfluid velocities and free energies are found. In obtaining the free energy, the effects of terms in the Hamiltonian that are linear in excitation operators are included, and such terms have not been studied previously in this context. Minimization of the free energy at zero temperature is used to confirm the phase diagrams reported in the literature, where it has usually been obtained by neglecting the effect of excitations. The plane and stripe wave phases in the phase diagram are bosonic analogues of Fulde-Ferrell-Larkin-Ovchinnikov states in superconductors involving nonzero condensate momenta.

\clearpage

%% file: oppsummering.tex
\section*{\Huge Sammendrag}
\addcontentsline{toc}{chapter}{Sammendrag}	
$\\[0.5cm]$

\noindent En svakt vekselvirkende, spinn-bane koblet, to-komponent, ultrakald Bose-gass bundet til et Bravais gitter blir studert. En analytisk framgangsmåte for å beskrive slike systemer i superfluid regimet blir presentert, motivert av nylig fremgang innen eksperimenter på syntetisk spinn-bane koblede, ultrakalde gasser av nøytrale atomer som viser Bose-Einstein kondensasjon. Dette blir så anvendt på en Rashba spinn-bane koblet Bose-gass i et todimensjonalt optisk gitter. Bose-Einstein kondensasjon sammen med spinn-bane kobling er en interessant studie i seg selv. Videre, ved å introdusere et optisk gitter, gir systemet en høyst kontrollerbar eksperimentell fremgangsmåte for å teste flerfoldige fenomener i faste stoffers fysikk. Eksitasjonsspektre, kritisk superfluid hastighet og fri energi blir funnet for fem faser av systemet. Ledd i Hamiltonoperatoren som er lineære i eksitasjonsoperatorer blir behandlet for å finne fri energi, og slike ledd har ikke blitt studert tidligere i denne sammenhengen. Minimering av fri energi ved null temperatur brukes til å finne et fasediagram som stemmer overens med litteraturen, der det oftest er funnet uten å ta hensyn til eksitasjoner. Plan- og stripebølgefasene i fasediagrammet er bosoniske analogier til Fulde-Ferrell-Larkin-Ovchinnikov tilstander i superledere som involverer kondensering ved ikke-null impuls.

\clearpage

%% file: preface.tex
\section*{\Huge Preface}
\addcontentsline{toc}{chapter}{Preface}
$\\[0.5cm]$

\noindent This Master's thesis presents the results of research conducted in the field of theoretical condensed matter physics. The research was carried out in the final year of the two year Master of Science in Physics program at the Norwegian University of Science and Technology (NTNU). I also completed by Bachelor in Physics at the same university, and I would like to thank NTNU for providing a great arena for the study of physics. Many thanks go to my supervisor Professor Asle Sudbø, whose excellent availability and guidance has been a great help. Furthermore, his excitement for the subject has been a terrific motivation. I would also like to thank fellow Master student Jonas Halse Rygh for rewarding discussions on the topic of this thesis. My gratitude is extended to my other friends and my family for their support.

\begin{flushright}
\noindent Kristian Mæland

\noindent Trondheim, Norway

\noindent May 2020
\end{flushright}

\clearpage 

\section*{\Huge Preface to arXiv Version}
\addcontentsline{toc}{chapter}{Preface to arXiv Version}
$\\[0.5cm]$

\noindent This Master's thesis was based on work done by Andreas T. G. Janssønn, in his Master's thesis \cite{master}. While working on this, I had many valuable discussions with my fellow Master student Jonas H. Rygh. Additionally, the insight of my supervisor Professor Asle Sudbø was instrumental. Together with some continued work after submitting the thesis in May 2020, this resulted in a publication in Physical Review A \cite{PhysRevA.102.053318}. My Master's thesis is now uploaded to arXiv, to act as a comprehensive overview of the methods involved in obtaining the final results in \cite{PhysRevA.102.053318}. In this arXiv version some typos have been corrected, and a longer version of appendix \ref{app:LW} is included. Additionally, some notes are added in italics. I would like to highlight chapters \ref{sec:BV}, \ref{sec:MFTH} and \ref{chap:Excitations} as particularly useful when it comes to the method involved in obtaining the excitation spectrum. Note that while \cite{PhysRevA.102.053318} considers an external Zeeman field, this Master's thesis does not. 

The reader is advised of the following error made in the thesis: the slope of linear excitation spectra is identified as the critical superfluid velocity. Especially for phases with nonzero condensate momenta, this is highly questionable. The reader is referred to \cite{37} which covers this topic, wherein this is also cast into doubt for zero-momentum condensates in a lattice. Hence, whenever superfluidity is discussed, it should have been a discussion of whether phonon-like excitations exist in the system. And when results of the critical superfluid velocity are presented, it should have been called the sound velocity of the excitations.

\begin{flushright}
\noindent Kristian Mæland

\noindent Trondheim, Norway

\noindent November 2020
\end{flushright}

\cleardoublepage

%% file: list.tex
\phantomsection 
\addcontentsline{toc}{chapter}{Table of Contents} 
\tableofcontents
\clearpage



%% file: chapter1.tex

\chapter{Introduction}

Bosons, like the photon for instance, are particles with integer spin which separates them from fermions, like the electron, with half-integer spin. An important consequence is that bosons are not influenced by the Pauli exclusion principle. Unlike fermions, there is in principle no limit to how many bosons that can occupy the same quantum mechanical state. Thus, in certain bosonic systems when cooled below a critical temperature, a macroscopic number of particles can occupy the ground state. This is what is known as Bose-Einstein condensation, named after S. N. Bose and A. Einstein who first studied the concept \cite{Bose, Einstein1, Einstein2}. 

After the discovery of superfluid liquid helium in 1938 in the experiments \cite{Kapitza, Allen}, Bose-Einstein condensation was suggested as a way of describing the system \cite{London}. L. D. Landau further explored the system, accounting for interactions between the condensate and the excitations \cite{Landau}. When dragging an impurity through the condensate below a critical velocity, excitations become energetically unfavorable. Hence, the dissipation is eliminated, and below this critical superfluid velocity the system permits frictionless flow, explaining the superfluid behavior. Later, N. N. Bogoliubov calculated the excitation spectrum and found a linear dispersion close to the minimum adding to the microscopic theory of superfluidity \cite{bogoliubov1947theory}. 

The constituents of atoms are fermions, but due to addition of spins, some atoms have integer spin in total and thus behave like bosons. Hence, ultracold dilute atomic gases can exhibit Bose-Einstein condensation. Unlike the strongly interacting superfluid liquid helium such atomic gasses can be weakly interacting allowing for greater occupation fractions of the condensate. Dilute gases are used to avoid the formation of liquids or solids during the cooling. Typically, both laser cooling and evaporative cooling techniques are used to bring the system down to nanokelvin temperatures. After decades of technological advances in said cooling techniques, Bose-Einstein condensation in ultracold dilute atomic gases was first realized experimentally in 1995 using rubidium atoms in a group led by E. A. Cornell and C. E. Wieman \cite{Cornell}. Bose-Einstein condensation was later achieved in other alkali metals as well, including for lithium atoms by C. C. Bradley et al. \cite{Bradley} and sodium atoms in a group led by W. Ketterle \cite{Ketterle}. For this work, E. A. Cornell, C. E. Wieman and W. Ketterle were awarded the 2001 Nobel Prize in Physics \cite{Nobel2001}.  

One can also use lasers to set up a periodic potential landscape that generates an optical Bravais lattice. With the atoms bound by the periodic potential, the system resembles that of electrons in a crystal lattice. This means it can be used to simulate many phenomena of condensed matter physics. Among the applications of neutral atoms trapped in optical lattices is quantum computing, because the system is highly controllable \cite{Jaksch}. Additionally, the system can be further expanded to study spin-orbit coupling.

Spin-orbit coupling describes the interesting appearance of a coupling between a particle’s spin and its momentum when subjected to an electric field. It is a relativistic effect, derived from the Dirac equation, and therefore breaks Galilean invariance \cite{bransdenQM}. An example is how an electron’s spin couples to its orbital angular momentum in an atom, from which spin-orbit coupling derives its name. One way to understand this interaction, is by thinking of an electron moving in an electric field. If a Lorentz boost to the rest frame of the electron is performed, one finds an effective nonzero magnetic field. The electron has a magnetic dipole moment proportional to its spin, and therefore interacts with this effective magnetic field \cite{PCH, DJGriffithsQM}. Spin-orbit coupling has applications in data storage \cite{manchon2015newRashbaApp}, is important for the quantum spin Hall effect \cite{KaneMele}, for topological insulators \cite{KaneTopologicalInsulators}, and in general the rapidly expanding field of spintronics, in which manipulation of the spins in condensed matter systems is of interest \cite{galitskispielmanSOCrev}.

The first proposals for an experimentally realizable method to introduce a synthetic spin-orbit coupling to a dilute atomic Bose gas were reported in 2002 and 2005 \cite{Higbie, osterloh2005cold, ruseckas2005non}. This was first achieved experimentally in 2011 with a one-dimensional spin-orbit coupling in a group led by I. B. Spielman \cite{lin2011expSOC}. In later years the methods have been refined, and two-dimensional spin orbit couplings have also been achieved \cite{wu2016realization2DSOC}. Many proposals exist for methods to realize any linear combination of Rashba \cite{RashbaSOC} and Dresselhaus \cite{DresselhausSOC} spin-orbit couplings in two dimensions and beyond \cite{galitskispielmanSOCrev, generalSOC, wang2DRashba, spielman2016rashba, zhaiSWRashbaRev}. 

Experimentalists can pick out two states of the atoms with very similar energies, called two hyperfine states, and make sure the occupation numbers of other states are negligible. These two states are then labeled pseudospin up and pseudospin down. The name pseudospin is used because having picked out two states, one can use the same formalism as in a spin-1/2 system. The more mathematical explanation is that the two-dimensional Hilbert-space is isomorphic to a spin-1/2 system. The two pseudospin states are considered as two different components of the condensate. Multi-component condensates with more than two components are also possible.

To introduce a synthetic spin-orbit coupling to the system requires generation of momentum dependent transitions between the two pseudospin states. This can be achieved by lasers, with energies slightly detuned from transition energies of the atoms. The lasers generate transitions between the pseudospin states, and the Doppler effect ensures that the transition rates are dependent on the momenta of the atoms. As mentioned, spin-orbit coupling breaks Galilean invariance. The same is true for this synthetic version, and therefore systems of ultracold bosonic atoms with synthetic spin-orbit coupling are not Galilean invariant, something which has been proven experimentally \cite{HamnerSOCGalileanExp}. 

A reason why systems of ultracold, dilute atomic gasses have garnered so much interest, is because they offer high experimental tunability. Just by changing the frequencies, directions or intensities of the lasers used to generate the optical lattice or the spin-orbit coupling, one can tune parameters like the hopping parameter, interaction strength \cite{LS} and the strength of the spin-orbit coupling \cite{spielmantunable}. The hopping parameter is an energy associated with atoms tunneling between lattice sites that appears in the Bose-Hubbard model. The two-component Bose-Einstein condensate bound to an optical lattice can be described using the Bose-Hubbard model, and as such provides a method to experimentally test the predictions of the model \cite{LS}.

Due to the aforementioned tunability, these systems are also good probes of quantum phenomena that are often difficult to detect in solid state materials. Of particular interest to this thesis, one can study the concept of spin-orbit coupling in great detail. Being a relativistic effect, the effects of true spin-orbit coupling are often difficult to measure, and one is not able to tune its strength. Additionally, this thesis studies states which can be thought of as bosonic analogues to fermionic Fulde-Ferrell-Larkin-Ovchinnikov states in superconductors \cite{FF, LO}. This further connects the system to superconductors, which have many technological applications \cite{SFsuperconductivity}.

The structure of the thesis is as follows. In chapter \ref{chap:Prelim} we present preliminary material regarding the Bose-Hubbard model, spin-orbit coupling and superfluidity. The special cases of non-interacting, spin-orbit coupled Bose gas and a weakly interacting Bose gas with no spin-orbit coupling are presented. In addition, a generalized diagonalization method for Hamiltonians quadratic in bosonic operators is studied extensively, due to its heavy usage in the thesis. A mean field theory is applied to the Bose-Hubbard model describing the two-component, weakly interacting, spin-orbit coupled Bose-Einstein condensate in chapter \ref{chap:MFT}. The most interesting phases of the system identified in chapter \ref{chap:MFT} are then studied in chapter \ref{chap:Excitations}, wherein the elementary excitations and critical superfluid velocities are found. In addition, the free energy at zero temperature, i.e. the ground state energy, is obtained, which allows for the construction of a phase diagram in chapter \ref{chap:PDdiscuss}, presented together with a discussion of the overall results of the thesis. The conclusions are summarized in chapter \ref{chap:ConcOutlook} together with an outlook on potential continuations and applications of the results. The appendices give further details of the calculations.

\cleardoublepage

%% file: chapter2.tex

\chapter{Preliminaries} \label{chap:Prelim}
\section{Notation}
In this thesis vector quantities are denoted in bold font, e.g. $\boldsymbol{x}$. Unit vectors are denoted $\hat{\boldsymbol{x}} = \boldsymbol{x}/|\boldsymbol{x}|$. Operators and matrices are not given a special notation, the fact that they are operators and matrices should be clear from context. For a matrix $M$ we will use the notation $M^T$ for its transpose, $M^\dagger$ for its Hermitian conjugate and $M^*$ for its complex conjugate. The identity matrix will be denoted $I$, its size will be left implicit. The Pauli matrices are represented by $\sigma_i$ for $i=x,y,z$, and the usual definitions
\begin{equation}
    \sigma_x = \begin{pmatrix} 0 & 1 \\ 1 & 0    \end{pmatrix}, \mbox{\quad} \sigma_y = \begin{pmatrix} 0 & -i \\ i & 0    \end{pmatrix} \mbox{\quad and \quad} \sigma_z = \begin{pmatrix} 1 & 0 \\ 0 & -1    \end{pmatrix}
\end{equation}
are used. We will let $\alpha, \beta =\{ \uparrow, \downarrow \}$ represent spin indices. A $2 \cross 2$ matrix labelled $\eta^{\alpha\beta}$ represents the elements of the matrix in the sense that
\begin{equation}
    \eta = \begin{pmatrix} \eta^{\uparrow\uparrow} & \eta^{\uparrow\downarrow} \\ \eta^{\downarrow\uparrow} & \eta^{\downarrow\downarrow}   \end{pmatrix}.
\end{equation}
Planck's constant divided by $2\pi$ is set equal to one throughout the thesis, i.e. $\hbar = 1$. To simplify some expressions, $ab^\dagger + ba^\dagger = ab^\dagger + \textrm{H.c.}$ will be used, where H.c. indicates that a term is the same as the Hermitian conjugate of the preceding term. 

\section{Bose-Hubbard Model}
This thesis is concerned with ultracold bosonic atoms bound to optical lattices. The formation of optical lattices in one, two and three dimensions is described in \cite{bloch2008many, PethickSmith}. The simplest configurations utilize counterpropagating lasers with the same frequency that generate standing waves. Through the ac Stark effect, the energy of an atom is shifted in the presence of an electric field. With the periodic electric field from the lasers, this can be thought of as a periodic external potential acting on the atom \cite{PethickSmith}. 

One of the reasons we introduce an optical lattice is that the system then resembles electrons in a crystal potential. Hence, experiments on cold atom systems in optical lattices can be used to test theories from condensed matter systems \cite{PethickSmith}. The advantage of the cold atom experiments is the high degree of tunability of the parameters in the system. The optical lattice is generated by controllable external lasers. Hence, the lattice constant, the hopping parameter and the interaction parameters can be tuned by changing the frequency or intensity of the lasers \cite{PethickSmith}. E.g. by increasing the intensity of the laser the periodic potential becomes deeper, thus reducing the hopping parameter and increasing the on-site interactions \cite{LS}. The interactions can also be tuned using Feshback resonance which can alter the scattering lengths, as described in \cite{PethickSmith}. This appears when the total energy $E$ of the particles in the interaction is close to the energy of a bound state in the system, $E_{\textrm{res}}$. The scattering length then has a contribution \cite{PethickSmith}
\begin{equation}
    a_s \sim \frac{1}{E-E_{\textrm{res}}}.
\end{equation}
The energy of the bound states can e.g. be controlled by an external magnetic field, making it possible to tune the interaction parameters \cite{PethickSmith}.

Since we will introduce a synthetic spin-$1/2$ spin-orbit coupling (SOC) to the system we need to have two components that act as the two pseudospin states. Hence, we are considering a weakly interacting, SOC, two-component Bose gas bound to a Bravais lattice. We will assume the temperature is below the critical temperature for Bose-Einstein condensation (BEC) to occur such that only the low energy contribution to the scattering amplitude is of importance. This is described by the s-wave scattering length, $a_s$ \cite{Pitaevskii, abrikosov}. The Bose gas is also assumed to be dilute enough that any scatterings beyond two-body scatterings can be neglected. The condition for this is $n|a_s|^3 \ll 1$, where $n=N/V$ is the number of particles per volume, i.e. the average separation between particles is much greater that the s-wave scattering length \cite{Pitaevskii}. 

Our starting point is the same as the Hamiltonian used by Linder and Sudbø \cite{LS} to describe a weakly interacting, two-component BEC without SOC. This Hamiltonian was also used by Janssønn \cite{master} and the following derivations follow these references closely. In second quantization we describe the system in terms of bosonic field operators $\psi^{\alpha\dagger} (\boldsymbol{r}), \psi^\alpha (\boldsymbol{r})$ creating or annihilating bosons of particle species $\alpha$ at position $\boldsymbol{r}$. We have two bosonic species labeled $\alpha, \beta = \{\uparrow, \downarrow \}$, for pseudospin up and down, with masses $m^\alpha$. The Hamiltonian is
\begin{align}
    \begin{split}
    \label{eq:firstH}
        H = &\sum_\alpha \int d\boldsymbol{r} \psi^{\alpha\dagger} (\boldsymbol{r}) h^\alpha (\boldsymbol{r}) \psi^\alpha (\boldsymbol{r}) \\
        &+\frac{1}{2}\sum_{\alpha\beta} \int d\boldsymbol{r} d\boldsymbol{r}' \psi^{\alpha\dagger} (\boldsymbol{r}) \psi^{\beta\dagger} (\boldsymbol{r}') v^{\alpha\beta}(\abs{\boldsymbol{r}-\boldsymbol{r}'})\psi^{\beta} (\boldsymbol{r}') \psi^{\alpha} (\boldsymbol{r}).
    \end{split}
\end{align}
Here, $h^\alpha (\boldsymbol{r})$ is the single particle Hamiltonian, while $v^{\alpha\beta}(\abs{\boldsymbol{r}-\boldsymbol{r}'})$ represents the two-body scattering potential. The single particle Hamiltonian is given by
\begin{equation}
    h^\alpha (\boldsymbol{r}) = -\frac{\nabla^2}{2m^\alpha} -\mu^\alpha + V (\boldsymbol{r}), 
\end{equation}
where $\mu^\alpha$ is a species dependent chemical potential and  $V(\boldsymbol{r})$ represents the external potential generating the optical lattice. Hence, if $\boldsymbol{a}_n, n=1,\dots,d$ are the $d$ primitive vectors of the $d$- dimensional ($d$D) Bravais lattice we have
\begin{align}
    \begin{split}
    \label{eq:hlatticeperiodic}
        V(\boldsymbol{r} + c_1 \boldsymbol{a}_1 + \dots + c_d \boldsymbol{a}_d) &= V(\boldsymbol{r}) \Rightarrow \\
        h^\alpha(\boldsymbol{r} + c_1 \boldsymbol{a}_1 + \dots + c_d \boldsymbol{a}_d) &= h^\alpha(\boldsymbol{r}), \mbox{\qquad} c_n \in \mathbb{Z}.
    \end{split}
\end{align}
We will also assume $v^{\alpha\beta}(\abs{\boldsymbol{r}-\boldsymbol{r}'}) = v^{\beta\alpha}(\abs{\boldsymbol{r}-\boldsymbol{r}'})$, i.e. that the interspecies interaction only depends on the relative presence of particle species. The terms in \eqref{eq:firstH} are visualized by Feynman diagrams in figure 2.1 of \cite{master}. 

As done in \cite{master, LS} and discussed in \cite{bloch2008many} we assume we can expand the bosonic field operators using a basis of Wannier functions $w^\alpha(\boldsymbol{r}-\boldsymbol{r}_i)$ located at the lattice sites $\boldsymbol{r}_i$. This is done to obtain a lattice formulation of the Hamiltonian in terms of bosonic operators $b_i^{\alpha\dagger}, b_i^\alpha$ creating or annihilating bosons of particle species $\alpha$ at specific lattice sites $i$ . Inserting
\begin{equation}
    \psi^\alpha (\boldsymbol{r}) = \sum_i w^\alpha(\boldsymbol{r}-\boldsymbol{r}_i)b_i^\alpha
\end{equation}
in \eqref{eq:firstH} yields
\begin{align}
    \begin{split}
        H &= \sum_\alpha \int d\boldsymbol{r}\sum_{ij} w^{\alpha*}(\boldsymbol{r}-\boldsymbol{r}_i)b_i^{\alpha\dagger} h^\alpha (\boldsymbol{r}) w^\alpha (\boldsymbol{r}-\boldsymbol{r}_j)b_j^\alpha \\
        &+\frac{1}{2}\sum_{\alpha\beta}\int d\boldsymbol{r}d\boldsymbol{r}' \sum_{ijkl} w^{\alpha*}(\boldsymbol{r}-\boldsymbol{r}_i)b_i^{\alpha\dagger} w^{\beta*}(\boldsymbol{r}'-\boldsymbol{r}_j)b_j^{\beta\dagger} \\
        &\mbox{\qquad\qquad\qquad\qquad} \cdot v^{\alpha\beta}(\abs{\boldsymbol{r}-\boldsymbol{r}'}) w^{\beta}(\boldsymbol{r}'-\boldsymbol{r}_k)b_k^{\beta}w^{\alpha}(\boldsymbol{r}-\boldsymbol{r}_l)b_l^{\alpha}\\
        &= -\sum_\alpha \sum_{i\neq j } t_{ij}^\alpha b_i^{\alpha\dagger}b_j^\alpha + \sum_\alpha \sum_i T_i^\alpha b_i^{\alpha\dagger}b_i^\alpha \\
        &+\frac{1}{2}\sum_{\alpha\beta}\sum_{ijkl} U_{ijkl}^{\alpha\beta} b_i^{\alpha\dagger}b_j^{\beta\dagger}b_k^\beta b_l^\alpha,
    \end{split}
\end{align}
where the hopping parameter
\begin{equation}
    t_{ij}^\alpha = -\int d\boldsymbol{r} w^{\alpha*}(\boldsymbol{r}-\boldsymbol{r}_i) h^\alpha (\boldsymbol{r}) w^\alpha (\boldsymbol{r}-\boldsymbol{r}_j) 
\end{equation}
is an energy associated with particles hopping between lattice sites $i$ and $j$. The quantity
\begin{align}
    \begin{split}
        T_i^\alpha &= \int d\boldsymbol{r} w^{\alpha*}(\boldsymbol{r}-\boldsymbol{r}_i) h^\alpha (\boldsymbol{r}) w^\alpha (\boldsymbol{r}-\boldsymbol{r}_i) \\
        & \stackrel{(\ref{eq:hlatticeperiodic})}{=} \int d\boldsymbol{r} w^{\alpha*}(\boldsymbol{r}) h^\alpha (\boldsymbol{r}) w^\alpha (\boldsymbol{r}) \equiv T^\alpha
    \end{split}
\end{align}
is a species dependent energy offset at each lattice site \cite{LS}. The interaction parameters are
\begin{align}
\begin{split}
    U_{ijkl}^{\alpha\beta} &= \int d\boldsymbol{r}d\boldsymbol{r}'  w^{\alpha*}(\boldsymbol{r}-\boldsymbol{r}_i) w^{\beta*}(\boldsymbol{r}'-\boldsymbol{r}_j)\\
        &\mbox{\qquad\qquad} \cdot v^{\alpha\beta}(\abs{\boldsymbol{r}-\boldsymbol{r}'})w^{\beta}(\boldsymbol{r}'-\boldsymbol{r}_k)w^{\alpha}(\boldsymbol{r}-\boldsymbol{r}_l).
\end{split}
\end{align}
From now on, it is assumed that the lattice depth is sufficiently large to ensure neighboring Wannier functions have negligible overlap. In such a tight-binding limit, the Wannier functions decay exponentially away from the lattice sites \cite{bloch2008many}, and it is assumed that only nearest neighbor hopping and on-site interactions are relevant. We then have
\begin{equation}
    v^{\alpha\beta}(\abs{\boldsymbol{r}-\boldsymbol{r}'}) = \gamma^{\alpha\beta}\delta(\boldsymbol{r}-\boldsymbol{r}'),
\end{equation}
where \cite{LS}
\begin{equation}
    \gamma^{\alpha\beta} = \gamma^{\beta\alpha} = \frac{2\pi (m^\alpha + m^\beta)a_s^{\alpha\beta}}{m^\alpha m^\beta}.
\end{equation}
Hence, the particles are subjected to interactions only when they occupy the same lattice site. The interaction strength is proportional to the inter- and intraspecies s-wave scattering lengths $a_s^{\alpha\beta}$. The only relevant interaction parameters are $U_{iiii}^{\alpha\beta}$ that now become
\begin{align}
    \begin{split}
        U_{iiii}^{\alpha\beta} &= \int d\boldsymbol{r}d\boldsymbol{r}'  w^{\alpha*}(\boldsymbol{r}-\boldsymbol{r}_i) w^{\beta*}(\boldsymbol{r}'-\boldsymbol{r}_i)\\
        &\mbox{\qquad\qquad} \cdot \gamma^{\alpha\beta}\delta(\boldsymbol{r}-\boldsymbol{r}')w^{\beta}(\boldsymbol{r}'-\boldsymbol{r}_i)w^{\alpha}(\boldsymbol{r}-\boldsymbol{r}_i) \\
        &= \int d\boldsymbol{r} \gamma^{\alpha\beta} \abs{w^{\alpha}(\boldsymbol{r})}^2\abs{w^{\beta}(\boldsymbol{r})}^2 \equiv U^{\alpha\beta} = U^{\beta\alpha}.
    \end{split}
\end{align}
Additionally, it is assumed that the hopping parameter is the same for all nearest neighbor hoppings, i.e.
\begin{equation}
    t_{\langle i, j \rangle}^\alpha \equiv t^\alpha,
\end{equation}
where $\langle i, j \rangle$ denotes nearest neighbors. The final Bose-Hubbard Hamiltonian in real space is then
\begin{align}
    \begin{split}
    \label{eq:realspaceBoseHubbard}
        H &= -\sum_\alpha t^\alpha\sum_{\langle i, j \rangle}  b_i^{\alpha\dagger}b_j^\alpha + \sum_\alpha T^\alpha \sum_i  b_i^{\alpha\dagger}b_i^\alpha \\
        &+\frac{1}{2}\sum_{\alpha\beta} U^{\alpha\beta} \sum_{i}  b_i^{\alpha\dagger}b_i^{\beta\dagger}b_i^\beta b_i^\alpha.
    \end{split}
\end{align}
The parameters $t^\alpha$, $T^\alpha$ and $U^{\alpha\beta}$ will be assumed real. Also, we assume $t^\alpha$ and $U^{\alpha\beta}$ are positive, such that hopping it energetically favorable, and interactions are energetically unfavorable. Repulsive interactions are also a natural choice together with diluteness to ensure the Bose gas does not form a liquid or a solid during the cooling process \cite{abrikosov}. 

BEC is closely related to the momentum distribution of the particles. It will therefore be favorable to study the system in momentum space by performing a Fourier transform of the bosonic operators
\begin{equation}
\label{eq:FTbosonicop}
    b_i^\alpha = \frac{1}{\sqrt{N_s}}\sum_{\boldsymbol{k}} A_{\boldsymbol{k}}^\alpha e^{-i \boldsymbol{k}\cdot \boldsymbol{r}_i}. 
\end{equation}
Here, $N_s$ is the number of lattice sites and $A_{\boldsymbol{k}}^\alpha$ is a bosonic operator annihilating a boson of particle species $\alpha$ with momentum $\boldsymbol{k}$. Inserting \eqref{eq:FTbosonicop} into \eqref{eq:realspaceBoseHubbard} yields
\begin{align}
    \begin{split}
        H &= -\frac{1}{N_s} \sum_\alpha t^\alpha \sum_{\langle i, j \rangle } \sum_{\boldsymbol{k}\boldsymbol{k}'} A_{\boldsymbol{k}}^{\alpha\dagger} e^{i \boldsymbol{k}\cdot \boldsymbol{r}_i} A_{\boldsymbol{k}'}^\alpha e^{-i \boldsymbol{k}'\cdot \boldsymbol{r}_j} \\
        &+ \frac{1}{N_s} \sum_\alpha T^\alpha  \sum_i \sum_{\boldsymbol{k}\boldsymbol{k}'} A_{\boldsymbol{k}}^{\alpha\dagger} e^{i \boldsymbol{k}\cdot \boldsymbol{r}_i} A_{\boldsymbol{k}'}^\alpha e^{-i \boldsymbol{k}' \cdot \boldsymbol{r}_i} \\
        &+\frac{1}{2N_s^2} \sum_{\alpha\beta}U^{\alpha\beta}\sum_i \sum_{\boldsymbol{k}\boldsymbol{k}'\boldsymbol{p}\boldsymbol{p}'} A_{\boldsymbol{k}}^{\alpha\dagger}A_{\boldsymbol{k}'}^{\beta\dagger}A_{\boldsymbol{p}}^{\beta}A_{\boldsymbol{p}'}^{\alpha}e^{i(\boldsymbol{k}+\boldsymbol{k}'-\boldsymbol{p}-\boldsymbol{p}')\cdot \boldsymbol{r}_i}.
    \end{split}
\end{align}
Using 
\begin{equation}
\label{eq:deltasum}
    \frac{1}{N_s}\sum_i e^{i(\boldsymbol{k}-\boldsymbol{k}')\cdot \boldsymbol{r}_i} = \delta_{\boldsymbol{k}\boldsymbol{k}'}
\end{equation}
and
\begin{equation}
    e^{i\boldsymbol{k}\cdot \boldsymbol{r}_i}e^{-i\boldsymbol{k}' \cdot \boldsymbol{r}_j } = e^{-i\boldsymbol{k}'\cdot(\boldsymbol{r}_j -\boldsymbol{r}_i)}e^{i(\boldsymbol{k}-\boldsymbol{k}')\cdot\boldsymbol{r}_i} = e^{-i\boldsymbol{k}'\cdot \boldsymbol{\delta}_{ji}}e^{i(\boldsymbol{k}-\boldsymbol{k}')\cdot\boldsymbol{r}_i}
\end{equation}
allows for some simplifications. Applied to the hopping term we find
\begin{align}
    \begin{split}
        &\frac{1}{N_s} t^\alpha \sum_{\langle i, j \rangle } \sum_{\boldsymbol{k}'} e^{i \boldsymbol{k}\cdot \boldsymbol{r}_i} e^{-i \boldsymbol{k}'\cdot \boldsymbol{r}_j} = \sum_{\boldsymbol{k}'}\sum_{\boldsymbol{\delta}\in \boldsymbol{\delta}_{\langle i, j \rangle}} t^\alpha e^{-i\boldsymbol{k}'\cdot \boldsymbol{\delta}} \frac{1}{N_s} \sum_i e^{i(\boldsymbol{k}-\boldsymbol{k}')\cdot \boldsymbol{r}_i}\\
        &= \sum_{\boldsymbol{k}'}\sum_{\boldsymbol{\delta}\in \boldsymbol{\delta}_{\langle i, j \rangle}} t^\alpha e^{-i\boldsymbol{k}'\cdot \boldsymbol{\delta}} \delta_{\boldsymbol{k}\boldsymbol{k}'} = \sum_{\boldsymbol{\delta}\in \boldsymbol{\delta}_{\langle i, j \rangle}} t^\alpha e^{-i\boldsymbol{k}\cdot \boldsymbol{\delta}}.
    \end{split}
\end{align}
The nearest neighbor vectors are 
\begin{equation}
    \boldsymbol{\delta}_{\langle i, j \rangle} \equiv \{ \pm \boldsymbol{a}_1, \dots, \pm \boldsymbol{a}_d \}.
\end{equation}
Using these, we define
\begin{align}
\label{eq:generalek}
    \begin{split}
    \epsilon_{\boldsymbol{k}}^\alpha &\equiv - t^\alpha  \sum_{\boldsymbol{\delta}\in \boldsymbol{\delta}_{\langle i, j \rangle}} e^{-i\boldsymbol{k}\cdot \boldsymbol{\delta}} = -  t^\alpha \sum_{n=1}^d \left(e^{i\boldsymbol{k}\cdot \boldsymbol{a}_n} +e^{-i\boldsymbol{k}\cdot \boldsymbol{a}_n} \right) \\
    &= -2t^\alpha \sum_{n=1}^d \cos(\boldsymbol{k}\cdot \boldsymbol{a}_n).
    \end{split}
\end{align}
In total, we find the Hamiltonian 
\begin{align}
    \begin{split}
    \label{eq:boseHubbardnoSOC}
        H &= \sum_{\boldsymbol{k}}\sum_\alpha (\epsilon_{\boldsymbol{k}}^\alpha + T^\alpha) A_{\boldsymbol{k}}^{\alpha\dagger}A_{\boldsymbol{k}}^\alpha \\
        &+\frac{1}{2N_s}\sum_{\boldsymbol{k}\boldsymbol{k}'\boldsymbol{p}\boldsymbol{p}'}\sum_{\alpha\beta}U^{\alpha\beta} A_{\boldsymbol{k}}^{\alpha\dagger}A_{\boldsymbol{k}'}^{\beta\dagger}A_{\boldsymbol{p}}^{\beta}A_{\boldsymbol{p}'}^{\alpha} \delta_{\boldsymbol{k}+\boldsymbol{k}', \boldsymbol{p}+\boldsymbol{p}'}.
    \end{split}
\end{align}
In the next subchapter we discuss a synthetic SOC using the two particle species as pseudospin states, and how it can be modeled analytically and added to the above Hamiltonian. 

\section{Synthetic Spin-Orbit Coupling}
A 2D electron gas in the $xy$-plane subjected to an electric field in the $z$-direction, $\boldsymbol{E} = E \hat{\boldsymbol{z}}$, experiences a spin-orbit coupling
\begin{equation}
    H_{\textrm{SOC}} \propto \boldsymbol{\sigma} \cdot (\boldsymbol{E}\cross \boldsymbol{k})
\end{equation}
as used by Bychov and Rashba to explain spin-resonance in 2D semi-conductors \cite{RashbaSOC}. With $k_z = 0$ in 2D, this is
\begin{equation}
    H_{\textrm{SOC}} = \lambda_R (\sigma_x k_y -\sigma_y k_x),
\end{equation}
where $\lambda_R$ is the Rashba SOC strength. Dresselhaus also proposed a coupling of higher order in momentum that can be represented as  \cite{DresselhausSOC}
\begin{equation}
    H_{\textrm{SOC}} = \lambda_D (\sigma_x k_x -\sigma_y k_y)
\end{equation}
in 2D. 

As mentioned, SOC is derived from the Dirac equation and is therefore a relativistic effect \cite{bransdenQM}. Hence its effects are only significant in electron systems when the electrons have relativistic speed or are subjected to strong electric fields. The latter is the case for electrons in numerous condensed matter systems. However, in condensed matter systems the parameters are largely constrained by the properties of the material. The synthetic SOC introduced to cold atom systems can however be controlled externally, and thus provides a platform to study the effects of SOC in greater detail experimentally. The first realization of SOC in neutral bosonic atoms engineered a 1D SOC that displayed an equal combination of Rashba and Dresselhaus SOC \cite{lin2011expSOC}. A highly tunable version was later reported in 2015 \cite{spielmantunable}. Many proposals exist for ways to generalize these methods to obtain higher dimensional SOC and arbitrary linear combinations of Rashba and Dresselhaus SOC \cite{galitskispielmanSOCrev, wang2DRashba, spielman2016rashba, zhaiSWRashbaRev}. A tunable 2D SOC was achieved for bosons in 2016 by Wu et al. \cite{wu2016realization2DSOC, generalSOC}.

The most widely used method of introducing a synthetic SOC to a system of cold neutral atoms employs Raman transitions. Raman transitions are transitions between two atomic states via an intermediate state induced by absorption and emission of two photons. Versions of this were used in \cite{lin2011expSOC, spielmantunable, wu2016realization2DSOC} among others. Though we will focus on pure Rashba SOC in 2D, we will below give a short and simplified introduction to the experimental method proposed in \cite{Higbie} and used in \cite{lin2011expSOC} to produce an equal combination of Rashba and Dresselhaus SOC affecting one dimension. As was stated, many of the proposals to create pure 2D Rashba SOC are generalizations of this procedure.

Let $\ket{a}$ and $\ket{b}$ represent two states of the atoms of approximately equal energy, i.e. two hyperfine states. These will be labeled pseudospin up and pseudospin down, and represents the two components of the system. Experimentalists can ensure that the occupation numbers of other states are negligible. The intermediate excited state is labeled $\ket{e}$ and the energy difference of the states $\ket{a}$ and $\ket{b}$ is $\omega_0$. The illustration in figure \ref{fig:Raman} accompanies the following description of the Raman transition.

\begin{figure}
    \centering
    \includegraphics[width = 0.5\linewidth]{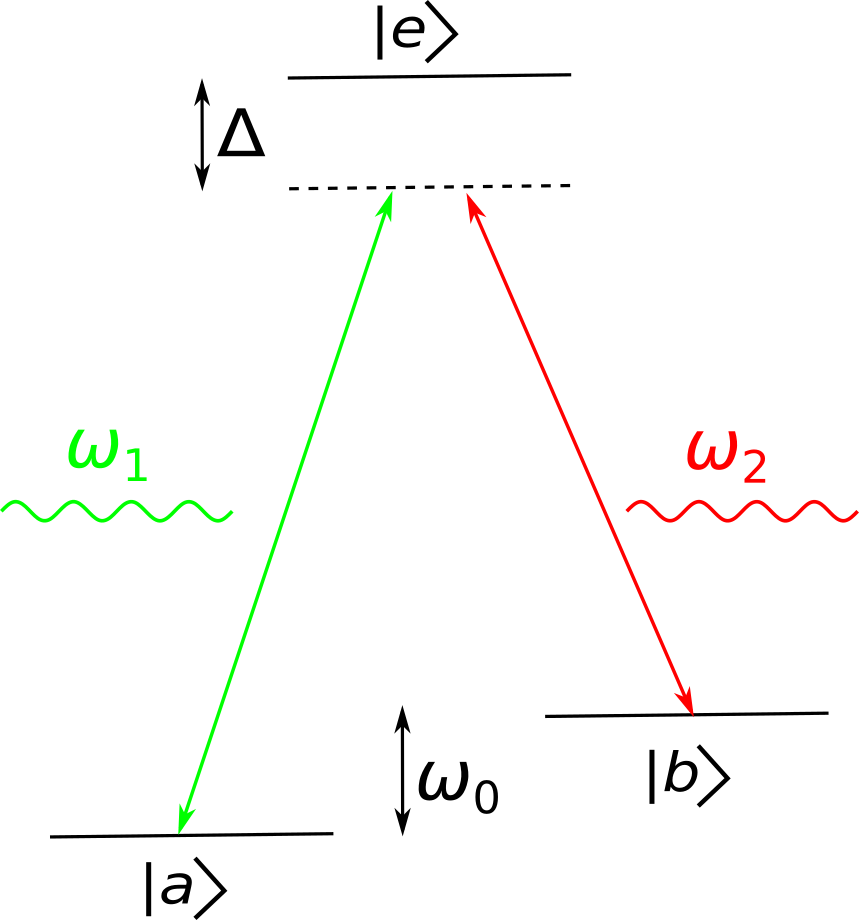}
    \caption{A Raman transition between hyperfine states $\ket{a}$ and $\ket{b}$ via an excited state $\ket{e}$ induced by lasers with frequency $\omega_1$ and $\omega_2$. Figure adapted from \cite{Higbie, galitskispielmanSOCrev}.}
    \label{fig:Raman}
\end{figure}

A laser with frequency $\omega_1$ detuned $\Delta$ from the energy difference of $\ket{e}$ and $\ket{a}$ is introduced along with a laser with frequency $\omega_2$ detuned $\Delta$ from the relative energy of $\ket{e}$ and $\ket{b}$. These lasers induce transitions between the hyperfine states via the intermediate state $\ket{e}$ by absorption of a photon from one laser, and stimulated emission of a photon with the same frequency as the other laser. Since the frequency experienced by an atom depends on its velocity through the Doppler effect, the transition rates will depend on the momenta of the atoms. Thus, a momentum dependent transition between two pseudospin states has been achieved, emulating the SOC experienced by spin-$1/2$ particles. 

From now on, this thesis is concerned with modeling a pure Rashba SOC in 2D due to its numerous applications in condensed matter systems like the aforementioned quantum spin-Hall effect and topological insulators as discussed in \cite{manchon2015newRashbaApp}.  The starting point is the Rashba SOC Hamiltonian
\begin{equation}
\label{eq:RashbaSOC}
    H_{\textrm{SOC}} = \lambda_R (\sigma_x k_y -\sigma_y k_x).
\end{equation}
A heuristic discretization of the above Hamiltonian to a 2D Bravais lattice was performed by Solli \cite{Solli}, with corrections provided by Janssønn \cite{master}, based on work by Sjømark \cite{smark} in 1D. The same will be presented here, with minor adjustments due to some typos in \cite{master}. The end result will be the same that was found by Thingstad \cite{Thingstad} using an alternate method, suggesting the heuristic approach is valid. The goal is to write $H_{\textrm{SOC}}$ on a form which can be incorporated in the Bose-Hubbard Hamiltonian \eqref{eq:boseHubbardnoSOC}

In terms of the lattice operators
\begin{equation}
    b_i = \begin{pmatrix}
    b_i^\uparrow \\ b_i^\downarrow
    \end{pmatrix},
\end{equation}
the component $k_{\boldsymbol{a}_n} = \boldsymbol{k}\cdot \hat{\boldsymbol{a}}_n$ of the momentum along the direction $\hat{\boldsymbol{a}}_n = \boldsymbol{a}_n/|\boldsymbol{a}_n|$ of the primitive lattice vector $\boldsymbol{a}_n$, is discretized as
\begin{align}
    \begin{split}
        k_{\boldsymbol{a}_n} &= -i \sum_i (b_i^\dagger b_{i+n} - b_i^\dagger b_{i-n}) = -i \sum_i (b_i^\dagger b_{i+n} - b_{i+n}^\dagger b_{i}).
    \end{split}
\end{align}
Periodic boundary conditions were used when shifting the summation variable in the second term and the indices $i\pm n$ indicate the operators create or annihilate bosons at lattice sites $\boldsymbol{r}_i \pm \boldsymbol{a}_n$. Then,
\begin{align}
    \begin{split}
        k_x = \sum_n k_{\boldsymbol{a}_n} \hat{\boldsymbol{a}}_n \cdot \hat{\boldsymbol{x}} = -i\sum_i\sum_n (b_i^\dagger b_{i+n} - b_{i+n}^\dagger b_{i}) \hat{\boldsymbol{a}}_n \cdot \hat{\boldsymbol{x}}, \\
        k_y = \sum_n k_{\boldsymbol{a}_n} \hat{\boldsymbol{a}}_n \cdot \hat{\boldsymbol{y}} = -i\sum_i\sum_n (b_i^\dagger b_{i+n} - b_{i+n}^\dagger b_{i}) \hat{\boldsymbol{a}}_n \cdot \hat{\boldsymbol{y}}.
    \end{split}
\end{align}
We insert this in \eqref{eq:RashbaSOC} and heuristically move the Pauli matrices inside the operator products to produce scalars. 
\begin{align}
    \begin{split}
        H_{\textrm{SOC}} &= i\lambda_R \sum_{\alpha\beta}\sum_i\sum_n \Big(b_i^{\alpha\dagger}\big(-\sigma_x^{\alpha\beta} (\hat{\boldsymbol{a}}_n \cdot \hat{\boldsymbol{y}}) + \sigma_y^{\alpha\beta} (\hat{\boldsymbol{a}}_n \cdot \hat{\boldsymbol{x}})\big)b_{i+n}^\beta \\
        &\mbox{\qquad\qquad\qquad\qquad} - b_{i+n}^{\alpha\dagger}\big(-\sigma_x^{\alpha\beta} (\hat{\boldsymbol{a}}_n \cdot \hat{\boldsymbol{y}}) + \sigma_y^{\alpha\beta} (\hat{\boldsymbol{a}}_n \cdot \hat{\boldsymbol{x}})\big)b_{i}^\beta \Big) \\
        &=i\lambda_R \sum_{\alpha\beta}\sum_i\sum_n \Big(b_i^{\alpha\dagger}\big(-\sigma_x^{\alpha\beta} (\hat{\boldsymbol{a}}_n \cdot \hat{\boldsymbol{y}}) + \sigma_y^{\alpha\beta} (\hat{\boldsymbol{a}}_n \cdot \hat{\boldsymbol{x}})\big)b_{i+n}^\beta \\
        &\mbox{\qquad\qquad\qquad\qquad} - b_{i+n}^{\beta\dagger}\big(-\sigma_x^{\beta\alpha} (\hat{\boldsymbol{a}}_n \cdot \hat{\boldsymbol{y}}) + \sigma_y^{\beta\alpha} (\hat{\boldsymbol{a}}_n \cdot \hat{\boldsymbol{x}})\big)b_{i}^\beta \Big) \\
        &=i\lambda_R \sum_{\alpha\beta}\sum_i\sum_n \Big(b_i^{\alpha\dagger}\big(-\sigma_x^{\alpha\beta} (\hat{\boldsymbol{a}}_n \cdot \hat{\boldsymbol{y}}) + \sigma_y^{\alpha\beta} (\hat{\boldsymbol{a}}_n \cdot \hat{\boldsymbol{x}})\big)b_{i+n}^\beta \\
        &\mbox{\qquad\qquad\qquad\qquad} - \textrm{H.c.} \Big).
    \end{split}
\end{align}
The summation indices $\alpha$ and $\beta$ were interchanged in the second term, and the Hermiticity of the Pauli matrices allowed for the identification of the second term as the Hermitian conjugate (H.c.) of the first. Next, \eqref{eq:FTbosonicop} together with \eqref{eq:deltasum} is applied to transform to momentum space.
\begin{align}
    \begin{split}
        H_{\textrm{SOC}} &= i\lambda_R \sum_{\alpha\beta}\sum_i\sum_n \Bigg[ \left( \frac{1}{\sqrt{N_s}} \sum_{\boldsymbol{k}}A_{\boldsymbol{k}}^{\alpha\dagger}e^{i\boldsymbol{k}\cdot \boldsymbol{r}_i} \right) \big(-\sigma_x^{\alpha\beta} (\hat{\boldsymbol{a}}_n \cdot \hat{\boldsymbol{y}}) \\
        &\mbox{\qquad} + \sigma_y^{\alpha\beta} (\hat{\boldsymbol{a}}_n \cdot \hat{\boldsymbol{x}})\big) \left(  \frac{1}{\sqrt{N_s}} \sum_{\boldsymbol{k}'}A_{\boldsymbol{k}'}^{\beta}e^{-i\boldsymbol{k}'\cdot (\boldsymbol{r}_i+\boldsymbol{a}_n)} \right) - \textrm{H.c.} \Bigg] \\
        &= i\lambda_R \sum_{\boldsymbol{k}}\sum_{\alpha\beta}\sum_n \Big( A_{\boldsymbol{k}}^{\alpha\dagger} \big(-\sigma_x^{\alpha\beta} (\hat{\boldsymbol{a}}_n \cdot \hat{\boldsymbol{y}}) + \sigma_y^{\alpha\beta} (\hat{\boldsymbol{a}}_n \cdot \hat{\boldsymbol{x}})\big)A_{\boldsymbol{k}}^{\beta} e^{-i \boldsymbol{k} \cdot \boldsymbol{a}_n} \\
        &\mbox{\qquad\qquad\qquad\qquad}-\textrm{H.c.} \Big).
    \end{split}
\end{align}
Performing the sum over pseudospin indices yields
\begin{align}
    \begin{split}
        H_{\textrm{SOC}} &= \lambda_R \sum_{\boldsymbol{k}}\sum_n \Big( A_{\boldsymbol{k}}^{\uparrow\dagger}(-i \hat{\boldsymbol{a}}_n \cdot \hat{\boldsymbol{y}} + \hat{\boldsymbol{a}}_n \cdot \hat{\boldsymbol{x}} ) (e^{-i \boldsymbol{k}\cdot \boldsymbol{a}_n}- e^{i \boldsymbol{k}\cdot \boldsymbol{a}_n}) A_{\boldsymbol{k}}^{\downarrow} \\
        & \mbox{\qquad\qquad\qquad}+ A_{\boldsymbol{k}}^{\downarrow\dagger}(-i \hat{\boldsymbol{a}}_n \cdot \hat{\boldsymbol{y}} - \hat{\boldsymbol{a}}_n \cdot \hat{\boldsymbol{x}} ) (e^{-i \boldsymbol{k}\cdot \boldsymbol{a}_n}- e^{i \boldsymbol{k}\cdot \boldsymbol{a}_n}) A_{\boldsymbol{k}}^{\uparrow} \Big) \\
        &= \sum_{\boldsymbol{k}}\left[ A_{\boldsymbol{k}}^{\uparrow\dagger}\left( -2\lambda_R \sum_n (\hat{\boldsymbol{a}}_n \cdot \hat{\boldsymbol{y}} + i \hat{\boldsymbol{a}}_n \cdot \hat{\boldsymbol{x}} )\sin(\boldsymbol{k} \cdot \boldsymbol{a}_n) \right) A_{\boldsymbol{k}}^{\downarrow} + \textrm{H.c.} \right] \\
        &=   \sum_{\boldsymbol{k}}\big( A_{\boldsymbol{k}}^{\uparrow\dagger}s_{\boldsymbol{k}} A_{\boldsymbol{k}}^{\downarrow} + \textrm{H.c.} \big),
    \end{split}
\end{align}
where we defined the Rashba SOC term
\begin{equation}
    \label{eq:Bravaissk}
    s_{\boldsymbol{k}} = -2\lambda_R \sum_n (\hat{\boldsymbol{a}}_n \cdot \hat{\boldsymbol{y}} + i \hat{\boldsymbol{a}}_n \cdot \hat{\boldsymbol{x}} )\sin(\boldsymbol{k} \cdot \boldsymbol{a}_n) .
\end{equation}
Notice that it is momentum dependent and is involved in spin-flip processes as expected. The full Bose-Hubbard Hamiltonian with SOC is now
\begin{align}
    \begin{split}
    \label{eq:FullH}
        H &= \sum_{\boldsymbol{k}}\sum_\alpha (\epsilon_{\boldsymbol{k}}^\alpha + T^\alpha) A_{\boldsymbol{k}}^{\alpha\dagger}A_{\boldsymbol{k}}^\alpha + \sum_{\boldsymbol{k}}\big( A_{\boldsymbol{k}}^{\uparrow\dagger}s_{\boldsymbol{k}} A_{\boldsymbol{k}}^{\downarrow} + \textrm{H.c.} \big)\\
        &+\frac{1}{2N_s}\sum_{\boldsymbol{k}\boldsymbol{k}'\boldsymbol{p}\boldsymbol{p}'}\sum_{\alpha\beta}U^{\alpha\beta} A_{\boldsymbol{k}}^{\alpha\dagger}A_{\boldsymbol{k}'}^{\beta\dagger}A_{\boldsymbol{p}}^{\beta}A_{\boldsymbol{p}'}^{\alpha} \delta_{\boldsymbol{k}+\boldsymbol{k}', \boldsymbol{p}+\boldsymbol{p}'} \\
        &= \sum_{\boldsymbol{k}}\sum_{\alpha\beta}\eta_{\boldsymbol{k}}^{\alpha\beta}A_{\boldsymbol{k}}^{\alpha\dagger}A_{\boldsymbol{k}}^\beta   +\frac{1}{2N_s}\sum_{\boldsymbol{k}\boldsymbol{k}'\boldsymbol{p}\boldsymbol{p}'}\sum_{\alpha\beta}U^{\alpha\beta} A_{\boldsymbol{k}}^{\alpha\dagger}A_{\boldsymbol{k}'}^{\beta\dagger}A_{\boldsymbol{p}}^{\beta}A_{\boldsymbol{p}'}^{\alpha} \delta_{\boldsymbol{k}+\boldsymbol{k}', \boldsymbol{p}+\boldsymbol{p}'},
    \end{split}
\end{align}
where we introduced the matrix
\begin{gather}
    \eta_{\boldsymbol{k}} = \begin{pmatrix} \epsilon_{\boldsymbol{k}}^{\uparrow} +T^{\uparrow} & s_{\boldsymbol{k}} \\ s_{\boldsymbol{k}}^* &  \epsilon_{\boldsymbol{k}}^{\downarrow} +T^{\downarrow} \end{pmatrix}.
\end{gather}

\section{Superfluidity}

Superfluids are fluids that can flow without dissipating any energy. When Landau \cite{Landau} first provided a theoretical understanding of the superfluidity found experimentally by Kapitza \cite{Kapitza} and Allen and Jones \cite{Allen} for liquid helium at sufficiently low temperature, he proposed one can view the system as a mixture of two fluids. One normal fluid that does experience friction, and one superfluid component that can support frictionless flow. Imagine the sample is placed in a container initially at rest. If one rotates the container, the normal fluid part will follow the walls of the container, while the superfluid part remains stationary \cite{Landau}.

Landau's criterion for superfluidity is derived using Galilean invariance in \cite{pitaevskiisuperfluid}. This will be presented, and the consequence of SOC breaking Galilean invariance will then be discussed afterwards. We consider a fluid inside a cylindrical container that is in motion relative to the container. In the reference frame $K$ where the fluid is at rest we allow for elementary excitations $\Omega(\boldsymbol{k})$ away from the ground state energy $E_0$. The formation of such excitations is the dissipative process under consideration. The reference frame $K'$ in which the container is at rest moves with velocity $-\boldsymbol{v}$ relative to $K$. Performing a Galilean transformation of the total energy, $E = E_0 + \Omega(\boldsymbol{k})$, yields
\begin{equation}
    E' = E_0 + \Omega(\boldsymbol{k}) + \boldsymbol{k} \cdot \boldsymbol{v} + \frac{1}{2}mv^2,
\end{equation}
where $m$ is the total mass of the fluid. It is clear that $\Omega(\boldsymbol{k}) + \boldsymbol{k} \cdot \boldsymbol{v}$ is the change in energy due to the presence of the excitation with momentum $\boldsymbol{k}$. Dissipation occurs if creation of the excitation is energetically favorable, i.e. if
\begin{equation}
    \Omega(\boldsymbol{k}) + \boldsymbol{k} \cdot \boldsymbol{v} < 0.
\end{equation}
This condition becomes $v > \Omega(\boldsymbol{k})/k $, where $v = |\boldsymbol{v}|$ and $k = \abs{\boldsymbol{k}}$. When this is satisfied the fluid will transfer energy to the container, and kinetic energy is lost to heat. The minimal value of such a velocity is
\begin{equation}
\label{eq:Landaucrit}
    v_c = \min_{\boldsymbol{k}} \frac{\Omega(\boldsymbol{k})}{k}.
\end{equation}
This is called the critical superfluid velocity, and the minimum is found by considering all values of $\boldsymbol{k}$. Landau's criterion for superfluidity is
\begin{equation}
    v < v_c,
\end{equation}
and if satisfied, elementary excitations will not lead to a reduction in energy, meaning the fluid can flow without friction and displays superfluid behavior. Superfluidity and BEC are closely related, but not equivalent \cite{pitaevskiisuperfluid}. For instance, an ideal Bose gas in 3D displays BEC below a critical temperature with dispersion $\Omega(\boldsymbol{k}) \sim k^2$, meaning $v_c = 0$ and no superfluidity. Meanwhile, we will see that the excitation spectrum of a weakly interacting Bose gas is linear close to its minimum. For such a phonon spectrum, $\Omega(\boldsymbol{k}) = ck$, the critical superfluid velocity corresponds to the speed of sound, $v_c = c$.

\subsection{Two Kinds of Critical Superfluid Velocity} 
Synthetic SOC introduced to a BEC will break the Galilean invariance of the system. Theoretical consequences are discussed in \cite{zhaiSWRashbaRev, 36, 37}, and experimental observation was made in \cite{HamnerSOCGalileanExp}. The main consequence is that there are two kinds of critical superfluid velocity in our system. In a system with Galilean invariance, the case (a) where a superfluid is flowing through a stationary container and the case (b) where a container is dragged against a stationary superfluid are equivalent. These two cases are connected by a Galilean transformation, and since our system is not Galilean invariant they are no longer equivalent. Thus the critical flowing velocity of case (a) is different from the critical dragging velocity of case (b) \cite{37}. These cases are illustrated in figure \ref{fig:TwoVel}. Also note that case (b) is equivalent to case (c), considering an impurity moving in a superfluid at rest.

\begin{figure}
    \centering
    \includegraphics[width = 0.9\linewidth]{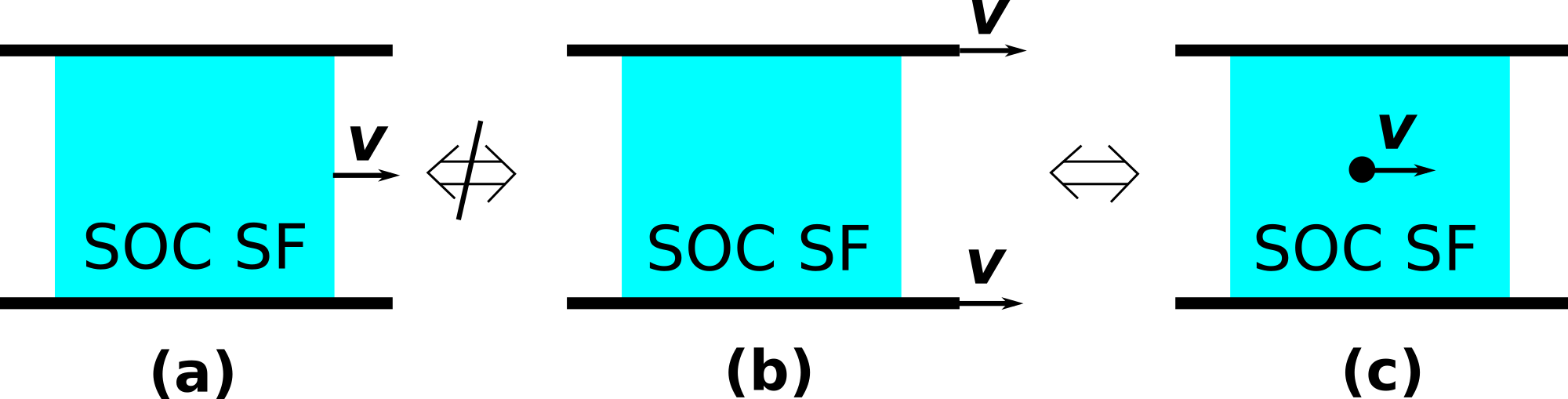}
    \caption{An illustration of possible superfluid (SF) flows in the lab frame. A SOC superfluid moving against a stationary container is shown in (a). Due to lack of Galilean invariance, this is not equivalent to case (b), where the container is dragged, and the SOC superfluid is at rest. Case (b) is however equivalent to case (c) showing an impurity moving through the stationary SOC superfluid. Figure adapted from \cite{37}.}
    \label{fig:TwoVel}
\end{figure}

These two kinds of critical velocities are named $v_\textrm{flow}$ for case (a) and $v_\textrm{drag}$ for case (b). It is argued in \cite{zhaiSWRashbaRev, 36, 37} that because the condensate is at rest in case (b) Landau's criterion is still valid even though it was derived using Galilean invariance \cite{Landau}. Additionally, \cite{37} gives an alternate argument based on conservation of energy and momentum that does not rely on Galilean invariance. Imagine the case of a static SOC superfluid with an impurity. The critical superfluid velocity is a measure of the maximum speed with which the impurity can move without dissipation. Consider an excitation formed in the static superfluid by the moving impurity. Conservation of momentum and energy reads
\begin{align}
    \begin{split}
        m_0 \boldsymbol{v}_i &= m_0 \boldsymbol{v}_f + \boldsymbol{k}, \\
        \frac{m_0 \boldsymbol{v}_i^2}{2} &= \frac{m_0 \boldsymbol{v}_f^2}{2} + \Omega_0 (\boldsymbol{k}).
    \end{split}
\end{align}
Here, $m_0$ is the mass of the impurity, $\boldsymbol{v}_i$ its initial velocity and $\boldsymbol{v}_f$ its velocity after the formation of the excitation $\Omega_0 (\boldsymbol{k})$ with momentum $\boldsymbol{k}$. The subscript indicates that the excitation energy is calculated for a condensate at rest. Once again, the question is if such a formation of an excitation is possible. Inserting the momentum conservation into the energy conservation yields
\begin{equation}
\label{eq:vik2m0}
    v_i = \frac{\Omega_0(\boldsymbol{k})}{k} + \frac{k}{2m_0}.
\end{equation}
The minimal velocity capable of satisfying this is the critical dragging velocity
\begin{equation}
\label{eq:vcdrag}
    v_\textrm{drag} = \min_{\boldsymbol{k}}\frac{\Omega_0 (\boldsymbol{k})}{k},
\end{equation}
which is the same as \eqref{eq:Landaucrit} given that the superfluid is at rest. When $v_i < v_\textrm{drag}$ the formation of an excitation is not energetically favorable, and the impurity moves without loosing energy. 


Without knowing what transformation our system is invariant under, we would have to find the spectrum of a moving condensate directly. An example of such a calculation for a Rashba SOC continuum BEC is found in \cite{37}. Nevertheless, it is noted in \cite{37} that the dragging velocity is much easier to probe experimentally than the flowing velocity. Our approach is also best suited to find the critical dragging velocity, and so we will focus solely on this kind of critical superfluid velocity. Therefore, the critical dragging velocity will from now on be referred to as the critical superfluid velocity, $v_c$.

\textit{Note: As mentioned in the preface to the arXiv version, the following is highly questionable. Equations \eqref{eq:vs} and \eqref{eq:vc} are however valid methods to find the sound velocity of phonon-like excitations.}

We will however study condensates at nonzero momenta as well, in which case the condensate is not at rest. The excitation spectra we find are then for moving condensates. In the case of condensation at zero momentum, we have argued that the critical superfluid velocity corresponds to the slope of an excitation spectrum which is linear close to its minimum. We propose the same is true if the minimum occurs at a nonzero condensate momentum, $\boldsymbol{k}_0$. The important point to remember is that the value obtained is frame dependent, and thus only valid in the lab frame where the optical lattice is at rest. The critical superfluid velocity obtained in such cases will be calculated using \cite{LS, Toniolo} 
\begin{equation}
\label{eq:vs}
    \boldsymbol{v}_c =\left. \pdv{\Omega(\boldsymbol{k})}{\boldsymbol{k}}\right\rvert_{\boldsymbol{k} \to \boldsymbol{k}_0}.
\end{equation}
In isotropic cases, the $x$ and $y$ components will be equal, and we will give the result as a scalar, $v_c$, equal to the components. Alternatively one can use the discretized version 
\begin{equation}
\label{eq:vc}
    v_c = \lim_{\boldsymbol{q}\to \boldsymbol{0}}\frac{\Omega(\boldsymbol{k}_0+\boldsymbol{q})}{\abs{\boldsymbol{q}}},
\end{equation}
assuming $\Omega(\boldsymbol{k}_0) = 0$.

\section{Non-Interacting Spin-Orbit Coupled Bose Gas} \label{sec:SOCU0}
In preparation for treating the weakly interacting, synthetically SOC Bose gas we first investigate its behavior if the interactions are set to zero. The Hamiltonian \eqref{eq:FullH} then reduces to
\begin{equation}
    H = \sum_{\boldsymbol{k}}\sum_{\alpha\beta}\eta_{\boldsymbol{k}}^{\alpha\beta}A_{\boldsymbol{k}}^{\alpha\dagger}A_{\boldsymbol{k}}^\beta.
\end{equation}
Here,
\begin{gather}
    \eta_{\boldsymbol{k}} = \begin{pmatrix} \epsilon_{\boldsymbol{k}}^{\uparrow} +T^{\uparrow} & s_{\boldsymbol{k}} \\ s_{\boldsymbol{k}}^* &  \epsilon_{\boldsymbol{k}}^{\downarrow} +T^{\downarrow} \end{pmatrix},
\end{gather}
where
\begin{equation}
    \label{eq:epsk}
    \epsilon_{\boldsymbol{k}}^\alpha \stackrel{(\ref{eq:generalek})}{=} -2t^{\alpha}\left(\cos(k_x a) + \cos(k_y a)\right),
\end{equation}
and
\begin{equation}
    \label{eq:sk}
    s_{\boldsymbol{k}} \stackrel{(\ref{eq:Bravaissk})}{=} -2\lambda_R \left( \sin(k_y a) + i \sin(k_x a)    \right),
\end{equation}
for a 2D square lattice with lattice constant $a$. Defining the operator vector $\boldsymbol{A}_{\boldsymbol{k}} = (A_{\boldsymbol{k}}^{\uparrow}, A_{\boldsymbol{k}}^{\downarrow})^T$ we can write
\begin{equation}
    H = \sum_{\boldsymbol{k}}\boldsymbol{A}_{\boldsymbol{k}}^\dagger \eta_{\boldsymbol{k}} \boldsymbol{A}_{\boldsymbol{k}}.
\end{equation}
We now attempt to diagonalize the problem using a unitary transformation. One should check that such a transformation is in fact a canonical transformation, i.e. that the new operators one defines are bosonic. Our goal is to find a unitary matrix $P_{\boldsymbol{k}}$ such that 
\begin{equation}
    \boldsymbol{A}_{\boldsymbol{k}}^\dagger \eta_{\boldsymbol{k}} \boldsymbol{A}_{\boldsymbol{k}} = \boldsymbol{A}_{\boldsymbol{k}}^\dagger P_{\boldsymbol{k}} P_{\boldsymbol{k}}^\dagger \eta_{\boldsymbol{k}} P_{\boldsymbol{k}} P_{\boldsymbol{k}}^\dagger \boldsymbol{A}_{\boldsymbol{k}} = \boldsymbol{C}_{\boldsymbol{k}}^\dagger \lambda_{\boldsymbol{k}} \boldsymbol{C}_{\boldsymbol{k}}.
\end{equation}
We defined the new operators $\boldsymbol{C}_{\boldsymbol{k}} = (C_{\boldsymbol{k}}^+, C_{\boldsymbol{k}}^-)^T = P_{\boldsymbol{k}}^\dagger \boldsymbol{A}_{\boldsymbol{k}}$. If the transformation matrix $P_{\boldsymbol{k}}$ contains the eigenvectors of $\eta_{\boldsymbol{k}}$ as its columns, then the matrix $\lambda_{\boldsymbol{k}}$ is diagonal, with the eigenvalues of $\eta_{\boldsymbol{k}}$ on its diagonal,
\begin{equation}
    \lambda_{\boldsymbol{k}} = 
    \begin{pmatrix} 
    \lambda_{\boldsymbol{k}}^+ & 0 \\ 
    0 & \lambda_{\boldsymbol{k}}^- 
    \end{pmatrix}.
\end{equation}
The eigenvalues of $\eta_{\boldsymbol{k}}$ are found to be
\begin{align}
    \begin{split}
        \lambda_{\boldsymbol{k}}^\pm = \frac{1}{2}\Big( &(\epsilon_{\boldsymbol{k}}^{\uparrow} + \epsilon_{\boldsymbol{k}}^{\downarrow}) + (T^\uparrow + T^\downarrow) \\
        & \pm \sqrt{4\abs{s_{\boldsymbol{k}}}^2+\big((\epsilon_{\boldsymbol{k}}^{\uparrow} - \epsilon_{\boldsymbol{k}}^{\downarrow}) - (T^\uparrow - T^\downarrow)\big)^2} \Big).
    \end{split}
\end{align}
At $\boldsymbol{k} = \boldsymbol{0}$ there is a Zeeman splitting
\begin{equation}
    \lambda_{\boldsymbol{0}}^+-\lambda_{\boldsymbol{0}}^- = \abs{(\epsilon_{\boldsymbol{0}}^{\uparrow} - \epsilon_{\boldsymbol{0}}^{\downarrow}) - (T^\uparrow -T^\downarrow)}
\end{equation}
due to differences in hopping parameters $t^\uparrow$ and $t^\downarrow$ and differences in the energy offsets $T^\downarrow$ and $T^\uparrow$. We choose to assume $t^\uparrow= t^\downarrow = t$ and let the energy offsets parametrize the Zeeman splitting. Defining $T = (T^\uparrow + T^\downarrow)/2$ and $\Delta T = T^\uparrow - T^\downarrow$ the energies are
\begin{align}
    \begin{split}
        \lambda_{\boldsymbol{k}}^\pm = \epsilon_{\boldsymbol{k}} + T \pm \sqrt{\abs{s_{\boldsymbol{k}}}^2+ \left(\frac{\Delta T}{2}\right)^2} .
    \end{split}
\end{align}
These are plotted for increasing $\Delta T$ in figure \ref{fig:SOCU0DT}. The minima of $\lambda_{\boldsymbol{k}}^-$ are in general four-fold degenerate, however, as one can see, the minima at nonzero $\boldsymbol{k}$ converge to $\boldsymbol{k} = \boldsymbol{0}$ as the Zeeman splitting $\Delta T$ is increased. These one-fold and four-fold cases are illustrated in figure \ref{fig:1fold4fold} for the 2D square lattice in momentum space.
\begin{figure}
    \centering
    \includegraphics[width = 0.9\linewidth]{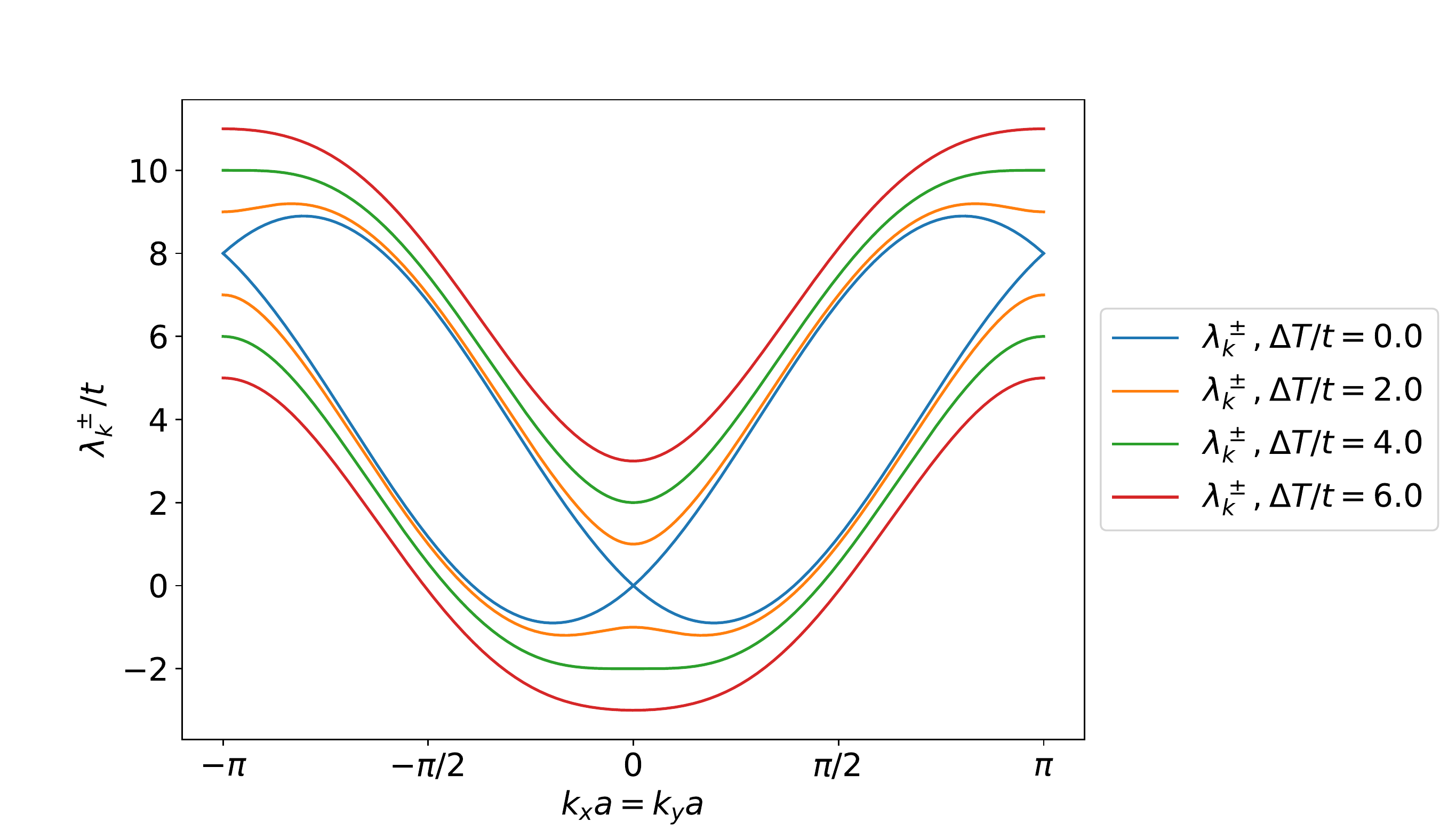}
    \caption{The energies $\lambda_{\boldsymbol{k}}^\pm$ for several $\Delta T$. The parameters in the plot are $\lambda_R/t = 1.0$ and $T/t = 4.0$.}
    \label{fig:SOCU0DT}
\end{figure}


\begin{figure}
    \centering
    \begin{subfigure}{.45\textwidth}
      \includegraphics[width=0.9\linewidth]{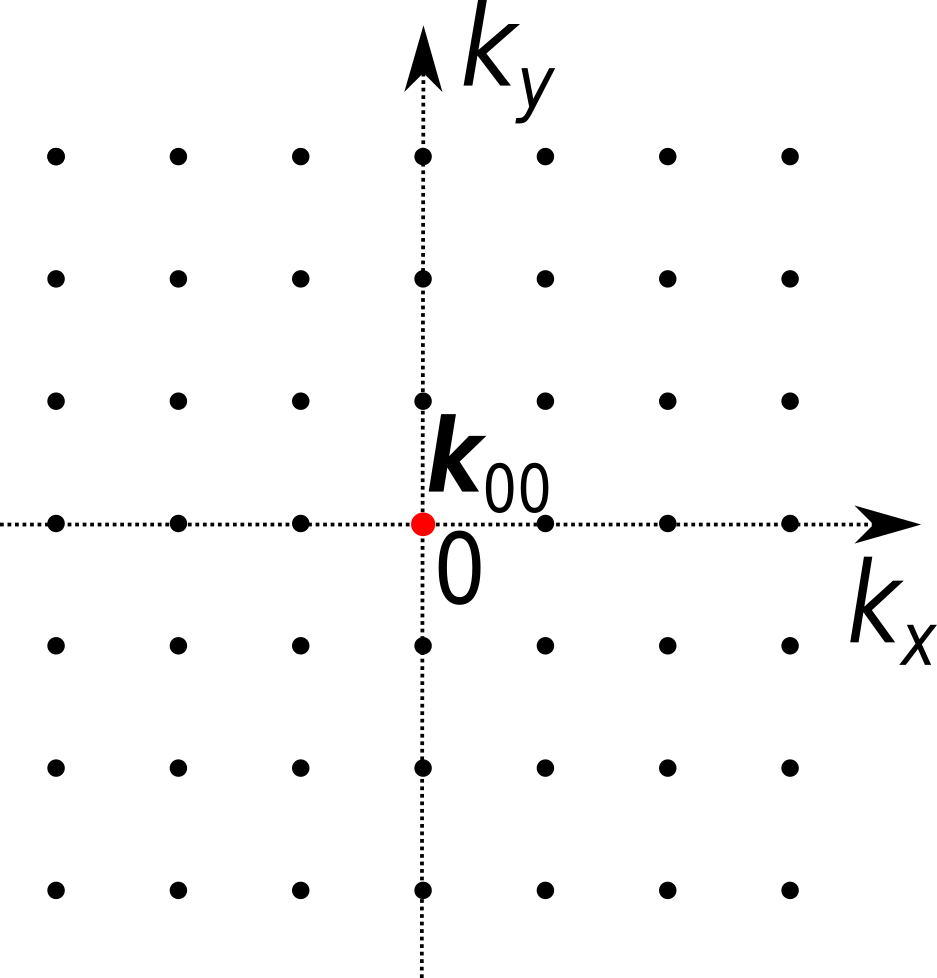}
      \caption{ }
    \end{subfigure}%
    \begin{subfigure}{.45\textwidth}
      \includegraphics[width=0.9\linewidth]{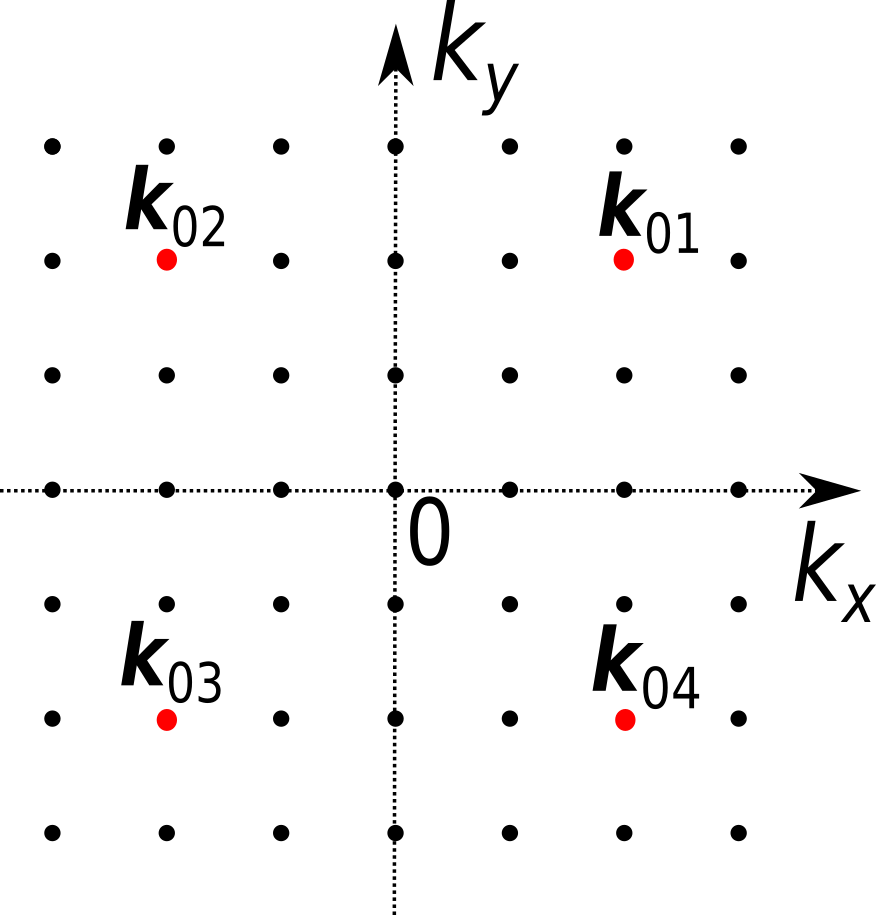}
      \caption{ }
    \end{subfigure}
    \caption{An illustration of a minimum at $\boldsymbol{k} = \boldsymbol{k}_{00} = \boldsymbol{0}$ (a) compared to the SOC induced four-fold degenerate minima $\boldsymbol{k} = \boldsymbol{k}_{0i}$ (b). The black points represent lattice sites, while the red points represent the minima. How far the $\boldsymbol{k}_{0i}$ are placed from zero momentum depends on the Zeeman splitting and the strength of the SOC. Figure adapted from \cite{master}. \label{fig:1fold4fold}}
\end{figure}

From now on, we focus on the case of no Zeeman splitting. Assuming $t^\uparrow= t^\downarrow = t$ and $T^\uparrow= T^\downarrow = T$, the energies reduce to
\begin{equation}
\label{eq:helicitybands}
     \lambda_{\boldsymbol{k}}^\pm = \epsilon_{\boldsymbol{k}} + T \pm \abs{s_{\boldsymbol{k}}}
\end{equation}
\begin{figure}
    \centering
    \includegraphics[width = 0.7\linewidth]{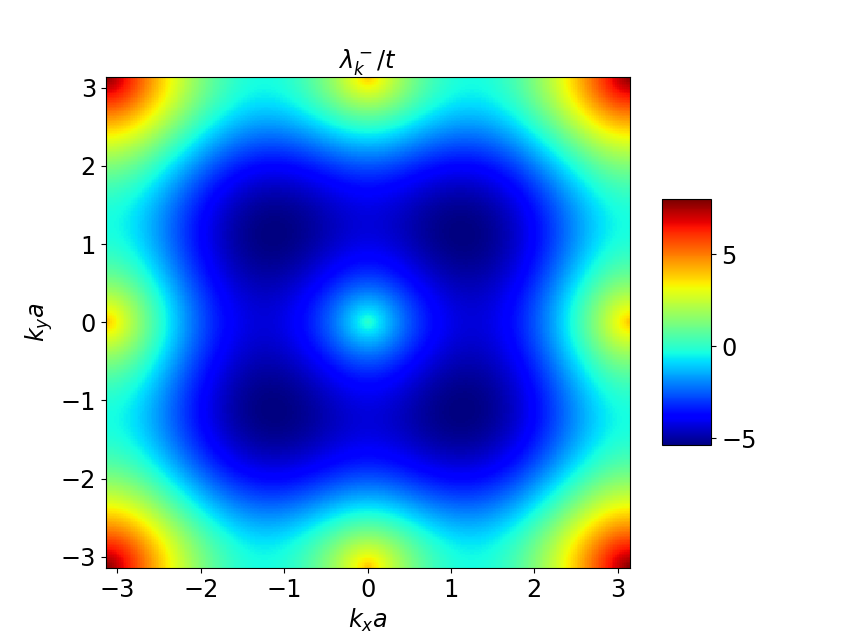}
    \caption{The lowest energy $\lambda_{\boldsymbol{k}}^-$ for $\lambda_R/t = 3.0$ and $T/t = 4.0$.}
    \label{fig:SOCU0eig}
\end{figure}
The lowest eigenvalue $\lambda_{\boldsymbol{k}}^-$ is plotted in the first Brillouin zone (1BZ) in figure \ref{fig:SOCU0eig}. Its minima occur at the four points $\boldsymbol{k}_{01} = (k_0, k_0), \boldsymbol{k}_{02} = (-k_0, k_0), \boldsymbol{k}_{03} = (-k_0, -k_0)$ and $\boldsymbol{k}_{04} = (k_0, -k_0)$ with
\begin{equation}
    k_0 a = k_{0m} a \equiv \arctan(\frac{\lambda_R}{\sqrt{2}t}).
\end{equation}
Hence, with no Zeeman splitting any nonzero $\lambda_R$ will lead to minima at nonzero $\boldsymbol{k}$. The minimal value of $\lambda_{\boldsymbol{k}}^-$ is
\begin{equation}
    \lambda_0 = T - 4t\sqrt{\frac{\lambda_R^2}{2t^2}+1}
\end{equation}
Whether or not this is negative is a matter of the choice of value for $T$. In figures \ref{fig:SOCU0DT} and \ref{fig:SOCU0eig} the value for $T$ was chosen such that $\lambda_{\boldsymbol{0}}^\pm = 0$ when $\Delta T = 0$ and hence $\lambda_0 < 0$. If one wishes to avoid negative energies, one can e.g. tune $T$ such that $\lambda_0 = 0$. The final expression for $H$ is 
\begin{equation}
    H =  \sum_{\boldsymbol{k}}\sum_{\sigma = \pm}\lambda_{\boldsymbol{k}}^{\sigma}C_{\boldsymbol{k}}^{\sigma\dagger}C_{\boldsymbol{k}}^\sigma.
\end{equation}
Provided $s_{\boldsymbol{k}} \neq 0$, the eigenvectors of $\eta_{\boldsymbol{k}}$ are
\begin{equation}
    \boldsymbol{\chi}^\pm = \frac{1}{\sqrt{2}}\begin{pmatrix} \pm \frac{s_{\boldsymbol{k}}}{\abs{s_{\boldsymbol{k}}}} \\ 1 \end{pmatrix}.
\end{equation}
If we define $s_{\boldsymbol{k}} \equiv \abs{s_{\boldsymbol{k}}}e^{-i\gamma_{\boldsymbol{k}}}$ this is
\begin{equation}
\label{eq:SOCU0eigenvec}
    \boldsymbol{\chi}^\pm = \frac{1}{\sqrt{2}}\begin{pmatrix} \pm e^{-i\gamma_{\boldsymbol{k}}} \\ 1 \end{pmatrix}.
\end{equation}
Hence the definitions of the new operators are,
\begin{equation}
\label{eq:Helicitybasis}
    \begin{pmatrix} C_{\boldsymbol{k}}^+ \\ C_{\boldsymbol{k}}^- \end{pmatrix} = P_{\boldsymbol{k}}^\dagger  \begin{pmatrix} A_{\boldsymbol{k}}^\uparrow \\ A_{\boldsymbol{k}}^\downarrow \end{pmatrix} = \frac{1}{\sqrt{2}}\begin{pmatrix} A_{\boldsymbol{k}}^\downarrow + e^{i\gamma_{\boldsymbol{k}}}  A_{\boldsymbol{k}}^\uparrow \\ A_{\boldsymbol{k}}^\downarrow - e^{i\gamma_{\boldsymbol{k}}}  A_{\boldsymbol{k}}^\uparrow\end{pmatrix}.
\end{equation}
With $\sigma, \rho = \pm$ we find that
\begin{align}
    \begin{split}
        [ C_{\boldsymbol{k}}^\sigma,  C_{\boldsymbol{k}'}^{\rho\dagger}] &= \frac{1}{2}\left[A_{\boldsymbol{k}}^\downarrow + \sigma e^{i\gamma_{\boldsymbol{k}}}  A_{\boldsymbol{k}}^\uparrow, (A_{\boldsymbol{k}'}^\downarrow + \rho e^{i\gamma_{\boldsymbol{k}'}}  A_{\boldsymbol{k}'}^\uparrow)^\dagger\right]\\
        &= \frac{1}{2}\left([A_{\boldsymbol{k}}^\downarrow, A_{\boldsymbol{k}'}^{\downarrow\dagger}] + \sigma\rho e^{i(\gamma_{\boldsymbol{k}}-\gamma_{\boldsymbol{k}'})}[A_{\boldsymbol{k}}^\uparrow, A_{\boldsymbol{k}'}^{\uparrow\dagger}]\right) = \delta_{\boldsymbol{k}\boldsymbol{k}'}\delta^{\sigma\rho}.
    \end{split}
\end{align}
As required, the new operators are bosonic. By inversion, the old operators in terms of the new are
\begin{equation}
    \begin{pmatrix} A_{\boldsymbol{k}}^\uparrow \\ A_{\boldsymbol{k}}^\downarrow \end{pmatrix} = \frac{1}{\sqrt{2}} \begin{pmatrix} e^{-i\gamma_{\boldsymbol{k}}} \left( C_{\boldsymbol{k}}^+ - C_{\boldsymbol{k}}^- \right)\\ C_{\boldsymbol{k}}^+ + C_{\boldsymbol{k}}^-  \end{pmatrix}.
\end{equation}

Finally, we may compare the eigenvectors \eqref{eq:SOCU0eigenvec} to the general helicity eigenvectors \cite{UnconventionalStates}
\begin{equation}
    \boldsymbol{\xi}^+ = \begin{pmatrix} e^{-i\phi}\cos(\frac{\theta}{2}) \\ \sin(\frac{\theta}{2}) \end{pmatrix} \mbox{\qquad and \qquad} \boldsymbol{\xi}^- = \begin{pmatrix} -e^{-i\phi}\sin(\frac{\theta}{2}) \\ \cos(\frac{\theta}{2}) \end{pmatrix}.
\end{equation}
This leads to the identifications $\theta = \pi/2$ and $\phi = \gamma_{\boldsymbol{k}}$. The former fits well with the fact that our synthetic SOC for a pseudospin-$1/2$ system models the SOC induced in a spin-$1/2$ system constrained to the $xy$-plane by an electric field along the $z$-axis. The latter identification requires some care. $\gamma_{\boldsymbol{k}}$ is defined by $s_{\boldsymbol{k}} \equiv \abs{s_{\boldsymbol{k}}}e^{-i\gamma_{\boldsymbol{k}}}$. Consulting \eqref{eq:sk} it becomes clear that $\gamma_{\boldsymbol{k}}$ can not be identified with the azimuth angle $\boldsymbol{k}$ makes with the $k_x$-axis. In fact, such an interpretation can only make sense if $k_x=0$, $k_y=0$ or $k_x=\pm k_y$ and in those cases $\gamma_{\boldsymbol{k}}$ is the angle $\boldsymbol{k}$ makes with the negative $k_y$-axis. This imperfect correspondence between $\phi$ and $\gamma_{\boldsymbol{k}}$ leads us to define the eigenvectors \eqref{eq:SOCU0eigenvec} as pseudohelicity eigenvectors, and the basis \eqref{eq:Helicitybasis} as a pseudohelicity basis. We will however refer to \eqref{eq:Helicitybasis} as a helicity basis. We note for posterity that
\begin{equation}
\label{eq:GSgammas}
    \gamma_{\boldsymbol{k}_{01}} = \frac{3\pi}{4}, \mbox{\qquad} \gamma_{\boldsymbol{k}_{02}} = -\frac{3\pi}{4}, \mbox{\qquad} \gamma_{\boldsymbol{k}_{03}} = -\frac{\pi}{4} \mbox{\qquad and \qquad} \gamma_{\boldsymbol{k}_{04}} = \frac{\pi}{4}.
\end{equation}


\section{Weakly Interacting Dilute Bose Gas} \label{sec:Weaklyinteracting}
As a further precursor to treating a two-component, SOC, weakly interacting BEC, we study the one-component, weakly interacting, dilute Bose gas. In the process we will review the Bogoliubov transformation and discover that the presence of interactions makes the dispersion relation linear close to the minimum. We will follow the treatments in \cite{PethickSmith, Pitaevskii, abrikosov} with the exception that we will treat a Bose gas bound to a 2D square Bravais lattice. The Hamiltonian is 
\begin{equation}
    H = \sum_{\boldsymbol{k}} (\epsilon_{\boldsymbol{k}} +T) A_{\boldsymbol{k}}^\dagger A_{\boldsymbol{k}} + \frac{U}{2N_s}\sum_{\boldsymbol{k}\boldsymbol{k'}\boldsymbol{p}\boldsymbol{p'}} A_{\boldsymbol{k}}^\dagger A_{\boldsymbol{k'}}^\dagger A_{\boldsymbol{p}} A_{\boldsymbol{p'}} \delta_{\boldsymbol{k}+\boldsymbol{k'},\boldsymbol{p}+\boldsymbol{p'}},
\end{equation}
where
\begin{equation}
    \label{epskpre}
    \epsilon_{\boldsymbol{k}} = -2t\left(\cos(k_x a) + \cos(k_y a)\right),
\end{equation}
which has a quadratic minimum at $\boldsymbol{k} = \boldsymbol{0}$. When including weak interactions, our aim is to find new bosonic quasiparticle operators defined as linear combinations of the original operators. In terms of these quasiparticle operators the Hamiltonian will be diagonal, and the coefficient of the number operators is the quasiparticle energy spectrum we are interested in.

As the interactions are weak, we expect the quasiparticle energy spectrum will also have its minimum at $\boldsymbol{0}$. We also assume the temperature is low enough that BEC occurs, such that the occupation of the states with $\boldsymbol{k} = \boldsymbol{0}$ is macroscopic. The number of particles in the condensate is denoted $N_0$ while the total number of particles in the system is denoted $N$. We assume that $(N-N_0)/N \ll 1$ and the Bogoliubov approach then suggests replacing the condensate operators $A_{\boldsymbol{0}}$ and $A_{\boldsymbol{0}}^\dagger$ by $\sqrt{N_0}$ since the mean value of the number operator $A_{\boldsymbol{0}}^\dagger A_{\boldsymbol{0}}$ is $N_0$. 

In this thesis we will however include a complex phase such that $A_{\boldsymbol{0}}$ is replaced by $\sqrt{N_0}e^{-i\theta_0}$. Such an approach will prove to be significant when SOC is included in the problem. The angle $\theta_0$ is at this point an arbitrary variational parameter. Variational parameters can be determined by minimization of the free energy in case the free energy depends on them, as discussed in chapter 4 of \cite{bruusflensberg}. If not, they are arbitrary, in the sense that any choice gives the same free energy, and hence the same physics. 

The excitations represent small perturbations from a pure condensate, and so we may neglect terms that are more than quadratic in excitation operators. One may then write the Hamiltonian as $H= H'_0 + H'_2$ with
\begin{equation}
    H'_0 = (\epsilon_{\boldsymbol{0}} + T) A_{\boldsymbol{0}}^\dagger A_{\boldsymbol{0}} + \frac{U}{2N_s} A_{\boldsymbol{0}}^\dagger A_{\boldsymbol{0}}^\dagger A_{\boldsymbol{0}} A_{\boldsymbol{0}}
\end{equation}
and
\begin{align}
    \begin{split}
        H'_2 &= \sum_{\boldsymbol{k}\neq\boldsymbol{0}} (\epsilon_{\boldsymbol{k}} + T) A_{\boldsymbol{k}}^\dagger A_{\boldsymbol{k}} \\
        &+ \frac{U}{2N_s}\sum_{{\boldsymbol{k}}\neq {\boldsymbol{0}}} \left(A_{\boldsymbol{0}}^\dagger A_{\boldsymbol{0}}^\dagger A_{\boldsymbol{k}} A_{-\boldsymbol{k}} + 4A_{\boldsymbol{0}}^\dagger A_{\boldsymbol{k}}^\dagger A_{\boldsymbol{k}} A_{\boldsymbol{0}} + A_{\boldsymbol{k}}^\dagger A_{-\boldsymbol{k}}^\dagger A_{\boldsymbol{0}} A_{\boldsymbol{0}} \right).
    \end{split}
\end{align}
We now make the replacement
\begin{equation}
    A_{\boldsymbol{0}} \rightarrow \sqrt{N_0}e^{-i\theta_0}.
\end{equation}
Additionally, following \cite{Pitaevskii},
\begin{equation}
    N_0 = N - \sum_{{\boldsymbol{k}} \neq {\boldsymbol{0}}} A_{\boldsymbol{k}}^\dagger A_{\boldsymbol{k}}
\end{equation}
is used to replace $N_0$ by $N$ in the Hamiltonian. In $H'_2$ we may replace $N_0$ by $N$ directly to the same order of approximation as done in \cite{Pitaevskii}. From $H'_0$ this gives
\begin{equation}
    -(\epsilon_{\boldsymbol{0}}+T)\sum_{{\boldsymbol{k}}\neq {\boldsymbol{0}}} A_{\boldsymbol{k}}^\dagger A_{\boldsymbol{k}} -\frac{UN}{N_s}\sum_{{\boldsymbol{k}} \neq {\boldsymbol{0}}} A_{\boldsymbol{k}}^\dagger A_{\boldsymbol{k}}
\end{equation}
which we move into the new quadratic part $H_2$. Then $H = H_0 + H_2$ with
\begin{equation}
    H_0 = (\epsilon_{\boldsymbol{0}} + T) N + \frac{UN^2}{2N_s}
\end{equation}
and
\begin{align}
    \begin{split}
        H_2 =  &\sum_{{\boldsymbol{k}} \neq {\boldsymbol{0}}} \left(\mathcal{E}_{\boldsymbol{k}} + \frac{UN}{N_s}\right) A_{\boldsymbol{k}}^\dagger A_{\boldsymbol{k}} \\
        &+ \frac{UN}{2N_s}\sum_{{\boldsymbol{k}} \neq {\boldsymbol{0}}} \left(e^{i2\theta_0}A_{\boldsymbol{k}} A_{-\boldsymbol{k}} + e^{-i2\theta_0}A_{\boldsymbol{k}}^\dagger A_{-\boldsymbol{k}}^\dagger \right).
    \end{split}
\end{align}
Here, we defined
\begin{equation}
    \mathcal{E}_{\boldsymbol{k}} \equiv \epsilon_{\boldsymbol{k}} - \epsilon_0 = 4t -2t\left(\cos(k_x a) + \cos(k_y a)\right).
\end{equation}
In order to diagonalize the problem we attempt a Bogoliubov transformation. We postulate that the new quasiparticle operators are given by
\begin{align}
    \begin{split}
        B_{\boldsymbol{k}} &= u_{\boldsymbol{k}}^* A_{\boldsymbol{k}}+v_{\boldsymbol{k}}A_{-\boldsymbol{k}}^\dagger, \\
        B_{-\boldsymbol{k}}^\dagger &= v_{\boldsymbol{k}}^*A_{\boldsymbol{k}}+u_{\boldsymbol{k}}A_{-\boldsymbol{k}}^\dagger.
    \end{split}
\end{align}
In order for the transformation to be canonical we must have $[B_{\boldsymbol{k}}, B_{\boldsymbol{k}}^\dagger] = 1$. This requirement reduces to $|u_{\boldsymbol{k}}|^2- |v_{\boldsymbol{k}}|^2 = 1$. We use this to identify
\begin{align}
    \begin{split}
    \label{eq:Bogoliubovtran}
        A_{\boldsymbol{k}} &= u_{\boldsymbol{k}}B_{\boldsymbol{k}}-v_{\boldsymbol{k}}B_{-\boldsymbol{k}}^\dagger, \\
        A_{-\boldsymbol{k}}^\dagger &= -v_{\boldsymbol{k}}^*B_{\boldsymbol{k}}+u_{\boldsymbol{k}}^*B_{-\boldsymbol{k}}^\dagger.
    \end{split}
\end{align}
Using that $\mathcal{E}_{-\boldsymbol{k}} = \mathcal{E}_{\boldsymbol{k}}$ we rewrite $H_2$ to
\begin{align}
    \begin{split}
        H_2 = \frac{1}{2}\sum_{{\boldsymbol{k}} \neq {\boldsymbol{0}}}\bigg[ &\left(\mathcal{E}_{\boldsymbol{k}} + \frac{UN}{N_s}\right) \left(A_{\boldsymbol{k}}^\dagger A_{\boldsymbol{k}} +A_{-\boldsymbol{k}} A_{-\boldsymbol{k}}^\dagger  \right) \\
        &+ \frac{UN}{N_s}\left(e^{i2\theta_0}A_{\boldsymbol{k}} A_{-\boldsymbol{k}} + e^{-i2\theta_0}A_{\boldsymbol{k}}^\dagger A_{-\boldsymbol{k}}^\dagger \right)\bigg],
    \end{split}
\end{align}
simultaneously shifting $H_0$ by $$-\frac{1}{2}\sum_{{\boldsymbol{k}} \neq {\boldsymbol{0}}}\left(\mathcal{E}_{\boldsymbol{k}} + \frac{UN}{N_s}\right)$$ because a commutation relation was used. Inserting \eqref{eq:Bogoliubovtran} yields
\begin{equation}
    H_2 = \frac{1}{2}\sum_{{\boldsymbol{k}} \neq {\boldsymbol{0}}}\bigg[ \omega_{\boldsymbol{k}}\left(B_{\boldsymbol{k}}^\dagger B_{\boldsymbol{k}} + B_{-\boldsymbol{k}} B_{-\boldsymbol{k}}^\dagger\right) + a_{\boldsymbol{k}}B_{\boldsymbol{k}} B_{-\boldsymbol{k}} + a_{\boldsymbol{k}}^*B_{\boldsymbol{k}}^\dagger B_{-\boldsymbol{k}}^\dagger \bigg],
\end{equation}
where
\begin{align}
    \begin{split}
        \omega_{\boldsymbol{k}} &= \left(\abs{u_{\boldsymbol{k}}}^2+\abs{v_{\boldsymbol{k}}}^2\right)\left(\mathcal{E}_{\boldsymbol{k}} + \frac{UN}{N_s}\right) -\left(u_{\boldsymbol{k}}v_{\boldsymbol{k}}e^{i2\theta_0} + u_{\boldsymbol{k}}^* v_{\boldsymbol{k}}^* e^{-i2\theta_0}\right)\frac{UN}{N_s},\\
        a_{\boldsymbol{k}} &= \left(u_{\boldsymbol{k}}^2e^{i2\theta_0} +  (v_{\boldsymbol{k}}^*)^2 e^{-i2\theta_0}\right)\frac{UN}{N_s}-2u_{\boldsymbol{k}}v_{\boldsymbol{k}}^*\left(\mathcal{E}_{\boldsymbol{k}} + \frac{UN}{N_s}\right).
    \end{split}
\end{align}
Insisting that the Hamiltonian is diagonal in terms of the quasiparticle operators we must have $a_{\boldsymbol{k}} = 0$. Upon choosing $u_{\boldsymbol{k}} = |u_{\boldsymbol{k}}|\exp(-i\theta_0)$ and $v_{\boldsymbol{k}} = |v_{\boldsymbol{k}}|\exp(-i\theta_0)$ the equations are the same as in \cite{PethickSmith}. The solution, using that $|u_{\boldsymbol{k}}|^2- |v_{\boldsymbol{k}}|^2 = 1$ is
\begin{align}
        \abs{u_{\boldsymbol{k}}}^2 &= \abs{v_{\boldsymbol{k}}}^2+1 = \frac{1}{2}\left(\frac{\mathcal{E}_{\boldsymbol{k}} + \frac{UN}{N_s}}{\omega_{\boldsymbol{k}}}+1\right), \label{eq:Boguv} \\
        \omega_{\boldsymbol{k}} &= \sqrt{\mathcal{E}_{\boldsymbol{k}}\left(\mathcal{E}_{\boldsymbol{k}} + 2\frac{UN}{N_s}\right)}. \label{eq:noSOCom}
\end{align}
Using that $\omega_{-\boldsymbol{k}} = \omega_{\boldsymbol{k}}$, the Hamiltonian may now be written
\begin{align}
    \begin{split}
        H =& (\epsilon_{\boldsymbol{0}} + T) N + \frac{UN^2}{2N_s} -\frac{1}{2}\sum_{{\boldsymbol{k}} \neq {\boldsymbol{0}}}\left(\mathcal{E}_{\boldsymbol{k}} + \frac{UN}{N_s}\right) \\
        &+ \sum_{{\boldsymbol{k}} \neq {\boldsymbol{0}}}\omega_{\boldsymbol{k}}\left(B_{\boldsymbol{k}}^\dagger B_{\boldsymbol{k}} +\frac12\right).
    \end{split}
\end{align}
The quasiparticle energy spectrum, $\omega_{\boldsymbol{k}}$, is linear for small $|\boldsymbol{k}|$ since $\mathcal{E}_{\boldsymbol{k}}$ is zero at $\boldsymbol{k}=0$ and quadratic for small $|\boldsymbol{k}|$. This represents new physics due to the interactions. The critical superfluid velocity has become nonzero, and to be specific it is $v_c = \sqrt{2UNta^2/N_s}$.


Without interactions we would find that all particles are in the condensate at zero temperature. Let us investigate the ground state depletion in the presence of interactions. We have
\begin{equation}
    N = N_0 + \sum_{\boldsymbol{k}\neq \boldsymbol{0}} \langle A_{\boldsymbol{k}}^\dagger A_{\boldsymbol{k}} \rangle.
\end{equation}
To obtain the mean value, we transform to the diagonal basis and get
\begin{align}
    \begin{split}
        N &= N_0 + \sum_{\boldsymbol{k}\neq \boldsymbol{0}} \bigg( \abs{u_{\boldsymbol{k}}}^2 \langle  B_{\boldsymbol{k}}^\dagger B_{\boldsymbol{k}} \rangle + \abs{v_{\boldsymbol{k}}}^2 \langle  B_{-\boldsymbol{k}} B_{-\boldsymbol{k}}^\dagger \rangle \\
        & \mbox{\qquad\qquad\qquad\qquad} -u_{\boldsymbol{k}}^* v_{\boldsymbol{k}} \langle  B_{\boldsymbol{k}}^\dagger B_{-\boldsymbol{k}}^\dagger \rangle - u_{\boldsymbol{k}} v_{\boldsymbol{k}}^* \langle  B_{\boldsymbol{k}} B_{-\boldsymbol{k}} \rangle \bigg)\\
        &= N_0 + \sum_{\boldsymbol{k}\neq \boldsymbol{0}} \bigg( (\abs{u_{\boldsymbol{k}}}^2+\abs{v_{\boldsymbol{k}}}^2) \langle  B_{\boldsymbol{k}}^\dagger B_{\boldsymbol{k}} \rangle + \abs{v_{\boldsymbol{k}}}^2 \bigg),
    \end{split}
\end{align}
where we used a commutator along with the fact that $v_{-\boldsymbol{k}} = v_{\boldsymbol{k}}$. Because the Hamiltonian is diagonal in terms of the quasiparticle operators, the quasiparticles behave like an ideal Bose gas \cite{Pitaevskii}. Therefore the mean values of the off-diagonal terms are zero. Furthermore, the mean value of $B_{\boldsymbol{k}}^\dagger B_{\boldsymbol{k}}$ follows Bose-Einstein statistics. Thus,
\begin{equation}
    N = N_0 + \sum_{\boldsymbol{k}\neq \boldsymbol{0}} \bigg( \frac{\abs{u_{\boldsymbol{k}}}^2+\abs{v_{\boldsymbol{k}}}^2}{e^{\beta\omega_{\boldsymbol{k}}}-1} + \abs{v_{\boldsymbol{k}}}^2 \bigg),
\end{equation}
where $\beta = 1/k_B T'$, $k_B$ is Boltzmann's constant and $T'$ is the temperature. At zero temperature, we are left with
\begin{equation}
\label{eq:GSdepnoSoc}
    N = N_0 + \sum_{\boldsymbol{k}\neq \boldsymbol{0}}\abs{v_{\boldsymbol{k}}}^2.
\end{equation}
Hence, there is a depletion of the ground state even at zero temperature. For our initial assumption that the depletion is small to hold, we see that we must require $U\ll t$ such that $|v_{\boldsymbol{k}}|^2$ given in \eqref{eq:Boguv} is small. This is what is meant by weakly interacting Bose gas in the context of a Bravais lattice. On the other hand, for $U \gg t$ the ground state depletion is severe, and the system is expected to be in the Mott insulator phase for such strong interactions \cite{Oosten}. This thesis is concerned with the superfluid phase, where $U\ll t$.

\subsection{Free Energy}
To determine the variational parameter $\theta_0$ we must calculate the free energy. We first derive a general procedure for finding the free energy based on the calculation in \cite{master}, and then apply it to the weakly interacting Bose gas. The Hamiltonian is assumed to be on the form
\begin{equation}
    H = H'_0 + \left.\sum_{\boldsymbol{k}}\right.^{'}\sum_\sigma \Omega_\sigma(\boldsymbol{k}) \left(  B_{\boldsymbol{k}, \sigma}^{\dagger}B_{\boldsymbol{k}, \sigma} + \frac{1}{2} \right), 
\end{equation}
where the sum $\left.\sum_{\boldsymbol{k}}\right.^{'}$ excludes any condensate momenta and the sum over $\sigma$ takes into account the possibility of several branches in the excitation spectrum $\Omega_\sigma(\boldsymbol{k})$. Assume $\ket{\Tilde{N}_m} = \prod_{i=1}^m \ket{N_i}$, where $N_i = N_{\boldsymbol{k},\sigma} = B_{\boldsymbol{k}, \sigma}^{\dagger}B_{\boldsymbol{k},\sigma}$, is a many-particle Fock basis. Then, the partition function is
\begin{align}
    \begin{split}
        Z &= \Tr(e^{-\beta H}) = \sum_m \bra{\Tilde{N}_m} e^{-\beta H} \ket{\Tilde{N}_m} \\
        &=e^{-\beta H'_0}e^{-\frac{\beta}{2} \left.\sum_{\boldsymbol{k}}\right.^{'}\sum_\sigma \Omega_\sigma(\boldsymbol{k})}\\
        &\mbox{\qquad} \cdot\sum_{m} \bra{\Tilde{N}_m} e^{-\beta \left.\sum_{\boldsymbol{k}}\right.^{'}\sum_\sigma \Omega_\sigma(\boldsymbol{k}) N_{\boldsymbol{k},\sigma} } \ket{\Tilde{N}_m} \\
        &=e^{-\beta H'_0}e^{-\frac{\beta}{2} \left.\sum_{\boldsymbol{k}}\right.^{'}\sum_\sigma \Omega_\sigma(\boldsymbol{k})} \left.\prod_{\boldsymbol{k}, \sigma}\right.^{'}\sum_{N_{\boldsymbol{k},\sigma} = 0}^{\infty} e^{-\beta \Omega_\sigma(\boldsymbol{k}) N_{\boldsymbol{k},\sigma} }  \\
        &=e^{-\beta H'_0}e^{-\frac{\beta}{2} \left.\sum_{\boldsymbol{k}}\right.^{'}\sum_\sigma \Omega_\sigma(\boldsymbol{k})} \left.\prod_{\boldsymbol{k}, \sigma}\right.^{'} \frac{1}{1-e^{-\beta \Omega_\sigma(\boldsymbol{k})} }.
    \end{split}
\end{align}
The computation of the $N_{\boldsymbol{k},\sigma}$ sum requires $\Omega_\sigma(\boldsymbol{k})>0$ which is assumed to be true when $\boldsymbol{k}$ is not a condensate momentum. 
Using $F=-\ln(Z)/\beta$ for the free energy, we get
\begin{align}
    \begin{split}
    F &= H'_0 +\frac{1}{2}\left.\sum_{\boldsymbol{k}}\right.^{'}\sum_\sigma \Omega_\sigma(\boldsymbol{k}) + \frac{1}{\beta} \left.\sum_{\boldsymbol{k}}\right.^{'}\sum_\sigma \ln\Big(1-\exp\big(-\beta\Omega_\sigma(\boldsymbol{k})\big)\Big).
\end{split}
\end{align}
We will focus on the effects of the elementary excitations due to interactions and SOC rather than thermal effects. Therefore we set the temperature to zero, or $\beta \to \infty$. Then, $F = \langle H \rangle$, which is the ground state energy. Thus, we finally get
\begin{align}
\begin{split}
\label{eq:F}
    F \stackrel{\beta\to\infty}{=} \langle H \rangle = H'_0 +\frac{1}{2}\left.\sum_{\boldsymbol{k}}\right.^{'}\sum_\sigma \Omega_\sigma(\boldsymbol{k}).
\end{split}
\end{align}

For the weakly interacting Bose gas we find 
\begin{equation}
    F = (\epsilon_{\boldsymbol{0}} + T) N + \frac{UN^2}{2N_s} -\frac{1}{2}\sum_{{\boldsymbol{k}} \neq {\boldsymbol{0}}}\left(\mathcal{E}_{\boldsymbol{k}} + \frac{UN}{N_s} - \omega_{\boldsymbol{k}}\right).
\end{equation}
As this is independent of $\theta_0$, the angle is arbitrary and may be set to $0$ as is usually done a priori in the literature \cite{abrikosov, PethickSmith,Pitaevskii}.

\section{Generalized Diagonalization Theory} \label{sec:BV} 
In the previous sections we have seen two examples of canonical transformations used to diagonalize Hamiltonians that are quadratic in bosonic operators. For the SOC, non-interacting Bose gas we could use a unitary transformation, while for the one-component, weakly interacting Bose gas we used a Bogoliubov transformation of a two-component basis. When the size of the basis becomes larger, it is convenient to introduce a matrix generalization of the Bogoliubov transformation. This section is concerned with the theory of the resulting Bogoliubov-Valatin transformation that will be used extensively in the remainder of the thesis. Due to this extensive use, the method will be presented in great detail based on papers by Tsallis \cite{Tsallis}, Xiao \cite{Xiao} and van Hemmen \cite{Hemmen}.

The most general Hamiltonian which is quadratic in bosonic operators is \cite{Tsallis}
\begin{equation}
    H = \sum_{i=1}^n \sum_{j=1}^n \left((M_1)_{ij}A_i^\dagger A_j + + (M_1)_{ij}^* A_i A_j^\dagger + (M_2)_{ij}A_i^\dagger A_j^\dagger + (M_2)_{ij}^* A_i A_j\right),
\end{equation}
where $A_i^\dagger$ and $A_i$ are bosonic creation and annihilation operators, satisfying $[A_i, A_j] = 0, [A_i^\dagger, A_j^\dagger] = 0$ and $[A_i, A_j^\dagger] = \delta_{ij}$. The $n \cross n$ matrices $M_1$ and $M_2$ must be Hermitian and symmetric respectively \cite{Tsallis}. We now seek to rewrite this Hamiltonian in matrix notation and define operator vectors
\begin{align}
\begin{split}
    \boldsymbol{A} &= (A_1, \dots, A_n, A_1^\dagger, \dots, A_n^\dagger)^T \mbox{\qquad and}\\
    \boldsymbol{A}^{\dagger} &= (A_1^\dagger, \dots, A_n^\dagger, A_1, \dots, A_n).
\end{split}
\end{align} 
Given that $A_i^\dagger$ and $A_i$ are bosonic creation and annihilation operators, $\boldsymbol{A}$ and $\boldsymbol{A}^{\dagger}$ satisfy the commutation relation $\boldsymbol{A}\otimes \boldsymbol{A}^{\dagger} -((\boldsymbol{A}^{\dagger})^T\otimes(\boldsymbol{A})^T)^T=J$, where we defined a matrix $J$ by
\begin{equation}
    J = \begin{pmatrix} I & 0 \\ 0&-I \end{pmatrix}.
\end{equation}
The matrix $J$ is its own inverse, i.e. $J^2 = I$. In terms of components the commutation relation is
\begin{equation}
    \boldsymbol{A}_i \boldsymbol{A}^{\dagger}_j - \boldsymbol{A}^{\dagger}_j \boldsymbol{A}_i = J_{ij} = \begin{cases} \delta_{ij} \mbox{\qquad if \qquad} i\leq n \\ -\delta_{ij} \mbox{\qquad if \qquad} i > n \end{cases}
\end{equation}
We can now write the Hamiltonian as 
\begin{equation}
    H = \boldsymbol{A}^{\dagger}M\boldsymbol{A},
\end{equation}
where $M$ is a $2n\cross 2n$ Hermitian matrix on the form
\begin{equation}
    M = 
    \begin{pmatrix}
    M_1 & M_2 \\
    M_2^* & M_1^* \\
    \end{pmatrix},
\end{equation}
where $M_1^\dagger = M_1$ and $M_2^T = M_2$, such that $M^\dagger = M$. $\boldsymbol{A}^{\dagger}M\boldsymbol{A} = \sum_{ij}\boldsymbol{A}^{\dagger}_i M_{ij} \boldsymbol{A}_j$ tells us that $M_{ij}$ is the coefficient in front of $\boldsymbol{A}^{\dagger}_i\boldsymbol{A}_j$ in the Hamiltonian.

\subsection{The Bogoliubov-Valatin Transformation} 
Whenever we attempt to diagonalize a Hamiltonian, we simultaneously define new operators $\boldsymbol{B}^{\dagger}$ and $\boldsymbol{B}$. When the original operators are bosonic we also want the new operators to be bosonic, and we have two requirements we need to fulfill. The transformation matrix must satisfy $T^{-1} = JT^\dagger J$, which will be shown later in theorem \ref{thm:non-unitary}, and we also want $\boldsymbol{B}$ to be the Hermitian conjugate of $\boldsymbol{B}^{\dagger}$. If we define the new operators as $\boldsymbol{B}^{\dagger} = \boldsymbol{A}^{\dagger}T$, we want $\boldsymbol{B} = T^\dagger \boldsymbol{A}$ which means we require $\boldsymbol{B} = JT^{-1}J \boldsymbol{A}$. With this choice, the diagonalization procedure is \cite{Tsallis} 
\begin{align}
    \begin{split}
        \boldsymbol{A}^{\dagger}M\boldsymbol{A} &= \boldsymbol{A}^{\dagger}(TT^{-1})M(J(T(JJ)T^{-1})J)\boldsymbol{A} \\
        &= (\boldsymbol{A}^{\dagger}T)(T^{-1}MJTJ)(JT^{-1}J\boldsymbol{A}) = \boldsymbol{B}^{\dagger}D\boldsymbol{B}.
    \end{split}
\end{align}
We will call such a transformation a Bogoliubov-Valatin (BV) transformation motivated by \cite{Xiao}, and define it more clearly later. The method is also known as the dynamic matrix method because $MJ$ is closely related to the dynamic matrix $JM$ in the Heisenberg equation of motion \cite{Xiao}. 

Notice that it is actually $MJ$ we are diagonalizing, and thus we should look for the eigenvalues, $\lambda$, of $MJ$ using $\det(MJ-\lambda I) = 0$. These eigenvalues go on the diagonal of a matrix $DJ$, which we then have to multiply from the right by $J$ to get the matrix $D$ in the Hamiltonian, $H = \boldsymbol{B}^{\dagger}D\boldsymbol{B}$. The new operators $\boldsymbol{B}$ describe bosonic quasiparticles that behave essentially like uncoupled harmonic oscillators. These quasiparticles describe collective excitations in the system, analogously to the way phonons describe collective vibrations of the atoms in a lattice.

\subsection{Complex Eigenvalues and Dynamical Instabilities}
$M$ is by definition Hermitian, $M^\dagger = M$, and so $MJ$ is not Hermitian, $(MJ)^\dagger = JM$, unless $M_2=0$. Hence, $MJ$ can in general have complex eigenvalues. 
There are different definitions in the literature for the transformation procedure we are using. While we follow Tsallis \cite{Tsallis} and diagonalize $MJ$, Xiao \cite{Xiao} and others define the transformation in an alternate way such that $JM$ is the matrix being diagonalized. This should all amount to a change of eigenvectors but not of eigenvalues, something which can be proven. 
If $\lambda$ is an eigenvalue of $MJ$ and $\boldsymbol{x}$ its corresponding eigenvector, we have that $MJ\boldsymbol{x} = \lambda \boldsymbol{x}$. Multiplying from the left by $J$ we get
\begin{equation}
    MJ\boldsymbol{x} = \lambda \boldsymbol{x} \iff JMJ\boldsymbol{x} = J \lambda \boldsymbol{x} \iff JM(J\boldsymbol{x}) = \lambda (J\boldsymbol{x}),
\end{equation}
showing that $\lambda$ is also an eigenvalue of $JM$. In conclusion, $JM$ and $MJ$ have the same set of eigenvalues, while their eigenvectors are related by a multiplication by $J$. 
It can even be shown that the new operators are defined equivalently.


Complex eigenvalues of $JM$, or equivalently of $MJ$,  are defined as dynamical instabilities by Pethick and Smith in chapter 14.3 of \cite{PethickSmith}. This is because it is proved in \cite{WuNiuStability} that if $\omega \in \mathbb{C}$ is an eigenvalue of $JM$ then $\omega^*$ is also an eigenvalue of $JM$, i.e. complex eigenvalues come in conjugate pairs. As the time-dependence of states are related to the eigenvalues of $JM$ by $\exp(-i\omega t)$ \cite{PethickSmith}, a complex eigenvalue of $JM$ will always mean there is an unstable mode, in the sense that small perturbations grow exponentially in time \cite{PethickSmith}. In conclusion, complex eigenvalues of $MJ$ at some parameters are equivalent to the system described by the Hamiltonian being dynamically unstable at those parameters. Furthermore, if the eigenvalues of $MJ$ are complex it is not possible to diagonalize $MJ$ in a way that defines new bosonic quasiparticles. Meanwhile, if $MJ$ has real eigenvalues and is diagonalizable, we will always be able to set up a transformation matrix $T$ such that the new operators are bosonic. To prove this, we need to prove some other properties as well. In the cases where the proofs offer little new insight, the reader is referred to the proofs in \cite{Xiao}. 


\subsection{Existence of the Bogoliubov-Valatin Transformation}

\begin{theorem}
\label{thm:non-unitary}
    Assume the original operators satisfy the bosonic commutation relation $\boldsymbol{A}\otimes \boldsymbol{A}^{\dagger} -((\boldsymbol{A}^{\dagger})^T\otimes(\boldsymbol{A})^T)^T=J$. For the new operators $\boldsymbol{B}^{\dagger} = \boldsymbol{A}^{\dagger}T$ and $\boldsymbol{B} = T^\dagger \boldsymbol{A}$ to satisfy the same commutation relation $\boldsymbol{B}\otimes \boldsymbol{B}^{\dagger} -((\boldsymbol{B}^{\dagger})^T\otimes(\boldsymbol{B})^T)^T=J$, we get the requirement $T^{-1} = JT^\dagger J$. 
\end{theorem}
\begin{proof}
If we define $\boldsymbol{B} = T^\dagger \boldsymbol{A}$ and $\boldsymbol{B}^{\dagger} = \boldsymbol{A}^{\dagger}T$, we get $(\boldsymbol{B}^{\dagger})^T = T^T (\boldsymbol{A}^{\dagger})^T $ and $(\boldsymbol{B})^T = (\boldsymbol{A})^T T^*$. Hence,
\begin{gather*}
    \boldsymbol{B}\otimes \boldsymbol{B}^{\dagger} -((\boldsymbol{B}^{\dagger})^T\otimes(\boldsymbol{B})^T)^T = J\\
    T^\dagger \boldsymbol{A}\otimes \boldsymbol{A}^{\dagger}T - \big(T^T(\boldsymbol{A}^{\dagger})^T\otimes(\boldsymbol{A})^T T^*\big)^T = J\\
    T^\dagger \boldsymbol{A}\otimes \boldsymbol{A}^{\dagger}T - T^\dagger \big((\boldsymbol{A}^{\dagger})^T\otimes(\boldsymbol{A})^T)\big)^T T = J\\
    T^\dagger \left(\boldsymbol{A}\otimes \boldsymbol{A}^{\dagger} -\big((\boldsymbol{A}^{\dagger})^T\otimes(\boldsymbol{A})^T\big)^T\right) T = J\\
    T^\dagger J T = J \iff (J T^\dagger J) T = I =  T^{-1}T \iff T^{-1} = JT^\dagger J.
\end{gather*}
\end{proof}

\begin{theorem}
\label{thm:Tform}
    The transformation matrix $T$ takes the form \cite{Xiao}
    \begin{equation}
        T = 
    \begin{pmatrix}
    T_1 & T_2 \\
   T_2^* & T_1^* \\
    \end{pmatrix}.
    \end{equation}
\end{theorem}
If $T_2 = 0$ we get $J T^\dagger J = T^\dagger$, and the BV transformation becomes a unitary transformation. Tsallis states that $T_2 = 0 \iff M_2=0$ \cite{Tsallis}, which explains why a unitary transformation was sufficient when considering the non-interacting, SOC Bose gas, and why it would fail in the case of the weakly interacting Bose gas.  

Let us define
\begin{equation}
    \Sigma_x = \begin{pmatrix} 0 & I \\ I&0    \end{pmatrix},
\end{equation}
and note that $\Sigma_x^2 = I$.
We notice that $((\Sigma_x\boldsymbol{A})^T)^\dagger = \boldsymbol{A}$, which is because $\boldsymbol{A}_{i+n} = (\boldsymbol{A}_i)^\dagger$, i.e. $A_i^\dagger$ is the Hermitian conjugate of $A_i$.  One may ask if the transformation $\boldsymbol{B} = T^\dagger \boldsymbol{A}$ preserves this. The following theorem proves this.
\begin{theorem}
    \label{thm:formket}
    If $((\Sigma_x\boldsymbol{A})^T)^\dagger = \boldsymbol{A}$ and $T^\dagger = JT^{-1}J$, then $\boldsymbol{B} = T^\dagger \boldsymbol{A}$ satisfies $((\Sigma_x\boldsymbol{B})^T)^\dagger = \boldsymbol{B}$, suggesting the last $n$ elements of $\boldsymbol{B}$ are the Hermitian conjugates of the first $n$ elements \cite{Xiao}. 
\end{theorem}
Hence, $\boldsymbol{B} = (B_1, \dots, B_n, B_1^\dagger, \dots, B_n^\dagger)^T$ when $\boldsymbol{A} = (A_1, \dots, A_n, A_1^\dagger, \dots, A_n^\dagger)^T$. It can also be shown that the eigenvalues, when real, are equally distributed around $0$.
\begin{theorem}
\label{thm:dist}
    Real eigenvalues of $MJ$ are equally distributed around $0$.  
\end{theorem}
\begin{proof}
The proof involves introducing an operator $K$ such that \cite{Xiao,Hemmen}
\begin{equation}
    K \begin{pmatrix} u \\ v \end{pmatrix} = \Sigma_x \begin{pmatrix} u \\ v \end{pmatrix}^* = \begin{pmatrix} v^* \\ u^* \end{pmatrix},
\end{equation}
where $u$ and $v$ represent column vectors of length $n$.
It is easy to show that $\{J, K\} = 0$ and $[M, K] = 0$ \cite{Hemmen}. Thus, if $MJ\boldsymbol{x} = \lambda \boldsymbol{x}$,
\begin{equation}
    MJK\boldsymbol{x} = -KMJ\boldsymbol{x} = -K\lambda\boldsymbol{x} =  -\lambda^* K\boldsymbol{x},
\end{equation}
which shows that if $\boldsymbol{x}$ is an eigenvector of $MJ$ with eigenvalue $\lambda$, then $K\boldsymbol{x}$ is an eigenvector with eigenvalue $-\lambda^*$. In particular, when the eigenvalues of $MJ$ are real, we have that if $\boldsymbol{x}$ is an eigenvector of $MJ$ with eigenvalue $\lambda$, then $K\boldsymbol{x}$ is an eigenvector with eigenvalue $-\lambda$. Hence, when $\pm\omega_i \in \mathbb{R}$, $i=1,\dots,n$ are the eigenvalues of $MJ$, $D$ can be written $D = \textrm{diag}(\omega_1, \dots, \omega_n, \omega_1, \dots, \omega_n)$.
\end{proof}

We can now define what we mean by $M$ being what Xiao \cite{Xiao} calls Bogoliubov-Valatinianly (BV) diagonalizable: There exists a matrix $T$ on the form
\begin{equation}
\label{eq:Tform}
    T = 
    \begin{pmatrix}
    T_1 & T_2 \\
    T_2^* & T_1^* \\
    \end{pmatrix},
\end{equation}
with the property $T^{-1} = JT^\dagger J$, such that $$\boldsymbol{A}^{\dagger}M\boldsymbol{A} = (\boldsymbol{A}^{\dagger}T)(T^{-1}MJTJ)(JT^{-1}J\boldsymbol{A}) = \boldsymbol{B}^{\dagger}D\boldsymbol{B},$$ where $D$ is diagonal with real entries. Here we defined $\boldsymbol{B} = T^\dagger \boldsymbol{A}$ and $\boldsymbol{B}^{\dagger} = \boldsymbol{A}^{\dagger}T$. These will satisfy the commutation relation $\boldsymbol{B}\otimes \boldsymbol{B}^{\dagger} -((\boldsymbol{B}^{\dagger})^T\otimes(\boldsymbol{B})^T)^T=J$ by theorem \ref{thm:non-unitary}, and thus consist of bosonic operators $B_i$ and $B_i^\dagger$. Using theorem \ref{thm:formket}, theorem \ref{thm:dist} and commutators, the diagonalized Hamiltonian can be written
\begin{equation}
\label{eq:diagtheorydiagH}
    H = \boldsymbol{B}^{\dagger}D\boldsymbol{B} = 2\sum_{i=1}^{n} \omega_i \left(B_i^\dagger B_i +\frac{1}{2}\right).
\end{equation}
We see that real entries in $D$ are required such that the Hamiltonian remains Hermitian. We are now ready to prove the main result. This is the same as Theorem 29 in \cite{Xiao}.
\begin{theorem}
    The fact that $MJ$ is digonalizable and has real eigenvalues is equivalent to the fact that the BV diagonalization procedure we have defined for $M$ exists.
\end{theorem}
\begin{proof}
Assume $MJ$ is diagonalizable, and the eigenvalues are real. Then there exists a matrix $T$ with the property that $T^{-1}(MJ)T = DJ$, where $DJ$ is diagonal. I.e. the matrix $T$ is invertible, which is equivalent to its columns being linearly independent. Its columns are the eigenvectors of $MJ$, and so $MJ$ being diagonalizable is equivalent to saying that $MJ$ has $2n$ linearly independent eigenvectors. We will discuss further in chapter \ref{sec:setupT} why this, together with $MJ$ having real eigenvalues is enough to ensure that we can construct a matrix $T$ with the property $JT^\dagger J = T^{-1}$ that simultaneously obeys $T^{-1}(MJ)T = DJ$, where $DJ$ is diagonal. Hence, the new operators defined during the diagonalization are bosonic. By theorem \ref{thm:dist} the eigenvalues, when real, can be written $\pm\omega_i$, with $i=1, \dots, n$. Thus, $DJ = \textrm{diag}(\omega_1, \dots, \omega_n, -\omega_1, \dots, -\omega_n)$ and $D = \textrm{diag}(\omega_1, \dots, \omega_n, \omega_1, \dots, \omega_n)$. Hence, we can write the Hamiltonian as in \eqref{eq:diagtheorydiagH}.
Because we assume $\omega_i \in \mathbb{R}$, this Hamiltonian is diagonal and Hermitian, and thus we conclude that $M$ can be BV diagonalized.

To prove equivalence we must also show the opposite implication. Assume $M$ can be BV diagonalized, i.e. that there exists a matrix $T$ such that $T^{-1}(MJ)TJ = D$, where $D$ is diagonal. By multiplying from the right by $J$ we obtain $T^{-1}(MJ)T = DJ$, where, by the definition of $J$, $DJ$ is diagonal if $D$ is diagonal. This proves that if $M$ is BV diagonalizable, $MJ$ is diagonalizable. By the definition of BV diagonalization, $D$ has real entries. Thus, $DJ$ has real entries. As these will be the eigenvalues of $MJ$, it is clear that the eigenvalues of $MJ$ are real. For a more rigorous proof, see \cite{Xiao}.
\end{proof}

\subsection{Setting Up the Transformation Matrix} \label{sec:setupT}
For $T$ to be invertible its $2n$ columns must be linearly independent, i.e. $\textrm{Rank}(T) = 2n$, meaning that $MJ$ has to have $2n$ linearly independent eigenvectors. Theorem \ref{thm:non-unitary} might lead one to believe that $T$ satisfies $JT^\dagger J = T^{-1}$ automatically. This is not true, it is in fact a requirement for the diagonalization procedure to describe the system in terms of bosonic quasiparticles. Therefore, we have to be careful in setting up $T$, such that $JT^\dagger J = T^{-1}$, or equivalently $T^\dagger J T = J$, is satisfied. Naming the eigenvectors $\boldsymbol{x}_i$ the requirement $T^\dagger J T = J$ can be written
\begin{equation}
    \boldsymbol{x}_i^\dagger J \boldsymbol{x}_j = J_{ij} = \begin{cases} \delta_{ij} \textrm{\quad if\quad} i\leq n \\
    -\delta_{ij} \textrm{\quad if\quad} i>n \end{cases}.
\end{equation}
We name this requirement BV orthonormalization, and we have to choose our eigenvectors such that they satisfy this. We notice that the BV norm of an eigenvector $\boldsymbol{x}$ in principle can be $\boldsymbol{x}^\dagger J \boldsymbol{x} = 0$. If so, we will not be able to construct $T$. This is the case for complex eigenvalues. For an eigenvalue $\lambda$ we have $MJ\boldsymbol{x} = \lambda \boldsymbol{x}$. Multiplying from the right by $J$ and then by $\boldsymbol{x}^\dagger$ we get
\begin{equation}
    \boldsymbol{x}^\dagger JMJ \boldsymbol{x} = \lambda \boldsymbol{x}^\dagger J \boldsymbol{x}.
\end{equation}
Both sides of this equation are at first glance complex numbers. However, the left hand side must be real because it is its own Hermitian conjugate due to the Hermiticity of $M$. Hence, when $\lambda \in \mathbb{C}$ we get $\boldsymbol{x}^\dagger JMJ \boldsymbol{x} = \boldsymbol{x}^\dagger J \boldsymbol{x} = 0$. I.e. complex eigenvalues have BV norm zero eigenvectors, and these can not be used to construct a matrix $T$ that satisfies $JT^\dagger J = T^{-1}$. Assuming $MJ$ is diagonalizable with real eigenvalues, lemmas 22 and 23 in \cite{Xiao} proves that a BV orthonormalized set of eigenvectors exists for the eigenspace corresponding to any $\omega_i$. 

Assuming real eigenvalues, we can prove that eigenvectors corresponding to different eigenvalues are BV orthogonal:
\begin{equation}
    \omega_i \boldsymbol{x}_i^\dagger J \boldsymbol{x}_j = (MJ\boldsymbol{x}_i)^\dagger J \boldsymbol{x}_j = \boldsymbol{x}_i^\dagger JMJ \boldsymbol{x}_j = \omega_j \boldsymbol{x}_i^\dagger J \boldsymbol{x}_j.
\end{equation}
Thus, when $i\neq j$ and $\omega_i \neq \omega_j$ we must have $\boldsymbol{x}_i^\dagger J \boldsymbol{x}_j = 0$. The problem if we had complex eigenvalues is that we have to replace $\omega_i$ by $\omega_i^*$ on the left hand side. And $\omega_i \neq \omega_j$ does not exclude $\omega_i^* = \omega_j$, meaning that for complex eigenvalues there is no guarantee that different eigenvalues can have BV orthogonal eigenvectors. From now on we assume the eigenvalues are real.

When we are forced to work numerically, the eigenvectors provided by the numerical routine for degenerate eigenvalues are not in general BV orthogonal. If we have several equal eigenvalues, we can use the given set of eigenvectors to BV orthonormalize the eigenspace corresponding to these eigenvalues. This can be accomplished by a BV modified Gram-Schmidt process (BVMGS), and the resulting BV orthonormalized vectors will still be eigenvectors corresponding to the original eigenvalue. 

A modified Gram-Schmidt (MGS) process suited for numerics is explained in \cite{GramS}. If one has a set of vectors $\boldsymbol{v}_i$ to be orthonormalized, one can use the following process. Let $\boldsymbol{u}_1 = \boldsymbol{v}_1$. Then, for $k>1$
\begin{align}
    \begin{split}
        \boldsymbol{u}_k^{(1)} &= \boldsymbol{v}_k - \textrm{proj}_{\boldsymbol{u}_1} \boldsymbol{v}_k, \\
        \boldsymbol{u}_k^{(i)} &= \boldsymbol{u}_k^{(i-1)} - \textrm{proj}_{\boldsymbol{u}_{i}} \boldsymbol{u}_k^{(i-1)}, \mbox{\qquad for \qquad} i = 2, \dots, k-1, \\
        \boldsymbol{u}_k &= \frac{\boldsymbol{u}_k^{(k-1)}}{\abs{\boldsymbol{u}_k^{(k-1)}}}.
    \end{split}
\end{align}
BVMGS has two main differences from MGS. First, we replace the inner product by the definition $\langle \boldsymbol{u}, \boldsymbol{v} \rangle = \boldsymbol{u}^\dagger J \boldsymbol{v}$, and make sure the order in these products are such that the new vectors are in fact BV orthogonal. Thus, we change the definition of the projection operator to
\begin{equation}
    \textrm{proj}_{\boldsymbol{u}} \boldsymbol{v} = \boldsymbol{u}\frac{\langle \boldsymbol{u}, \boldsymbol{v}\rangle}{\langle \boldsymbol{u}, \boldsymbol{u}\rangle} = \boldsymbol{u}\frac{\boldsymbol{u}^\dagger J \boldsymbol{v}}{\boldsymbol{u}^\dagger J \boldsymbol{u}}.
\end{equation}
As an example, let us say we have two vectors $\boldsymbol{x}_1$ and $\boldsymbol{x}_2$. Then $\boldsymbol{u}_1 = \boldsymbol{x}_1$ and $\boldsymbol{u}_2 = \boldsymbol{x}_2-\boldsymbol{u}_1 (\boldsymbol{u}_1^\dagger J \boldsymbol{x}_2)/\boldsymbol{u}_1^\dagger J \boldsymbol{u}_1$. As we can see $$\boldsymbol{u}_1^\dagger J \boldsymbol{u}_2 = \boldsymbol{u}_1^\dagger J \boldsymbol{x}_2 - \boldsymbol{u}_1^\dagger J \boldsymbol{x}_2 = 0,$$ meaning the two new vectors are BV orthogonal. In the end it is just a matter of BV normalizing the set by the rule $\boldsymbol{e}_i = \boldsymbol{u}_i/\sqrt{|\boldsymbol{u}_i^\dagger J \boldsymbol{u}_i|}$. 

The second change we make, is that we find two vectors at a time instead of one vector at a time. Let us say the eigenvalue $\lambda$ has multiplicity $m$. Then we use the modified Gram-Schmidt process on the $2m$ eigenvectors provided for $\lambda$ and $-\lambda$. The reason we include the eigenvectors for $-\lambda$ as well, is that theorem \ref{thm:dist} tells us there is a close relationship between the eigenvectors of $\lambda$ and $-\lambda$. We choose one of these $2m$ eigenvectors that has a nonzero BV norm as our start, $\boldsymbol{u}_1$. Then, we also include the vector resulting from applying the operator $K$ ($K \boldsymbol{u} = \Sigma_x \boldsymbol{u}^*$) on the first vector, $K\boldsymbol{u_1}$. Note that by the definition of the operator $K$, $\boldsymbol{y}^\dagger J (K\boldsymbol{y}) = 0$ for any $\boldsymbol{y}$ of length $2n$, i.e. $K\boldsymbol{y}$ is BV orthogonal to $\boldsymbol{y}$. 

Next, we find a new vector BV orthogonal on the first two, $\boldsymbol{u}_2$, make sure that the vector resulting from applying the operator $K$ to this vector, $K\boldsymbol{u}_2$, is also BV orthogonal to the first two, and then include both of these. This is continued, until we have a set of $2m$ new BV orthonormalized vectors. Finally, the $m$ vectors with BV norm $1$ are put in the left half of $T$. Once we have constructed the left half of $T$, it is a simple matter to fill in the right half, as we know that the form of $T$ is \eqref{eq:Tform}.
The same method can also be used in the case that $\lambda = 0$ with multiplicity $2m$, one simply thinks of the first $m$ occurrences of 0 as $\lambda$ and the last $m$ occurrences of 0 as $-\lambda$. 

Note that for $\lambda>0$ there is no guarantee that the BV norm $1$ vectors will correspond to $\lambda$ and the BV norm $-1$ vectors correspond to $-\lambda$. The important part for the diagonalization procedure is that the eigenvectors with BV norm $1$ are put in the left half of $T$, which automatically puts the eigenvectors with BV norm $-1$ in the right half. The consequence of this, is that the diagonalized matrix $D$ may contain some eigenvalues with a negative sign. See e.g. example 30 in \cite{Xiao}.


\subsection{Summary of Diagonalization Theory}
In the context of diagonalizing Hamiltonians that are quadratic in bosonic operators, we have defined BV diagonalization of a matrix $M$ and shown that it is equivalent to the matrix $MJ$ being diagonalizable with real eigenvalues. This means that if we can show that $MJ$ has real eigenvalues and $2n$ linearly independent eigenvectors then $M$ is BV diagonalizable. We have also made some rules one should follow in setting up the transformation matrix $T$. Additionally, we discussed how complex eigenvalues of $MJ$ are related to instabilities in the system described by the Hamiltonian. 

\cleardoublepage

%% file: chapter3.tex

\chapter{Mean Field Theory and Phases} \label{chap:MFT}

\section{Mean Field Theory} \label{sec:MFTH}
In \cite{master}, Janssønn developed a framework to describe a SOC, weakly interacting BEC in a Bravais lattice by employing mean field theory (MFT) to reduce the Hamiltonian to a form that was at most quadratic in excitation operators. The operator independent part of the Hamiltonian was then used to study a pure condensate in a square lattice, wherein the most interesting phases of the system were identified, and a phase diagram was presented. The objective of this thesis is to obtain the excitation spectrum, critical superfluid velocity and free energy in these phases. Finally, constructing a phase diagram based on the free energy will be interesting, in order to investigate if the effects of the excitations change the conclusions in \cite{master}. Janssønn used a MFT approach based on van Oosten et. al. \cite{Oosten}. This describes the system as a grand canonical ensemble where the chemical potential, $\mu$, determines the number of particles in the condensate. Condensate operators are replaced by their mean values plus a fluctuation, and terms linear in fluctuations are used to determine the chemical potential \cite{Oosten}.

Due to difficulties encountered regarding the BV diagonalization at the condensate momenta in the many-fold cases, we instead employ the method used by Bogoliubov \cite{bogoliubov1947theory} to the continuum dilute Bose gas. This is also presented by Pethick and Smith \cite{PethickSmith}, Pitaevskii and Stringari \cite{Pitaevskii} and Abrikosov, Gorkov and Dzyaloshinski \cite{abrikosov}. The same MFT approach was applied by Linder and Sudbø \cite{LS} and Toniolo and Linder \cite{Toniolo} in the presence of an optical lattice. Additionally, this was the method we followed in chapter \ref{sec:Weaklyinteracting} when treating the weakly interacting Bose gas in a square lattice. The fluctuations are set to zero by assumption, and hence the condensate operators are replaced by their mean value only \cite{bogoliubov1947theory, PethickSmith, Pitaevskii, abrikosov, LS, Toniolo}. This is what is usually called the Bogoliubov approach and is argued to be valid in 3D for $na_s^3 \ll 1$ in \cite{Bogoliubovvalid}, where $n$ is the total number of atoms per volume. This is the same as the requirement of diluteness, and will soon be discussed in conjunction with our 2D system.


We will use 
\begin{equation}
    \label{eq:NN0excitup}
    N^\uparrow = N_0^\uparrow + \left.\sum_{\boldsymbol{k}}\right.^{'} A_{\boldsymbol{k}}^{\uparrow\dagger} A_{\boldsymbol{k}}^{\uparrow}
\end{equation}
and
\begin{equation}
    \label{eq:NN0excitdown}
    N^\downarrow = N_0^\downarrow + \left.\sum_{\boldsymbol{k}}\right.^{'} A_{\boldsymbol{k}}^{\downarrow\dagger} A_{\boldsymbol{k}}^{\downarrow}
\end{equation}
to replace the number of particles with pseudospin $\alpha$ in the condensate, $N_0^\alpha$, by $N^\alpha$, the total number of particles with pseudospin $\alpha$ in the system. The sums $\left.\sum_{\boldsymbol{k}}\right.^{'}$ exclude any condensate momenta. We fix $N^\uparrow$, $N^\downarrow$ and hence also fix $N = N^\uparrow + N^\downarrow$, the total number of particles in the system. Thus, we consider the system as a canonical ensemble. The chemical potentials, $\mu^\alpha$, are removed from the description, it is now the total number of particles of each pseudospin type, $N^\alpha$, that are interpreted as the input parameters. The sum of \eqref{eq:NN0excitup} and \eqref{eq:NN0excitdown} must also be true,
\begin{equation}
\label{eq:NN0excit}
    N = N_0 + \left.\sum_{\boldsymbol{k}}\right.^{'}\sum_\alpha A_{\boldsymbol{k}}^{\alpha\dagger} A_{\boldsymbol{k}}^{\alpha}.
\end{equation}
In the cases where $N_0^\uparrow = 0$ or $N_0^\downarrow = 0$ this will be the most relevant equation.

In typical experiments the total number of particles, $N$, is often set equal to the number of lattice sites $N_s$. Both \cite{PethickSmith} and \cite{Pitaevskii} mention in their chapters concerning Bose gases in optical lattices that the filling $N/N_s$ is of order unity in 3D. In an experiment in a 2D optical lattice studying the Mott insulator phase $N/N_s=1$ was used \cite{expBEC2dfillingSpielman}. An experiment in a 3D optical lattice studying the superfluid to Mott insulator transition also used $N/N_s=1$ \cite{expBEC3dfillingKetterle}. In the same experiments, the typical lattice size is $N_s = (1-3) \cdot 10^{5}$. Furthermore, $N/N_s=1$ seems to be a typical assumption in several theoretical papers \cite{Toniolo, Oosten}. 

We assumed that only two-body scatterings are relevant when constructing our Hamiltonian. The condition for this to be valid is that the the Bose gas is sufficiently dilute. It is important that the number of atoms in an interaction volume is small. In 3D with $n=N/V$, where $V$ is the volume of the system, one requires $na_s^3 \ll 1$ \cite{Pitaevskii}, where $a_s$ is the s-wave scattering length. As mentioned, this is the same requirement that is used to check the validity of the Bogoliubov approach  \cite{Bogoliubovvalid} in which condensate operators are replaced by their mean values. In 2D, this interaction ``volume'' is $a_s^2$. The number of atoms in a ``volume'' $a^2$ is given by the filling $N/N_s$, where $a$ is the lattice constant. Hence, we require
\begin{equation}
\label{eq:Nreq}
    \frac{Na_s^2}{N_sa^2} \ll 1 \iff a_s \ll \frac{a}{\sqrt{N/N_s}}.
\end{equation}
This requirement on $a_s$ becomes stricter the greater the filling $N/N_s$ is. We therefore follow experiments and theoretical papers in assuming $N/N_s = 1$ whenever a numerical value is needed. However, the treatment should be valid for any $N/N_s$ as long as \eqref{eq:Nreq} is fulfilled.

Our starting point is the Bose-Hubbard Hamiltonian with SOC \eqref{eq:FullH}
\begin{equation}
    H = \sum_{\boldsymbol{k}}\sum_{\alpha\beta}\eta_{\boldsymbol{k}}^{\alpha\beta}A_{\boldsymbol{k}}^{\alpha\dagger}A_{\boldsymbol{k}}^\beta + \frac{1}{2N_s}\sum_{\boldsymbol{k}\boldsymbol{k}'\boldsymbol{p}\boldsymbol{p}'}\sum_{\alpha\beta}U^{\alpha\beta}A_{\boldsymbol{k}}^{\alpha\dagger}A_{\boldsymbol{k}'}^{\beta\dagger}A_{\boldsymbol{p}}^\beta A_{\boldsymbol{p}'}^\alpha \delta_{\boldsymbol{k}+\boldsymbol{k}',\boldsymbol{p}+\boldsymbol{p}'},
\end{equation}
where
\begin{gather*}
    \eta_{\boldsymbol{k}} = \begin{pmatrix} \epsilon_{\boldsymbol{k}}^{\uparrow} +T^{\uparrow} & s_{\boldsymbol{k}} \\ s_{\boldsymbol{k}}^* &  \epsilon_{\boldsymbol{k}}^{\downarrow} +T^{\downarrow} \end{pmatrix}.
\end{gather*}

Due to the nature of BEC the Bogoliubov approach amounts to treating condensate operators $A_{\boldsymbol{k}_{0i}}^\alpha$ differently than excitation operators $A_{\boldsymbol{k}}^\alpha$, where $\boldsymbol{k}_{0i}$ is any occupied condensate momentum and $\boldsymbol{k}$ is any non-condensate momentum. The condensate operators are assumed dominant, and only terms that are at most quadratic in excitation operators are included. Contributions from terms that are cubic or quartic in excitation operators are assumed negligible. Rewriting the Hamiltonian in this way, enables us to later employ the BV transformation to diagonalize the Hamiltonian and obtain the quasiparticle excitation spectrum. For now, the treatment concerns a general Bravais lattice. In a square lattice, possible condensate momenta are represented in figure \ref{fig:1fold4fold}. It is shown in \cite{master} that the Hamiltonian can be written $H \approx H_0 + H_1 + H_2$, where
\begin{align}
\label{eq:H0before}
    \begin{split}
        H_0 = &\sum_{i}  \sum_{\alpha\beta}\eta_{\boldsymbol{k}_{0i}}^{\alpha\beta}A_{\boldsymbol{k}_{0i}}^{\alpha\dagger}A_{\boldsymbol{k}_{0i}}^\beta \\
        &+ \frac{1}{2N_s}\sum_{iji'j'}\sum_{\alpha\beta}U^{\alpha\beta}A_{\boldsymbol{k}_{0i}}^{\alpha\dagger}A_{\boldsymbol{k}_{0j}}^{\beta\dagger}A_{\boldsymbol{k}_{0i'}}^\beta A_{\boldsymbol{k}_{0j'}}^\alpha \delta_{\boldsymbol{k}_{0i}+\boldsymbol{k}_{0j},\boldsymbol{k}_{0i'}+\boldsymbol{k}_{0j'}},
    \end{split}
\end{align}
\begin{align}
\label{eq:H1before}
    \begin{split}
        H_1 = \frac{1}{N_s}\left.\sum_{\boldsymbol{k}}\right.^{'}\sum_{iji'}\sum_{\alpha\beta}U^{\alpha\beta}&\big(A_{\boldsymbol{k}_{0i}}^{\alpha\dagger}A_{\boldsymbol{k}_{0j}}^{\beta\dagger}A_{\boldsymbol{k}_{0i'}}^\beta A_{\boldsymbol{k}}^\alpha \\
        &+A_{\boldsymbol{k}}^{\alpha\dagger}A_{\boldsymbol{k}_{0i'}}^{\beta\dagger}A_{\boldsymbol{k}_{0j}}^\beta A_{\boldsymbol{k}_{0i}}^\alpha \big)\delta_{\boldsymbol{k}+\boldsymbol{k}_{0i'},\boldsymbol{k}_{0i}+\boldsymbol{k}_{0j}}
    \end{split}
\end{align}
and
\begin{align}
    \begin{split}
    \label{eq:H2before}
        H_2 = &\left.\sum_{\boldsymbol{k}}\right.^{'}\sum_{\alpha\beta}\eta_{\boldsymbol{k}}^{\alpha\beta}A_{\boldsymbol{k}}^{\alpha\dagger}A_{\boldsymbol{k}}^\beta \\
        &+\frac{1}{2N_s}\left.\sum_{\boldsymbol{k}\boldsymbol{k}'}\right.^{''}\sum_{ij}\sum_{\alpha\beta} U^{\alpha\beta}\Big(\big( A_{\boldsymbol{k}_{0i}}^{\alpha\dagger}A_{\boldsymbol{k}_{0j}}^{\beta\dagger}A_{\boldsymbol{k}}^\beta A_{\boldsymbol{k}'}^\alpha\\
        &\mbox{\qquad\qquad\qquad\qquad\qquad}+ A_{\boldsymbol{k}}^{\alpha\dagger}A_{\boldsymbol{k}'}^{\beta\dagger}A_{\boldsymbol{k}_{0j}}^\beta A_{\boldsymbol{k}_{0i}}^\alpha\big)\delta_{\boldsymbol{k}+\boldsymbol{k}',\boldsymbol{k}_{0i}+\boldsymbol{k}_{0j}} \\
        &\mbox{\qquad\qquad\qquad\qquad\qquad}+2\big(A_{\boldsymbol{k}_{0i}}^{\alpha\dagger}A_{\boldsymbol{k}}^{\beta\dagger}A_{\boldsymbol{k}_{0j}}^\beta A_{\boldsymbol{k}'}^\alpha \\
        &\mbox{\qquad\qquad\qquad\qquad\qquad}+A_{\boldsymbol{k}_{0i}}^{\alpha\dagger}A_{\boldsymbol{k}}^{\beta\dagger}A_{\boldsymbol{k}'}^\beta A_{\boldsymbol{k}_{0j}}^\alpha \big)\delta_{\boldsymbol{k}+\boldsymbol{k}_{0i},\boldsymbol{k}'+\boldsymbol{k}_{0j}}\Big).
    \end{split}
\end{align}
The sums $\left.\sum_{\boldsymbol{k}}\right.^{'}$ exclude any condensate momenta, while $\left.\sum_{\boldsymbol{k}\boldsymbol{k}'}\right.^{''}$ excludes any terms where at least one of $\boldsymbol{k}$ and $\boldsymbol{k}'$ is equal to a condensate momentum. All possible momentum configurations in the interaction terms used to derive the expressions above are given in table \ref{tab:interactionmomentumconfig}. The possible presence of terms that are linear in excitation operators given in $H_1$ was pointed out by Janssønn, and have to our knowledge not been explored in the literature \cite{master}. Such terms stem from the possibility that $\boldsymbol{k} + \boldsymbol{k}' = \boldsymbol{p} + \boldsymbol{p}'$ may be fulfilled by three condensate momenta and one non-condensate momentum in many-fold cases as represented by cases 2-5 in table \ref{tab:interactionmomentumconfig}. When there is only one condensate momentum this would be impossible. 

\begin{table}[ht]
    \centering
    \caption{All possible momentum configurations in the interaction terms where at least two momenta are condensate momenta, $\boldsymbol{k}_{0i}$. Table reproduced from \cite{master}.}
    \begin{tabular}{|c|c|c|c|c|}
        \hline
        Case & $\boldsymbol{k}$ & $\boldsymbol{k}'$ & $\boldsymbol{p}$ & $\boldsymbol{p}'$ \\
        \hline
        1 & $\boldsymbol{k}_{0i}$ & $\boldsymbol{k}_{0j}$ & $\boldsymbol{k}_{0i'}$ & $\boldsymbol{k}_{0j'}$  \\
        \hline
        2 & $\boldsymbol{k}_{0i}$ & $\boldsymbol{k}_{0j}$ & $\boldsymbol{k}_{0i'}$ & $\boldsymbol{p}'$ \\
        \hline
        3 & $\boldsymbol{k}_{0i}$ & $\boldsymbol{k}_{0j}$ & $\boldsymbol{p}$ &  $\boldsymbol{k}_{0j'}$ \\
        \hline
        4 & $\boldsymbol{k}_{0i}$ & $\boldsymbol{k}'$ & $\boldsymbol{k}_{0i'}$ &  $\boldsymbol{k}_{0j'}$ \\
        \hline
        5 & $\boldsymbol{k}$ & $\boldsymbol{k}_{0j}$ & $\boldsymbol{k}_{0i'}$ &  $\boldsymbol{k}_{0j'}$ \\
        \hline
        6 & $\boldsymbol{k}_{0i}$ & $\boldsymbol{k}_{0j}$ & $\boldsymbol{p}$ & $\boldsymbol{p}'$ \\
        \hline
        7 & $\boldsymbol{k}_{0i}$ & $\boldsymbol{k}'$ & $\boldsymbol{k}_{0i'}$ & $\boldsymbol{p}'$ \\
        \hline
        8 & $\boldsymbol{k}_{0i}$ & $\boldsymbol{k}'$ & $\boldsymbol{p}$ &  $\boldsymbol{k}_{0j'}$ \\
        \hline
        9 & $\boldsymbol{k}$ & $\boldsymbol{k}'$ & $\boldsymbol{k}_{0i'}$ &  $\boldsymbol{k}_{0j'}$ \\
        \hline
        10 & $\boldsymbol{k}$ & $\boldsymbol{k}_{0j}$ & $\boldsymbol{p}$ &  $\boldsymbol{k}_{0j'}$ \\
        \hline
        11 & $\boldsymbol{k}$ & $\boldsymbol{k}_{0j}$ & $\boldsymbol{k}_{0i'}$ &  $\boldsymbol{p}'$ \\
        \hline
    \end{tabular}
    \label{tab:interactionmomentumconfig}
\end{table}

We now employ the Bogoliubov approach and replace the condensate operators by 
\begin{equation}
\label{eq:BogApp}
    A_{\boldsymbol{k}_{0i}}^\alpha \to \sqrt{N_{\boldsymbol{k}_{0i}}^\alpha}e^{-i\theta_{\boldsymbol{k}_{0i}}^\alpha},
\end{equation}
where $N_{\boldsymbol{k}_{0i}}^\alpha = \langle A_{\boldsymbol{k}_{0i}}^{\alpha\dagger} A_{\boldsymbol{k}_{0i}}^\alpha \rangle $ is the number of condensate particles in pseudospin state $\alpha$ with momentum $\boldsymbol{k}_{0i}$. The factor $e^{-i\theta_{\boldsymbol{k}_{0i}}^\alpha}$ is a phase factor that can be determined by minimizing the free energy with respect to the angle $\theta_{\boldsymbol{k}_{0i}}^\alpha$ \cite{bruusflensberg}. Such phase factors determined by the angles $\theta_{\boldsymbol{k}_{0i}}^\alpha$ are usually omitted, but we will find they play an important role in phases that appear due to SOC. Inserting \eqref{eq:BogApp} in \eqref{eq:H0before}, \eqref{eq:H1before} and \eqref{eq:H2before} we find $H \approx H_0 + H_1 + H_2$, with
\begin{align}
    \begin{split}
        H_0 = &\sum_{i}  \sum_{\alpha\beta}\eta_{\boldsymbol{k}_{0i}}^{\alpha\beta}\sqrt{N_{\boldsymbol{k}_{0i}}^\alpha N_{\boldsymbol{k}_{0i}}^\beta}e^{i(\theta_{\boldsymbol{k}_{0i}}^\alpha-\theta_{\boldsymbol{k}_{0i}}^\beta)} \\
        &+ \frac{1}{2N_s}\sum_{iji'j'}\sum_{\alpha\beta}U^{\alpha\beta}\sqrt{N_{\boldsymbol{k}_{0i}}^{\alpha}N_{\boldsymbol{k}_{0j}}^{\beta}N_{\boldsymbol{k}_{0i'}}^\beta N_{\boldsymbol{k}_{0j'}}^\alpha}\\
        &\mbox{\qquad\qquad\qquad\qquad} \cdot e^{i(\theta_{\boldsymbol{k}_{0i}}^\alpha+\theta_{\boldsymbol{k}_{0j}}^\beta-\theta_{\boldsymbol{k}_{0i'}}^\beta -\theta_{\boldsymbol{k}_{0j'}}^\alpha)} \delta_{\boldsymbol{k}_{0i}+\boldsymbol{k}_{0j},\boldsymbol{k}_{0i'}+\boldsymbol{k}_{0j'}},
    \end{split}
\end{align}
\begin{align}
    \begin{split}
        H_1 = \frac{1}{N_s}\left.\sum_{\boldsymbol{k}}\right.^{'}&\sum_{iji'}\sum_{\alpha\beta}U^{\alpha\beta}\bigg(\sqrt{N_{\boldsymbol{k}_{0i}}^{\alpha}N_{\boldsymbol{k}_{0j}}^{\beta}N_{\boldsymbol{k}_{0i'}}^\beta} e^{i(\theta_{\boldsymbol{k}_{0i}}^\alpha+\theta_{\boldsymbol{k}_{0j}}^\beta -\theta_{\boldsymbol{k}_{0i'}}^\beta)} A_{\boldsymbol{k}}^\alpha \\
        &+\sqrt{N_{\boldsymbol{k}_{0i}}^{\alpha}N_{\boldsymbol{k}_{0j}}^{\beta}N_{\boldsymbol{k}_{0i'}}^\beta} e^{-i(\theta_{\boldsymbol{k}_{0i}}^\alpha+\theta_{\boldsymbol{k}_{0j}}^\beta -\theta_{\boldsymbol{k}_{0i'}}^\beta)} A_{\boldsymbol{k}}^{\alpha\dagger} \bigg)\delta_{\boldsymbol{k}+\boldsymbol{k}_{0i'},\boldsymbol{k}_{0i}+\boldsymbol{k}_{0j}}
    \end{split}
\end{align}
and
\begin{align}
    \begin{split}
        H_2 = &\left.\sum_{\boldsymbol{k}}\right.^{'}\sum_{\alpha\beta}\eta_{\boldsymbol{k}}^{\alpha\beta}A_{\boldsymbol{k}}^{\alpha\dagger}A_{\boldsymbol{k}}^\beta +\frac{1}{2N_s}\left.\sum_{\boldsymbol{k}\boldsymbol{k}'}\right.^{''}\sum_{ij}\sum_{\alpha\beta} U^{\alpha\beta}\\
        &\cdot\Bigg(\Big( \sqrt{N_{\boldsymbol{k}_{0i}}^{\alpha}N_{\boldsymbol{k}_{0j}}^{\beta}}e^{i(\theta_{\boldsymbol{k}_{0i}}^\alpha+\theta_{\boldsymbol{k}_{0j}}^\beta)}A_{\boldsymbol{k}}^\beta A_{\boldsymbol{k}'}^\alpha\\
        &+ \sqrt{N_{\boldsymbol{k}_{0i}}^{\alpha}N_{\boldsymbol{k}_{0j}}^{\beta}}e^{-i(\theta_{\boldsymbol{k}_{0i}}^\alpha+\theta_{\boldsymbol{k}_{0j}}^\beta)}A_{\boldsymbol{k}}^{\alpha\dagger}A_{\boldsymbol{k}'}^{\beta\dagger}\Big)\delta_{\boldsymbol{k}+\boldsymbol{k}',\boldsymbol{k}_{0i}+\boldsymbol{k}_{0j}} \\
        &+2\Big(\sqrt{N_{\boldsymbol{k}_{0i}}^{\alpha}N_{\boldsymbol{k}_{0j}}^\beta}e^{i(\theta_{\boldsymbol{k}_{0i}}^\alpha-\theta_{\boldsymbol{k}_{0j}}^\beta)}A_{\boldsymbol{k}}^{\beta\dagger} A_{\boldsymbol{k}'}^\alpha \\
        &+\sqrt{N_{\boldsymbol{k}_{0i}}^{\alpha}N_{\boldsymbol{k}_{0j}}^\alpha}e^{i(\theta_{\boldsymbol{k}_{0i}}^\alpha-\theta_{\boldsymbol{k}_{0j}}^\alpha)} A_{\boldsymbol{k}}^{\beta\dagger}A_{\boldsymbol{k}'}^\beta  \Big)\delta_{\boldsymbol{k}+\boldsymbol{k}_{0i},\boldsymbol{k}'+\boldsymbol{k}_{0j}}\Bigg).
    \end{split}
\end{align}
Comparing to equation (3.91) in \cite{master} we see the major change is that the sums in $H_2$ remain constrained here, but are unconstrained in \cite{master}. This is because we have now neglected terms containing fluctuation operators that in \cite{master} were moved into $H_2$ by removing the restrictions on the sums over $\boldsymbol{k}$. Since $H^\dagger = H$ we must have $H_0$ real and $H_1$ and $H_2$ Hermitian. 
The fact that $\Im(H_0) = 0$ can be shown by rewriting the sum in terms of possible momentum configurations and using $\sin(-x) = -\sin(x)$.
We also rewrite $H_1$ and $H_2$ to make it more obvious that they are their own Hermitian conjugates. Take for instance the terms
\begin{equation*}
    \left.\sum_{\boldsymbol{k}\boldsymbol{k}'}\right.^{''}\sum_{ij}\sum_{\alpha\beta} U^{\alpha\beta}\sqrt{N_{\boldsymbol{k}_{0i}}^{\alpha}N_{\boldsymbol{k}_{0j}}^\beta}\big(e^{i(\theta_{\boldsymbol{k}_{0i}}^\alpha-\theta_{\boldsymbol{k}_{0j}}^\beta)}A_{\boldsymbol{k}}^{\beta\dagger} A_{\boldsymbol{k}'}^\alpha + e^{i(\theta_{\boldsymbol{k}_{0i}}^\alpha-\theta_{\boldsymbol{k}_{0j}}^\beta)}A_{\boldsymbol{k}}^{\beta\dagger} A_{\boldsymbol{k}'}^\alpha \big).
\end{equation*}
We let $\alpha \leftrightarrow \beta$, $i \leftrightarrow j$ and $\boldsymbol{k} \leftrightarrow \boldsymbol{k}'$ in the second term, and recognize it as the Hermitian conjugate of the first. Finally we may write
\begin{equation}
\label{eq:MFTH}
    H = H_0 + H_1 + H_2,
\end{equation}
where
\begin{align}
    \begin{split}
    \label{eq:H0}
    H_0 = &\sum_{i}  \sum_{\alpha\beta}\eta_{\boldsymbol{k}_{0i}}^{\alpha\beta}\sqrt{N_{\boldsymbol{k}_{0i}}^\alpha N_{\boldsymbol{k}_{0i}}^\beta}e^{i(\theta_{\boldsymbol{k}_{0i}}^\alpha-\theta_{\boldsymbol{k}_{0i}}^\beta)} \\
    &+ \frac{1}{2N_s}\sum_{iji'j'}\sum_{\alpha\beta}U^{\alpha\beta} \sqrt{N_{\boldsymbol{k}_{0i}}^{\alpha}N_{\boldsymbol{k}_{0j}}^{\beta}N_{\boldsymbol{k}_{0i'}}^\beta N_{\boldsymbol{k}_{0j'}}^\alpha} \\
    &\mbox{\qquad\qquad\qquad} \cdot \cos(\theta_{\boldsymbol{k}_{0i}}^\alpha+\theta_{\boldsymbol{k}_{0j}}^\beta-\theta_{\boldsymbol{k}_{0i'}}^\beta -\theta_{\boldsymbol{k}_{0j'}}^\alpha) \delta_{\boldsymbol{k}_{0i}+\boldsymbol{k}_{0j},\boldsymbol{k}_{0i'}+\boldsymbol{k}_{0j'}},
    \end{split}
\end{align}
\begin{align}
    \begin{split}
    \label{eq:H1}
        H_1 = \frac{1}{N_s}\left.\sum_{\boldsymbol{k}}\right.^{'}\sum_{iji'}\sum_{\alpha\beta}U^{\alpha\beta}&\Big(\sqrt{N_{\boldsymbol{k}_{0i}}^{\alpha}N_{\boldsymbol{k}_{0j}}^{\beta}N_{\boldsymbol{k}_{0i'}}^\beta} e^{i(\theta_{\boldsymbol{k}_{0i}}^\alpha+\theta_{\boldsymbol{k}_{0j}}^\beta -\theta_{\boldsymbol{k}_{0i'}}^\beta)} A_{\boldsymbol{k}}^\alpha \\
        &+\textrm{H.c.} \Big)\delta_{\boldsymbol{k}+\boldsymbol{k}_{0i'},\boldsymbol{k}_{0i}+\boldsymbol{k}_{0j}}
    \end{split}
\end{align}
and
\begin{align}
    \begin{split}
    \label{eq:H2}
        H_2 = &\left.\sum_{\boldsymbol{k}}\right.^{'}\sum_{\alpha\beta}\eta_{\boldsymbol{k}}^{\alpha\beta}A_{\boldsymbol{k}}^{\alpha\dagger}A_{\boldsymbol{k}}^\beta +\frac{1}{2N_s}\left.\sum_{\boldsymbol{k}\boldsymbol{k}'}\right.^{''}\sum_{ij}\sum_{\alpha\beta} U^{\alpha\beta}\\
        &\cdot \Bigg[\sqrt{N_{\boldsymbol{k}_{0i}}^{\alpha}N_{\boldsymbol{k}_{0j}}^{\beta}}\bigg(\Big( e^{i(\theta_{\boldsymbol{k}_{0i}}^\alpha+\theta_{\boldsymbol{k}_{0j}}^\beta)}A_{\boldsymbol{k}}^\beta A_{\boldsymbol{k}'}^\alpha + \textrm{H.c.}\Big)\delta_{\boldsymbol{k}+\boldsymbol{k}',\boldsymbol{k}_{0i}+\boldsymbol{k}_{0j}} \\
        &\mbox{\qquad\qquad\qquad}+\Big(e^{i(\theta_{\boldsymbol{k}_{0i}}^\alpha-\theta_{\boldsymbol{k}_{0j}}^\beta)}A_{\boldsymbol{k}}^{\beta\dagger} A_{\boldsymbol{k}'}^\alpha +\textrm{H.c}\Big)\bigg)\delta_{\boldsymbol{k}+\boldsymbol{k}_{0i},\boldsymbol{k}'+\boldsymbol{k}_{0j}}\\
        &+\sqrt{N_{\boldsymbol{k}_{0i}}^{\alpha}N_{\boldsymbol{k}_{0j}}^\alpha}\Big(e^{i(\theta_{\boldsymbol{k}_{0i}}^\alpha-\theta_{\boldsymbol{k}_{0j}}^\alpha)} A_{\boldsymbol{k}}^{\beta\dagger}A_{\boldsymbol{k}'}^\beta + \textrm{H.c.} \Big)\delta_{\boldsymbol{k}+\boldsymbol{k}_{0i},\boldsymbol{k}'+\boldsymbol{k}_{0j}}\Bigg].
    \end{split}
\end{align}

\section{Phase Diagram When Neglecting Excitations}
To investigate the possible phases of the system, the operator independent part of the Hamiltonian, $H_0$, will be used. This describes a pure condensate, where one assumes the free energy $F \approx H_0$ and thus minimizes $H_0$ in terms of the free parameters $N_{\boldsymbol{k}_{0i}}^\alpha$, $\theta_{\boldsymbol{k}_{0i}}^\alpha$ and $k_0$. $k_0$ is defined by $\boldsymbol{k}_{01}  \equiv (k_0, k_0)$, while the total number of condensate particles is $N_0 = \sum_{i}\sum_\alpha N_{\boldsymbol{k}_{0i}}^\alpha$. Neglecting excitations, $N_0$ is the same as the total number of particles, $N$, and is kept fixed. For the moment, we do not choose specific values of $N^\uparrow$ and $N^\downarrow$. Note that even when interactions are weak, there will always be excitations out of the condensate. The results found from minimizing $H_0$ are therefore only guidelines. The more accurate approach is to diagonalize the full Hamiltonian \eqref{eq:MFTH}, and then minimize the free energy. This is the aim of the next chapter.


A 2D square optical lattice is assumed for the remainder of the thesis. We will also assume $t^\uparrow = t^\downarrow = t$, $T^{\uparrow} = T^{\downarrow} = T$, $U^{\uparrow\uparrow}=U^{\downarrow\downarrow} = U$ and 
\begin{equation}
    U^{\uparrow\downarrow}=U^{\downarrow\uparrow} \equiv \alpha U
\end{equation}
Naturally, since we assumed repulsive interactions, it is assumed that $\alpha \geq 0$. First considering the case where the condensation occurs at zero momentum we get
\begin{align}
    \begin{split}
        H_0 &= N_0(\epsilon_{\boldsymbol{0}}+T)+\frac{U}{2N_s}\left((N_0^\uparrow)^2+2\alpha N_0^\uparrow (N_0-N_0^\uparrow) + (N_0-N_0^\uparrow)^2\right),
    \end{split}
\end{align}
where we used that $N_0^\uparrow + N_0^\downarrow = N_0$ is fixed. Also, $N_0^\alpha$ is a shorthand for $N_{\boldsymbol{k}_{00}=\boldsymbol{0}}^\alpha$. It is clear the only dependence on $N_0^\uparrow$ lies in the second term. It is easy to show that when $\alpha < 1$ $N_0^\uparrow=N_0^\downarrow = N_0/2$, i.e. balance between pseudospin states, minimizes $H_0$, while for $\alpha>1$ complete imbalance is preferred. For concreteness $N_0^\uparrow = N_0$ and $N_0^\downarrow = 0$ without loss of generality. The former phase is denoted NZ, the latter PZ for non-polarized and polarized zero-momentum phase respectively. 

Next, we assume the condensation occurs into any of the four momenta $\boldsymbol{k}_{0i},$ $i = 1,2,3,4$ introduced in chapter \ref{sec:SOCU0}. We adopt the shorthand notations $N_{\boldsymbol{k}_{0i}}^\alpha=N_{0i}^\alpha$ and $\theta_{\boldsymbol{k}_{0i}}^\alpha = \theta_i^\alpha$ from now on. A priori, any distribution of particles between the four possible momenta found for the non-interacting, SOC Bose gas is possible. It is however expected that including interactions will lead to certain ground states being preferred \cite{SOCOLRev}. Defining $s_{\boldsymbol{k}} \equiv |s_{\boldsymbol{k}}|\exp(-i\gamma_{\boldsymbol{k}})$ and $\Delta\theta_i \equiv \theta_i^\downarrow-\theta_i^\uparrow$ and using \eqref{eq:H0}, the expression for $H_0$ becomes
\begin{align}
    \begin{split}
    \label{eq:fullH0}
        H_0 &= N_0(\epsilon_{\boldsymbol{k}_{01}}+T) + \sum_{i=1}^4 2\sqrt{N_{0i}^{\uparrow} N_{0i}^{\downarrow}}\abs{s_{\boldsymbol{k}_{01}}}\cos(\gamma_{\boldsymbol{k}_{0i}}+\Delta\theta_i)\\
        &+\frac{U}{2N_s}\bigg[  \sum_{i=1}^4 \left((N_{0i}^{\uparrow})^2 + 2\alpha N_{0i}^{\uparrow} N_{0i}^{\downarrow}  + (N_{0i}^{\downarrow})^2 \right) \\
        &+ \sum_{i=1}^3 \sum_{j>i}^4 \bigg(4N_{0i}^{\uparrow}N_{0j}^{\uparrow} + 4N_{0i}^{\downarrow}N_{0j}^{\downarrow} + 2\alpha(N_{0i}^\uparrow N_{0j}^\downarrow + N_{0i}^\downarrow N_{0j}^\uparrow)  \\
        & \mbox{\qquad\qquad}+4\alpha\sqrt{N_{0i}^{\uparrow}N_{0i}^{\downarrow}N_{0j}^{\uparrow}N_{0j}^{\downarrow}}\cos(\Delta\theta_i-\Delta\theta_j) \bigg)\\
        &+8\sqrt{N_{01}^{\uparrow}N_{03}^{\uparrow}N_{02}^{\uparrow}N_{04}^{\uparrow}}\cos(\theta_1^\uparrow+\theta_3^\uparrow-\theta_2^\uparrow-\theta_4^\uparrow) \\
        &+8\sqrt{N_{01}^{\downarrow}N_{03}^{\downarrow}N_{02}^{\downarrow}N_{04}^{\downarrow}}\cos(\theta_1^\downarrow+\theta_3^\downarrow-\theta_2^\downarrow-\theta_4^\downarrow) \\
        &+4\alpha\sqrt{N_{01}^{\uparrow}N_{03}^{\downarrow}N_{02}^{\downarrow}N_{04}^{\uparrow}}\cos(\theta_1^\uparrow+\theta_3^\downarrow-\theta_2^\downarrow-\theta_4^\uparrow) \\
        &+4\alpha\sqrt{N_{01}^{\uparrow}N_{03}^{\downarrow}N_{02}^{\uparrow}N_{04}^{\downarrow}}\cos(\theta_1^\uparrow+\theta_3^\downarrow-\theta_2^\uparrow-\theta_4^\downarrow) \\
        &+4\alpha\sqrt{N_{01}^{\downarrow}N_{03}^{\uparrow}N_{02}^{\uparrow}N_{04}^{\downarrow}}\cos(\theta_1^\downarrow+\theta_3^\uparrow-\theta_2^\uparrow-\theta_4^\downarrow) \\
        &+4\alpha\sqrt{N_{01}^{\downarrow}N_{03}^{\uparrow}N_{02}^{\downarrow}N_{04}^{\uparrow}}\cos(\theta_1^\downarrow+\theta_3^\uparrow-\theta_2^\downarrow-\theta_4^\uparrow) \bigg].
    \end{split}
\end{align}
Several comments can be made here. The SOC dependent terms are minimized when
\begin{equation}
\label{eq:gammathetapi}
    \gamma_{\boldsymbol{k}_{0i}} + \Delta\theta_i = \gamma_{\boldsymbol{k}_{0i}} + \theta_i^\downarrow -\theta_i^\uparrow =  \pi
\end{equation}
and for a pseudospin balanced condensate, $N_{0i}^\uparrow = N_{0i}^\downarrow$. Equation \eqref{eq:gammathetapi} was also found in \cite{master} as a requirement for the chemical potential to be real. The first $U$ dependent terms will, as discussed for the zero-momentum case, prefer balance when $\alpha < 1$ and complete imbalance when $\alpha>1$. Hence, for $\alpha>1$ there will be a competition between SOC and interactions as to whether balance or complete imbalance between pseudospin states is preferred. Also, if SOC ensures balance between pseudospin states, and several momenta are occupied, these first $U$ dependent terms will prefer balance between the momenta as well. 

The terms proportional to $\cos(\Delta\theta_i-\Delta\theta_j)$ are minimized if $\cos(\Delta\theta_i-\Delta\theta_j)$ $=$ $-1$. When \eqref{eq:gammathetapi} is fulfilled, we have $\cos(\Delta\theta_1-\Delta\theta_3) = -1$ and $\cos(\Delta\theta_2-\Delta\theta_4) = -1$ while the other angle combinations render the cosine $0$. Thus, when only two momenta are occupied, either $\boldsymbol{k}_{01}$ and $\boldsymbol{k}_{03}$ or  $\boldsymbol{k}_{02}$ and $\boldsymbol{k}_{04}$ are preferred. Which are chosen is arbitrary, and without loss of generality one may assume $\boldsymbol{k}_{01}$ and $\boldsymbol{k}_{03}$. This is defined as the stripe wave (SW) phase. Its name is derived from its striped spin polarization \cite{SOCOLRev}. 

Another interesting phase is a condensate at a single nonzero momentum, called the plane wave (PW) phase. 
A phase that  is not expected to appear as a ground state is mentioned in \cite{SOCOLRev}, namely the lattice wave (LW) phase, where all four momenta are equally occupied.
In \cite{SOCOLRev} this is called a Skyrmion state. The final six interaction dependent terms in $H_0$ are only relevant if all four condensate momenta are occupied. Their effect will be discussed in the context of the LW phase.

A numeric investigation of \eqref{eq:fullH0} assuming \eqref{eq:gammathetapi} holds was made. As a check, it appeared the choice \eqref{eq:gammathetapi} for the angles was always at least a local minimum of $H_0$. In general the pseudospin balanced PW phase minimizes $H_0$ when $\alpha < 1$ and the momentum and pseudospin balanced SW phase is preferred when $\alpha>1$. When $\lambda_R$ decreases for $\alpha > 1$ a point is reached where a completely imbalanced PW phase is preferred. This happens at very weak SOC, i.e. $\lambda_R \lesssim U$. We have assumed $U \ll t$ and thus it is only when $\lambda_R \ll t$ this state appears. We therefore focus on the two cases when there is no SOC, and when there is SOC with a strength such that SOC dominates over interactions in the minimization of $H_0$. Hence, the completely imbalanced PW phase at weak SOC will be ignored.



The possible phases mentioned so far are PZ, NZ, PW, SW and LW as in \cite{master}. In the next chapter we will derive the elementary excitations in these phases. Two other possible phases are occupation of $\boldsymbol{k}_{01}$ and $\boldsymbol{k}_{02}$ named C1 phase and occupation of $\boldsymbol{k}_{01}$, $\boldsymbol{k}_{02}$ and $\boldsymbol{k}_{03}$ named C2 phase \cite{master}. Shortly, we will calculate $H_0$ in all these phases and construct a phase diagram analogously to what was done in \cite{master}.

We again point out that minimization of $H_0$ is not the most accurate approach. One should minimize the free energy, or equivalently in the case of zero temperature, the ground state energy $\langle H\rangle$. Also remember that we can control $N^\uparrow$ and $N^\downarrow$ and they are thus not variational parameters. The intuition afforded us by investigating $H_0$ tells us that whenever nonzero condensate momenta are occupied, the most natural phases are ones where there are equally many particles in the two pseudospin states and in the different momenta if several condensate momenta are occupied. We will therefore choose $N^\uparrow = N^\downarrow$ (except in the PZ phase) and assume $N_{0i}^\alpha= N_{0j}^\alpha$ for all occupied momenta $\boldsymbol{k}_{0i}$ and $\boldsymbol{k}_{0j}$. The latter was found to be a requirement in \cite{master} to ensure the chemical potentials did not depend on an arbitrary momentum index.

We illustrate the possible phases in figure \ref{fig:phaseillustation}. The figure is a reproduction of a similar figure in \cite{master}, wherein it was shown that if a nonzero condensate momentum is occupied, there will be particles of both pseudospin states present in the condensate. This was needed to cancel the terms linear in condensate fluctuations. As we have now set these fluctuations to zero, the result is not necessarily valid. We however know that SOC is required to obtain nonzero condensate momenta, and SOC will be most operative in the system if there are particles of both pseudospin states in the condensate. We have also found that the operator independent part of the Hamiltonian tends to prefer pseudospin balance in the condensate when SOC dominates the minimization, suggesting complete pseudospin imbalance is unlikely in the nonzero condensate momentum cases.  


\begin{figure}
    \raggedright
    \begin{subfigure}{.33\textwidth}
      \includegraphics[width=0.9\linewidth]{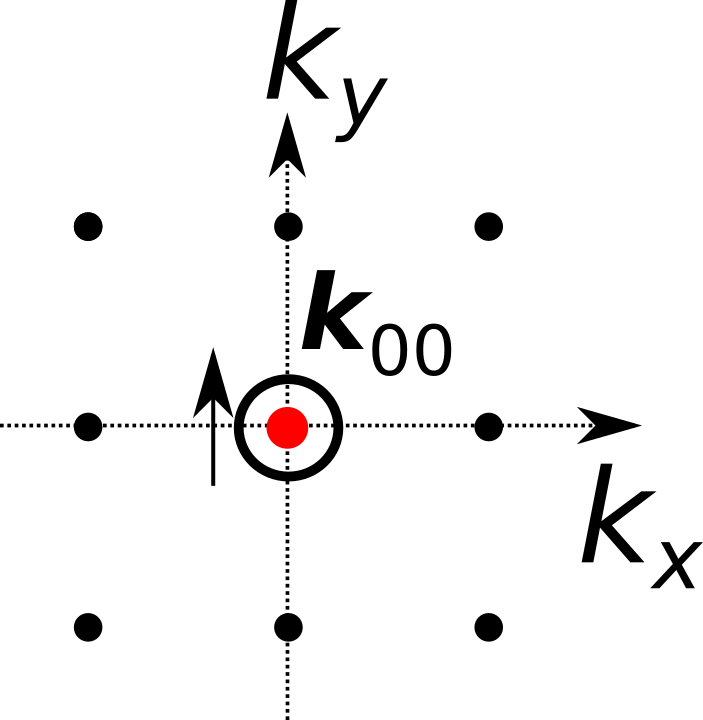}
      \caption{PZ.}
      \label{fig:PZill}
    \end{subfigure}%
    \begin{subfigure}{.33\textwidth}
      \includegraphics[width=0.9\linewidth]{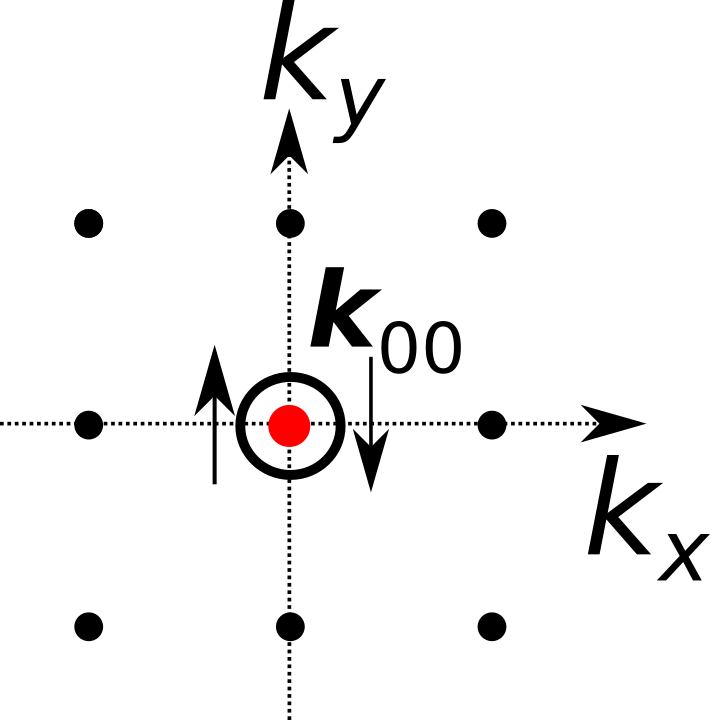}
      \caption{NZ.}
      \label{fig:NZill}
    \end{subfigure}%
    \vskip\baselineskip
    \raggedright
    \begin{subfigure}{.33\textwidth}
      \includegraphics[width=0.9\linewidth]{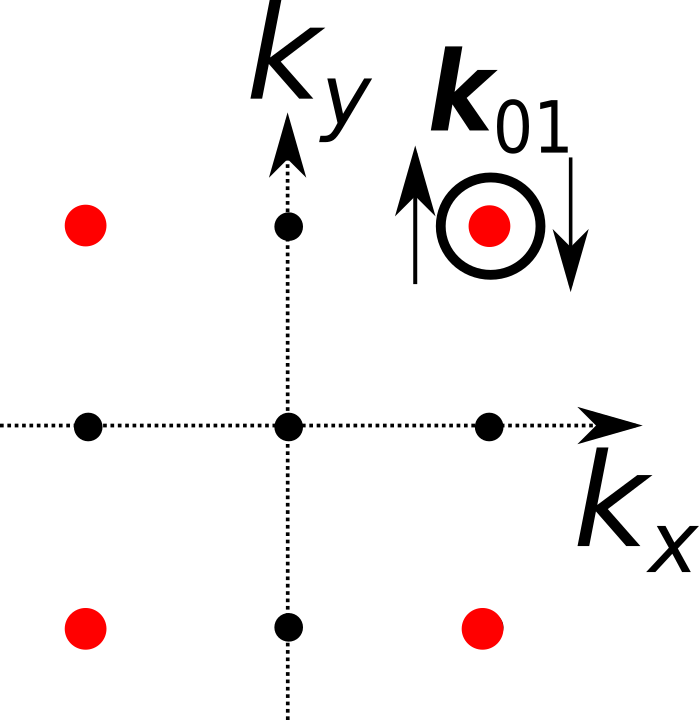}
      \caption{PW.}
      \label{fig:PWill}
    \end{subfigure}%
    \begin{subfigure}{.33\textwidth}
      \includegraphics[width=0.9\linewidth]{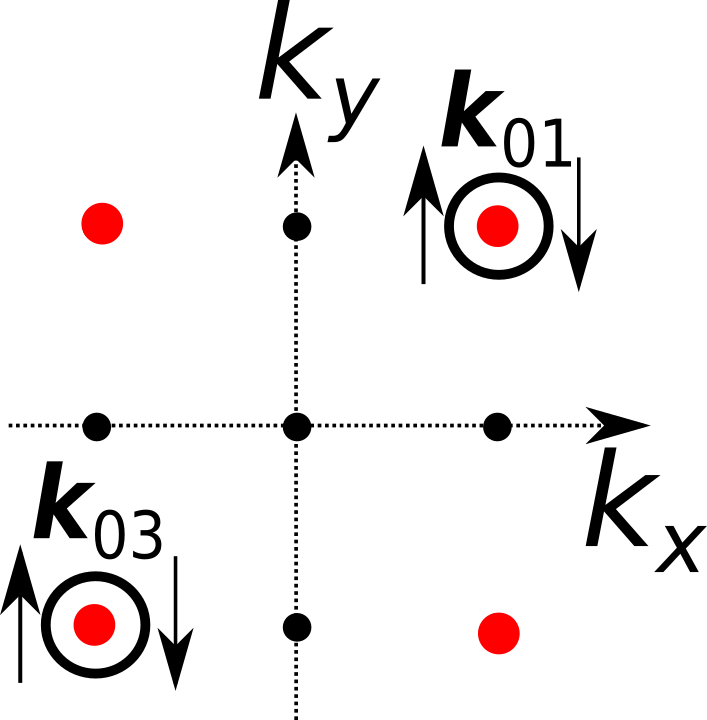}
      \caption{SW.}
      \label{fig:SWill}
    \end{subfigure}%
    \begin{subfigure}{.33\textwidth}
      \includegraphics[width=0.9\linewidth]{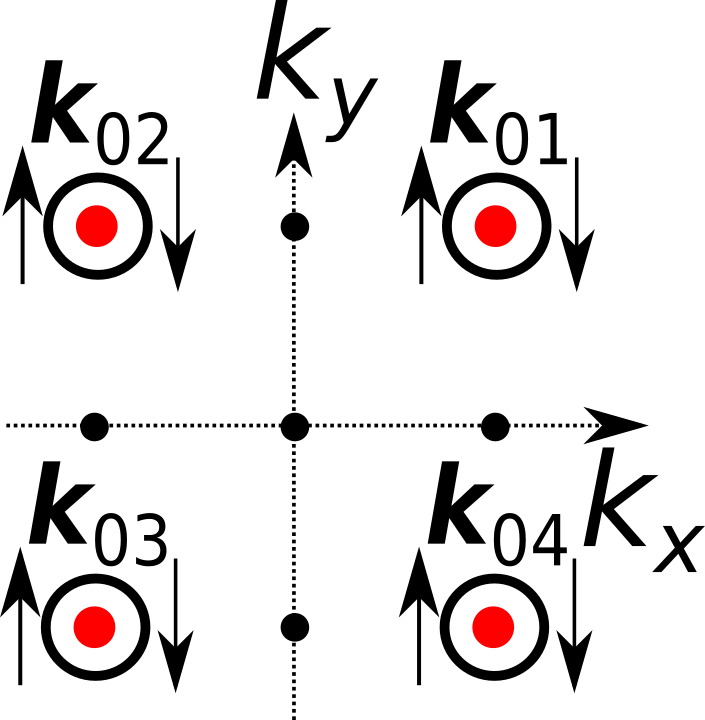}
      \caption{LW.}
      \label{fig:LWill}
    \end{subfigure}%
    \vskip\baselineskip
    \raggedright
    \begin{subfigure}{.33\textwidth}
      \includegraphics[width=0.9\linewidth]{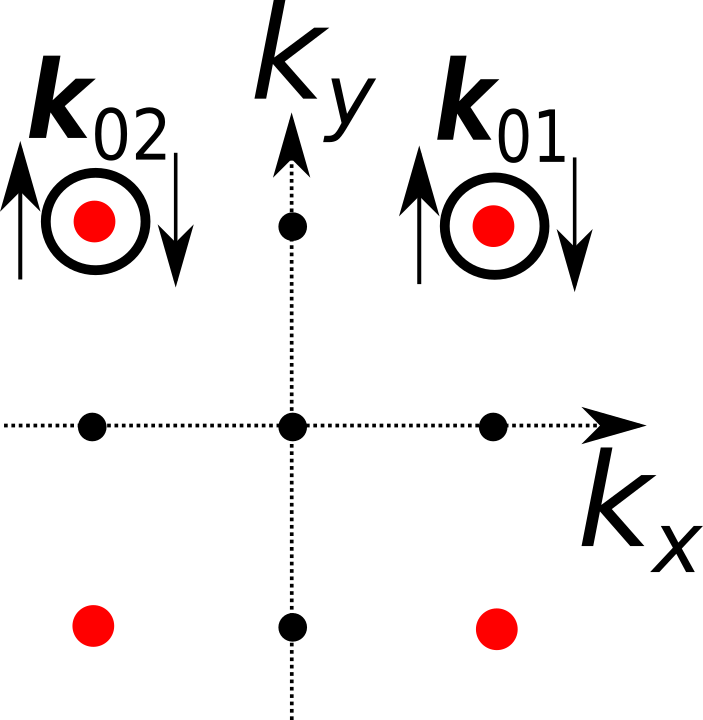}
      \caption{C1.}
      \label{fig:C1ill}
    \end{subfigure}%
    \begin{subfigure}{.33\textwidth}
      \includegraphics[width=0.9\linewidth]{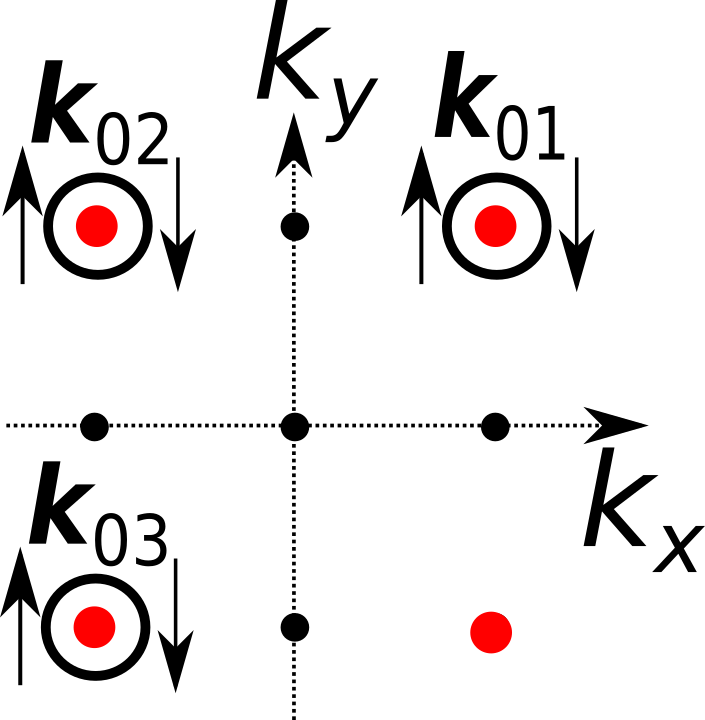}
      \caption{C2.}
      \label{fig:C2ill}
    \end{subfigure}%
    \caption{An illustation of the possible phases. The black points represent lattice sites in momentum space for the 2D square lattice, while the red points represent possible condensate momenta. Lattice sites between $\boldsymbol{k}_{00} = \boldsymbol{0}$ and $\boldsymbol{k}_{0i}$ are not shown. Encircled red points indicate the condensate momentum is occupied, the arrows indicate the presence of pseudospin up and down  atoms in the condensate. (a) and (b) are the polarized (PZ) and non-polarized (NZ) zero momentum phases, named after their degree of pseudospin imbalance and condensate momentum. (c), (d) and (e) show the plane (PW), stripe (SW) and lattice (LW) wave phases, named after the wave patters they generate in real space. (f) and (g) show the arbitrarily named C1 and C2 phases whose excitation spectra will not be explored in this thesis. This figure is adapted from figure 4.1 by Janssønn \cite{master} who could later exclude the C1 and C2 phases in the grand canonical ensemble.
    \label{fig:phaseillustation}}
\end{figure}

\subsection{PZ Phase}
We assume the system condenses at zero momentum into only one pseudospin state such that $N_{\boldsymbol{k}_{00}}^\uparrow = N_0^\uparrow = N_0 = N$, while $N_0^\downarrow = 0$. Thus, $H_0$ is
\begin{equation}
\label{eq:H0PZ}
    H_0^{\textrm{PZ}} = N(\epsilon_{\boldsymbol{0}}+T)+\frac{UN^2}{2N_s}.
\end{equation}
\subsection{NZ Phase}
We assume the condensate has zero momentum and equal number of particles in both pseudospin states. I.e. $N_0^\uparrow = N_0^\downarrow = N_0/2 = N/2$. Hence,
\begin{equation}
\label{eq:H0NZ}
    H_0^{\textrm{NZ}} = N(\epsilon_{\boldsymbol{0}}+T)+\frac{UN^2}{4N_s}(1+\alpha).
\end{equation}
We note that for $\alpha<1$, $H_0^{\textrm{NZ}} < H_0^{\textrm{PZ}}$ while for $\alpha>1$ the opposite is true.
\subsection{PW Phase} \label{sec:PWphaseH0}
Without loss of generality we assume the condensate momentum is $\boldsymbol{k}_{01} = (k_0, k_0)$. Also assuming balance between pseodospin states we find
\begin{equation}
\label{eq:H0PW}
    H_0^{\textrm{PW}} = N(\epsilon_{\boldsymbol{k}_{01}}+T)+N\abs{s_{\boldsymbol{k}_{01}}}\cos(\gamma_{\boldsymbol{k}_{01}}+\Delta\theta_1) + \frac{UN^2}{4N_s}(1+\alpha).
\end{equation}
This $H_0$ is minimized when $\theta_1^\downarrow-\theta_1^\uparrow = \pi/4$ i.e. when \eqref{eq:gammathetapi} holds. Inserting \eqref{eq:gammathetapi} we have
\begin{equation}
    H_0^{\textrm{PW}} = N(\epsilon_{\boldsymbol{k}_{01}}+T)-N\abs{s_{\boldsymbol{k}_{01}}}+ \frac{UN^2}{4N_s}(1+\alpha).
\end{equation}
In terms of the variational parameter $k_0$ the minimum of $H_0$ appears at
\begin{equation}
    k_0 a = k_{0m} a \equiv \arctan(\frac{\lambda_R}{\sqrt{2}t}),
\end{equation}
which can be shown by differentiating with respect to $k_0$. Notice that this is the same $k_0$ found to be the minimum of $\lambda_{\boldsymbol{k}}^-$ for the non-interacting SOC Bose gas.
\subsection{SW Phase}
We assume the condensate momenta are $\boldsymbol{k}_{01} = (k_0, k_0)$ and $\boldsymbol{k}_{03} = -\boldsymbol{k}_{01}$ and that $N_{0i}^\alpha = N/4$ for $i=1,3$ and $\alpha = \uparrow, \downarrow$. Hence,
\begin{align}
    \begin{split}
    \label{eq:H0SW}
        H_0^{\textrm{SW}} &= N(\epsilon_{\boldsymbol{k}_{01}}+T) +\frac{UN^2}{8N_s}\Big(3+\alpha\big(2+\cos(\Delta\theta_1-\Delta\theta_3)\big)\Big) \\
        &+\frac{N}{2}\abs{s_{\boldsymbol{k}_{01}}}\cos(\gamma_{\boldsymbol{k}_{01}}+\Delta\theta_1) +\frac{N}{2}\abs{s_{\boldsymbol{k}_{01}}}\cos(\gamma_{\boldsymbol{k}_{03}}+\Delta\theta_3).
    \end{split}
\end{align}
Once again angles satisfying \eqref{eq:gammathetapi} minimize $H_0$, while $k_0=k_{0m}$ is found to be the value of $k_0$ that minimizes $H_0$. Inserting \eqref{eq:gammathetapi} we have
\begin{align}
    \begin{split}
         H_0^{\textrm{SW}} &= N(\epsilon_{\boldsymbol{k}_{01}}+T) -N\abs{s_{\boldsymbol{k}_{01}}} +\frac{UN^2}{8N_s}(3+\alpha).
    \end{split}
\end{align}
\subsection{LW Phase} \label{sec:H0LWphase}
We assume the condensate momenta are $\pm\boldsymbol{k}_{01} = \pm(k_0, k_0)$ and $\pm\boldsymbol{k}_{02} = \pm(-k_0, k_0)$ and that $N_{0i}^\alpha = N/8$ for all $i$ and $\alpha = \uparrow, \downarrow$. Hence,
\begin{align}
    \begin{split}
    \label{eq:H0LW}
        H_0^{\textrm{LW}} &= N(\epsilon_{\boldsymbol{k}_{01}}+T) + \frac{N}{4}\abs{s_{\boldsymbol{k}_{01}}}\sum_{i=1}^4 \cos(\gamma_{\boldsymbol{k}_{0i}}+\Delta\theta_i)\\
        +&\frac{UN^2}{32N_s}\Bigg(14+2\cos(\theta_1^\uparrow+\theta_3^\uparrow-\theta_2^\uparrow-\theta_4^\uparrow)+2\cos(\theta_1^\downarrow+\theta_3^\downarrow-\theta_2^\downarrow-\theta_4^\downarrow)\\
        &+\alpha \bigg[8+\cos(\Delta\theta_1-\Delta\theta_2) + \cos(\Delta\theta_1-\Delta\theta_3) + \cos(\Delta\theta_1-\Delta\theta_4) \\
        &\mbox{\qquad\quad}+ \cos(\Delta\theta_2-\Delta\theta_3) + \cos(\Delta\theta_2-\Delta\theta_4) + \cos(\Delta\theta_3-\Delta\theta_4)\\
        &\mbox{\qquad\quad} + \cos(\theta_1^\uparrow+\theta_3^\downarrow-\theta_2^\downarrow-\theta_4^\uparrow) + \cos(\theta_1^\downarrow+\theta_3^\uparrow-\theta_2^\uparrow-\theta_4^\downarrow) \\
        &\mbox{\qquad\quad} + \cos(\theta_1^\uparrow+\theta_3^\downarrow-\theta_2^\uparrow-\theta_4^\downarrow) + \cos(\theta_1^\downarrow+\theta_3^\uparrow-\theta_2^\downarrow-\theta_4^\uparrow) \bigg]\Bigg).
    \end{split}
\end{align}
It will be shown in appendix \ref{app:LW} that the choices \eqref{eq:gammathetapi} together with 
$$\theta_1^\uparrow+\theta_3^\uparrow-\theta_2^\uparrow-\theta_4^\uparrow = -\pi/2$$ minimize $\langle H_{\textrm{LW}} \rangle$. Let us define $\theta_1^\alpha+\theta_3^\beta-\theta_2^\gamma-\theta_4^\delta \equiv \alpha\beta\gamma\delta$. Using \eqref{eq:gammathetapi} we find that
\begin{align}
    \begin{split}
        \downarrow\downarrow\downarrow\downarrow = \uparrow\uparrow\uparrow\uparrow + \pi, \mbox{\qquad} \uparrow\downarrow\uparrow\downarrow = \uparrow\downarrow\downarrow\uparrow + \pi \mbox{\qquad and \qquad} \downarrow\uparrow\downarrow\uparrow = \downarrow\uparrow\uparrow\downarrow + \pi,
    \end{split}
\end{align}
which means that all the $\cos(\alpha\beta\gamma\delta)$-terms cancel since $\cos(x+\pi) = -\cos(x)$. Additionally, we find that $\uparrow\downarrow\downarrow\uparrow = \downarrow\uparrow\uparrow\downarrow = \uparrow\uparrow\uparrow\uparrow - \pi/2$ and $\uparrow\downarrow\uparrow\downarrow = \downarrow\uparrow\downarrow\uparrow =  \uparrow\uparrow\uparrow\uparrow + \pi/2$.
Once again it is $k_{0m}$ that minimizes $H_0^{\textrm{LW}}$. Inserting the choices for the angles, we get
\begin{align}
    \begin{split}
        H_0^{\textrm{LW}} &= N(\epsilon_{\boldsymbol{k}_{01}}+T) - N\abs{s_{\boldsymbol{k}_{01}}}+\frac{UN^2}{16N_s}(7+3\alpha).
    \end{split}
\end{align}
This is greater than $H_0^{\textrm{SW}}$ for all $\alpha \geq 0$.

\subsection{C1 and C2 Phases}
In the C1 phase, we assume the condensate momenta are $\boldsymbol{k}_{01} = (k_0, k_0)$ and $\boldsymbol{k}_{02} = (-k_0, k_0)$ and that $N_{0i}^\alpha = N/4$ for $i=1,2$ and $\alpha = \uparrow, \downarrow$. 
Once again angles satisfying \eqref{eq:gammathetapi} minimizes $H_0$, at least assuming that SOC dominates the minimization. Furthermore, $k_0=k_{0m}$ is found to be the value of $k_0$ that minimizes $H_0$. Inserting \eqref{eq:gammathetapi} we have
\begin{align}
    \begin{split}
         H_0^{\textrm{C1}} &= N(\epsilon_{\boldsymbol{k}_{01}}+T) -N\abs{s_{\boldsymbol{k}_{01}}} +\frac{UN^2}{8N_s}(3+2\alpha),
    \end{split}
\end{align}
which is greater than $H_0^{\textrm{PW}}$ and $H_0^{\textrm{SW}}$ for all $\alpha \geq 0$. For $\alpha>1$ is is also greater than $H_0^{\textrm{LW}}$.

In the C2 phase, we assume the condensate momenta are $\boldsymbol{k}_{01} = (k_0, k_0)$, $\boldsymbol{k}_{02} = (-k_0, k_0)$ and $\boldsymbol{k}_{03} = -\boldsymbol{k}_{01}$ and that $N_{0i}^\alpha = N/6$ for $i=1,2,3$ and $\alpha = \uparrow, \downarrow$. 
Angles satisfying \eqref{eq:gammathetapi} minimizes $H_0$, at least assuming that SOC dominates the minimization. Additionally, $k_0=k_{0m}$ is found to be the value of $k_0$ that minimizes $H_0$. Inserting \eqref{eq:gammathetapi} we have
\begin{align}
    \begin{split}
         H_0^{\textrm{C2}} &= N(\epsilon_{\boldsymbol{k}_{01}}+T) -N\abs{s_{\boldsymbol{k}_{01}}} +\frac{UN^2}{36N_s}(15+7\alpha),
    \end{split}
\end{align}
which is greater than $H_0^{\textrm{PW}}$ for $\alpha<3$, greater than $H_0^{\textrm{SW}}$ for all $\alpha \geq 0$ and greater than $H_0^{\textrm{LW}}$ for $\alpha>3$. Though further investigation is required, we will exclude the C1 and C2 phases from now on. The arguments being that at least two of the PW, SW and LW phases have lower $H_0$ than the C1 and C2 phases at any $\alpha \geq 0, \alpha \neq 3$ and that they are not mentioned as possible states in the review article \cite{SOCOLRev}. Furthermore, the C1 and C2 phases could be dismissed in \cite{master} as they would render a complex chemical potential, and there is therefore reason to suspect these will not be relevant. 

\subsection{Phase Diagram}
In the PZ phase $N_0^\downarrow = N^\downarrow = 0$ is an input parameter. Hence, the PZ phase is special, in the sense that it has different input parameters than the other phases. Thus, including it in a phase diagram no longer makes sense, as opposed to what was done in \cite{master}. If one sets $N^\uparrow = N$ and $N^\downarrow = 0$ one will get the PZ phase for all parameters based on a treatment of $H_0$ only. Its excitation spectrum is investigated in the next chapter. We ignored a completely pseudospin imbalanced version of the PW phase. However, it can be shown that this phase will always have a higher value of $H_0$ than the PZ phase. Essentially, $\epsilon_{\boldsymbol{0}} = -4t$ is replaced by $\epsilon_{\boldsymbol{k}_{01}} = -4t\cos(k_0 a)$ which is always greater when $k_0 \neq 0$. 

The remaining phases may be compared at equal input parameters. The dependence on the energy offset $T$ is the same in all phases, and it is therefore arbitrary. The variational parameters are set to the values that minimize $H_0$ in the respective phases. Neglecting excitations, and assuming NZ, PW, SW and LW are the only possible phases the phase diagram is shown in figure \ref{fig:PDH0new}. For $\alpha<1$ this is the same as the phase diagram in \cite{master}. For $\alpha>1$ the PZ phase has been replaced by the SW phase for nonzero SOC and the NZ phase for zero SOC. At zero SOC and $N^\uparrow = N^\downarrow$, the NZ phase is the only possible phase. For nonzero SOC $H_0^{\textrm{PW}}< H_0^{\textrm{SW}}$ for $\alpha<1$ and $H_0^{\textrm{PW}}> H_0^{\textrm{SW}}$ for $\alpha>1$, in agreement with \cite{SOCOLRev}. Also, as mentioned $H_0^{\textrm{LW}}> H_0^{\textrm{SW}}$ for $\alpha \geq 0$ meaning the LW phase does not enter the phase diagram when neglecting excitations. 

\begin{figure}
    \centering
    \includegraphics[width=0.7\linewidth]{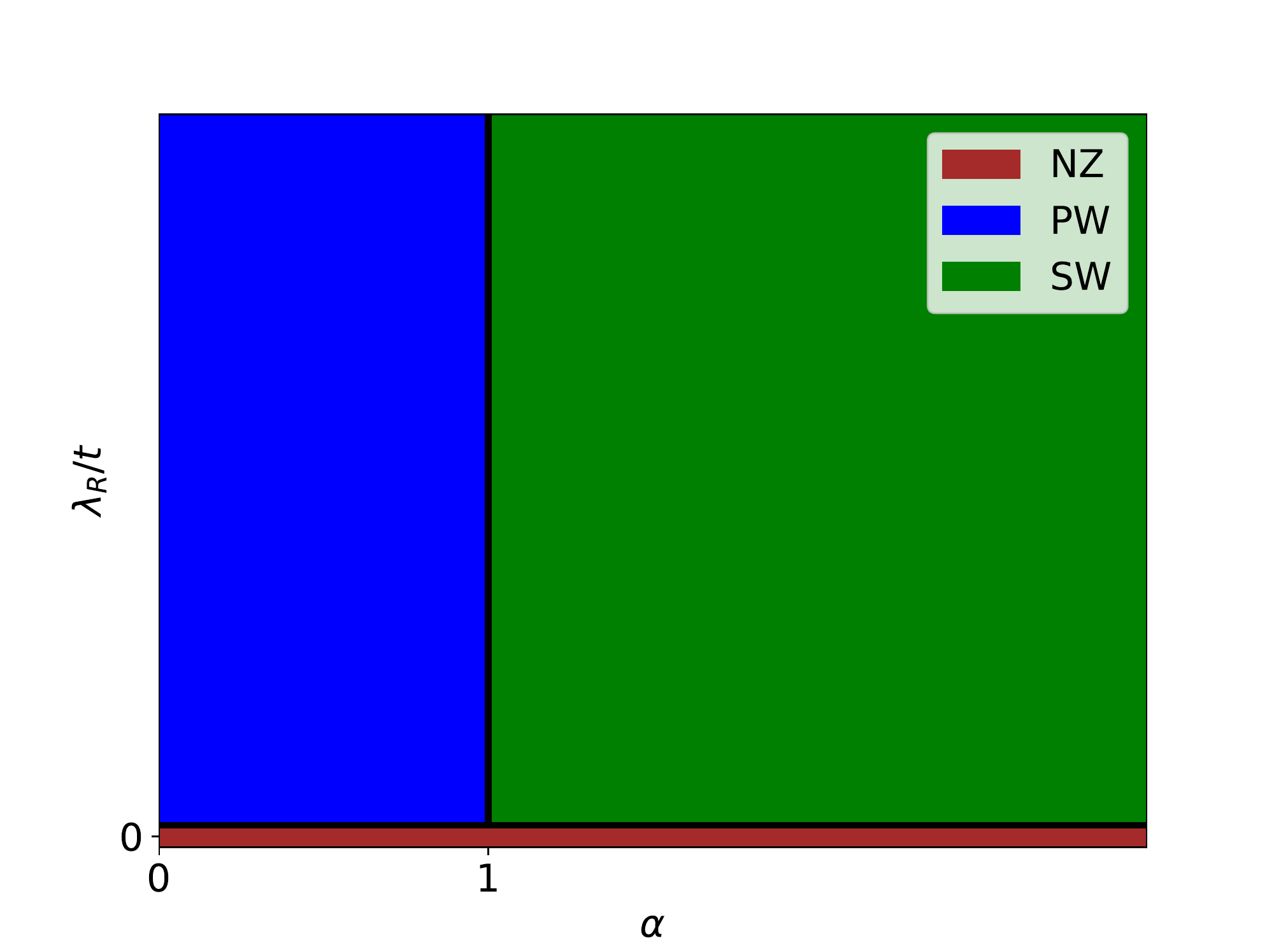}
    \caption[Phase diagram]{Phase diagram when neglecting excitations and $N^\uparrow = N^\downarrow$ is chosen. The area of the NZ phase is exaggerated. This phase only occurs for no SOC, i.e. $\lambda_R = 0$.}
    \label{fig:PDH0new}
\end{figure}



If we instead overlook the fact that the PZ phase requires different input parameters, it is found that PZ, NZ, PW and SW are the only possible phases when neglecting excitations. The phase diagram would be similar to the phase diagram in figure 4.2 of \cite{master} and that reported using numerical calculations in \cite{galteland2016competing}. The fundamental difference between the current approach and that used in \cite{master}, is that a transition from the SW to the PZ phase requires a change of the input parameters $N^\uparrow$ and $N^\downarrow$, not a change of $\alpha$ and $\lambda_R$. Hence, the PZ phase is removed when $N^\uparrow = N^\downarrow$. Also note that this difference in the approaches would vanish when considering the excitations. To make the PZ phase stable in the grand canonical ensemble would require $\mu^\uparrow \neq \mu^\downarrow$ which represents different input parameters than in the other phases, where $\mu^\uparrow = \mu^\downarrow$ is assumed \cite{master}. This was not shown in \cite{master}, but is analogous to the result that will be obtained in this thesis in the canonical ensemble, namely that stability of the PZ phase requires different energy offsets $T^\uparrow \neq T^\downarrow$.

\cleardoublepage

%% file: chapter4.tex

\chapter{Excitation Spectra and Critical Superfluid Velocity} \label{chap:Excitations} 

\input{1ResultsIntroduction.tex}

\input{2PZ.tex}

\input{3NZ.tex}
\input{4PW.tex}
\input{5SW.tex}

\input{7LWShort}

\cleardoublepage

%% file: 1ResultsIntroduction.tex
We have specialized to a 2D square optical lattice with lattice constant $a$. It will be assumed that $t^\uparrow = t^\downarrow = t$, $U^{\uparrow\uparrow}=U^{\downarrow\downarrow} = U$ and $U^{\uparrow\downarrow}=U^{\downarrow\uparrow} = \alpha U$. Then $\epsilon_{\boldsymbol{k}}^\uparrow =\epsilon_{\boldsymbol{k}}^\downarrow = \epsilon_{\boldsymbol{k}}$. The expressions for $\epsilon_{\boldsymbol{k}}^\alpha$ and $s_{\boldsymbol{k}}$ are given in \eqref{eq:epsk} and \eqref{eq:sk}. Apart from the PZ phase we will choose the input parameters such that $T^{\uparrow} = T^{\downarrow} = T$ and $N^\uparrow = N^\downarrow = N/2$. 

When setting up the phase diagram in figure \ref{fig:PDH0new}, we neglected elementary excitations, i.e. set $H\approx H_0$. The purpose of this chapter is to include elementary excitations to see if their effects change the conclusions in figure \ref{fig:PDH0new}. We will include $H_1$ and $H_2$ in the treatment, in order to obtain the quasiparticle excitation spectrum, the free energy and the critical superfluid velocity in the phases PZ, NZ, PW, SW and LW.

%% file: 2PZ.tex
\section{PZ Phase}

In the PZ phase only $\boldsymbol{k}_{00} = \boldsymbol{0}$ with spin up is occupied. Thus, $N_0^\uparrow = N_0$ and $N_0^\downarrow = 0$. We define $U_s$ by $2U_s = U^{\uparrow\uparrow}N^\uparrow/N_s = UN/N_s$. We decided to define $2U_s = UN/N_s$ such that we use the same $U_s$ in all phases. Furthermore, the term $4t+\epsilon_{\boldsymbol{k}}$ will appear often. This term varies between $0$ and $8t$, and we give it a new name to better the notation. In other words we define
\begin{align}
    U_s &\equiv \frac{UN}{2N_s} \mbox{\qquad and}\\
    \mathcal{E}_{\boldsymbol{k}} &\equiv \epsilon_{\boldsymbol{k}} - \epsilon_{\boldsymbol{0}} = 4t + \epsilon_{\boldsymbol{k}} = 4t -2t\left(\cos(k_x a) + \cos(k_y a)\right).
\end{align}
From \eqref{eq:H0} we find
\begin{equation}
    H_0^{''} = N_0(\epsilon_{\boldsymbol{0}}+T)+\frac{UN_0^2}{2N_s},
\end{equation}
where the double prime is used to separate it from quantities $H_0$ and $H'_0$ to be defined later. For the PZ phase, there is a problem with the assumption that the energy offset is equal for both pseudospin states, $T^{\uparrow}=T^{\downarrow} = T$, because such a choice does not agree with the assumption that all particles condense into the pseudospin up state. $H_0^{''}$ only depends on $T^{\uparrow}$ and we have already let $T^{\uparrow} \equiv T$. For now, we leave $T^{\downarrow}$ undetermined.

Inserting \eqref{eq:NN0excit} into $H_0^{''}$ we get
\begin{align}
    \begin{split}
    \label{eq:H0PZrewrite}
        H_0^{''} &= N(\epsilon_{\boldsymbol{0}}+T)+\frac{UN^2}{2N_s} \\
        &-(\epsilon_{\boldsymbol{0}}+T)\sum_{\boldsymbol{k}\neq \boldsymbol{0}}\left(A_{\boldsymbol{k}}^{\uparrow\dagger}A_{\boldsymbol{k}}^{\uparrow} + A_{\boldsymbol{k}}^{\downarrow\dagger}A_{\boldsymbol{k}}^{\downarrow}\right) -\frac{UN}{N_s}\sum_{\boldsymbol{k}\neq \boldsymbol{0}}\left(A_{\boldsymbol{k}}^{\uparrow\dagger}A_{\boldsymbol{k}}^{\uparrow} + A_{\boldsymbol{k}}^{\downarrow\dagger}A_{\boldsymbol{k}}^{\downarrow}\right).
    \end{split}
\end{align}
We define the first line as $H_0$ and move the second line to $H_2$ as it is quadratic in excitation operators. Note that terms more than quadratic in excitation operators, or equivalently of order less than $N/N_s$ have been neglected to the same order of approximation as the MFT Hamiltonian \eqref{eq:MFTH}. The Kronecker delta in \eqref{eq:H1} gives for the PZ phase $\boldsymbol{k}=\boldsymbol{k}_{00}=\boldsymbol{0}$ which is excluded from the sum, so $H_1 = 0$. 

In $H_2$ we may replace $N_0$ by $N$ directly to the same order of approximation \cite{Pitaevskii}. Writing out the sums in \eqref{eq:H2} and including the contribution from \eqref{eq:H0PZrewrite} yields
\begin{align}
    \begin{split}
    \label{eq:H2PZfirst}
        H_2 = \sum_{\boldsymbol{k}\neq \boldsymbol{0}} \Bigg( &(\mathcal{E}_{\boldsymbol{k}} +2U_s)A_{\boldsymbol{k}}^{\uparrow\dagger}A_{\boldsymbol{k}}^{\uparrow} +(\mathcal{E}_{\boldsymbol{k}}+2\Delta) A_{\boldsymbol{k}}^{\downarrow\dagger}A_{\boldsymbol{k}}^{\downarrow} \\
        &+s_{\boldsymbol{k}}A_{\boldsymbol{k}}^{\uparrow\dagger}A_{\boldsymbol{k}}^{\downarrow} + s_{\boldsymbol{k}}^{*}A_{\boldsymbol{k}}^{\downarrow\dagger}A_{\boldsymbol{k}}^{\uparrow} \\
        &+U_s e^{i2\theta_{\boldsymbol{0}}^{\uparrow}}A_{\boldsymbol{k}}^{\uparrow}A_{\boldsymbol{-k}}^{\uparrow} + U_s e^{-i2\theta_{\boldsymbol{0}}^{\uparrow}}A_{\boldsymbol{-k}}^{\uparrow\dagger}A_{\boldsymbol{k}}^{\uparrow\dagger}\Bigg).
    \end{split}
\end{align}
A new quantity $\Delta$ has been defined. The coefficient of $A_{\boldsymbol{k}}^{\downarrow\dagger}A_{\boldsymbol{k}}^{\downarrow}$ is $\mathcal{E}_{\boldsymbol{k}} +T^{\downarrow}-T +2U_s(\alpha-1) \equiv \mathcal{E}_{\boldsymbol{k}}+2\Delta$. Hence, 
\begin{equation}
    \label{eq:Delta}
    2\Delta \equiv T^{\downarrow}-T+ 2U_s(\alpha-1).
\end{equation}
For the assumption that all particles condense at pseudospin up to make sense, the coefficient of $A_{\boldsymbol{k}}^{\uparrow\dagger}A_{\boldsymbol{k}}^{\uparrow}$ should be lower than the coefficient of $A_{\boldsymbol{k}}^{\downarrow\dagger}A_{\boldsymbol{k}}^{\downarrow}$. I.e. we require $\Delta > U_s$, and in terms of the input parameter $T^{\downarrow}$ this requirement is
\begin{equation}
    T^{\downarrow} > T + 2U_s(2-\alpha).
\end{equation}
A weakly interacting two-component BEC without SOC was treated by Linder and Sudbø in \cite{LS}. It is apparent that in \cite{LS} $\Delta=0$ if one sets $n_A = N_0/N_s$ and $n_B=0$, i.e. try to use their results in the PZ phase. As we have just argued, this means the PZ phase does not make sense and it could never be stable. It also means that if we want to compare our PZ phase results regarding excitation spectra and critical superfluid velocity to the results in \cite{LS}, we should set both $s_{\boldsymbol{k}}=0$ and $\Delta=0$.


With the aim of BV diagonalizing the Hamiltonian, we define the operator vectors
\begin{align}
\begin{split}
    \boldsymbol{A}_{\boldsymbol{k}} &= (A_{\boldsymbol{k}}^{\uparrow},A_{-\boldsymbol{k}}^{\uparrow}, A_{\boldsymbol{k}}^{\downarrow},A_{-\boldsymbol{k}}^{\downarrow},A_{\boldsymbol{k}}^{\uparrow\dagger},A_{-\boldsymbol{k}}^{\uparrow\dagger}, A_{\boldsymbol{k}}^{\downarrow\dagger},A_{-\boldsymbol{k}}^{\downarrow\dagger})^T \mbox{\qquad and}\\
    \boldsymbol{A}_{\boldsymbol{k}}^{\dagger} &= (A_{\boldsymbol{k}}^{\uparrow\dagger},A_{-\boldsymbol{k}}^{\uparrow\dagger}, A_{\boldsymbol{k}}^{\downarrow\dagger},A_{-\boldsymbol{k}}^{\downarrow\dagger}, A_{\boldsymbol{k}}^{\uparrow},A_{-\boldsymbol{k}}^{\uparrow}, A_{\boldsymbol{k}}^{\downarrow},A_{-\boldsymbol{k}}^{\downarrow}).
\end{split}
\end{align} 
These satisfy the commutator $\boldsymbol{A}_{\boldsymbol{k}}\otimes \boldsymbol{A}_{\boldsymbol{k}}^{\dagger} -((\boldsymbol{A}_{\boldsymbol{k}}^{\dagger})^T\otimes(\boldsymbol{A}_{\boldsymbol{k}})^T)^T=J$ when $\boldsymbol{k} \neq \boldsymbol{0}$. We can now write 
\begin{equation}
    H = H_0 + \sum_{\boldsymbol{k} \neq \boldsymbol{0}} \boldsymbol{A}_{\boldsymbol{k}}^{\dagger}M_{\boldsymbol{k}}\boldsymbol{A}_{\boldsymbol{k}},
\end{equation}
where $M_{\boldsymbol{k}}$ is an $8\cross8$ matrix that should be written on the form
\begin{gather}
\label{eq:Mform}
    M_{\boldsymbol{k}} = 
    \begin{pmatrix}
    M_1 & M_2 \\
   M_2^* & M_1^*
    \end{pmatrix},
\end{gather}
where $M_1^\dagger = M_1$, $M_2^T=M_2$ and we have suppressed the $\boldsymbol{k}$-dependence of the submatrices in the notation. We can do this with our $H_2$ by using the commutation relations, which give $A_{\boldsymbol{k}}^{\alpha\dagger}A_{\boldsymbol{k}}^{\alpha} = (A_{\boldsymbol{k}}^{\alpha\dagger}A_{\boldsymbol{k}}^{\alpha} + A_{\boldsymbol{k}}^{\alpha}A_{\boldsymbol{k}}^{\alpha\dagger} - 1)/2$ and for commuting operators simply e.g. $A_{\boldsymbol{k}}^{\alpha}A_{\boldsymbol{-k}}^{\alpha} = (A_{\boldsymbol{k}}^{\alpha}A_{\boldsymbol{-k}}^{\alpha} + A_{\boldsymbol{-k}}^{\alpha}A_{\boldsymbol{k}}^{\alpha})/2$. Note that this simultaneously shifts $H_0$,
\begin{align}
    \begin{split}
        H'_0&=H_0  - \frac{1}{2}\sum_{\boldsymbol{k} \neq \boldsymbol{0}}\left(2\mathcal{E}_{\boldsymbol{k}}+2U_s +2\Delta \right) \\
        &= H_0  - \sum_{\boldsymbol{k} \neq \boldsymbol{0}}\left(\mathcal{E}_{\boldsymbol{k}}+U_s+\Delta \right).
    \end{split}
\end{align}
This is a quantum mechanical correction to the ground state because it stems from a commutator. Remembering that we have $N_s$ lattice sites, we note that there are $N_s-1$ different $\boldsymbol{k}$ in the sum. Thus, the sum over $\boldsymbol{k}$ can be computed for the $\boldsymbol{k}$ independent parts. Firstly,
\begin{equation}
    \sum_{\boldsymbol{k} \neq \boldsymbol{0}}\left(4t  +U_s +\Delta \right) = (4t  +U_s +\Delta)\sum_{\boldsymbol{k} \neq \boldsymbol{0}} 1  = (N_s-1)(4t+U_s+\Delta).
\end{equation}
Secondly,
\begin{equation}
    \sum_{\boldsymbol{k} \neq \boldsymbol{0}} \epsilon_{\boldsymbol{k}} = \sum_{\boldsymbol{k} } \epsilon_{\boldsymbol{k}} - \epsilon_{\boldsymbol{0}} = \sum_{\boldsymbol{k} } \epsilon_{\boldsymbol{k}} +4t. 
\end{equation}
Here,
\begin{equation}
    \sum_{\boldsymbol{k} } \epsilon_{\boldsymbol{k}} = -2t \sum_{k_x k_y} \left( \cos(k_x a) + \cos(k_y a) \right) = -4t  \sum_{k_x k_y} \cos(k_x a).
\end{equation}
The possible $\boldsymbol{k}$ are equally distributed in the first Brillouin zone (1BZ), i.e. $-\pi \leq k_x a < \pi$, $-\pi \leq k_y a < \pi$. Looking at the form of $\epsilon_{\boldsymbol{k}}$ this means the sum $\sum_{\boldsymbol{k}} \epsilon_{\boldsymbol{k}}$ has to be zero, similar to how $\int_{-\pi}^\pi \cos(x) dx = 0$. Thus,
\begin{equation}
\label{eq:H'0PZ}
    H'_0 = H_0 -4tN_s-(N_s-1)(U_s+\Delta).
\end{equation}

We also use 
\begin{equation}
    \sum_{\boldsymbol{k}}C(\boldsymbol{k})A_{\boldsymbol{k}}^{\uparrow\dagger}A_{\boldsymbol{k}}^{\uparrow} = \sum_{\boldsymbol{k}}\frac{1}{2}\left( C(\boldsymbol{k})A_{\boldsymbol{k}}^{\uparrow\dagger}A_{\boldsymbol{k}}^{\uparrow}+C(-\boldsymbol{k})A_{-\boldsymbol{k}}^{\uparrow\dagger}A_{-\boldsymbol{k}}^{\uparrow} \right)
\end{equation}
and similar relations to rewrite $H_2$, simultaneously applying the relations $\epsilon_{-\boldsymbol{k}} = \epsilon_{\boldsymbol{k}} \iff \mathcal{E}_{-\boldsymbol{k}} = \mathcal{E}_{\boldsymbol{k}}$ and $s_{-\boldsymbol{k}} = -s_{\boldsymbol{k}}$. Starting from \eqref{eq:H2PZfirst}
we find
\begin{align}
    \begin{split}
        H_2 = \sum_{\boldsymbol{k} \neq \boldsymbol{0}} \Bigg( &\frac{\mathcal{E}_{\boldsymbol{k}} +2U_s}{4}(A_{\boldsymbol{k}}^{\uparrow\dagger}A_{\boldsymbol{k}}^{\uparrow} + A_{\boldsymbol{k}}^{\uparrow} A_{\boldsymbol{k}}^{\uparrow\dagger} + A_{-\boldsymbol{k}}^{\uparrow\dagger}A_{-\boldsymbol{k}}^{\uparrow} + A_{-\boldsymbol{k}}^{\uparrow} A_{-\boldsymbol{k}}^{\uparrow\dagger}) \\
        &+\frac{\mathcal{E}_{\boldsymbol{k}}+2\Delta}{4}(A_{\boldsymbol{k}}^{\downarrow\dagger}A_{\boldsymbol{k}}^{\downarrow}+ A_{\boldsymbol{k}}^{\downarrow} A_{\boldsymbol{k}}^{\downarrow\dagger}+ A_{-\boldsymbol{k}}^{\downarrow\dagger}A_{-\boldsymbol{k}}^{\downarrow}+ A_{-\boldsymbol{k}}^{\downarrow} A_{-\boldsymbol{k}}^{\downarrow\dagger})\\
        &+\frac{s_{\boldsymbol{k}}}{4}(A_{\boldsymbol{k}}^{\uparrow\dagger}A_{\boldsymbol{k}}^{\downarrow} +A_{\boldsymbol{k}}^{\downarrow}A_{\boldsymbol{k}}^{\uparrow\dagger} - A_{-\boldsymbol{k}}^{\uparrow\dagger}A_{-\boldsymbol{k}}^{\downarrow} -A_{-\boldsymbol{k}}^{\downarrow}A_{-\boldsymbol{k}}^{\uparrow\dagger} )\\ &+\frac{s_{\boldsymbol{k}}^{*}}{4}(A_{\boldsymbol{k}}^{\downarrow\dagger}A_{\boldsymbol{k}}^{\uparrow}+ A_{\boldsymbol{k}}^{\uparrow}A_{\boldsymbol{k}}^{\downarrow\dagger} - A_{-\boldsymbol{k}}^{\downarrow\dagger}A_{-\boldsymbol{k}}^{\uparrow}- A_{-\boldsymbol{k}}^{\uparrow}A_{-\boldsymbol{k}}^{\downarrow\dagger}) \\
        &+\frac{2U_s}{4}e^{i2\theta_{\boldsymbol{0}}^{\uparrow}}(A_{\boldsymbol{k}}^{\uparrow}A_{\boldsymbol{-k}}^{\uparrow}+A_{\boldsymbol{-k}}^{\uparrow}A_{\boldsymbol{k}}^{\uparrow})\\
        &+ \frac{2U_s}{4}e^{-i2\theta_{\boldsymbol{0}}^{\uparrow}}(A_{\boldsymbol{-k}}^{\uparrow\dagger}A_{\boldsymbol{k}}^{\uparrow\dagger}+ A_{\boldsymbol{k}}^{\uparrow\dagger}A_{\boldsymbol{-k}}^{\uparrow\dagger})\Bigg).
    \end{split}
\end{align}
Moving the factor $1/4$ outside the sum, we get
\begin{equation}
    H = H'_0 + \frac{1}{4}\sum_{\boldsymbol{k} \neq \boldsymbol{0}} \boldsymbol{A}_{\boldsymbol{k}}^{\dagger}M_{\boldsymbol{k}}\boldsymbol{A}_{\boldsymbol{k}}.
\end{equation}
Here, $M_{\boldsymbol{k}}$ is
\begin{gather}
\label{eq:PZM}
    M_{\boldsymbol{k}} = 
    \begin{pmatrix}
    M_1 & M_2 \\
    M_2^* & M_1^* \\
    \end{pmatrix},
\end{gather}
with
\begin{gather*}
    M_1 = 
    \begin{pmatrix}
    M_{11}(\boldsymbol{k}) & 0 & s_{\boldsymbol{k}} & 0 \\
    0 & M_{11}(\boldsymbol{k}) & 0 & -s_{\boldsymbol{k}} \\
    s_{\boldsymbol{k}}^{*} & 0 & M_{33}(\boldsymbol{k}) & 0 \\
    0 & -s_{\boldsymbol{k}}^{*} & 0 & M_{33}(\boldsymbol{k}) \\
    \end{pmatrix}
\end{gather*}
and
\begin{gather*}
    M_2^* = 
    \begin{pmatrix}
    0 & M_{52} & 0 & 0 \\
    M_{52} & 0 & 0 & 0 \\
    0 & 0 & 0 & 0 \\
    0 & 0 & 0 & 0 \\
    \end{pmatrix}.
\end{gather*}
The matrix elements are
\begin{align}
\begin{split}
    M_{11}(\boldsymbol{k}) =  \mathcal{E}_{\boldsymbol{k}} +2U_s, \mbox{\qquad} M_{33}(\boldsymbol{k}) = \mathcal{E}_{\boldsymbol{k}} + 2\Delta \mbox{\qquad and \qquad} M_{52} = 2U_s e^{i2\theta_{\boldsymbol{0}}^{\uparrow}}.
\end{split}
\end{align}

\subsection{Excitation Spectrum}
We want to find eigenvalues of
\begin{gather*}
    M_{\boldsymbol{k}}J = 
    \begin{pmatrix}
    M_1 & -M_2 \\
   M_2^* & -M_1^* \\
    \end{pmatrix},
\end{gather*}
i.e. all solutions $\lambda$ of  $\det(M_{\boldsymbol{k}}J-\lambda I) = 0$. Analytic eigenvalues are in this thesis calculated using the symbolic computing environment \textit{Maple}. This yields the four double eigenvalues $\lambda(\boldsymbol{k}) = \pm \Omega_\pm(\boldsymbol{k})$, with
\begin{equation}
\label{eq:PZeigen}
    \Omega_{\pm}(\boldsymbol{k}) = \sqrt{C_{1\boldsymbol{k}} \pm 2\sqrt{C_{2\boldsymbol{k}}}},
\end{equation}
where we have defined (suppressing the $\boldsymbol{k}$ dependence of $M_{11}$ and $M_{33}$ in the notation)
\begin{align}
\begin{split}
    2C_{1\boldsymbol{k}} &= 2\abs{s_{\boldsymbol{k}}}^2 -\abs{M_{52}}^2 +M_{11}^2 + M_{33}^2  \mbox{\quad and}\\
    16C_{2\boldsymbol{k}} &= 4\big(M_{11} + M_{33} + \abs{M_{52}}\big)\big(M_{11} + M_{33} - \abs{M_{52}}\big)\abs{s_{\boldsymbol{k}}}^2 \\
    &+\big(M_{11}^2 - M_{33}^2 - \abs{M_{52}}^2\big)^2.
\end{split}
\end{align}
More explicitly this is
\begin{align}
\begin{split}
    C_{1\boldsymbol{k}} = & \abs{s_{\boldsymbol{k}}}^2 +\mathcal{E}_{\boldsymbol{k}}^2 +2(U_s+\Delta)\mathcal{E}_{\boldsymbol{k}} + 2\Delta^2   \mbox{\quad and}\\
    C_{2\boldsymbol{k}} = &\abs{s_{\boldsymbol{k}}}^2\big( \mathcal{E}_{\boldsymbol{k}}^2 +2(U_s+\Delta)\mathcal{E}_{\boldsymbol{k}} +\Delta^2+2U_s\Delta \big)  \\
    & +(\Delta-U_s)^2\mathcal{E}_{\boldsymbol{k}}^2 +2\Delta^2(\Delta-U_s)\mathcal{E}_{\boldsymbol{k}} +\Delta^4.
\end{split}
\end{align}
These eigenvalues satisfy $\Omega_{+}(\boldsymbol{k}=\boldsymbol{0}) = 2\Delta$ and $\Omega_-(\boldsymbol{k}=\boldsymbol{0}) = 0$.

We arrive at 
\begin{equation}
    D_{\boldsymbol{k}} = \textrm{diag}(\Omega_+, \Omega_+, \Omega_-,\Omega_-, \Omega_+, \Omega_+, \Omega_-,\Omega_-),
\end{equation}
and
\begin{equation}
    H = H'_0 + \frac{1}{4}\sum_{\boldsymbol{k} \neq \boldsymbol{0}}\boldsymbol{B}_{\boldsymbol{k}}^{\dagger}D_{\boldsymbol{k}}\boldsymbol{B}_{\boldsymbol{k}}.
\end{equation}
The new operators,
$$\boldsymbol{B}_{\boldsymbol{k}} = (B_{\boldsymbol{k},1}, B_{\boldsymbol{k},3}, B_{\boldsymbol{k},2}, B_{\boldsymbol{k},4}, B_{\boldsymbol{k},1}^{\dagger}, B_{\boldsymbol{k},3}^{\dagger}, B_{\boldsymbol{k},2}^{\dagger},B_{\boldsymbol{k},4}^{\dagger})^T,$$
are defined by the transformation matrix, $T_{\boldsymbol{k}}$, as $\boldsymbol{B}_{\boldsymbol{k}} = T_{\boldsymbol{k}}^\dagger \boldsymbol{A}_{\boldsymbol{k}}$. It is clear that the new operators are defined as linear combinations of the old, where the coefficients are given by the complex conjugate of the eigenvectors of $M_{\boldsymbol{k}}J$. Thus, $B_{\boldsymbol{k},1}$ and $B_{\boldsymbol{k},3}$ are defined using the eigenvectors of the largest eigenvalue. Investigating the transformation matrix $T_{\boldsymbol{k}}$ numerically at several trial momenta $\boldsymbol{k}_t$, we are able to confirm a relation $B_{-\boldsymbol{k},3} = B_{\boldsymbol{k},1}$ and similarly $B_{-\boldsymbol{k},4} = B_{\boldsymbol{k},2}$. 
This is not a rigorous proof, but we feel confident the transformation matrix can be set up in such a way that this holds for any $\boldsymbol{k}_t$. In fact, we can give a more analytic argument for why this should be true. Let us look at the equations used to find the eigenvectors for the eigenvalues $\Omega_+(\boldsymbol{k})$ and $\Omega_+(-\boldsymbol{k}) = \Omega_+(\boldsymbol{k})$. Let $\boldsymbol{x} = (x_1, \dots, x_8)^T$ be a general $8\cross 1$ column vector. The equation $M_{\boldsymbol{k}}J\boldsymbol{x} = \Omega_+(\boldsymbol{k})\boldsymbol{x}$ can be used to determine the eigenvectors and it gives
\begin{align}
    \begin{split}
    \label{eq:PZeigenvec+}
        M_{11}(\boldsymbol{k})x_1 +s_{\boldsymbol{k}}x_3 -M_{52}^* x_6 &= \Omega_+(\boldsymbol{k})x_1,\\
        M_{11}(-\boldsymbol{k})x_2+s_{-\boldsymbol{k}}x_4 -M_{52}^* x_5 &= \Omega_+(\boldsymbol{k})x_2,\\
        s_{\boldsymbol{k}}^* x_1 + M_{33}(\boldsymbol{k})x_3 &= \Omega_+(\boldsymbol{k})x_3, \\
        s_{-\boldsymbol{k}}^* x_2 + M_{33}(-\boldsymbol{k})x_4 &= \Omega_+(\boldsymbol{k})x_4, \\
        M_{52}x_2 -M_{11}(\boldsymbol{k})x_5-s_{\boldsymbol{k}}^* x_7 &= \Omega_+(\boldsymbol{k})x_5, \\
        M_{52}x_1 -M_{11}(-\boldsymbol{k})x_6-s_{-\boldsymbol{k}}^* x_8&= \Omega_+(\boldsymbol{k})x_6, \\
        -s_{\boldsymbol{k}} x_5 - M_{33}(\boldsymbol{k})x_7&= \Omega_+(\boldsymbol{k})x_7, \\
        -s_{-\boldsymbol{k}} x_6 - M_{33}(-\boldsymbol{k})x_8&= \Omega_+(\boldsymbol{k})x_8 \\.
    \end{split}
\end{align}
Meanwhile, the equation $M_{-\boldsymbol{k}}J\boldsymbol{x} = \Omega_+(-\boldsymbol{k})\boldsymbol{x}=\Omega_+(\boldsymbol{k})\boldsymbol{x}$ determines the eigenvectors at $-\boldsymbol{k}$. It gives
\begin{align}
    \begin{split}
        M_{11}(-\boldsymbol{k})x_1 +s_{-\boldsymbol{k}}x_3 -M_{52}^* x_6 &= \Omega_+(\boldsymbol{k})x_1,\\
        M_{11}(\boldsymbol{k})x_2+s_{\boldsymbol{k}}x_4 -M_{52}^* x_5 &= \Omega_+(\boldsymbol{k})x_2,\\
        s_{-\boldsymbol{k}}^* x_1 + M_{33}(-\boldsymbol{k})x_3 &= \Omega_+(\boldsymbol{k})x_3, \\
        s_{\boldsymbol{k}}^* x_2 + M_{33}(\boldsymbol{k})x_4 &= \Omega_+(\boldsymbol{k})x_4, \\
        M_{52}x_2 -M_{11}(-\boldsymbol{k})x_5-s_{-\boldsymbol{k}}^* x_7 &= \Omega_+(\boldsymbol{k})x_5, \\
        M_{52}x_1 -M_{11}(\boldsymbol{k})x_6-s_{\boldsymbol{k}}^* x_8&= \Omega_+(\boldsymbol{k})x_6, \\
        -s_{-\boldsymbol{k}} x_5 - M_{33}(-\boldsymbol{k})x_7&= \Omega_+(\boldsymbol{k})x_7, \\
        -s_{\boldsymbol{k}} x_6 - M_{33}(\boldsymbol{k})x_8&= \Omega_+(\boldsymbol{k})x_8 \\.
    \end{split}
\end{align}
We recognize that these sets of equations are the same, apart from an interchange $x_{2i-1} \leftrightarrow x_{2i}, i=1,2,3,4$. If we investigate the basis at $+\boldsymbol{k}$ and $-\boldsymbol{k}$ we see the same interchange:
\begin{align}
    \boldsymbol{A}_{\boldsymbol{k}} &= (A_{\boldsymbol{k}}^{\uparrow},A_{-\boldsymbol{k}}^{\uparrow}, A_{\boldsymbol{k}}^{\downarrow},A_{-\boldsymbol{k}}^{\downarrow},A_{\boldsymbol{k}}^{\uparrow\dagger},A_{-\boldsymbol{k}}^{\uparrow\dagger}, A_{\boldsymbol{k}}^{\downarrow\dagger},A_{-\boldsymbol{k}}^{\downarrow\dagger})^T \mbox{\qquad and}\\
    \boldsymbol{A}_{-\boldsymbol{k}} &= (A_{-\boldsymbol{k}}^{\uparrow},A_{\boldsymbol{k}}^{\uparrow}, A_{-\boldsymbol{k}}^{\downarrow},A_{\boldsymbol{k}}^{\downarrow},A_{-\boldsymbol{k}}^{\uparrow\dagger},A_{\boldsymbol{k}}^{\uparrow\dagger}, A_{-\boldsymbol{k}}^{\downarrow\dagger},A_{\boldsymbol{k}}^{\downarrow\dagger})^T .
\end{align}
Now, imagine we have found a set of two orthonormal eigenvectors from the set of equations in \eqref{eq:PZeigenvec+} that can be used in diagonalizing $M_{\boldsymbol{k}}J$. Then, in the case of $-\boldsymbol{k}$, we can choose the same eigenvectors with an interchange $x_{2i-1} \leftrightarrow x_{2i}, i=1,2,3,4$, and we are free to choose the opposite order of the eigenvectors. These can then be used in diagonalizing $M_{-\boldsymbol{k}}J$. Thus it is clear that row 2(1) of $T_{-\boldsymbol{k}}^\dagger$ will be the same as row 1(2) of $T_{\boldsymbol{k}}^\dagger$ apart from the interchange $x_{2i-1} \leftrightarrow x_{2i}, i=1,2,3,4$. This shows that $B_{-\boldsymbol{k},3} = B_{\boldsymbol{k},1}$ and $B_{-\boldsymbol{k},1} = B_{\boldsymbol{k},3}$. Similar arguments could be used to argue that the operators corresponding to $\Omega_-(\boldsymbol{k})$ obey $B_{-\boldsymbol{k},4} = B_{\boldsymbol{k},2}$. We will encounter similar relations in the other phases, and refer back to this argument as a method to support the relations between operators we find.


We now have the tools to simplify the diagonalized version of $H_2$, 
\begin{align}
    \begin{split}
        H_2 &= \frac{1}{4}\sum_{\boldsymbol{k} \neq \boldsymbol{0}}\big(\Omega_+(\boldsymbol{k}) B_{\boldsymbol{k},1}^{\dagger} B_{\boldsymbol{k},1} + \Omega_+(\boldsymbol{k}) B_{\boldsymbol{k},3}^{\dagger} B_{\boldsymbol{k},3} \\
        &\mbox{\qquad\qquad}+ \Omega_-(\boldsymbol{k}) B_{\boldsymbol{k},2}^{\dagger} B_{\boldsymbol{k},2} + \Omega_-(\boldsymbol{k}) B_{\boldsymbol{k},4}^{\dagger} B_{\boldsymbol{k},4} \\
        &\mbox{\qquad\qquad} + \Omega_+(\boldsymbol{k}) B_{\boldsymbol{k},1} B_{\boldsymbol{k},1}^{\dagger} + \Omega_+(\boldsymbol{k}) B_{\boldsymbol{k},3} B_{\boldsymbol{k},3}^{\dagger} \\
        &\mbox{\qquad\qquad}+ \Omega_-(\boldsymbol{k}) B_{\boldsymbol{k},2} B_{\boldsymbol{k},2}^{\dagger} + \Omega_-(\boldsymbol{k}) B_{\boldsymbol{k},4} B_{\boldsymbol{k},4}^{\dagger} \big) \\
        &=\frac{1}{2}\sum_{\boldsymbol{k} \neq \boldsymbol{0}}\bigg(\Omega_+(\boldsymbol{k}) \left(B_{\boldsymbol{k},1}^{\dagger} B_{\boldsymbol{k},1}+\frac12\right) + \Omega_+(-\boldsymbol{k}) \left(B_{-\boldsymbol{k},3}^{\dagger} B_{-\boldsymbol{k},3}+\frac12\right) \\
        &\mbox{\qquad\qquad}+ \Omega_-(\boldsymbol{k}) \left(B_{\boldsymbol{k},2}^{\dagger} B_{\boldsymbol{k},2}+\frac12\right) + \Omega_-(-\boldsymbol{k}) \left(B_{-\boldsymbol{k},4}^{\dagger} B_{-\boldsymbol{k},4}+\frac12\right)\bigg)\\
        &= \sum_{\boldsymbol{k} \neq \boldsymbol{0}}\bigg(\Omega_+(\boldsymbol{k}) \left(B_{\boldsymbol{k},1}^{\dagger} B_{\boldsymbol{k},1}+\frac12\right) + \Omega_-(\boldsymbol{k}) \left(B_{\boldsymbol{k},2}^{\dagger} B_{\boldsymbol{k},2}+\frac12\right) \bigg)\\
        &= \sum_{\boldsymbol{k} \neq \boldsymbol{0}}\sum_{\sigma=1}^{2} \Omega_\sigma(\boldsymbol{k}) \left(B_{\boldsymbol{k},\sigma}^{\dagger} B_{\boldsymbol{k},\sigma}+\frac12\right),
    \end{split}
\end{align}
where we defined $\Omega_1(\boldsymbol{k}) \equiv \Omega_+(\boldsymbol{k})$ and $\Omega_2(\boldsymbol{k}) \equiv \Omega_-(\boldsymbol{k})$. We let $\boldsymbol{k} \to -\boldsymbol{k}$ in some terms of the sum. Then we used that $\Omega_\pm(\boldsymbol{k})$ are inversion symmetric in $\boldsymbol{k}$, along with the relations between the new operators, to identify that some terms are equal in the fifth and sixth lines. The diagonal Hamiltonian is
\begin{equation}
\label{eq:HPZkneq0}
    H = H'_0 + \sum_{\boldsymbol{k} \neq \boldsymbol{0}}\sum_{\sigma=1}^{2} \Omega_\sigma(\boldsymbol{k}) \left(  B_{\boldsymbol{k}, \sigma}^{\dagger}B_{\boldsymbol{k}, \sigma} + \frac{1}{2} \right).
\end{equation}


Remember that the eigenvalues need to be real for the diagonalization procedure to be defined, and for the system to be stable. In other words, the PZ phase is only stable as long as the eigenvalues of $M_{\boldsymbol{k}}J$ are real. Numerical investigations suggest that the occurance of complex eigenvalues happens for small $\boldsymbol{k}$. Therefore an expansion for small $\boldsymbol{k}$ should yield a criterion for $\Omega_\pm(\boldsymbol{k})  \in \mathbb{R}$. As this is essentially what we are doing when calculating the critical superfluid velocity, we expect that this will be the same as the requirement for real critical superfluid velocity.  Later, when we find the critical superfluid velocity, we obtain a clear requirement on $\lambda_R/t$, which turns out to be $\lambda_R^2/t^2 \leq \Delta/2t$. We can think of the term $\Delta>U_s$ which is connected to the difference between the energy offsets for the two pseudospin states, as an analogue to a Zeeman splitting. It is this Zeeman splitting that can make the spectra real in the presence of SOC. Without Zeeman splitting, as in the NZ phase, it will not be possible to obtain real spectra for a condensed phase at $\boldsymbol{k}=\boldsymbol{0}$ with SOC. The reason is that any nonzero SOC will yield nonzero condensate momenta when there is no Zeeman splitting, as was found in chapter \ref{sec:SOCU0}.  

Figure \ref{fig:PZeigenvalues} shows an example of how the eigenvalues behave in the 1BZ, while figure \ref{fig:PZband} shows the band structure. We notice that both bands have their minimum at $\boldsymbol{k} = \boldsymbol{0}$ and that the lowest eigenvalue is linear close to the minimum. The parameters are chosen such that $U/t = 0.1$ and the average filling of particles per site is $N/N_s = 1$. Hence, $U_s/t = 0.05$ is used. 
\begin{figure}
    \centering
    \includegraphics[width=0.8\linewidth]{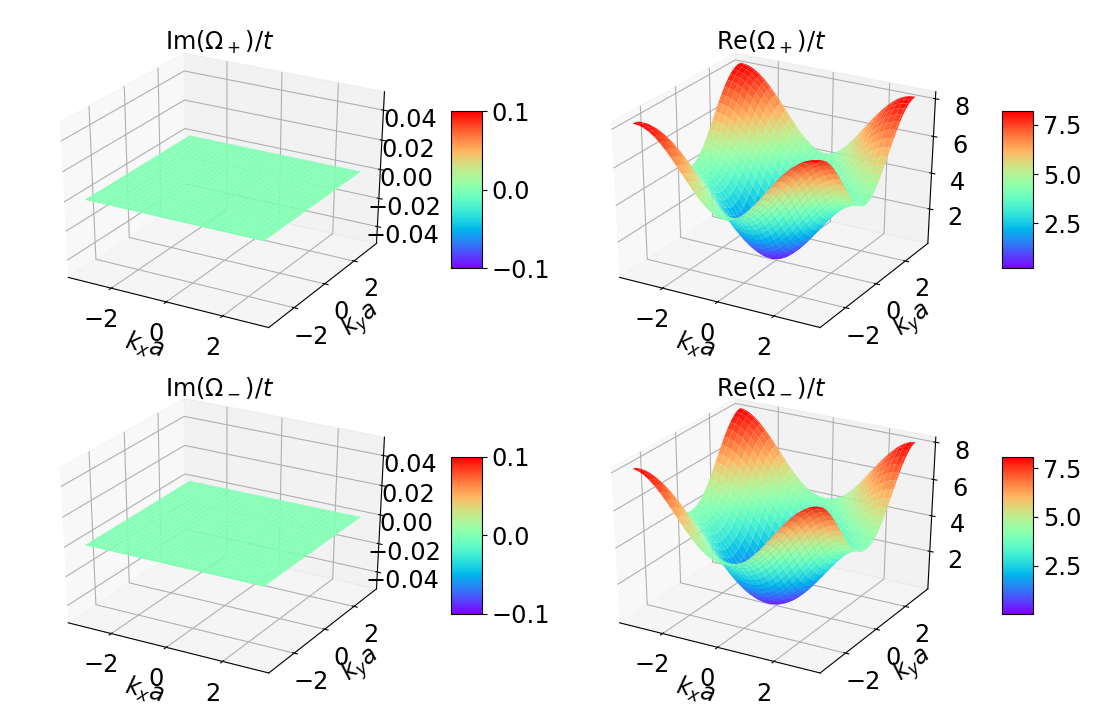}
    \caption[PZ phase excitation spectrum]{Shows real and imaginary parts of $\Omega_\pm(\boldsymbol{k})$ for $\Delta = 2U_s$, $U_s/t = 0.05$ and $\lambda_R/t = 0.1$, a set of parameters that render the eigenvalues real. }
    \label{fig:PZeigenvalues}
\end{figure}

\begin{figure}
    \centering
    \includegraphics[width=0.7\linewidth]{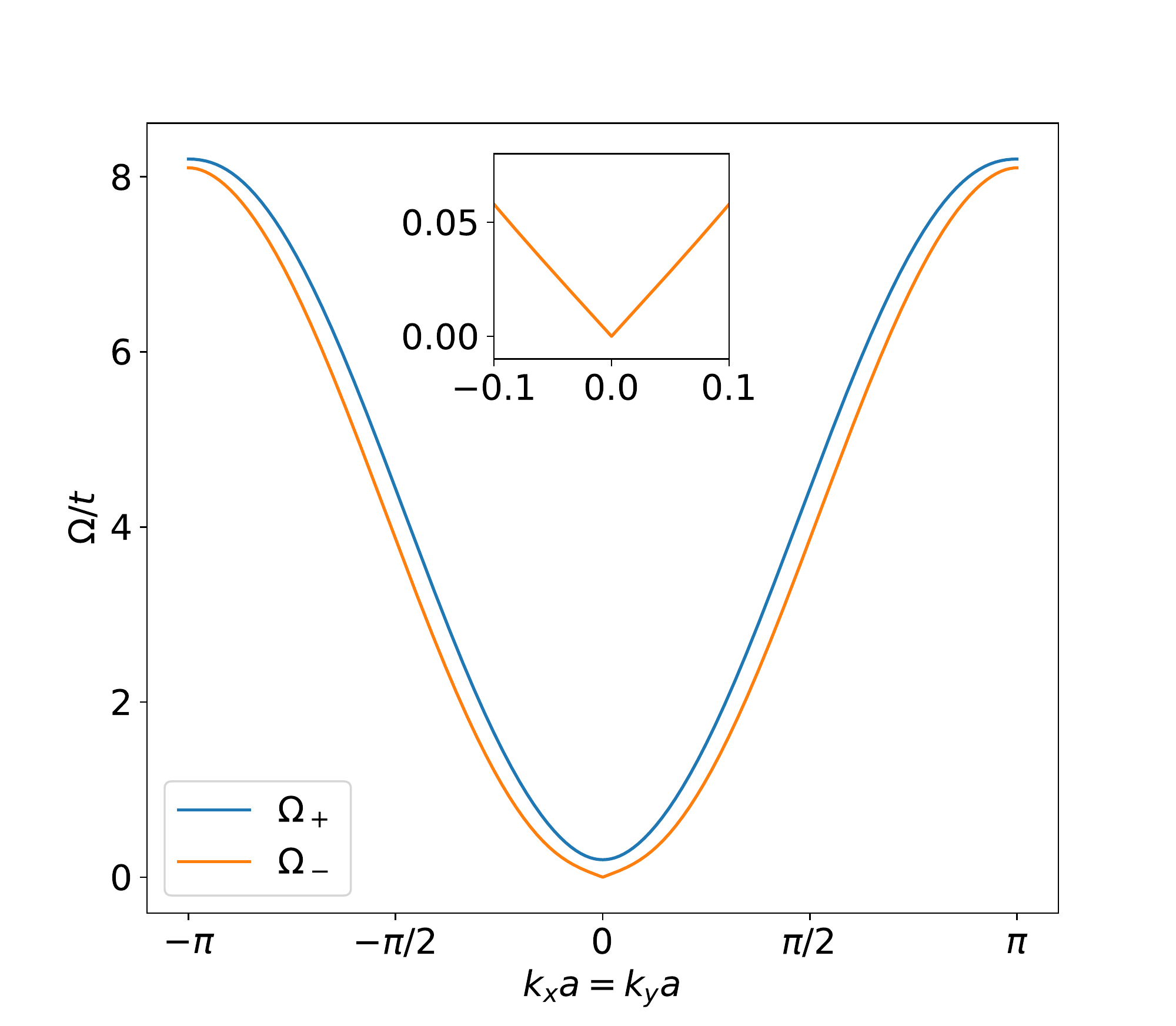}
    \caption[PZ phase band structure]{Shows the bands $\Omega_\pm(\boldsymbol{k})$ along the $k_x=k_y$ direction for $\Delta = 2U_s$, $U_s/t = 0.05$ and $\lambda_R/t = 0.1$, a set of parameters that render the eigenvalues real. A zoomed in portion close to $\boldsymbol{k}=\boldsymbol{0}$ is inserted. Notice how $\Omega_-(\boldsymbol{k})$ appears to be linear for small $k_x a$ and $k_y a$.}
    \label{fig:PZband}
\end{figure}

\subsection{Critical Superfluid Velocity} \label{sec:PZsupervel}
We use \eqref{eq:vc} to compute the critical superfluid velocity. We start by expanding for small $\boldsymbol{k}$, using also $\abs{\boldsymbol{k}}=k=\sqrt{k_x^2+k_y^2}$. Then, 
\begin{align}
\begin{split}
    \Omega_+(ka \ll 1) &= 2\Delta + \order{(ka)^2} \neq 0 \textrm{\quad if } \Delta \neq 0, \\
    \Omega_-(ka \ll 1) &\approx \left(4U_s ta^2 - \frac{8U_s\lambda_R^2a^2}{\Delta}\right)^{\frac12}k
\end{split}
\end{align}
These calculations show that $\Omega_+(\boldsymbol{k})$ is quadratic at small $\abs{\boldsymbol{k}}$, while $\Omega_-(\boldsymbol{k})$ is linear, in good agreement with figure \ref{fig:PZband}. Thus, $v_c^+=0$. For $v_c^- = \Omega_-(ka \ll 1)/k$ we find
\begin{equation}
\label{eq:PZvc}
    v_c^- = \sqrt{4U_s t a^2 - \frac{8U_s\lambda_R^2a^2}{\Delta}}. 
\end{equation}
If we instead had looked for the critical superfluid velocity using \eqref{eq:vs}, the $x$- and $y$-components would both be the same as \eqref{eq:PZvc}. Using our notation, we see that equation (33) of \cite{LS} states that one critical superfluid velocity is zero, corresponding to $v_c^+ = 0$, while the other is $\sqrt{4U_s t a^2}$. This is exactly the same as \eqref{eq:PZvc} with no SOC, i.e. setting $\lambda_R = 0$.

The requirement that $v_c^-$ is real is $\lambda_R^2 \leq t\Delta /2$. In terms of dimensionless variables, this is 
\begin{equation}
\label{eq:vPZreal}
    \frac{\lambda_R^2}{t^2}\leq \frac{1}{2}\frac{\Delta}{t}.
\end{equation}
Therefore, because we believe complex eigenvalues would occur for small $k$, we believe that the energies $\Omega_\pm(\boldsymbol{k})$ are real for all parameters such that $\lambda_R^2 \leq t\Delta /2$ in the PZ phase. A more careful, though numerical, investigation shows that this is correct. Thus we conclude that the PZ phase is stable in the presence of SOC, provided $\Delta>U_s$ and $\lambda_R^2 \leq t\Delta /2$. Since we can control the value of $\Delta$ by changing the input parameter $T^{\downarrow}$ we can always ensure stability of the PZ phase.

It appears that increasing $\lambda_R$ with fixed $\Delta$ reduces the critical superfluid velocity, which is shown in figure \ref{fig:PZbandlam}. As we have seen, the critical superfluid velocity is the slope of the energy spectra, given that they are linear close to their minima. The figures show that $\Omega_-(\boldsymbol{k})$ is linear for small $\abs{\boldsymbol{k}}$ as long as $\lambda_R$ obeys the $<$ sign in \eqref{eq:vPZreal}. We also see that increasing $\lambda_R$ reduces the slope of $\Omega_-(\boldsymbol{k})$ and thus reduces the superfluid velocity, in agreement with \eqref{eq:PZvc}. Furthermore, when $\lambda_R$ obeys the $=$ sign in \eqref{eq:vPZreal} we see that $\Omega_-(\boldsymbol{k})$ appears quadratic, and thus the critical superfluid velocity is zero, again in agreement with \eqref{eq:PZvc}.

\begin{figure}
    \centering
    \includegraphics[width=0.6\linewidth]{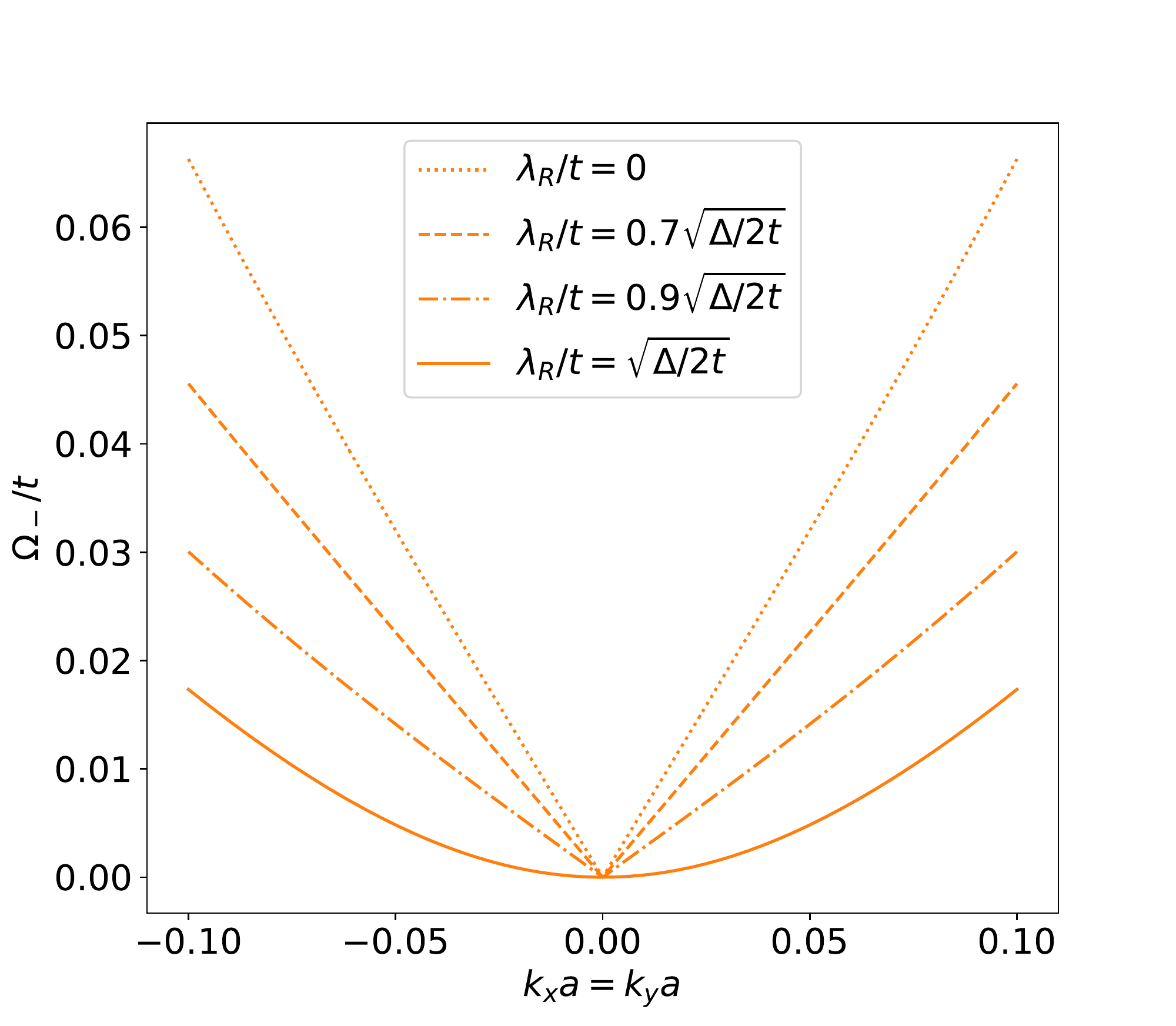}
    \caption[PZ phase superfluid velocity]{Shows the band $\Omega_-(\boldsymbol{k})$ along the $k_x=k_y$ direction for $\Delta = 2U_s$, $U_s/t = 0.05$, $\lambda_R/t = 0$ , $\lambda_R/t = 0.7\sqrt{\Delta/2t} \approx 0.16$, $\lambda_R/t = 0.9\sqrt{\Delta/2t} \approx 0.20$ and $\lambda_R/t = \sqrt{\Delta/2t} \approx 0.22$, a set of parameters that render the eigenvalues real. $\Omega_-(\boldsymbol{k})$ appears to be linear for small $\abs{\boldsymbol{k}}$ with decreasing slope as $\lambda_R/t$ is increased. When $\lambda_R/t = \sqrt{\Delta/2t}$ the apparent quadratic behavior agrees with the arguments in the text that the critical superfluid velocity should be zero. \label{fig:PZbandlam}}
\end{figure}

\subsection{Exitation Spectrum Without Interactions}
We set $U_s=0$, which is the same as setting $U=0$ i.e. no interactions. To compare to other results, we also set $\Delta=0$ (essentially letting $M_{11}=M_{33}=\mathcal{E}_{\boldsymbol{k}}$). Then we obtain
\begin{equation}
    \Omega_\pm(\boldsymbol{k}, U_s = 0, \Delta=0) = \mathcal{E}_{\boldsymbol{k}} \pm \abs{s_{\boldsymbol{k}}},
\end{equation}
which is the same as the spectrum found for a non-interacting SOC Bose gas in chapter \ref{sec:SOCU0} if $T = 4t$. Technically we would obtain $|\mathcal{E}_{\boldsymbol{k}} \pm \abs{s_{\boldsymbol{k}}}|$ from the general expression. However, investigating the eigenvectors numerically, setting $U_s = 0, \Delta=0$ and $\lambda_R\neq 0$, we can see that when $\mathcal{E}_{\boldsymbol{k}} \pm \abs{s_{\boldsymbol{k}}} > 0$ it is $|\mathcal{E}_{\boldsymbol{k}} \pm \abs{s_{\boldsymbol{k}}}|$ that enters the diagonalized Hamiltonian, while when $\mathcal{E}_{\boldsymbol{k}} \pm \abs{s_{\boldsymbol{k}}}<0$ it is $-|\mathcal{E}_{\boldsymbol{k}} \pm \abs{s_{\boldsymbol{k}}}|$ that enters the diagonalized Hamiltonian. Reporting the spectrum as $\mathcal{E}_{\boldsymbol{k}} \pm \abs{s_{\boldsymbol{k}}}$ is then the most correct representation.

The lowest energy $\Omega_{-}(\boldsymbol{k}, U_s = 0, \Delta=0) \stackrel{T = 4t}{=} \lambda_{\boldsymbol{k}}^-$ has its minima at nonzero $\boldsymbol{k}$ when $\lambda_R \neq 0$, which means the PZ phase is not stable when $U = \Delta = 0$ since it was assumed condensation occurs at $\boldsymbol{k} = \boldsymbol{0}$. The lowest energy also shows quadratic behavior close to the minima, and hence there is no superfluidity without interactions. 


\subsection{Exitation Spectrum Without SOC}
If we set $s_{\boldsymbol{k}} = 0$ in \eqref{eq:PZeigen} we get
\begin{equation}
    \Omega_\pm(\boldsymbol{k}, s_{\boldsymbol{k}} = 0)  = \frac{1}{\sqrt{2}}\sqrt{M_{11}^2 +M_{33}^2 -\abs{M_{52}}^2 \mp (M_{11}^2 -M_{33}^2 -\abs{M_{52}}^2)}.
\end{equation}
Thus, 
\begin{equation}
    \Omega_-(\boldsymbol{k}, s_{\boldsymbol{k}} = 0)  = \sqrt{\mathcal{E}_{\boldsymbol{k}}\left( \mathcal{E}_{\boldsymbol{k}} +4U_s \right)},
\end{equation}
which is the single component spectrum shown in \eqref{eq:noSOCom}. Meanwhile,
\begin{equation}
    \Omega_+(\boldsymbol{k}, s_{\boldsymbol{k}} = 0) = \mathcal{E}_{\boldsymbol{k}} + 2\Delta.
\end{equation}
We notice that if we set $U_s=0$, $\Delta = 0$ and $T=4t$ in the expressions above, we regain the energy in the case of no interactions, and no SOC, i.e. $\Omega_\pm(\boldsymbol{k}, s_{\boldsymbol{k}} = 0, U_s=0, \Delta=0) = \epsilon_{\boldsymbol{k}} + T$. If we were to derive the critical superfluid velocity from these expressions we would find one to be zero, while the other is $\sqrt{4U_s t a^2}$. This agrees with eq. (33) of \cite{LS}, and the general results in chapter \ref{sec:PZsupervel} if we set $\lambda_R=0$ in \eqref{eq:PZvc}. We also note that these eigenvalues in the $s_{\boldsymbol{k}} = 0, \Delta=0$ case corresponds to the eigenvalues in eq. (30) of \cite{LS}.

\subsection{Free Energy}
To find the free energy, $F_{\textrm{PZ}}$, we use the Hamiltonian on the form \eqref{eq:HPZkneq0}, remembering that these results will only be valid when $\Omega_\pm(\boldsymbol{k})$ are real for all $\boldsymbol{k}$, i.e. for $\lambda_R^2 \leq t\Delta/2$. 
We will focus on the effects of the elementary excitations due to interactions and SOC rather than thermal effects. Therefore we set the temperature to zero, i.e. $\beta \to \infty$. Then, the free energy is the same as the ground state energy, $F = \langle H \rangle$. 
Using \eqref{eq:F} we get
\begin{align}
\begin{split}
\label{eq:PZF}
    F_{\textrm{PZ}} \stackrel{\beta\to\infty}{=} \langle H_{\textrm{PZ}} \rangle = &N(T-4t)+\frac{UN^2}{2N_s} -4tN_s-(N_s-1)(U_s+\Delta) \\
    &+\frac{1}{2}\sum_{\boldsymbol{k}\neq\boldsymbol{0}}\sum_{\sigma=1}^2 \Omega_\sigma(\boldsymbol{k}),
\end{split}
\end{align}
which we see is independent of the angle $\theta_{\boldsymbol{0}}^{\uparrow}$ because $\Omega_\sigma(\boldsymbol{k})$ is independent of $\theta_{\boldsymbol{0}}^{\uparrow}$. Therefore, $\theta_{\boldsymbol{0}}^{\uparrow}$ is arbitrary. 

Notice that we never specified a choice for $N^\uparrow$ and $N^\downarrow$. We assumed $N_0^\uparrow = N_0$ and $N_0^\downarrow =0$, however for SOC to be operative there needs to be particles with pseudospin down as well. All of these must be excited particles, and we assumed there are few excited particles in total. Hence, we must let $N^\downarrow$ be a small nonzero number, while $N^\uparrow = N- N^\downarrow \approx N$.

%% file: 3NZ.tex
\section{NZ phase}
The NZ phase is similar to the PZ phase in that only  $\boldsymbol{k}_{00} = \boldsymbol{0}$ is occupied. However, now we have both pseudospin up and pseudospin down occupied in the condensate. From \eqref{eq:H0} we find
\begin{equation}
    H_0^{''} = (N_0^\uparrow+N_0^\downarrow)(\epsilon_{\boldsymbol{0}}+T)+\frac{U}{2N_s}\left((N_0^\uparrow)^2+(N_0^\downarrow)^2+2\alpha N_0^\uparrow N_0^\downarrow \right).
\end{equation}
We now use \eqref{eq:NN0excitup} and \eqref{eq:NN0excitdown} to replace $N_0^\alpha$ by $N^\alpha$. For $N_0^\uparrow N_0^\downarrow$ this yields
\begin{equation}
    N_0^\uparrow N_0^\downarrow = N^\uparrow N^\downarrow -N^\uparrow \left.\sum_{\boldsymbol{k}}\right.^{'} A_{\boldsymbol{k}}^{\downarrow\dagger}A_{\boldsymbol{k}}^{\downarrow}-N^\downarrow \left.\sum_{\boldsymbol{k}}\right.^{'} A_{\boldsymbol{k}}^{\uparrow\dagger}A_{\boldsymbol{k}}^{\uparrow},
\end{equation}
neglecting terms that are more than quadratic in excitation operators. Hence,
\begin{align}
    \begin{split}
    \label{eq:H0NZrewrite}
        H_0^{''} &= (N^\uparrow+N^\downarrow)(\epsilon_{\boldsymbol{0}}+T)+\frac{U}{2N_s}\left((N^\uparrow)^2+(N^\downarrow)^2+2\alpha N^\uparrow N^\downarrow \right) \\
        &-(\epsilon_{\boldsymbol{0}}+T)\sum_{\boldsymbol{k}\neq \boldsymbol{0}}\left(A_{\boldsymbol{k}}^{\uparrow\dagger}A_{\boldsymbol{k}}^{\uparrow} + A_{\boldsymbol{k}}^{\downarrow\dagger}A_{\boldsymbol{k}}^{\downarrow}\right) \\
        &-\frac{U}{2N_s}\Bigg( 2N^\uparrow \sum_{\boldsymbol{k}\neq \boldsymbol{0}}A_{\boldsymbol{k}}^{\uparrow\dagger}A_{\boldsymbol{k}}^{\uparrow} + 2N^\downarrow \sum_{\boldsymbol{k}\neq \boldsymbol{0}}A_{\boldsymbol{k}}^{\downarrow\dagger}A_{\boldsymbol{k}}^{\downarrow} \\
        &\mbox{\qquad\qquad} +2\alpha \bigg( N^\uparrow \sum_{\boldsymbol{k}\neq \boldsymbol{0}}A_{\boldsymbol{k}}^{\downarrow\dagger}A_{\boldsymbol{k}}^{\downarrow} + N^\downarrow \sum_{\boldsymbol{k}\neq \boldsymbol{0}}A_{\boldsymbol{k}}^{\uparrow\dagger}A_{\boldsymbol{k}}^{\uparrow} \bigg)\Bigg).
    \end{split}
\end{align}
Inserting $N^\uparrow = N^\downarrow = N/2$, we define $H_0 = H_0^{\textrm{NZ}}$ given in \eqref{eq:H0NZ}.
The rest of $H_0^{''}$ is moved to $H_2$ as it is quadratic in excitation operators.
The linear part $H_1=0$, and in $H_2$ we may replace $N_0^\alpha$ with $N^\alpha$ by the same arguments as for the PZ phase. For $H_2$ we get
\begin{align}
    \begin{split}
    \label{H2NZ}
        H_2 = \sum_{\boldsymbol{k}\neq \boldsymbol{0}} &\Bigg\{\big(\mathcal{E}_{\boldsymbol{k}} + U_s\big) \left(A_{\boldsymbol{k}}^{\uparrow\dagger}A_{\boldsymbol{k}}^{\uparrow} + A_{\boldsymbol{k}}^{\downarrow\dagger}A_{\boldsymbol{k}}^{\downarrow}\right) \\
        &+\left(s_{\boldsymbol{k}}+U_s\alpha e^{i(\theta_{\boldsymbol{0}}^{\downarrow}-\theta_{\boldsymbol{0}}^{\uparrow})}\right)A_{\boldsymbol{k}}^{\uparrow\dagger}A_{\boldsymbol{k}}^{\downarrow}\\
        &+\left(s_{\boldsymbol{k}}^{*}+U_s\alpha e^{-i(\theta_{\boldsymbol{0}}^{\downarrow}-\theta_{\boldsymbol{0}}^{\uparrow})}\right)A_{\boldsymbol{k}}^{\downarrow\dagger}A_{\boldsymbol{k}}^{\uparrow}\\
        &+\frac{U_s}{2}\left( e^{i2\theta_{\boldsymbol{0}}^{\uparrow}}A_{\boldsymbol{k}}^{\uparrow}A_{\boldsymbol{-k}}^{\uparrow} + e^{-i2\theta_{\boldsymbol{0}}^{\uparrow}}A_{\boldsymbol{-k}}^{\uparrow\dagger}A_{\boldsymbol{k}}^{\uparrow\dagger} \right)\\
        &+\frac{U_s}{2}\left( e^{i2\theta_{\boldsymbol{0}}^{\downarrow}}A_{\boldsymbol{k}}^{\downarrow}A_{\boldsymbol{-k}}^{\downarrow} + e^{-i2\theta_{\boldsymbol{0}}^{\downarrow}}A_{\boldsymbol{-k}}^{\downarrow\dagger}A_{\boldsymbol{k}}^{\downarrow\dagger} \right)\\
        &+\frac{U_s\alpha}{2}\bigg( e^{i(\theta_{\boldsymbol{0}}^{\uparrow} + \theta_{\boldsymbol{0}}^{\downarrow})}\left(A_{\boldsymbol{k}}^{\downarrow}A_{\boldsymbol{-k}}^{\uparrow} + A_{\boldsymbol{k}}^{\uparrow}A_{\boldsymbol{-k}}^{\downarrow}\right) \\
        & \mbox{\qquad\qquad} +e^{-i(\theta_{\boldsymbol{0}}^{\uparrow} + \theta_{\boldsymbol{0}}^{\downarrow})} \left(A_{\boldsymbol{-k}}^{\uparrow\dagger}A_{\boldsymbol{k}}^{\downarrow\dagger} + A_{\boldsymbol{-k}}^{\downarrow\dagger}A_{\boldsymbol{k}}^{\uparrow\dagger} \right)\bigg) \Bigg\}.
    \end{split}
\end{align}
All products of excitation operators commute, except for the first two. Thus, when we rewrite \eqref{H2NZ} using commutators, we
simultaneously shift $H_0$ to 
\begin{align}
    \begin{split}
        H'_0 &= H_0  - \sum_{\boldsymbol{k} \neq \boldsymbol{0}} \left(\mathcal{E}_{\boldsymbol{k}} + U_s\right) = H_0 - 4tN_s -(N_s-1)U_s.
    \end{split}
\end{align}
We also make $-\boldsymbol{k}$-terms explicit. For the diagonal terms, this has the effect of making all $M_{ii}$ equal. For the SOC dependent terms, noting that $s_{-\boldsymbol{k}} = -s_{\boldsymbol{k}}$, we get
\begin{align}
    \begin{split}
        &\sum_{\boldsymbol{k} \neq \boldsymbol{0}}\left(s_{\boldsymbol{k}}+U_s\alpha e^{i(\theta_{\boldsymbol{0}}^{\downarrow}-\theta_{\boldsymbol{0}}^{\uparrow})}\right)(A_{\boldsymbol{k}}^{\uparrow\dagger}A_{\boldsymbol{k}}^{\downarrow} + A_{\boldsymbol{k}}^{\downarrow}A_{\boldsymbol{k}}^{\uparrow\dagger} )/2 = \\
        &\sum_{\boldsymbol{k} \neq \boldsymbol{0}} \Bigg(\frac{1}{4}\left(s_{\boldsymbol{k}}+U_s\alpha e^{i(\theta_{\boldsymbol{0}}^{\downarrow}-\theta_{\boldsymbol{0}}^{\uparrow})}\right)(A_{\boldsymbol{k}}^{\uparrow\dagger}A_{\boldsymbol{k}}^{\downarrow} + A_{\boldsymbol{k}}^{\downarrow}A_{\boldsymbol{k}}^{\uparrow\dagger})  \\
        &\mbox{\qquad} + \frac{1}{4}\left(-s_{\boldsymbol{k}}+U_s\alpha e^{i(\theta_{\boldsymbol{0}}^{\downarrow}-\theta_{\boldsymbol{0}}^{\uparrow})}\right)(A_{-\boldsymbol{k}}^{\uparrow\dagger}A_{-\boldsymbol{k}}^{\downarrow} + A_{-\boldsymbol{k}}^{\downarrow}A_{-\boldsymbol{k}}^{\uparrow\dagger})\Bigg),
    \end{split}
\end{align}
and similarly for its Hermitian conjugate. We use this to write $M_{13}=M_{75} = K_{13} +s_{\boldsymbol{k}}$ and $M_{24}=M_{86} = K_{13} -s_{\boldsymbol{k}}$.

The Hamiltonian is
\begin{equation}
    H = H'_0 + \frac{1}{4} \sum_{\boldsymbol{k} \neq \boldsymbol{0}} \boldsymbol{A}_{\boldsymbol{k}}^{\dagger}M_{\boldsymbol{k}}\boldsymbol{A}_{\boldsymbol{k}},
\end{equation}
where $M_{\boldsymbol{k}}$ is of the form
\begin{gather}
    M_{\boldsymbol{k}} = 
    \begin{pmatrix}
    M_1 & M_2 \\
    M_2^* & M_1^* \\
    \end{pmatrix},
\end{gather}
with
\begin{gather*}
    M_1 = 
    \begin{pmatrix}
    M_{11}(\boldsymbol{k}) & 0 & K_{13} +s_{\boldsymbol{k}} & 0 \\
    0 & M_{11}(\boldsymbol{k}) & 0 & K_{13} -s_{\boldsymbol{k}} \\
    K_{13}^* +s_{\boldsymbol{k}}^* & 0 & M_{11}(\boldsymbol{k}) & 0 \\
    0 & K_{13}^* -s_{\boldsymbol{k}}^* & 0 & M_{11}(\boldsymbol{k}) \\
    \end{pmatrix}
\end{gather*}
and
\begin{gather*}
    M_2^* = 
    \begin{pmatrix}
    0 & M_{52} & 0 & M_{72} \\
    M_{52} & 0 & M_{72} & 0 \\
    0 & M_{72} & 0 & M_{74} \\
    M_{72} & 0 & M_{74} & 0 \\
    \end{pmatrix}.
\end{gather*}
The matrix elements are
\begin{equation}
\begin{alignedat}{2}
M_{11}(\boldsymbol{k}) &= \mathcal{E}_{\boldsymbol{k}} + U_s,  &\qquad  K_{13} &= U_s\alpha e^{i(\theta_{\boldsymbol{0}}^{\downarrow}-\theta_{\boldsymbol{0}}^{\uparrow})},\\
M_{52} &= U_s e^{i2\theta_{\boldsymbol{0}}^{\uparrow}},  &  M_{72} &= U_s\alpha e^{i(\theta_{\boldsymbol{0}}^{\uparrow} + \theta_{\boldsymbol{0}}^{\downarrow})}, \mbox{\qquad} M_{74} = U_s e^{i2\theta_{\boldsymbol{0}}^{\downarrow}}.\\
\end{alignedat}
\end{equation}


\subsection{Excitation Spectrum and Critical Superfluid Velocity}
The main structural difference between these matrices and the matrices of Linder and Sudbø \cite{LS}, is that $M_{13}(\boldsymbol{k}) \neq M_{13}(-\boldsymbol{k})$ because $s_{\boldsymbol{k}} = -s_{-\boldsymbol{k}}$. Also, in \cite{LS} all elements are real, whereas here, only the diagonal elements are real a priori. As it turns out, the fact that $M_{13} \neq \pm M_{24}$ and similar relations, make finding analytic eigenvalues difficult. Therefore, we focus first on the case of no SOC i.e. $\lambda_R = 0$ and thus $s_{\boldsymbol{k}}=0$. We name the matrix of this system $K_{\boldsymbol{k}}$, and it is the same as $M_{\boldsymbol{k}}$, upon setting $s_{\boldsymbol{k}}=0$.
This matrix has exactly the same form as the matrix considered in equation (19) of \cite{LS}, where the eigenvalues are found analytically. However, we do not assume all elements in $K_{\boldsymbol{k}}$ are real, and thus we can not use these eigenvalues for the $K_{\boldsymbol{k}}$ matrix directly. However, they serve as a nice test of the eigenvalues we do find.

We use \textit{Maple} to find eigenvalues of $K_{\boldsymbol{k}}J$, and obtain eigenvalues on the form $\lambda_K(\boldsymbol{k}) = \pm \Omega_{K\pm}(\boldsymbol{k})$, with
\begin{align}
    \begin{split}
         \Omega_{K\pm}(\boldsymbol{k}) = \frac{1}{\sqrt{2}}\Big\{&2M_{11}^2 +2\left( \abs{K_{13}}^2- \abs{M_{72}}^2 \right)\\
         & -\left( \abs{M_{52}}^2 + \abs{M_{74}}^2 \right) \pm \sqrt{R_{\boldsymbol{k}}} \Big\}^{1/2},
    \end{split}
\end{align}
where we defined
\begin{align}
    \begin{split}
        R_{\boldsymbol{k}} = &16M_{11}^2\abs{K_{13}}^2 + \left( \abs{M_{74}}^2 - \abs{M_{52}}^2 \right)^2 \\
        &+4\left(\abs{M_{72}}^2  - \abs{K_{13}}^2 \right)\left( \abs{M_{52}}^2 + \abs{M_{74}}^2 \right) \\
        &+ 8\Re\left( M_{52}(M_{72}^*)^2 M_{74} \right) + 8\Re\left( K_{13}^2M_{52}M_{74}^* \right) \\
        &-16M_{11}\big( \Re\left( K_{13}M_{52}M_{72}^* \right) +\Re\left( K_{13}M_{72}M_{74}^*\right) \big).
    \end{split}
\end{align}
Here, $\Re(z)$ is the real part of $z$. This is a general result that holds for any matrix on the form $K_{\boldsymbol{k}}$ and can be used as long as $M_{11}(\boldsymbol{k})$ is real. If we assume all matrix elements are real, $\Omega_{K\pm}(\boldsymbol{k})$ agree with the expression for the eigenvalues in \cite{LS}. We can investigate these expressions closer when we have definitions of the matrix elements. For the NZ phase, some of the matrix elements are very similar, especially their absolute values. Therefore, the expressions above can be greatly simplified, giving
\begin{equation}
    \Omega_{K\pm}(\boldsymbol{k}) = \sqrt{ \mathcal{E}_{\boldsymbol{k}}\left(\mathcal{E}_{\boldsymbol{k}} +2U_s(1\pm\alpha) \right)}
\end{equation}
The requirement for these eigenvalues to be real is $\alpha \leq 1$. We note that the requirement $\alpha \leq 1$ is in accordance with the phase diagram given in \cite{master}, where the system is in the NZ phase only for $\lambda_R=0$ and $\alpha \leq 1$ and it also agrees with the conclusions in \cite{LS}. We note that $\Omega_{K\pm}(-\boldsymbol{k}) = \Omega_{K\pm}(\boldsymbol{k})$ and that $\Omega_{K\pm}(\boldsymbol{0})=0$.
Furthermore, we see that it corresponds to eq. (30) of \cite{LS} with the following identifications: $\epsilon_{\boldsymbol{k}}^A = \epsilon_{\boldsymbol{k}}^B = 4t +\epsilon_{\boldsymbol{k}} = \mathcal{E}_{\boldsymbol{k}}$, $F_A = F_B = \abs{M_{52}} = \abs{M_{74}} = U_s = UN/2N_s$ and $F_{AB} = \abs{M_{72}} = \abs{K_{13}} = U_s\alpha$. Thus, consulting eq. (33) of \cite{LS} we expect two nonzero critical superfluid velocities
\begin{equation}
\label{eq:NZvc}
    v_c^\pm = \sqrt{2U_s t a^2 (1\pm\alpha)}.
\end{equation}
This is exactly what one obtains from doing the calculations. Also, we note that if we set $\alpha=0$, i.e. the case of two uncoupled components in the BEC, we get the well known single component Bogoliubov spectrum $\Omega_{K\pm}(\boldsymbol{k}, \alpha=0) = \sqrt{ \mathcal{E}_{\boldsymbol{k}}\left(\mathcal{E}_{\boldsymbol{k}} +2U_s \right)}$ \cite{LS}. This agrees completely with the similar result in \cite{LS}, and also compares favourably to \eqref{eq:noSOCom} though there appears to be a factor of 2 difference in the final term. This seeming discrepancy can be explained. Here, when we set $\alpha=0$, we are treating a system of two independent single component systems, each with $N/2$ particles, whereas chapter \ref{sec:Weaklyinteracting} treats a one-component system with $N$ particles. 

The eigenvalues with no SOC are shown in figs. \ref{fig:NZeigenvaluesSOC0} and \ref{fig:NZbandSOC0}. The apparent linearity agrees with our calculation that there should be two nonzero critical superfluid velocities.
\begin{figure}
    \centering
    \includegraphics[width=0.8\linewidth]{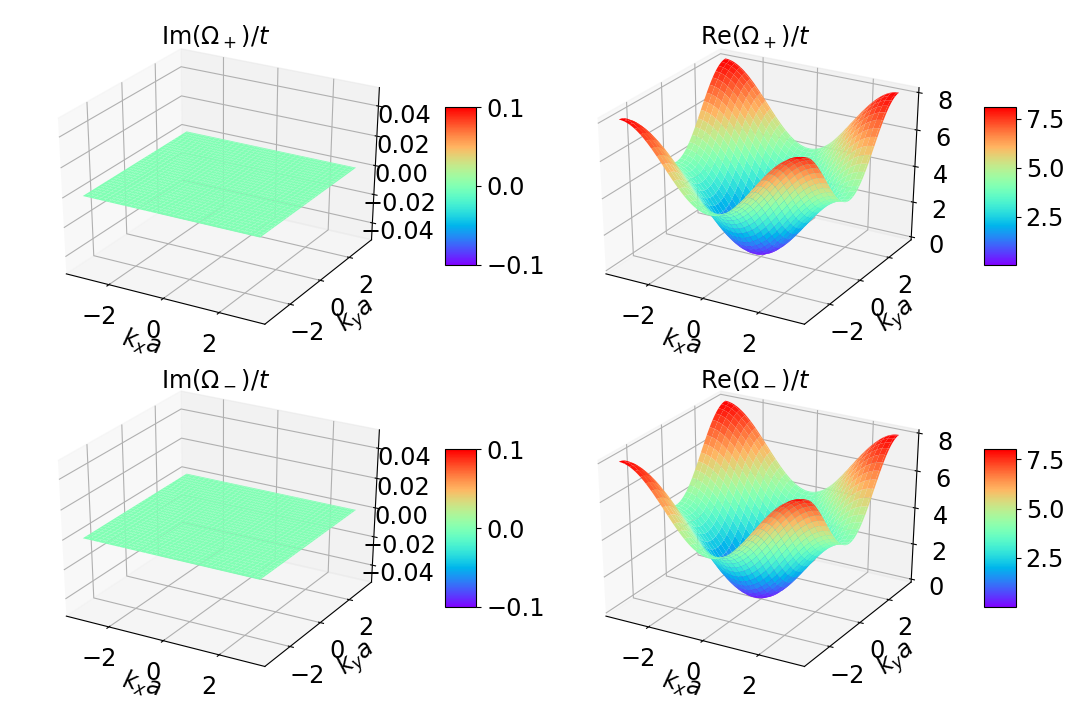}
    \caption{Shows real and imaginary parts of $\Omega_{K\pm}(\boldsymbol{k})$ for $U_s/t = 0.05$ and $\alpha = 0.9$ and thus if $U/t=0.1$ the average filling is $N/N_s =1$.}
    \label{fig:NZeigenvaluesSOC0}
\end{figure}

\begin{figure}
    \centering
    \includegraphics[width=0.7\linewidth]{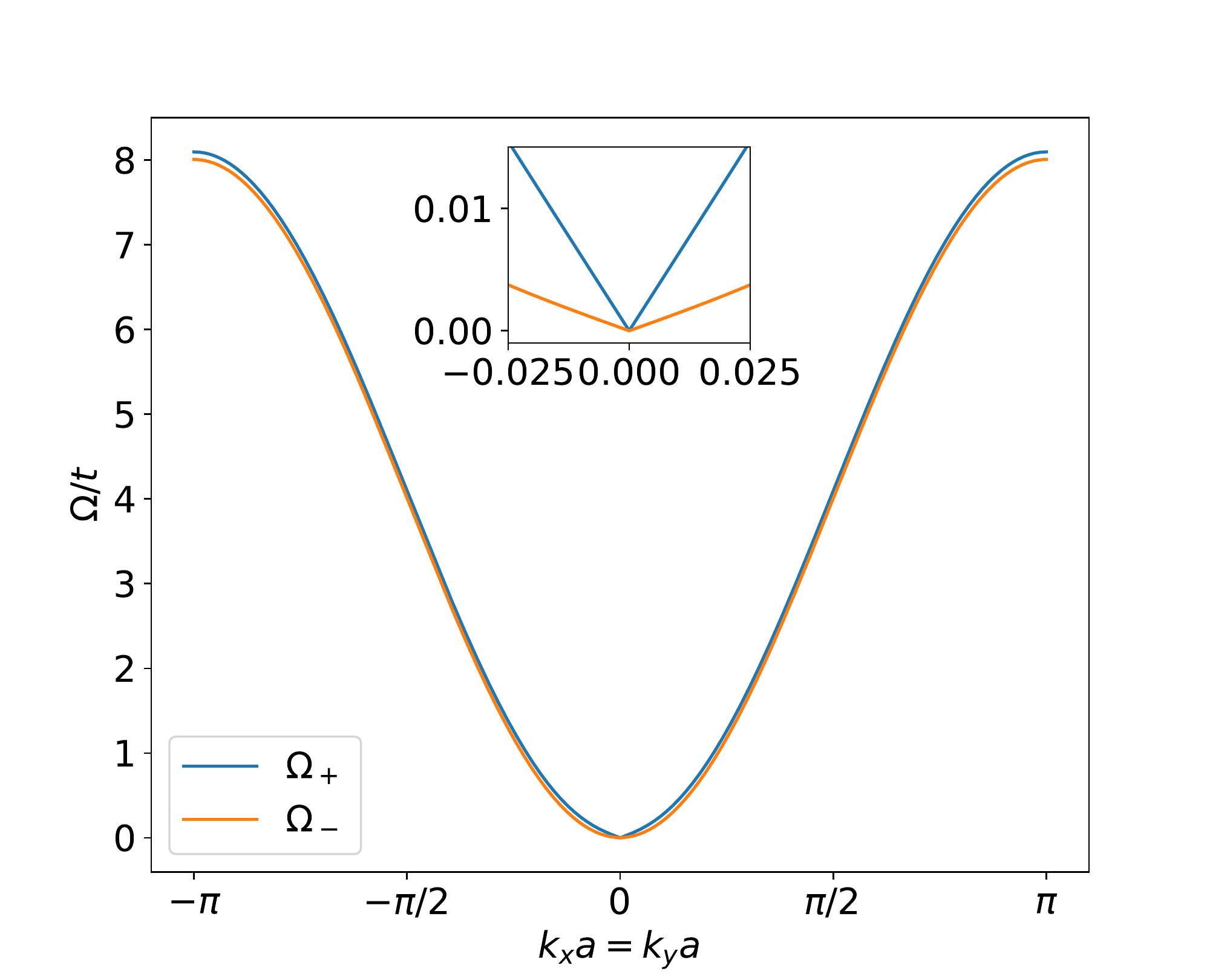}
    \caption{Shows the bands $\Omega_{K\pm}(\boldsymbol{k})$ along the $k_x=k_y$ direction for $U_s/t = 0.05$ and $\alpha = 0.9$. A zoomed in portion close to $\boldsymbol{k}=\boldsymbol{0}$ is inserted. Notice how $\Omega_{K\pm}(\boldsymbol{k})$ appear to be linear for small $\abs{\boldsymbol{k}}$, suggesting nonzero superfluid velocities. If $\alpha$ is decreased toward zero, the two bands become more and more similar.}
    \label{fig:NZbandSOC0}
\end{figure}

Numerical solutions of the eigenvalue problem show that the eigenvalues of $M_{\boldsymbol{k}}J$ are complex for any nonzero $\lambda_R$. This suggests the NZ phase in unstable in the presence of SOC. In figure \ref{fig:NZeigenSOC} we show the real and imaginary parts of the lowest band in the presence of SOC, and we see that a considerable area in $\boldsymbol{k}$-space has complex eigenvalues, with imaginary parts of order $\order{10^{-2}}$. The direction of the two areas with complex eigenvalues in the presence of SOC depend on the angles, which were set to $\theta_{\boldsymbol{0}}^\uparrow = 0$ and $\theta_{\boldsymbol{0}}^\downarrow = 3\pi/4$. It is in fact possible to find analytic eigenvalues of $M_{\boldsymbol{k}}J$ if one sets $\alpha=0$, however the lower branch there also turns out to be complex for any nonzero $\lambda_R$.

\begin{figure}
    \centering
    \includegraphics[width=0.9\linewidth]{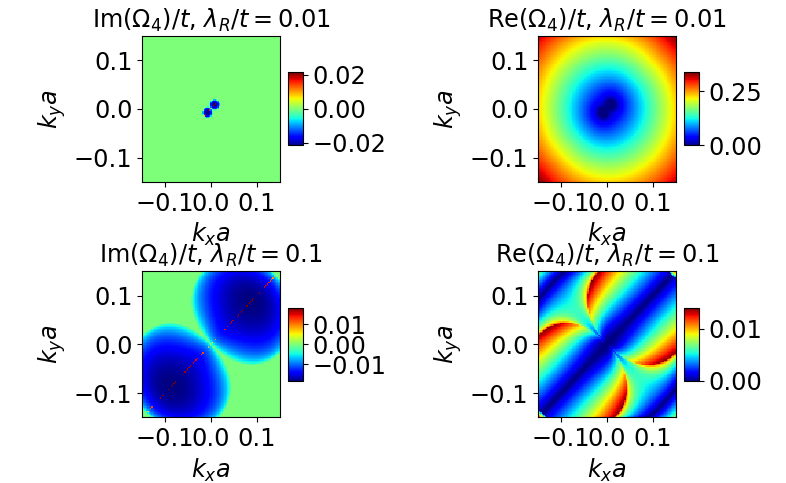}
    \caption{Plots of real and imaginary parts of the lowest numerical eigenvalue $\Omega_4(\boldsymbol{k})$ for $U_s/t = 0.05$ and $\alpha = 0.9$ in an area around $\boldsymbol{k}=\boldsymbol{0}$. In the first row, $\lambda_R/t = 0.01$ while in the second row $\lambda_R/t = 0.1$. Observe the significant imaginary parts of the eigenvalues in considerable areas.}
    \label{fig:NZeigenSOC}
\end{figure}


Furthermore, consider a point alluded to earlier when discussing the PZ phase. The SOC tries to move the minimum at $\boldsymbol{k} = \boldsymbol{0}$ down towards new minima at nonzero $\boldsymbol{k}$, but the result is complex eigenvalues and a smearing out of the zeros of the real part. We also mentioned that with no Zeeman splitting any nonzero SOC will lead to minima at nonzero $\boldsymbol{k}$. Hence, the NZ phase with no Zeeman term should not be possible for any nonzero SOC. In conclusion, the NZ phase is only stable for $\lambda_R =0$ and $\alpha\leq1$, and the excitation spectrum is $\Omega_{K\pm}(\boldsymbol{k})$.

\subsection{Free Energy}
As argued above, we must set $\lambda_R=0$, $\alpha \leq 1$ and use $\Omega_{K\pm}$ to have a real spectrum in the NZ phase.
\begin{equation}
    H_2 =  \frac{1}{4} \sum_{\boldsymbol{k}\neq\boldsymbol{0}}\boldsymbol{B}_{\boldsymbol{k}}^{\dagger}D_{\boldsymbol{k}}\boldsymbol{B}_{\boldsymbol{k}},
\end{equation}
where $D_{\boldsymbol{k}}$ is the matrix 
\begin{align}
    \begin{split}
        D_{\boldsymbol{k}} = \textrm{diag}\big(&\Omega_{K+}(\boldsymbol{k}), \Omega_{K+}(\boldsymbol{k}), \Omega_{K-}(\boldsymbol{k}), \Omega_{K-}(\boldsymbol{k}),\\
        &\Omega_{K+}(\boldsymbol{k}), \Omega_{K+}(\boldsymbol{k}), \Omega_{K-}(\boldsymbol{k}), \Omega_{K-}(\boldsymbol{k})\big).
    \end{split}
\end{align}
    
Just as in the PZ phase, numerical investigations of the transformation matrix $T_{\boldsymbol{k}}$ suggest a relation between the new operators corresponding to equal eigenvalues, and we find $B_{-\boldsymbol{k},3} = B_{\boldsymbol{k},1}$ and $B_{-\boldsymbol{k},4} = B_{\boldsymbol{k},2}$. Additionally, the analytic argument given in the PZ phase is also valid here. Thus, using commutators and the inversion symmetry of the eigenvalues, we obtain
\begin{equation}
    H_2 =  \sum_{\boldsymbol{k}\neq\boldsymbol{0}}\sum_{\sigma=1}^2 \Omega_{K\sigma}(\boldsymbol{k})\left(B_{\boldsymbol{k},\sigma}^{\dagger}B_{\boldsymbol{k},\sigma}+\frac{1}{2}\right),
\end{equation}
where we let $\Omega_{K1}(\boldsymbol{k})=\Omega_{K+}(\boldsymbol{k})$ and $\Omega_{K2}(\boldsymbol{k})=\Omega_{K-}(\boldsymbol{k})$. At zero temperature the free energy, $F_{\textrm{NZ}}$, is equal to $\langle H_{\textrm{NZ}} \rangle$. Using \eqref{eq:F} we find
\begin{equation}
\label{eq:NZF}
    F_{\textrm{NZ}} = H'_0 +\frac{1}{2}\sum_{\boldsymbol{k}\neq\boldsymbol{0}}\sum_{\sigma=1}^2 \Omega_{K\sigma}(\boldsymbol{k}).
\end{equation}
Here, $H'_0$ is
\begin{equation}
    H'_0 = N(\epsilon_{\boldsymbol{0}}+T)+\frac{UN^2}{4N_s}(1+\alpha)- 4tN_s -(N_s-1)U_s.
\end{equation}
$F_\textrm{NZ}$ is independent of the angles $\theta_{\boldsymbol{0}}^{\downarrow}$ and $\theta_{\boldsymbol{0}}^{\uparrow}$, and they are thus arbitrary. 

Notice that we did not assume $N_0^\uparrow = N_0^\downarrow$. However, once we set $N^\uparrow = N^\downarrow$ it is likely that $N_0^\uparrow \approx N_0^\downarrow$, since we assume there are few excitations. The same point can be made for the remaining phases.

%% file: 4PW.tex
\section{PW Phase}
The PW Phase is similar to the NZ phase, except that now $\boldsymbol{k}_{01}=(k_0, k_0)$ is the only occupied condensate momentum. We assume $N_{\boldsymbol{k}_{01}}^\uparrow = N_0^\uparrow$ and  $N_{\boldsymbol{k}_{01}}^\downarrow = N_0^\downarrow$. A single nonzero condensate momentum like this, can be thought of as an analogue to Fulde-Ferrell-Larkin-Ovchinnikov \cite{FFLO} states, usually discussed in the case of fermionic systems and in particular superconductors. To be specific, the PW phase is an analogue of Fulde-Ferrell states \cite{FF}, while the SW phase is an analogue of Larkin-Ovchinnikov states \cite{LO}.

From \eqref{eq:H0} we find
\begin{align}
    \begin{split}
    H_0^{''} &= (N_0^\uparrow+N_0^\downarrow)(\epsilon_{\boldsymbol{k}_{01}}+T)+2\sqrt{N_0^\uparrow N_0^\downarrow}\abs{s_{\boldsymbol{k}_{01}}}\cos(\gamma_{\boldsymbol{k}_{01}}+\Delta\theta_1)\\
    &+\frac{U}{2N_s}\left((N_0^\uparrow)^2+(N_0^\downarrow)^2+2\alpha N_0^\uparrow N_0^\downarrow \right).
    \end{split}
\end{align}
We now use \eqref{eq:NN0excitup} and \eqref{eq:NN0excitdown} to replace $N_0^\alpha$ by $N^\alpha$. For $\sqrt{N_0^\uparrow N_0^\downarrow}$ we Taylor expand the square root and keep only terms that are at most quadratic in excitation operators,
\begin{equation}
    \sqrt{N_0^\uparrow N_0^\downarrow} = \sqrt{N^\uparrow N^\downarrow}-\frac{1}{2}\sqrt{\frac{N^\uparrow}{N^\downarrow}}\left.\sum_{\boldsymbol{k}}\right.^{'} A_{\boldsymbol{k}}^{\downarrow\dagger}A_{\boldsymbol{k}}^{\downarrow} -\frac{1}{2} \sqrt{\frac{N^\downarrow}{N^\uparrow}} \left.\sum_{\boldsymbol{k}}\right.^{'} A_{\boldsymbol{k}}^{\uparrow\dagger}A_{\boldsymbol{k}}^{\uparrow}.
\end{equation}
Hence, inserting \eqref{eq:NN0excitup} and \eqref{eq:NN0excitdown} into $H_0^{''}$ we get
\begin{align}
    \begin{split}
    \label{eq:H0PWrewrite}
        H_0^{''} &= (N^\uparrow+N^\downarrow)(\epsilon_{\boldsymbol{k}_{01}}+T)+2\sqrt{N^\uparrow N^\downarrow}\abs{s_{\boldsymbol{k}_{01}}}\cos(\gamma_{\boldsymbol{k}_{01}}+\Delta\theta_1)\\
        &+\frac{U}{2N_s}\left((N^\uparrow)^2+(N^\downarrow)^2+2\alpha N^\uparrow N^\downarrow \right) \\
        &-(\epsilon_{\boldsymbol{k}_{01}}+T)\sum_{\boldsymbol{k}\neq \boldsymbol{k}_{01}}\left(A_{\boldsymbol{k}}^{\uparrow\dagger}A_{\boldsymbol{k}}^{\uparrow} + A_{\boldsymbol{k}}^{\downarrow\dagger}A_{\boldsymbol{k}}^{\downarrow}\right) \\
        &-\abs{s_{\boldsymbol{k}_{01}}}\left(\sqrt{\frac{N^\uparrow}{N^\downarrow}}\sum_{\boldsymbol{k}\neq \boldsymbol{k}_{01}} A_{\boldsymbol{k}}^{\downarrow\dagger}A_{\boldsymbol{k}}^{\downarrow} + \sqrt{\frac{N^\downarrow}{N^\uparrow}} \sum_{\boldsymbol{k}\neq \boldsymbol{k}_{01}} A_{\boldsymbol{k}}^{\uparrow\dagger}A_{\boldsymbol{k}}^{\uparrow}\right)\cos(\gamma_{\boldsymbol{k}_{01}}+\Delta\theta_1) \\
        &-\frac{U}{2N_s}\Bigg( 2N^\uparrow \sum_{\boldsymbol{k}\neq \boldsymbol{k}_{01}}A_{\boldsymbol{k}}^{\uparrow\dagger}A_{\boldsymbol{k}}^{\uparrow} + 2N^\downarrow \sum_{\boldsymbol{k}\neq \boldsymbol{k}_{01}}A_{\boldsymbol{k}}^{\downarrow\dagger}A_{\boldsymbol{k}}^{\downarrow} \\
        &\mbox{\qquad\qquad} +2\alpha \bigg( N^\uparrow \sum_{\boldsymbol{k}\neq \boldsymbol{k}_{01}}A_{\boldsymbol{k}}^{\downarrow\dagger}A_{\boldsymbol{k}}^{\downarrow} + N^\downarrow \sum_{\boldsymbol{k}\neq \boldsymbol{k}_{01}}A_{\boldsymbol{k}}^{\uparrow\dagger}A_{\boldsymbol{k}}^{\uparrow} \bigg)\Bigg).
    \end{split}
\end{align}
Choosing $N^\uparrow = N^\downarrow = N/2$ we define $H_0 = H_0^{\textrm{PW}}$ given in \eqref{eq:H0PW}.
The rest of $H_0^{''}$ is moved to $H_2$ as it is quadratic in excitation operators. Once again, because only one momentum is occupied, $H_1=0$, and in $H_2$ we may replace $N_0$ with $N$ by the same argument as in the PZ phase. 
The quadratic part has a similar form as in the NZ phase,
\begin{align}
    \begin{split}
        H_2 = \sum_{\boldsymbol{k}\neq \boldsymbol{k}_{01}} &\Bigg\{\left(\mathcal{E}_{\boldsymbol{k}} + U_s + G_{k_0}\right) (A_{\boldsymbol{k}}^{\uparrow\dagger}A_{\boldsymbol{k}}^{\uparrow} + A_{\boldsymbol{k}}^{\downarrow\dagger}A_{\boldsymbol{k}}^{\downarrow}) \\
        &+\left(s_{\boldsymbol{k}}+U_s\alpha e^{i(\theta_{1}^{\downarrow}-\theta_{1}^{\uparrow})}\right)A_{\boldsymbol{k}}^{\uparrow\dagger}A_{\boldsymbol{k}}^{\downarrow}\\
        &+\left(s_{\boldsymbol{k}}^{*}+U_s\alpha e^{-i(\theta_{1}^{\downarrow}-\theta_{1}^{\uparrow})}\right)A_{\boldsymbol{k}}^{\downarrow\dagger}A_{\boldsymbol{k}}^{\uparrow}\\
        &+\frac{U_s}{2}\left( e^{i2\theta_{1}^{\uparrow}}A_{\boldsymbol{k}}^{\uparrow}A_{2\boldsymbol{k}_{01}-\boldsymbol{k}}^{\uparrow} + e^{-i2\theta_{1}^{\uparrow}}A_{2\boldsymbol{k}_{01}-\boldsymbol{k}}^{\uparrow\dagger}A_{\boldsymbol{k}}^{\uparrow\dagger} \right)\\
        &+\frac{U_s}{2}\left( e^{i2\theta_{1}^{\downarrow}}A_{\boldsymbol{k}}^{\downarrow}A_{2\boldsymbol{k}_{01}-\boldsymbol{k}}^{\downarrow} + e^{-i2\theta_{1}^{\downarrow}}A_{2\boldsymbol{k}_{01}-\boldsymbol{k}}^{\downarrow\dagger}A_{\boldsymbol{k}}^{\downarrow\dagger} \right)\\
        &+\frac{U_s}{2}\alpha\bigg( e^{i(\theta_{1}^{\uparrow} + \theta_{1}^{\downarrow})}\left(A_{\boldsymbol{k}}^{\downarrow}A_{2\boldsymbol{k}_{01}-\boldsymbol{k}}^{\uparrow} + A_{\boldsymbol{k}}^{\uparrow}A_{2\boldsymbol{k}_{01}-\boldsymbol{k}}^{\downarrow}\right) \\
        &\mbox{\qquad \qquad} +e^{-i(\theta_{1}^{\uparrow} + \theta_{1}^{\downarrow})} \left(A_{2\boldsymbol{k}_{01}-\boldsymbol{k}}^{\uparrow\dagger}A_{\boldsymbol{k}}^{\downarrow\dagger} + A_{2\boldsymbol{k}_{01}-\boldsymbol{k}}^{\downarrow\dagger}A_{\boldsymbol{k}}^{\uparrow\dagger} \right) \bigg) \Bigg\},
    \end{split}
\end{align}
where we defined
\begin{align}
    \begin{split}
    G_{k_0} &\equiv  \epsilon_{\boldsymbol{0}}-\epsilon_{\boldsymbol{k}_{01}}  - \abs{s_{\boldsymbol{k}_{01}}}\cos(\gamma_{\boldsymbol{k}_{01}}+\Delta\theta_1)  \\
    &= 4t(\cos(k_0 a)-1)-2\sqrt{2}\lambda_R\abs{\sin(k_0 a)}\cos(\gamma_{\boldsymbol{k}_{01}}+\Delta\theta_1).
    \end{split}
\end{align}
Notice that upon setting $k_0=0$, this is equivalent to the NZ phase. 

\subsection{Approximate Analytic Eigenvalues in Helicity Basis}
It will prove impossible to find analytic eigenvalues in the above spin basis. We therefore first attempt an approximation along the lines of Toniolo and Linder \cite{Toniolo}, who treated the PW phase with the addition of a Zeeman field. We transform the Hamiltonian to the helicity basis \eqref{eq:Helicitybasis} which diagonalizes the non-interacting part of the Hamiltonian. Then, we claim that $C_{\boldsymbol{k}}^+$ is negligible because only the lowest band is relevant for BEC. I.e. the vast majority of the helicity quasiparticles will be placed in the minima of $\lambda_{\boldsymbol{k}}^-$ before introducing weak interactions. Defining $C_{\boldsymbol{k}} \equiv C_{\boldsymbol{k}}^-$ we find that
\begin{equation}
    \begin{pmatrix} A_{\boldsymbol{k}}^\uparrow \\ A_{\boldsymbol{k}}^\downarrow \end{pmatrix} = \frac{1}{\sqrt{2}} \begin{pmatrix} e^{-i\gamma_{\boldsymbol{k}}} \left( C_{\boldsymbol{k}}^+ - C_{\boldsymbol{k}}^- \right)\\ C_{\boldsymbol{k}}^+ + C_{\boldsymbol{k}}^-  \end{pmatrix} \approx \frac{1}{\sqrt{2}} \begin{pmatrix} -e^{-i\gamma_{\boldsymbol{k}}} C_{\boldsymbol{k}}\\ C_{\boldsymbol{k}}  \end{pmatrix}
\end{equation}
with this approximation. Then, $H_2$ becomes
\begin{align}
    \begin{split}
        H_2 = \frac{1}{2} \sum_{\boldsymbol{k}\neq \boldsymbol{k}_{01}} &\Bigg( 2N_{11}(\boldsymbol{k})C_{\boldsymbol{k}}^\dagger C_{\boldsymbol{k}} + N_{32}(\boldsymbol{k})C_{\boldsymbol{k}}C_{2\boldsymbol{k}_{01}-\boldsymbol{k}} +N_{32}^* (\boldsymbol{k})C_{2\boldsymbol{k}_{01}-\boldsymbol{k}}^\dagger C_{\boldsymbol{k}}^\dagger \Bigg),
    \end{split}
\end{align}
where we defined
\begin{align}
    \begin{split}
        N_{11}(\boldsymbol{k}) =& \mathcal{E}_{\boldsymbol{k}} + U_s + G_{k_0} -\abs{s_{\boldsymbol{k}}} -U_s\alpha \cos(\gamma_{\boldsymbol{k}}+\theta_{1}^{\downarrow}-\theta_{1}^{\uparrow}), \\
        N_{32}(\boldsymbol{k}) =& U_s e^{i2\theta_{1}^{\uparrow}}\frac{e^{-i(\gamma_{\boldsymbol{k}}+\gamma_{2\boldsymbol{k}_{01}-\boldsymbol{k}})}}{2}+\frac{U_s e^{i2\theta_{1}^{\downarrow}}}{2}\\
        &-  U_s\alpha e^{i(\theta_{1}^{\uparrow} + \theta_{1}^{\downarrow})}\frac{e^{-i\gamma_{\boldsymbol{k}}}+e^{-i\gamma_{2\boldsymbol{k}_{01}-\boldsymbol{k}}}}{2}.
    \end{split}
\end{align}
Defining
\begin{equation}
    \boldsymbol{C}_{\boldsymbol{k}} = (C_{\boldsymbol{k}}, C_{2\boldsymbol{k}_{01}-\boldsymbol{k}}, C_{\boldsymbol{k}}^\dagger, C_{2\boldsymbol{k}_{01}-\boldsymbol{k}}^\dagger)^T
\end{equation}
and using commutators, we can write the Hamiltonian on matrix form
\begin{equation}
    H_2 = \frac{1}{4}\sum_{\boldsymbol{k}\neq \boldsymbol{k}_{01}} \boldsymbol{C}_{\boldsymbol{k}}^\dagger N_{\boldsymbol{k}} \boldsymbol{C}_{\boldsymbol{k}},
\end{equation}
where
\begin{gather}
    N_{\boldsymbol{k}} = 
    \begin{pmatrix}
    N_{11}(\boldsymbol{k}) & 0 & 0 & N_{32}^* (\boldsymbol{k}) \\
    0 &  N_{11}(2\boldsymbol{k}_{01}-\boldsymbol{k}) & N_{32}^* (\boldsymbol{k}) & 0 \\
    0 & N_{32} (\boldsymbol{k}) & N_{11}(\boldsymbol{k})  & 0 \\
    N_{32} (\boldsymbol{k}) & 0 & 0 & N_{11}(2\boldsymbol{k}_{01}-\boldsymbol{k})  \\
    \end{pmatrix}.
\end{gather}
We used periodicity to fill in the diagonal.
Such a manipulation was also done in \cite{Toniolo} and will be justified later. 
The eigenvalues of $N_{\boldsymbol{k}}J$ obtained analytically with \textit{Maple} are $\lambda(\boldsymbol{k}) = \pm \Omega_i(\boldsymbol{k})$, $\Omega_1(\boldsymbol{k}) = \Omega_+(\boldsymbol{k})$, $\Omega_2(\boldsymbol{k}) = \Omega_-(\boldsymbol{k})$ with
\begin{align}
    \begin{split}
    \Omega_\pm(\boldsymbol{k}) = \frac{1}{2}\bigg(&\pm N_{11}(\boldsymbol{k}) \mp N_{11}(2\boldsymbol{k}_{01}-\boldsymbol{k}) \\
    &+ \sqrt{\big(N_{11}(\boldsymbol{k}) + N_{11}(2\boldsymbol{k}_{01}-\boldsymbol{k}) \big)^2 - 4\abs{N_{32} (\boldsymbol{k})}^2}  \bigg).
    \end{split}
\end{align}
We notice that $\Omega_-(2\boldsymbol{k}_{01}-\boldsymbol{k}) = \Omega_+(\boldsymbol{k})$ because $N_{32} (2\boldsymbol{k}_{01}-\boldsymbol{k}) = N_{32} (\boldsymbol{k})$. This can be used to represent the diagonalized Hamiltonian in terms of just one band. The procedure is similar to what was done in the PZ phase to combine two equal inversion symmetric bands into one. Let us investigate the equations governing the eigenvectors. For $N_{\boldsymbol{k}}J\boldsymbol{x} = \Omega_+(\boldsymbol{k})\boldsymbol{x}$ we find
\begin{align}
    \begin{split}
        N_{11}(\boldsymbol{k})x_1 -N_{32}^* (\boldsymbol{k})x_4 &= \Omega_+(\boldsymbol{k})x_1, \\
        N_{11}(2\boldsymbol{k}_{01}-\boldsymbol{k})x_2 -N_{32}^* (\boldsymbol{k})x_3 &= \Omega_+(\boldsymbol{k})x_2, \\
        N_{32} (\boldsymbol{k})x_2 - N_{11}(\boldsymbol{k})x_3 &=\Omega_+(\boldsymbol{k})x_3, \\
        N_{32} (\boldsymbol{k})x_1 - N_{11}(2\boldsymbol{k}_{01}-\boldsymbol{k})x_4 &=\Omega_+(\boldsymbol{k})x_4.
    \end{split}
\end{align}
If we instead look for the eigenvectors at $2\boldsymbol{k}_{01}-\boldsymbol{k}$, $N_{2\boldsymbol{k}_{01}-\boldsymbol{k}}J\boldsymbol{x} = \Omega_-(2\boldsymbol{k}_{01}-\boldsymbol{k})\boldsymbol{x} = \Omega_+(\boldsymbol{k})\boldsymbol{x}$ becomes
\begin{align}
    \begin{split}
        N_{11}(2\boldsymbol{k}_{01}-\boldsymbol{k})x_1 -N_{32}^* (\boldsymbol{k})x_4 &= \Omega_+(\boldsymbol{k})x_1, \\
        N_{11}(\boldsymbol{k})x_2 -N_{32}^* (\boldsymbol{k})x_3 &= \Omega_+(\boldsymbol{k})x_2, \\
        N_{32} (\boldsymbol{k})x_2 - N_{11}(2\boldsymbol{k}_{01}-\boldsymbol{k})x_3 &=\Omega_+(\boldsymbol{k})x_3, \\
        N_{32} (\boldsymbol{k})x_1 - N_{11}(\boldsymbol{k})x_4 &=\Omega_+(\boldsymbol{k})x_4.
    \end{split}
\end{align}
The relation $N_{32} (2\boldsymbol{k}_{01}-\boldsymbol{k}) = N_{32} (\boldsymbol{k})$ was used. These equation are the same, apart from an interchange $x_1 \leftrightarrow x_2$ and $x_3 \leftrightarrow x_4$. This is the same change we have in the basis
\begin{align}
    \begin{split}
        \boldsymbol{C}_{\boldsymbol{k}} &= (C_{\boldsymbol{k}}, C_{2\boldsymbol{k}_{01}-\boldsymbol{k}}, C_{\boldsymbol{k}}^\dagger, C_{2\boldsymbol{k}_{01}-\boldsymbol{k}}^\dagger)^T, \\
        \boldsymbol{C}_{2\boldsymbol{k}_{01}-\boldsymbol{k}} &= (C_{2\boldsymbol{k}_{01}-\boldsymbol{k}}, C_{\boldsymbol{k}}, C_{2\boldsymbol{k}_{01}-\boldsymbol{k}}^\dagger, C_{\boldsymbol{k}}^\dagger)^T .
    \end{split}
\end{align}
Imagine $\boldsymbol{x}_1$ is an eigenvector corresponding to $\Omega_1(\boldsymbol{k})= \Omega_+(\boldsymbol{k})$. Then the operator $B_{\boldsymbol{k},1}$ associated with $\Omega_1(\boldsymbol{k})$ is defined as $B_{\boldsymbol{k},1} = \boldsymbol{x}_1^\dagger \boldsymbol{C}_{\boldsymbol{k}}$. An eigenvector $\boldsymbol{x}_2$ corresponding to $\Omega_2 (2\boldsymbol{k}_{01}-\boldsymbol{k}) = \Omega_-(2\boldsymbol{k}_{01}-\boldsymbol{k})$ can then be chosen to be the same as $\boldsymbol{x}_1$ apart from the interchange $x_1 \leftrightarrow x_2$ and $x_3 \leftrightarrow x_4$. The operator $B_{2\boldsymbol{k}_{01}-\boldsymbol{k}, 2}$ associated with $\Omega_2 (2\boldsymbol{k}_{01}-\boldsymbol{k})$ is then defined as $B_{2\boldsymbol{k}_{01}-\boldsymbol{k}, 2} = \boldsymbol{x}_2^\dagger \boldsymbol{C}_{2\boldsymbol{k}_{01}-\boldsymbol{k}} = \boldsymbol{x}_1^\dagger \boldsymbol{C}_{\boldsymbol{k}} = B_{\boldsymbol{k},1}$. Hence, the operator $B_{2\boldsymbol{k}_{01}-\boldsymbol{k}, 2}$ can be defined to be equal to the operator $B_{\boldsymbol{k},1}$. Thus,
\begin{align}
    \begin{split}
        H_2 &= \frac{1}{2}\sum_{\boldsymbol{k}\neq \boldsymbol{k}_{01}} \left(\Omega_1(\boldsymbol{k}) \left(B_{\boldsymbol{k},1}^{\dagger}B_{\boldsymbol{k},1} +\frac12 \right) + \Omega_2(\boldsymbol{k}) \left(B_{\boldsymbol{k},2}^{\dagger}B_{\boldsymbol{k},2} +\frac12 \right)\right) \\
        &= \frac{1}{2}\sum_{\boldsymbol{k}\neq \boldsymbol{k}_{01}} \bigg(\Omega_1(\boldsymbol{k}) \left(B_{\boldsymbol{k},1}^{\dagger}B_{\boldsymbol{k},1} +\frac12 \right) \\
        &\mbox{\qquad\qquad\qquad} + \Omega_2(2\boldsymbol{k}_{01}-\boldsymbol{k}) \left(B_{2\boldsymbol{k}_{01}-\boldsymbol{k},2}^{\dagger}B_{2\boldsymbol{k}_{01}-\boldsymbol{k},2} +\frac12 \right)\bigg) \\
        &= \sum_{\boldsymbol{k}\neq \boldsymbol{k}_{01}} \Omega_1(\boldsymbol{k}) \left(B_{\boldsymbol{k},1}^{\dagger}B_{\boldsymbol{k},1} +\frac12 \right) \\
        &\equiv \sum_{\boldsymbol{k}\neq \boldsymbol{k}_{01}} \Omega_H(\boldsymbol{k}) \left(B_{\boldsymbol{k}}^{\dagger}B_{\boldsymbol{k}} +\frac12 \right) .
    \end{split}
\end{align}
Once more, periodicity was used to replace $\boldsymbol{k}$ by $2\boldsymbol{k}_{01}-\boldsymbol{k}$ in the second term. Then, we used that $\Omega_2(2\boldsymbol{k}_{01}-\boldsymbol{k}) = \Omega_1(\boldsymbol{k}) $ and $B_{2\boldsymbol{k}_{01}-\boldsymbol{k},2} = B_{\boldsymbol{k},1}$. We also defined
\begin{align}
    \begin{split}
    \label{eq:PWOMH}
    \Omega_H(\boldsymbol{k}) = \frac{1}{2}\bigg(&N_{11}(\boldsymbol{k}) - N_{11}(2\boldsymbol{k}_{01}-\boldsymbol{k})\\
    &+ \sqrt{\big(N_{11}(\boldsymbol{k}) + N_{11}(2\boldsymbol{k}_{01}-\boldsymbol{k}) \big)^2 - 4\abs{N_{32} (\boldsymbol{k})}^2}  \bigg).
    \end{split}
\end{align}
This single band is plotted in figure \ref{fig:PWHelApp} using the values of the variational parameters we will find minimizes the free energy in chapter \ref{sec:PWF}. It can be shown that in the case of no Zeeman field, these results are equivalent to the results in \cite{Toniolo}. There appears to be a typo in the definition of the coefficient $b_{\boldsymbol{k}}$ in equation (6) of \cite{Toniolo}. The term proportional to $U'$ should be divided by $2$. Then, we can show that $N_{11}(\boldsymbol{k}) = a_{\boldsymbol{k}}$ and $N_{32} (\boldsymbol{k}) = 2b_{\boldsymbol{k}}$ and so $\Omega_H(\boldsymbol{k})$ is the same band reported in \cite{Toniolo}. 

The helicity basis is undefined when $s_{\boldsymbol{k}} = 0$. Additionally, $\gamma_{\boldsymbol{k}}$ is undefined at $\boldsymbol{k}=\boldsymbol{0}$ and $\gamma_{2\boldsymbol{k}_{01}-\boldsymbol{k}}$ is undefined at $\boldsymbol{k} = 2\boldsymbol{k}_{01}$. This is the mathematical explanation for the discontinuities observed in figure \ref{fig:PWHelApp}. They become less pronounced for higher $\lambda_R$, but do not disappear. To avoid such discontinuities, obtain the free energy and hence determine the variational parameters, we attempt to solve the eigenvalue problem in the original spin basis. 


\begin{figure}
    \centering
    \begin{subfigure}{.45\textwidth}
      \includegraphics[width=\linewidth]{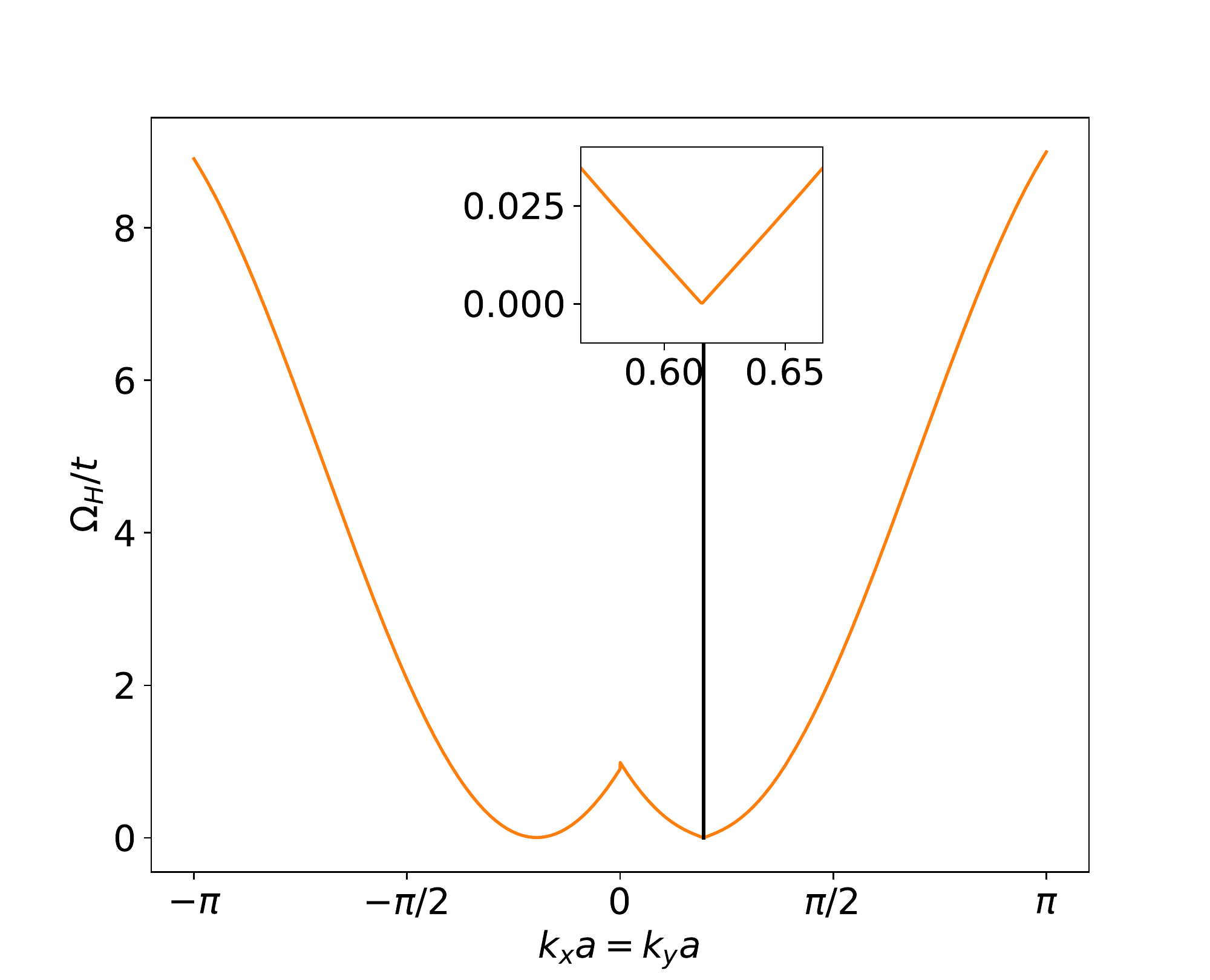}
      \caption{}
    \end{subfigure}%
    \begin{subfigure}{.45\textwidth}
      \includegraphics[width=\linewidth]{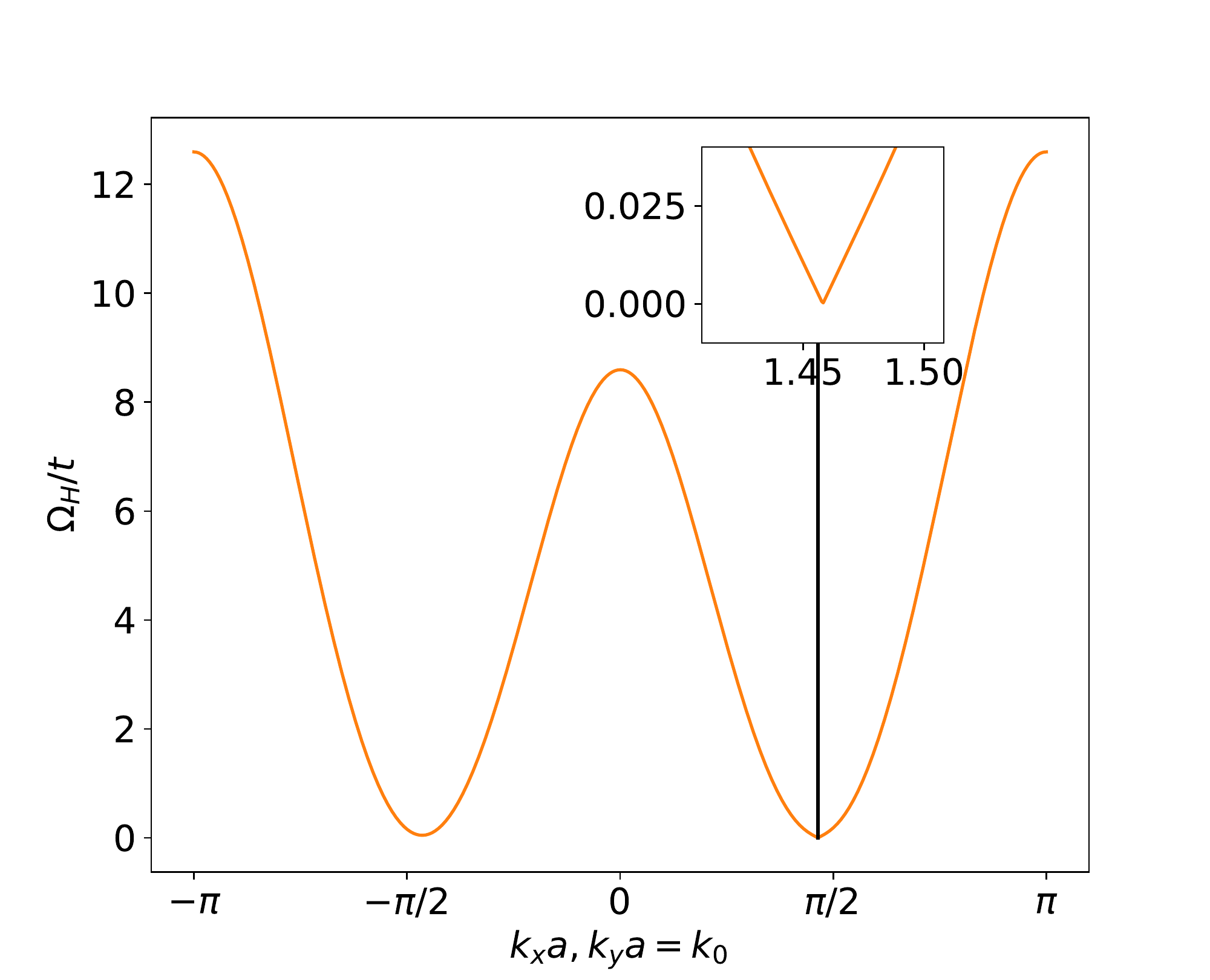}
      \caption{}
    \end{subfigure}%
    \caption{The band $\Omega_H(\boldsymbol{k})$ for  $U_s/t = 0.05$ and $\alpha = 0.9$. The black vertical line indicates the position of $\boldsymbol{k}_{01}$. A zoomed in portion is inserted, showing that the band is linear close to $\boldsymbol{k}_{01}$. In (a) $\Omega_H(\boldsymbol{k})$ is plotted along the direction $k_x=k_y$ with $\lambda_R/t = 1.0$. We observe a discontinuity at $\boldsymbol{k}= 0$ and for lower $\lambda_R/t$ a discontinuity at $2\boldsymbol{k}_{01}$ becomes visible as well. In (b) we plot along $k_x$ for $k_y = k_0 = k_{0m}$ with $\lambda_R/t = 12.5$. Though there is no Zeeman splitting here and a lower value for $U_s/t$ is used, this agrees qualitatively with figure 1 of \cite{Toniolo}. \label{fig:PWHelApp}}
\end{figure}

\subsection{Numeric Eigenvalues in Original Spin Basis}
Letting $-\boldsymbol{k} \to 2\boldsymbol{k}_{01}-\boldsymbol{k}$ in the definition of $\boldsymbol{A}_{\boldsymbol{k}}$ we get
\begin{equation}
    \boldsymbol{A}_{\boldsymbol{k}} = (A_{\boldsymbol{k}}^{\uparrow},A_{2\boldsymbol{k}_{01}-\boldsymbol{k}}^{\uparrow}, A_{\boldsymbol{k}}^{\downarrow},A_{2\boldsymbol{k}_{01}-\boldsymbol{k}}^{\downarrow},A_{\boldsymbol{k}}^{\uparrow\dagger},A_{2\boldsymbol{k}_{01}-\boldsymbol{k}}^{\uparrow\dagger}, A_{\boldsymbol{k}}^{\downarrow\dagger},A_{2\boldsymbol{k}_{01}-\boldsymbol{k}}^{\downarrow\dagger})^T.
\end{equation}
Again using commutator relations, we rewrite $H_2$, simultaneously shifting $H_0$ to 
\begin{align}
    \begin{split}
        H'_0 &= H_0 -\sum_{\boldsymbol{k}\neq \boldsymbol{k}_{01}} \left(4t+\epsilon_{\boldsymbol{k}} + U_s + G_{k_0}\right) \\
        &=  H_0 -4t\cos(k_0 a) - (N_s-1)(4t+U_s+G_{k_0}).
    \end{split}
\end{align}
Hence,
\begin{align}
    \begin{split}
         H &= H'_0 +\frac{1}{4} \sum_{\boldsymbol{k} \neq \boldsymbol{k}_{01}} \boldsymbol{A}_{\boldsymbol{k}}^{\dagger}M_{\boldsymbol{k}}\boldsymbol{A}_{\boldsymbol{k}}.
    \end{split}
\end{align}
The matrix $M_{\boldsymbol{k}}$ takes the form
\begin{gather}
    M_{\boldsymbol{k}} = 
    \begin{pmatrix}
    M_1 & M_2 \\
    M_2^* & M_1^* \\
    \end{pmatrix},
\end{gather}
with
\begin{gather*}
    M_1 = 
    \begin{pmatrix}
    2M_{11}(\boldsymbol{k}) & 0 & 2M_{13}(\boldsymbol{k}) & 0 \\
    0 & 0 & 0 & 0 \\
    2M_{13}^{*}(\boldsymbol{k}) & 0 & 2M_{11}(\boldsymbol{k}) & 0 \\
    0 & 0 & 0 & 0 \\
    \end{pmatrix}
\end{gather*}
and
\begin{gather*}
    M_2^* = 
    \begin{pmatrix}
    0 & M_{52} & 0 & M_{72} \\
    M_{52} & 0 & M_{72} & 0 \\
    0 & M_{72} & 0 & M_{74} \\
    M_{72} & 0 & M_{74} & 0 \\
    \end{pmatrix}.
\end{gather*}
The matrix elements are
\begin{align}
    \begin{split}
    \label{eq:PWelem}
        M_{11}(\boldsymbol{k}) &= \mathcal{E}_{\boldsymbol{k}}+U_s+G_{k_0}, \\
        M_{13}(\boldsymbol{k}) &= s_{\boldsymbol{k}} + U_s\alpha e^{i(\theta_{1}^{\downarrow}-\theta_{1}^{\uparrow})}, \\
        M_{52} &= U_s e^{i2\theta_{1}^{\uparrow}}, \mbox{\qquad} M_{72} = U_s\alpha e^{i(\theta_{1}^{\uparrow} + \theta_{1}^{\downarrow})}, \mbox{\qquad} M_{74} = U_s e^{i2\theta_{1}^{\downarrow}}.\\
    \end{split}
\end{align}
We want to find eigenvalues of $M_{\boldsymbol{k}}J$ using $\det(M_{\boldsymbol{k}}J-\lambda I) = 0$. On this form, it is not possible to get analytic eigenvalues using \textit{Maple}. Numerically, the eigenvalues prove to be complex in the presence of SOC and interactions. We therefore should try to rewrite $M_{\boldsymbol{k}}$ further. Using that we can always rewrite a sum by shifting the summation index, we have 
\begin{equation}
    \sum_{\boldsymbol{k}\in 1\textrm{BZ}}C(\boldsymbol{k})A_{\boldsymbol{k}}^{\uparrow\dagger}A_{\boldsymbol{k}}^{\uparrow} = \sum_{2\boldsymbol{k}_{01}-\boldsymbol{k}\in 1\textrm{BZ}} C(2\boldsymbol{k}_{01}-\boldsymbol{k})A_{2\boldsymbol{k}_{01}-\boldsymbol{k}}^{\uparrow\dagger}A_{2\boldsymbol{k}_{01}-\boldsymbol{k}}^{\uparrow}.
\end{equation}
We would prefer all sums to be over $\boldsymbol{k}\in 1\textrm{BZ}$. Fortunately, we expect the system to be periodic in $\boldsymbol{k}$-space by the size of the 1BZ, i.e. $2\pi/a$. In fact, values of $\boldsymbol{k}$ differing by $2\pi n/a$ with $n$ integer in one or both components are physically equivalent according to chapter 16.2 in \cite{Pitaevskii}. The sum $\sum_{2\boldsymbol{k}_{01}-\boldsymbol{k}\in 1\textrm{BZ}}$ is in fact a sum over an area in $\boldsymbol{k}$-space of equal size to the 1BZ, only shifted by $2\boldsymbol{k}_{01}$. Periodicity suggests that the part of the sum $\sum_{\boldsymbol{k}\in 1\textrm{BZ}}$ we are missing, is the same as the part of the sum that is outside the 1BZ. Thus, we can replace $2\boldsymbol{k}_{01}-\boldsymbol{k}\in 1\textrm{BZ}$ by $\boldsymbol{k}\in 1\textrm{BZ}$ in the sum. Then we can use
\begin{equation}
    \sum_{\boldsymbol{k}\neq \boldsymbol{k}_{01}}C(\boldsymbol{k})A_{\boldsymbol{k}}^{\uparrow\dagger}A_{\boldsymbol{k}}^{\uparrow} = \frac{1}{2}\sum_{\boldsymbol{k}\neq \boldsymbol{k}_{01}}\left(C(\boldsymbol{k})A_{\boldsymbol{k}}^{\uparrow\dagger}A_{\boldsymbol{k}}^{\uparrow}+ C(2\boldsymbol{k}_{01}-\boldsymbol{k})A_{2\boldsymbol{k}_{01}-\boldsymbol{k}}^{\uparrow\dagger}A_{2\boldsymbol{k}_{01}-\boldsymbol{k}}^{\uparrow}\right)
\end{equation}
to rewrite the form of $M_1$ to 
\begin{gather*}
    M_1 = 
    \begin{pmatrix}
    M_{11}(\boldsymbol{k}) & 0 & M_{13}(\boldsymbol{k}) & 0 \\
    0 & M_{11}(2\boldsymbol{k}_{01}-\boldsymbol{k}) & 0 & M_{13}(2\boldsymbol{k}_{01}-\boldsymbol{k}) \\
    M_{13}^{*}(\boldsymbol{k}) & 0 & M_{11}(\boldsymbol{k}) & 0 \\
    0 & M_{13}^{*}(2\boldsymbol{k}_{01}-\boldsymbol{k}) & 0 & M_{11}(2\boldsymbol{k}_{01}-\boldsymbol{k}) \\
    \end{pmatrix}.
\end{gather*}
We do this in all the terms in the Hamiltonian, however, it is only in the terms associated with $M_1$ that it makes a difference. \textit{Maple} is still unable to provide analytic eigenvalues, however, numerically we get four positive energies. As the eigenvalues are obtained numerically the most natural approach would be to name the eigenvalues $\Omega_{i'}(\boldsymbol{k})$ such that $\Omega_{i'} \geq \Omega_{j'}$ when $j'>i'$ for all $\boldsymbol{k}$. If so, all the eigenvalues would be inversion symmetric about $\boldsymbol{k}_{01}$. 

We however recognize two bands, where one is similar to the band $\Omega_H(\boldsymbol{k})$ we found in the preceding section and one seems like a natural generalization of the upper helicity band \eqref{eq:helicitybands}, $\lambda_{\boldsymbol{k}}^+$, that was neglected in the same section. We define these as $\Omega_1(\boldsymbol{k})$ and $\Omega_2(\boldsymbol{k})$. We are now able to recognize the other two energies as the inversions of $\Omega_1(\boldsymbol{k})$ and $\Omega_2(\boldsymbol{k})$ about $\boldsymbol{k}_{01}$. We name these bands $\Omega_1^{'}(\boldsymbol{k})$ and $\Omega_2^{'}(\boldsymbol{k})$ and note that $\Omega_1^{'}(2\boldsymbol{k}_{01}-\boldsymbol{k}) = \Omega_1(\boldsymbol{k})$ and $\Omega_2^{'}(2\boldsymbol{k}_{01}-\boldsymbol{k}) = \Omega_2(\boldsymbol{k})$. 

The argument that the operator $B_{2\boldsymbol{k}_{01}-\boldsymbol{k}, 1}^{'}$ associated with $\Omega_1^{'} (2\boldsymbol{k}_{01}-\boldsymbol{k})$ can be defined to be equal to the operator $B_{\boldsymbol{k},1}$ associated with $\Omega_1(\boldsymbol{k})$ still holds. Similarly, the operator $B_{2\boldsymbol{k}_{01}-\boldsymbol{k}, 2}^{'}$ associated with $\Omega_2^{'}(2\boldsymbol{k}_{01}-\boldsymbol{k})$ can be defined to be equal to the operator $B_{\boldsymbol{k},2}$ associated with $\Omega_2(\boldsymbol{k})$. Using this, along with $\Omega_1^{'}(2\boldsymbol{k}_{01}-\boldsymbol{k}) = \Omega_1(\boldsymbol{k})$ and $\Omega_2^{'}(2\boldsymbol{k}_{01}-\boldsymbol{k}) = \Omega_2(\boldsymbol{k})$ we obtain
\begin{align}
    \begin{split}
        H_2 = \frac{1}{4}\sum_{\boldsymbol{k} \neq \boldsymbol{k}_{01}}& \boldsymbol{B}_{\boldsymbol{k}}^{\dagger}D\boldsymbol{B}_{\boldsymbol{k}} \\
        =\frac{1}{2} \sum_{\boldsymbol{k} \neq \boldsymbol{k}_{01}}& \Bigg( \Omega_1(\boldsymbol{k})\left(B_{\boldsymbol{k}, 1}^{\dagger}B_{\boldsymbol{k}, 1} +\frac12 \right) + \Omega_2(\boldsymbol{k})\left(B_{\boldsymbol{k}, 2}^{\dagger}B_{\boldsymbol{k}, 2} +\frac12 \right)\\
        & + \Omega_1^{'}(\boldsymbol{k})\left(B_{\boldsymbol{k}, 1}^{'\dagger}B_{\boldsymbol{k}, 1}^{'} +\frac12 \right) + \Omega_2^{'}(\boldsymbol{k})\left(B_{\boldsymbol{k}, 2}^{'\dagger}B_{\boldsymbol{k}, 2}^{'} +\frac12 \right)\Bigg) \\
        =\frac{1}{2} \sum_{\boldsymbol{k} \neq \boldsymbol{k}_{01}}& \Bigg( \Omega_1(\boldsymbol{k})\left(B_{\boldsymbol{k}, 1}^{\dagger}B_{\boldsymbol{k}, 1} +\frac12 \right) + \Omega_2(\boldsymbol{k})\left(B_{\boldsymbol{k}, 2}^{\dagger}B_{\boldsymbol{k}, 2} +\frac12 \right)\\
        & + \Omega_1^{'}(2\boldsymbol{k}_{01}-\boldsymbol{k})\left(B_{2\boldsymbol{k}_{01}-\boldsymbol{k}, 1}^{'\dagger}B_{2\boldsymbol{k}_{01}-\boldsymbol{k}, 1}^{'} +\frac12 \right)\\
        & + \Omega_2^{'}(2\boldsymbol{k}_{01}-\boldsymbol{k})\left(B_{2\boldsymbol{k}_{01}-\boldsymbol{k}, 2}^{'\dagger}B_{2\boldsymbol{k}_{01}-\boldsymbol{k}, 2}^{'} +\frac12 \right)\Bigg) \\
        =\sum_{\boldsymbol{k} \neq \boldsymbol{k}_{01}}& \sum_{\sigma=1}^2 \Omega_\sigma(\boldsymbol{k})\left(B_{\boldsymbol{k}, \sigma}^{\dagger}B_{\boldsymbol{k}, \sigma} +\frac12 \right).
    \end{split}
\end{align}


We will think of, and present $\Omega_1(\boldsymbol{k})$ and $\Omega_2(\boldsymbol{k})$ as the two energy bands in the PW phase. For the purpose of numerics, it is however easier to use the four bands $\Omega_{\sigma'}(\boldsymbol{k})$, $\sigma' = 1',2',3',4'$, that are inversion symmetric about $\boldsymbol{k} = \boldsymbol{k}_{01}$. For $H_2$ we arrive at
\begin{equation}
    H_2 = \frac{1}{2} \sum_{\boldsymbol{k} \neq \boldsymbol{k}_{01}}\sum_{\sigma'=1'}^{4'} \Omega_{\sigma'}(\boldsymbol{k})\left(B_{\boldsymbol{k}, \sigma'}^{\dagger}B_{\boldsymbol{k}, \sigma'} +\frac12 \right),
\end{equation}
where we order the numerical eigenvalues such that $\Omega_{1'}\geq\Omega_{2'}\geq\Omega_{3'}\geq\Omega_{4'}$ for all $\boldsymbol{k}$. The results should be equal in the two approaches, the latter is used solely because it provides simpler numerical calculation of the free energy.

\subsection{Free Energy} \label{sec:PWF}
We have
\begin{align}
    \begin{split}
        H = &H'_0 + \frac{1}{2} \sum_{\boldsymbol{k} \neq \boldsymbol{k}_{01}}\sum_{\sigma'=1'}^{4'} \Omega_{\sigma'}(\boldsymbol{k})\left(B_{\boldsymbol{k}, \sigma'}^{\dagger}B_{\boldsymbol{k}, \sigma'} +\frac12 \right).
    \end{split}
\end{align}
At zero temperature, the free energy is equal to $\langle H \rangle$, and using \eqref{eq:F}
\begin{align}
    \begin{split}
    \label{eq:FPW}
        F_\textrm{PW} = &N(\epsilon_{\boldsymbol{k}_{01}}+T)+N\abs{s_{\boldsymbol{k}_{01}}}\cos(\gamma_{\boldsymbol{k}_{01}}+\Delta\theta_1) + \frac{UN^2}{4N_s}(1+\alpha) \\
        &-4t\cos(k_0 a) - (N_s-1)(4t+U_s+G_{k_0}) + \frac{1}{4}\sum_{\boldsymbol{k} \neq \boldsymbol{k}_{01}}\sum_{\sigma'=1'}^{4'} \Omega_{\sigma'}(\boldsymbol{k}).
    \end{split}
\end{align}

Before we start minimizing $F_\textrm{PW}$ we need to discuss how the excitation spectrum behaves at different parameters. As will be explained in chapter \ref{sec:PWexcitation} we need to keep $\alpha<1$ because for $\alpha\geq 1$ there are secondary minima in the excitation spectrum. If $k_0$ deviates too much from the value $k_{0m}$ found to minimize $H_{0}^\textrm{PW}$ in chapter \ref{sec:PWphaseH0} the excitation spectrum becomes complex.  In order to diagonalize the Hamiltonian the eigenvalues need to be real. It therefore only makes sense to investigate $F_\textrm{PW}$ for the values of $k_0$ such that all $\Omega_{\sigma'}(\boldsymbol{k})$ are real. The minimum of $F_\textrm{PW}$ within this set of $k_0$ values will then be used to estimate the value of $k_0$ that minimizes $F_\textrm{PW}$, which we name $k_{0\textrm{min}}$.  

Also, we fix $U \ll t$ and to be specific we set $U/t=1/10$. We also focus on $N/N_s = 1$, such that $U_s = 0.05$. The values of $U_s$, $\alpha$ and $\lambda_R$ are kept fixed as $k_0$ is varied. The dependence on $T$ is inconsequential, and therefore the energy offset is set to zero. We assume \eqref{eq:gammathetapi} holds while minimizing with respect to $k_0$.

Plots of $F_\textrm{PW}$ as a function of $k_0$ are given in figure \ref{fig:PWFs}. We see that the value of $k_0$ that minimizes $F_\textrm{PW}$ is slightly smaller than $k_{0m}$ for a lattice size of $9\cdot 10^{4}$. However, we also find that $k_{0\textrm{min}}$ approaches $k_{0m}$ from below as the lattice size is increased. It is found that $k_{0\textrm{min}} a \approx 0.3396$ minimizes $F_\textrm{PW}$ while $k_{0m} a \approx 0.3398$ and the relative difference in $k_0$ is of order $\order{10^{-4}}$. The corresponding relative reduction of $F_\textrm{PW}$ is of order $\order{10^{-8}}$. This will not alter the plots of the excitation spectrum in a visible way. We also expect these differences will approach zero as the lattice size becomes larger. We therefore state that $k_{0\textrm{min}} = k_{0m}$ and use $k_0 = k_{0m}$ when producing the figures in chapter \ref{sec:PWexcitation}. Apart from being a finite size effect, there is another reason we may neglect the small difference between $k_{0\textrm{min}}$ and $k_{0m}$. We have been thinking of $k_0$ as continuous, but in fact it is not. The point $\boldsymbol{k}_{01}$ has to be a lattice site, and so $k_0$ is discretized. Setting $k_0 = k_{0m}$ one has to take into account that it may be shifted a small amount to coincide with a lattice site. Hence, when the difference between $k_{0\textrm{min}}$ and $k_{0m}$ is small, they most likely correspond to the same lattice site.

The remaining free parameters are $\theta_{1}^\uparrow$ and $\theta_{1}^\downarrow$. By varying $\theta_1^\uparrow$ with $\theta_1^\downarrow$ fixed, it is found that $F_{\textrm{PW}}$ is minimized by $\theta_{1}^\downarrow-\theta_{1}^\uparrow = \pi/4$ which agrees with \eqref{eq:gammathetapi}. Once $\theta_{1}^\downarrow = \theta_{1}^\uparrow+\pi/4$ is set, it turns out the fluctuations in $F_\textrm{PW}(\theta_{1}^\uparrow)$ are negligible (of order $\order{10^{-13}}$ compared to the value of $F_{\textrm{PW}}$). We conclude that $F_\textrm{PW}$ is independent of $\theta_{1}^\uparrow$ meaning the two phase factors are constrained by $\theta_{1}^\downarrow - \theta_{1}^\uparrow = \pi/4$ only.

\begin{figure}
    \centering
    \begin{subfigure}{.49\textwidth}
      \includegraphics[width=\linewidth]{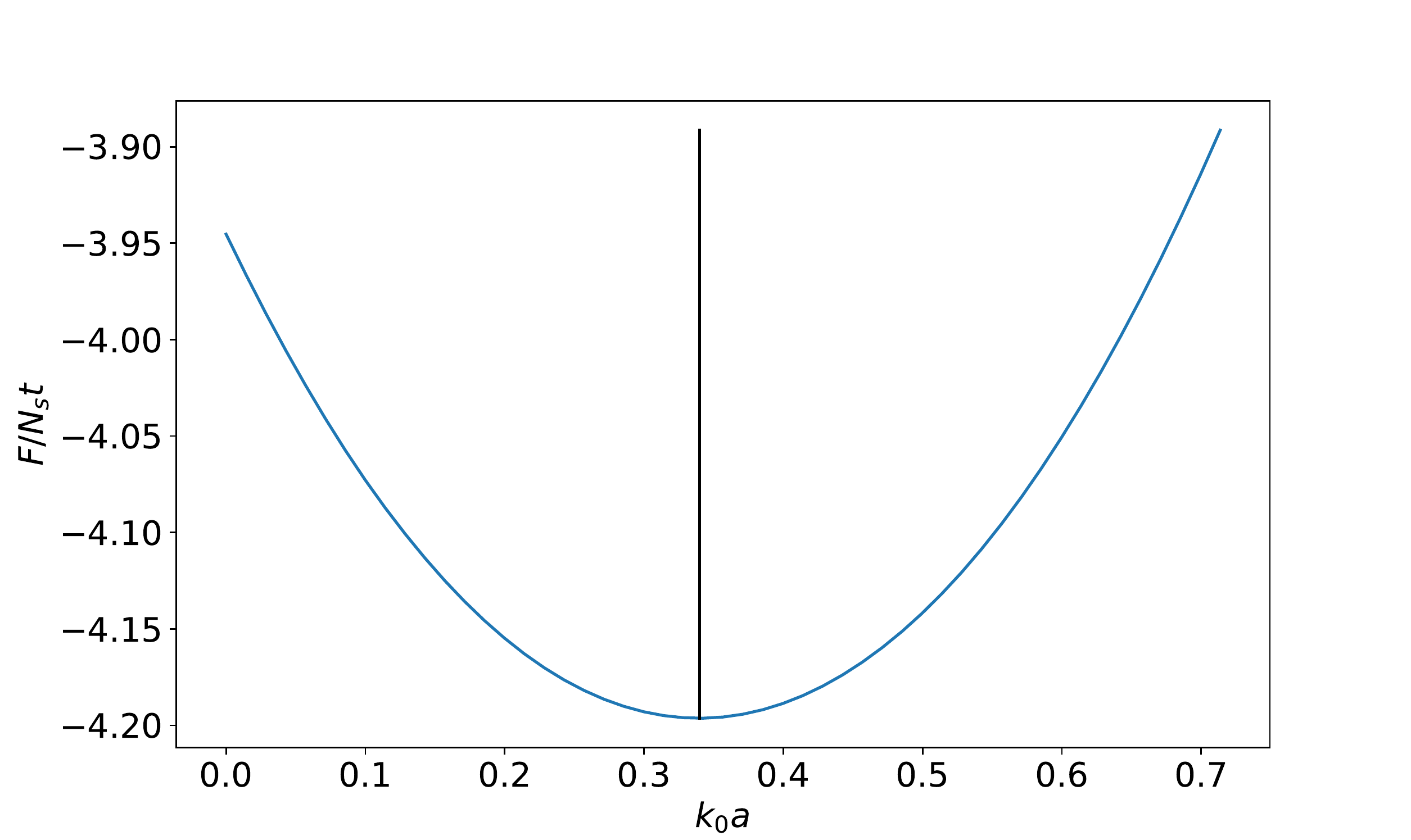}
      \caption{ }
      \label{fig:PWFzoom}
    \end{subfigure}%
    \begin{subfigure}{.49\textwidth}
      \includegraphics[width=0.9\linewidth]{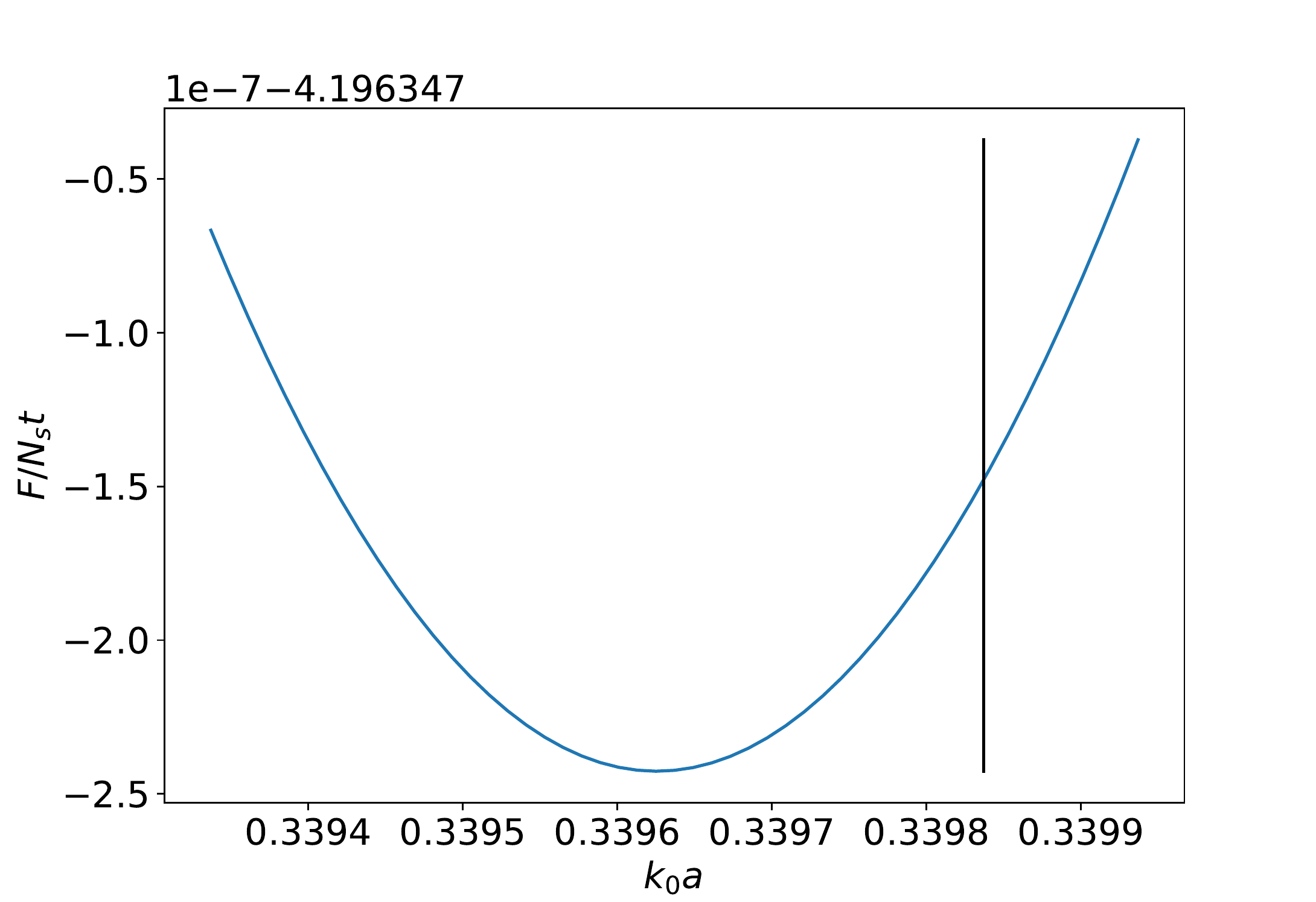}
      \caption{ }
      \label{fig:PWFzoom2}
    \end{subfigure}
    \caption{Shows $F_{\textrm{PW}}$ as a function of $k_0$ in (a) and zoomed into the minimum in (b). The black vertical lines indicate the position of $k_0 = k_{0m}$. For both figures $T=0$, $U_s/t = 0.05$, $\alpha = 0.9$ and $\lambda_R/t = 0.5$ were fixed, and 51 values of $k_0$ were considered. A lattice size of $9\cdot 10^{4}$ was used in (b). \label{fig:PWFs}}
\end{figure}

\subsection{Excitation Spectrum} \label{sec:PWexcitation}

\begin{figure}
    \centering
    \includegraphics[width=0.7\linewidth]{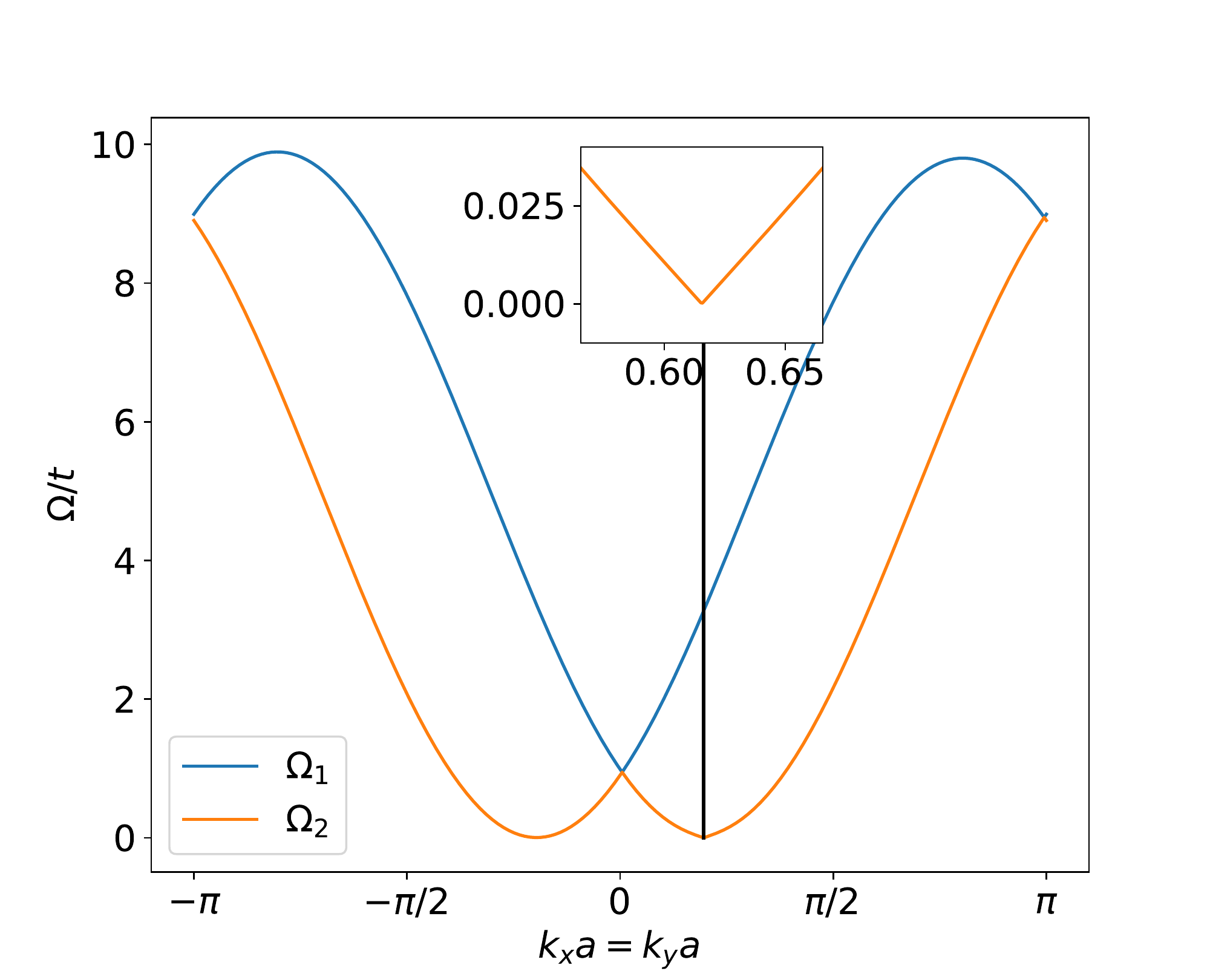}
    \caption{Shows the bands $\Omega_{\sigma}(\boldsymbol{k})$ along the $k_x=k_y$ direction for $U_s/t = 0.05$, $\alpha = 0.9$ and $\lambda_R/t = 1.0$. The black vertical line shows the position of $k_x=k_y=k_{0m}$. Inserted is a closer look at the behavior close to the minimum at $\boldsymbol{k}_{01}$ indicating linear dispersion.}
    \label{fig:PWbandl1}
\end{figure}

The bands $\Omega_1(\boldsymbol{k})$ and $\Omega_2(\boldsymbol{k})$ are shown in figure \ref{fig:PWbandl1} along the direction $k_x=k_y$. The analytic energy $\Omega_H(\boldsymbol{k})$ is shown in the 1BZ in figure \ref{fig:PWeigenl1H}. The energy $\Omega_2(\boldsymbol{k})$ should look very similar. It has been checked that the energy spectrum has its global minimum at $\boldsymbol{k}=\boldsymbol{k}_{01}$, consistent with our initial assumption that the Bose gas condenses into a state with this $\boldsymbol{k}$. In addition, gapped roton minima appear close to $-\boldsymbol{k}_{01}$ and $\pm\boldsymbol{k}_{02}$, as was reported in \cite{Toniolo}. Using a greater resolution it appears the lower bands might be complex close to $\boldsymbol{k}_{01}$. The eigenvalues at $\boldsymbol{k} = \boldsymbol{k}_{01}$ can be obtained analytically, and are in fact real. Hence, any imaginary parts there are numerical errors. However, we are excluding $\boldsymbol{k}_{01}$ from $H_2$, so what really matters is if $\Omega_\sigma(\boldsymbol{k}\neq\boldsymbol{k}_{01})$ are complex.

\begin{figure}
    \centering
    \includegraphics[width=0.9\linewidth]{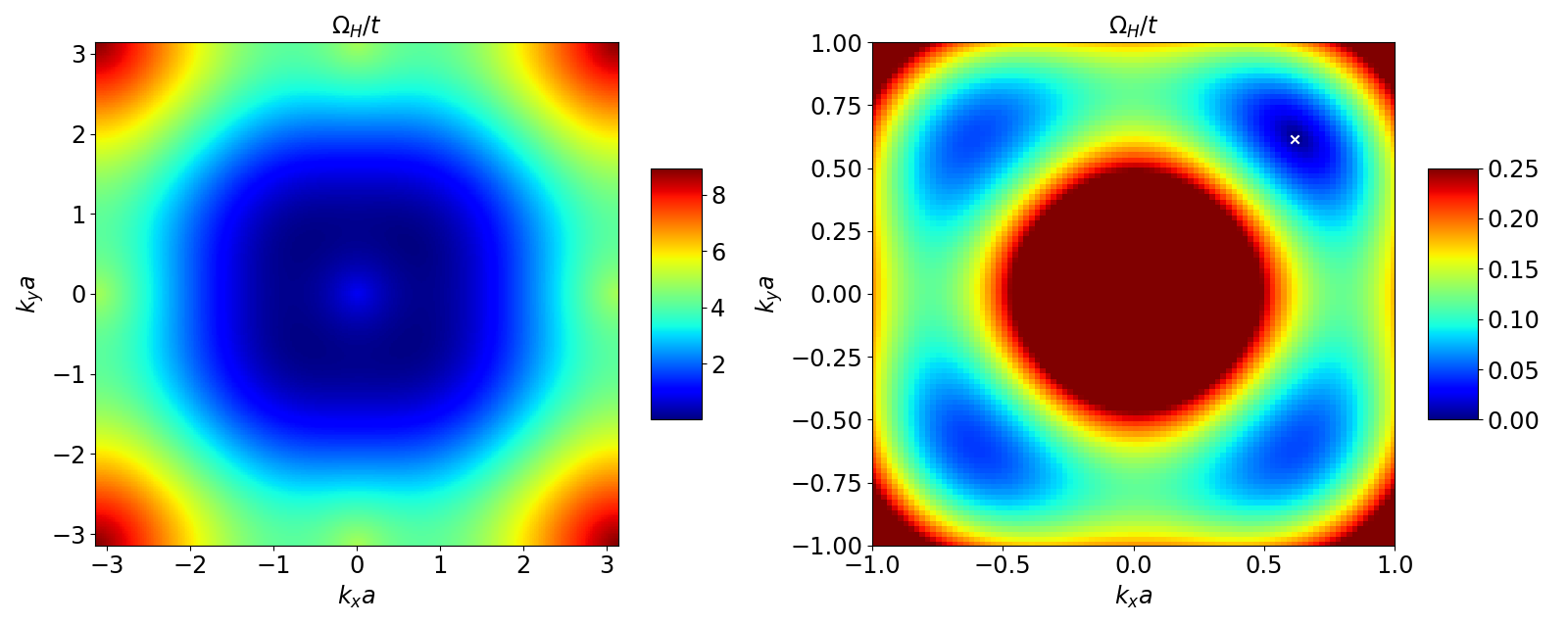}
    \caption{$\Omega_H(\boldsymbol{k})$ for $U_s/t = 0.05$, $\alpha = 0.1$ and $\lambda_R/t = 1.0$. For both $k_x$ and $k_y$ 301 points are considered. In the right figure we focus on the behavior at the phonon minimum at $\boldsymbol{k}_{01}$ and the gapped roton minima at $-\boldsymbol{k}_{01}$ and $\pm\boldsymbol{k}_{02}$. The white cross shows the position of $k_x=k_y=k_{0m}$.}
    \label{fig:PWeigenl1H}
\end{figure}



\begin{figure}
    \centering
    \includegraphics[width=0.8\linewidth]{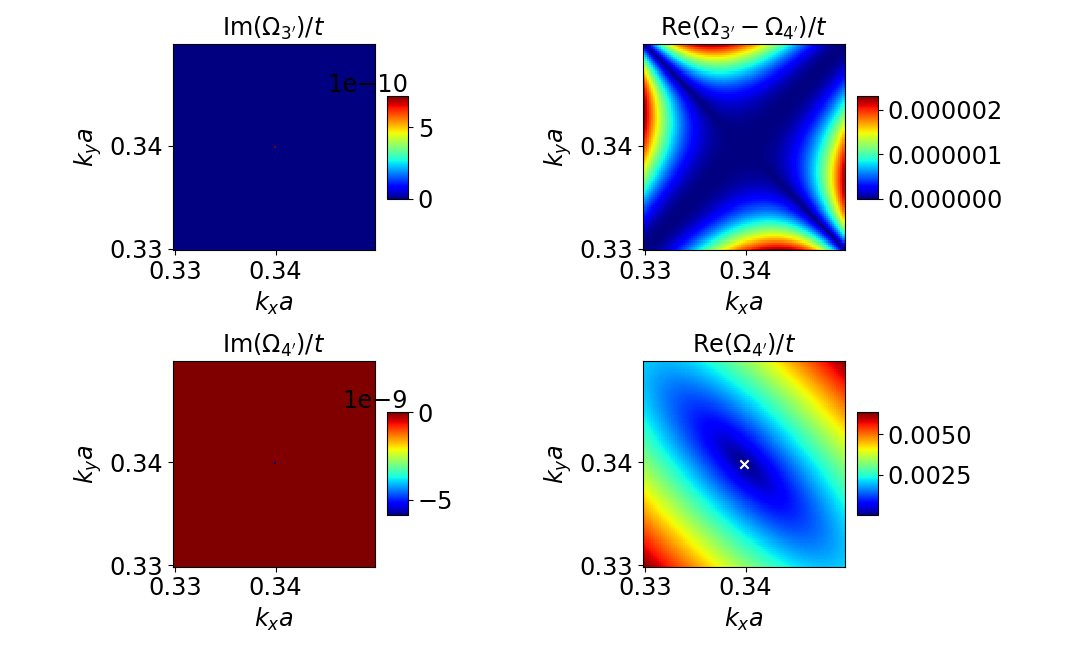}
    \caption{Shows real and imaginary parts of $\Omega_{3'}(\boldsymbol{k})$ and $\Omega_{4'}(\boldsymbol{k})$ in a small area around $\boldsymbol{k}_{01}$ for $U_s/t = 0.05$, $\alpha = 0.9$ and $\lambda_R/t = 0.5$. The white cross shows the position of $k_x=k_y=k_{0m}$. For both $k_x$ and $k_y$ 101 points are considered and this greater resolution reveals the possibility of complex eigenvalues close to $\boldsymbol{k}_{01}$.}
    \label{fig:PWeigenl1zoom}
\end{figure}

\begin{figure}
    \centering
    \includegraphics[width=\linewidth]{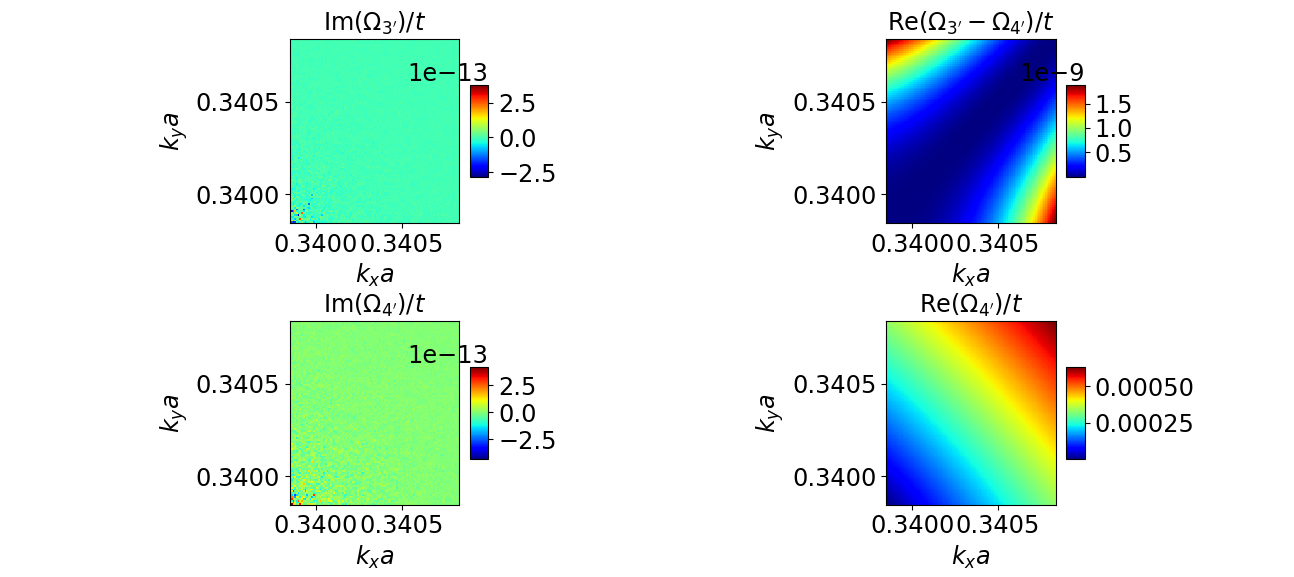}
    \caption{Shows real and imaginary parts of $\Omega_{3'}(\boldsymbol{k})$ and $\Omega_{4'}(\boldsymbol{k})$ in a small area close to $\boldsymbol{k}_{01}$ for $U_s/t = 0.05$, $\alpha = 0.9$ and $\lambda_R/t = 0.5$. The lower left edge is shifted $10^{-5}(1,1)$ from $\boldsymbol{k}_{01} a$. For both $k_x$ and $k_y$ 101 points are considered.}
    \label{fig:PWeigenl1close}
\end{figure}

Figure \ref{fig:PWeigenl1zoom} shows the numeric energies $\Omega_{3'}(\boldsymbol{k})$ and $\Omega_{4'}(\boldsymbol{k})$ for a small area around $\boldsymbol{k}_{01}$. In such a small area around $\boldsymbol{k}_{01}$ we find that $\Omega_{3'}(\boldsymbol{k}) \approx \Omega_{4'}(\boldsymbol{k})$ and they are therefore both approximately equal to $\Omega_2(\boldsymbol{k})$. The reason for this is that $\Omega_2(\boldsymbol{k})$ is always equal to either $\Omega_{3'}(\boldsymbol{k})$ or $\Omega_{4'}(\boldsymbol{k})$ depending on the value of $\boldsymbol{k}$. It appears the complex eigenvalues are contained in an area very close to $\boldsymbol{k}_{01}$. 

As the imaginary parts at $\boldsymbol{k} = \boldsymbol{k}_{01}$ can be shown to be numerical errors, it seems like similar numerical errors appear very close to $\boldsymbol{k} = \boldsymbol{k}_{01}$ as well.  To gain further insight into the problem, figure \ref{fig:PWeigenl1close} gives a version of figure \ref{fig:PWeigenl1zoom} for a small area close to but not including $\boldsymbol{k}_{01}$. Even at a distance as small as $10^{-5}/a$ from $\boldsymbol{k}_{01}$ the imaginary parts have already dropped several orders of magnitude, and quickly continues to drop as we move away from $\boldsymbol{k}_{01}$. We suggest treating these small imaginary parts as numerical errors, and if such an interpretation is valid, we conclude the PW phase is stable, at least for $\alpha < 1$.

The roton minimum at $-\boldsymbol{k}_{01}$ is shown in figure \ref{fig:PWbandBecomes} for $\alpha<1$ and $\alpha>1$. 
Regarding the stability of the PW phase, we investigate what happens to the roton minimum at $-\boldsymbol{k}_{01}$ as we increase $\alpha$. It becomes less gapped, and at $\alpha = 1$ it is found to be gapless. Also interesting, is what happens as we pass $\alpha = 1$. The lowest band develops secondary minima. The minimum at $\boldsymbol{k}_{01}$ is still $0$ and thus one among many global minima. However, for $\alpha>1$ this is not the true spectrum. If we investigate the BV norms of the eigenvectors for values of $\boldsymbol{k}$ between the secondary minima, we find that the lowest eigenvalue will enter the digonalized Hamiltonian with a negative sign in these areas. Thus, one has to imagine the lower band being mirrored around zero energy in the areas between the secondary minima. 

In the end we find a global minimum of the spectrum that is not the phonon minimum at $\boldsymbol{k}_{01}$ but the roton minimum at $-\boldsymbol{k}_{01}$. Such a thing happens for any $\lambda_R/t>0$ as $\alpha$ becomes larger than 1. When $\boldsymbol{k}_{01}$ is no longer the global minimum of the excitation spectrum, we have a violation of the initial assumption that the system condenses into a state with $\boldsymbol{k} = \boldsymbol{k}_{01}$, suggesting the PW phase becomes unstable. This is an energetic instability according to chapter 14.3 in \cite{PethickSmith} as opposed to dynamic instabilities which are connected to complex eigenvalues of $M_{\boldsymbol{k}}J$. The term energetic instability refers to the fact that there is a lower energy state available. We also interpret the gapless roton minimum at $\alpha=1$ as an indication of energetic instability.

\begin{figure}
    \centering
    \includegraphics[width=0.6\linewidth]{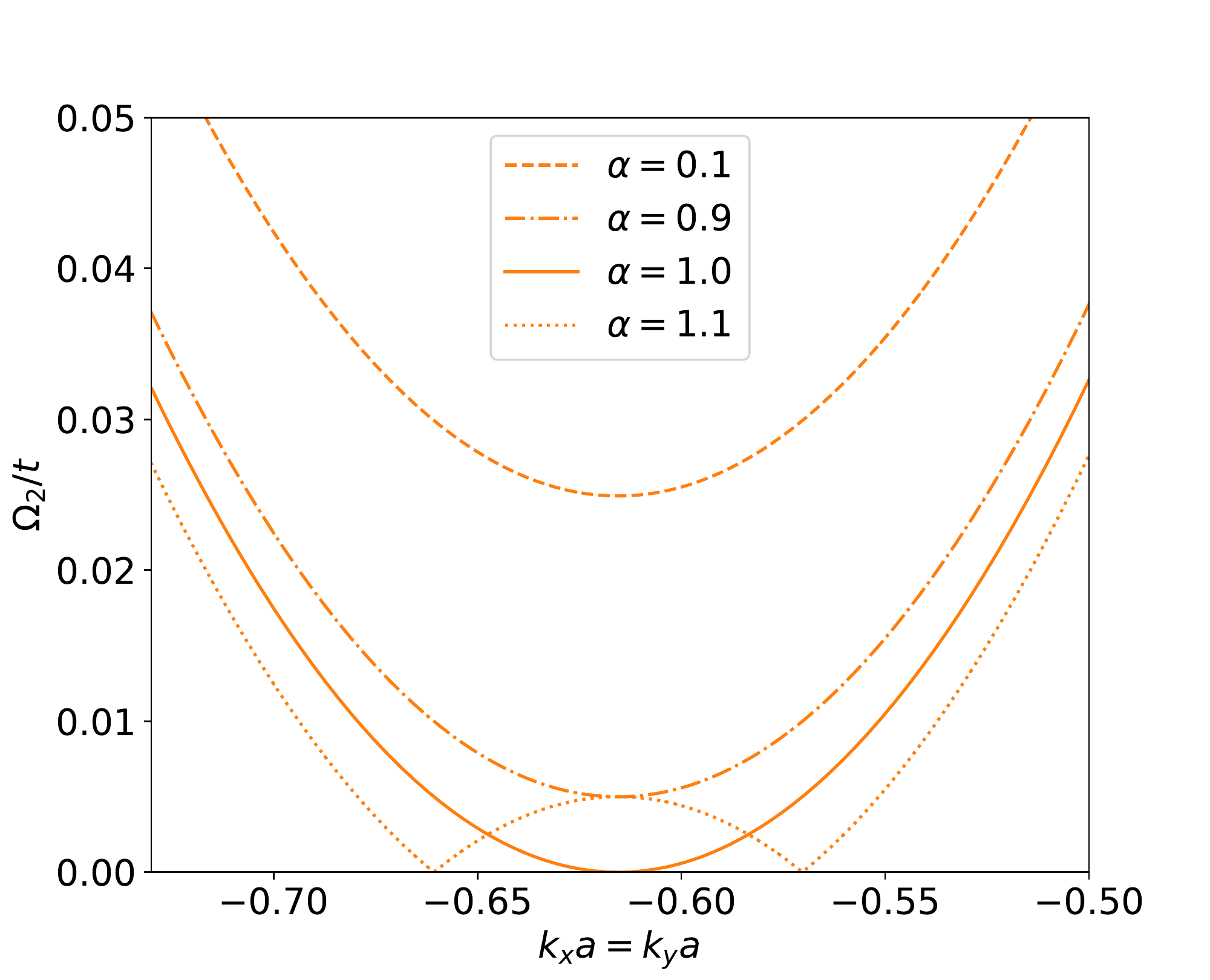}
    \caption{Shows the band $\Omega_2(\boldsymbol{k})$ along the $k_x=k_y$ direction for $U_s/t = 0.05$ and $\lambda_R/t = 1.0$, focusing on the behavior of the roton minimum at $-\boldsymbol{k}_{01}$ for several $\alpha$. Notice that it becomes ungapped for $\alpha=1.0$. The result for $\alpha > 1$ is discussed in the text. \label{fig:PWbandBecomes}}
\end{figure}

\subsection{Critical Superfluid Velocity}
In addition to having a minimal value of $0$, $\Omega_{2}(\boldsymbol{k})$ and $\Omega_H(\boldsymbol{k})$ are linear close to $\boldsymbol{k}_{01}$, suggesting we have a nonzero critical superfluid velocity. It also seems from figure \ref{fig:PWeigenl1zoom} that the critical superfluid velocity will depend on the direction. Such an anisotropic critical superfluid velocity was also found in \cite{Toniolo}. 

Using the analytic band $\Omega_H(\boldsymbol{k})$ in \eqref{eq:PWOMH} we can calculate an analytic expression for the critical superfluid velocity in e.g. the $k_x$-direction. Such a calculation was performed in \cite{Toniolo}, and using equation (12) in \cite{Toniolo} applied to $\Omega_H(\boldsymbol{k})$, we find
\begin{equation}
    v_x = a\sqrt{U_s(1+\alpha)}\sqrt{2t\cos(k_0 a)-2\lambda_R\frac{3\cos^2(k_0 a)-2}{2\sqrt{2}\sin(k_0 a)}}.
\end{equation}
We insert $k_0 a = k_{0m} a = \arctan(\lambda_R/\sqrt{2}t)$ and use
\begin{equation}
    \cos(\arctan(x)) = \frac{1}{\sqrt{1+x^2}} \mbox{\qquad and \qquad} \sin(\arctan(x)) = \frac{x}{\sqrt{1+x^2}}
\end{equation}
to obtain
\begin{equation}
\label{eq:PWanalyticvx}
    v_x = a\sqrt{U_s(1+\alpha)\frac{t^2+\lambda_R^2}{\sqrt{t^2+\lambda_R^2/2}}}.
\end{equation}
We see the superfluid velocity increases with increasing $\alpha$ and with increasing $\lambda_R$. Setting the Zeeman term to zero, this is the same result given in \cite{Toniolo}. Alternatively we can derive $v_x$ by setting $k_x = k_0 + q$, $k_y = k_0$ and expanding $\Omega_H(\boldsymbol{k})$ for small $q$ based on \eqref{eq:vc}. Finally a division by $q$ will yield the critical superfluid velocity, and the result is the same as above.

We can also find an analytic expression for the dependence on the angle $\phi$ made with the $k_x$-axis, which we will name $v_c^{\textrm{an}}(\phi)$. We set $k_x = k_0+q\cos(\phi)$ and $k_y = k_0 + q\sin(\phi)$. Expanding for small $q$, we find
\begin{align}
    \begin{split}
    \Omega_H(k_0 +q\cos(\phi)&, k_0+ q\sin(\phi)) \approx qa\sqrt{U_s(1+\alpha)}\\
    &\cdot\sqrt{2t\cos(k_0 a)-2\lambda_R\frac{\cos^2(k_0 a)(3-\sin(2\phi))-2}{2\sqrt{2}\sin(k_0 a)}} .
    \end{split}
\end{align}
Inserting $k_0 = k_{0m}$ yields
\begin{equation}
\label{eq:PWvcan}
    v_c^{\textrm{an}}(\phi) = a \sqrt{U_s(1+\alpha)\frac{t^2(1+\sin(2\phi))+\lambda_R^2}{\sqrt{t^2+\lambda_R^2/2}}}.
\end{equation}
As we can see, the expression is $\pi$-periodic as expected since $\Omega_H(\boldsymbol{k})$, $\Omega_{2}(\boldsymbol{k})$ and $\Omega_{4'}(\boldsymbol{k})$ appear to be inversion symmetric about $\boldsymbol{k}_{01}$ close to $\boldsymbol{k}_{01}$. The maximum value occurs for $\phi = \pi/4$ and the minimum value at $\phi = 3\pi/4$, which fits well with figure \ref{fig:PWeigenl1zoom}. In other words, the direction in which the critical superfluid velocity is largest, is parallel to $\boldsymbol{k}_{01}$, while the direction in which it is smallest is normal to $\boldsymbol{k}_{01}$. For $\phi = 0$ it is the same as $v_x$ in \eqref{eq:PWanalyticvx}. Interestingly, the critical superfluid velocity does not become the isotropic value found in the NZ phase, $\sqrt{2U_s t a^2 (1 - \alpha)}$, if the SOC is set to zero. This is an example of the fact that introducing SOC to the system is a highly nontrivial perturbation.

As we are working in natural units, where $\hbar = 1$, we find that energy times length has the same dimension as velocity. We therefore measure $v_c$ in units of $ta$ when we are plotting numeric results. Consulting the critical superfluid velocity in \eqref{eq:PWvcan}, we see that $v_c/ta$ is a natural choice when we are measuring all energies in units of $t$. The analytic critical superfluid velocity is shown in figure \ref{fig:PWsuperal} as a function of $\alpha$ for various $\lambda_R$. In figure \ref{fig:PWsuperlamnew} we plot it as a function of $\lambda_R$ for several values of $\alpha$. Both figures show that $v_c^{\textrm{an}}$ increases with increasing $\lambda_R$ and with increasing $\alpha$. Finally, figure \ref{fig:PWsuperphi} shows the critical superfluid velocity as a function of the angle made with the $k_x$-axis.


\begin{figure}
    \centering
    \includegraphics[width=0.7\linewidth]{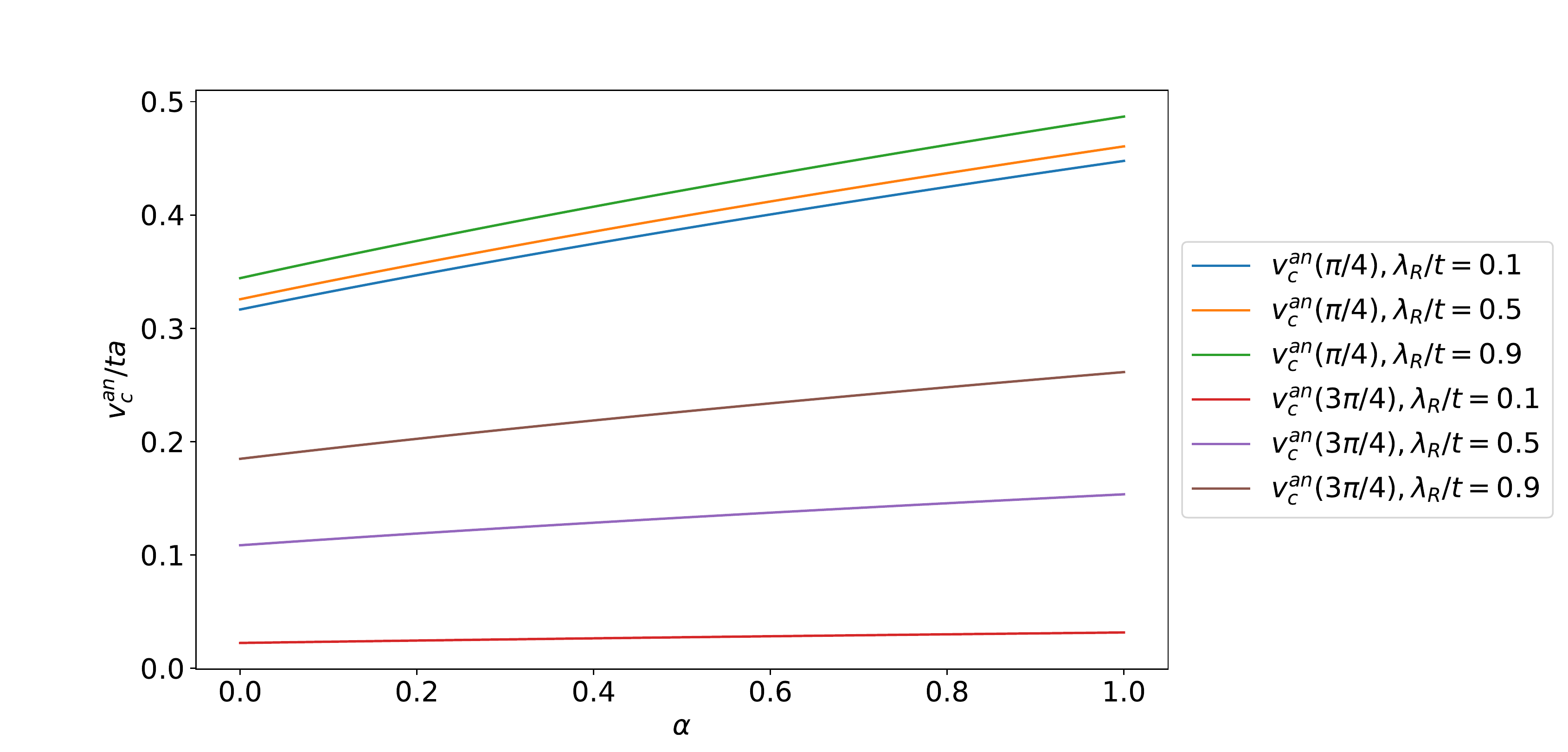}
    \caption{Maximum and minimum values of $v_c^{\textrm{an}}(\phi)$ are plotted against $\alpha$ for various $\lambda_R$ with $U_s/t = 0.05$.}
    \label{fig:PWsuperal}
\end{figure}

\begin{figure}
    \centering
    \includegraphics[width=0.7\linewidth]{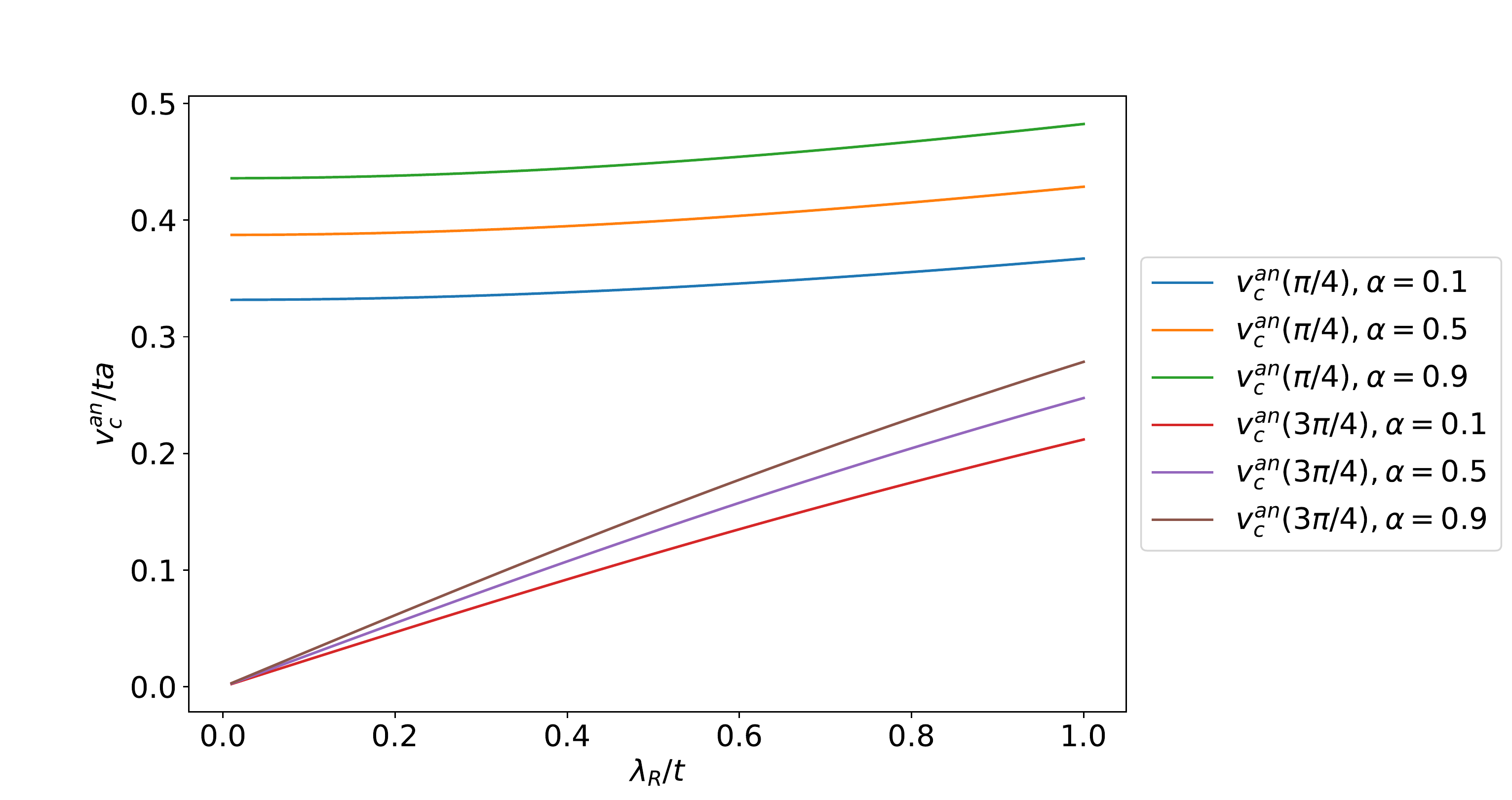}
    \caption{Maximum and minimum values of $v_c^{\textrm{an}}(\phi)$ are plotted against $\lambda_R$ for various $\alpha$, with $U_s/t = 0.05$. Note that $\lambda_R=0$ is not included, as that describes the NZ phase.}
    \label{fig:PWsuperlamnew}
\end{figure}

\begin{figure}
    \centering
    \includegraphics[width=0.7\linewidth]{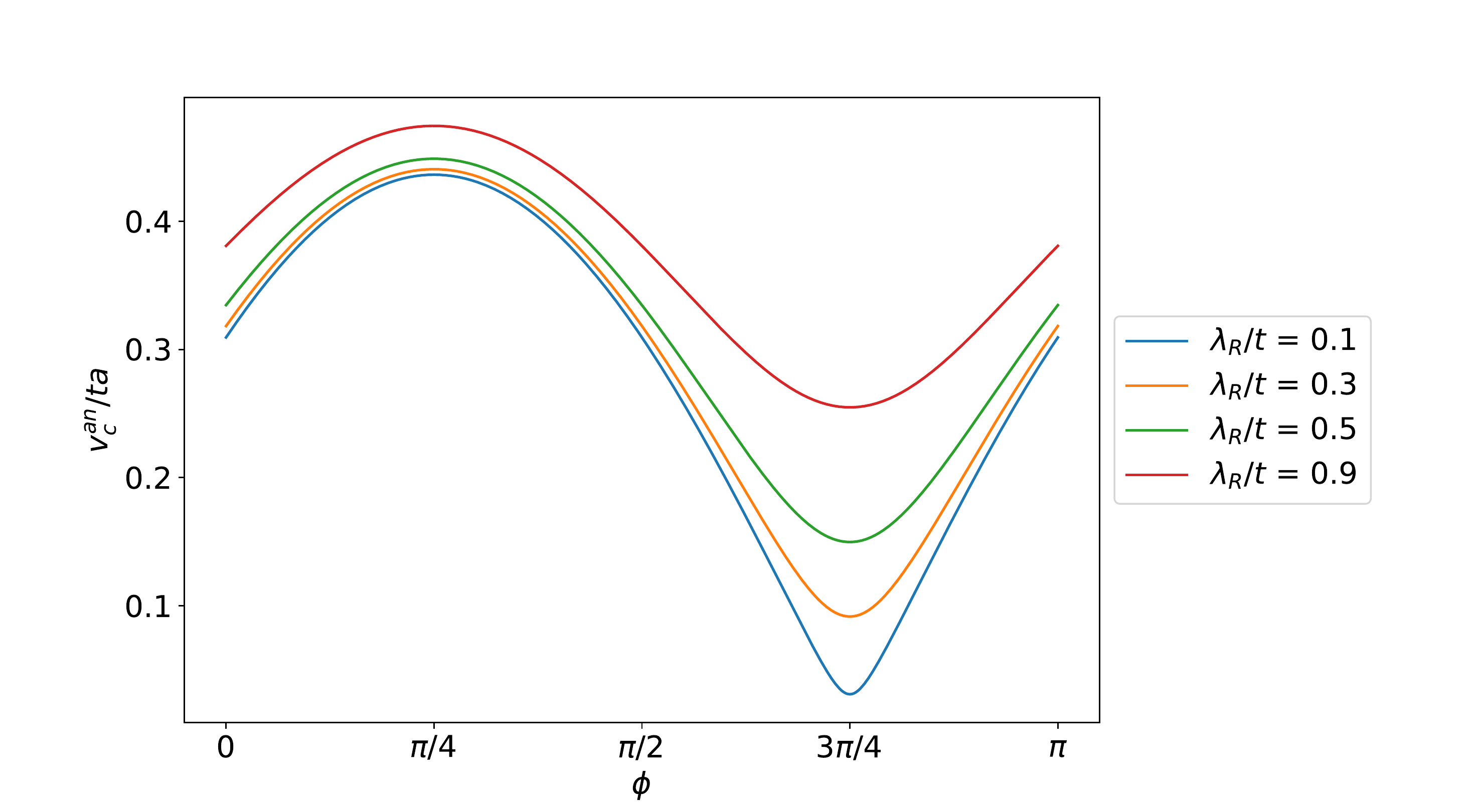}
    \caption{$v_c^{\textrm{an}}(\phi)$ is plotted against the angle with the $k_x$-axis, $\phi$, with $U_s/t = 0.05$ and $\alpha = 0.9$. The expression \eqref{eq:PWvcan} is $\pi$ periodic, which is why only $0$ to $\pi$ are included.}
    \label{fig:PWsuperphi}
\end{figure}

We could also have used numerical calculations to find the critical superfluid velocity as a function of the angle made with the $k_x$-axis, $v_c(\phi)$. Using the numeric eigenvalue $\Omega_{4'}(\boldsymbol{k})$ and parameterizing $\boldsymbol{q}$ by the angle $\phi$ made with the $k_x$-axis, a natural formula to use is
\begin{equation}
\label{eq:PWvcnumphi}
    v_c(\phi) =  \frac{\Omega_{4'}\big(\boldsymbol{k}_{01}+\abs{\boldsymbol{q}}(\cos(\phi), \sin(\phi))\big)}{\abs{\boldsymbol{q}}}.
\end{equation}
To ensure our results are valid, we will try different values of $|\boldsymbol{q}|$ to see that we get the same results, but $|\boldsymbol{q}|a=10^{-5}$ is used in producing the figures. Remember that $\Omega_{4'}(\boldsymbol{k})$ is not what we think of as the lowest energy in the full 1BZ. However, close to $\boldsymbol{k}_{01}$ we found that $\Omega_{4'}(\boldsymbol{k}) \approx \Omega_2(\boldsymbol{k})$. We also suspect they are both approximately the same as $\Omega_H(\boldsymbol{k})$ close to $\boldsymbol{k}_{01}$. Therefore, if we calculate the critical superfluid velocity numerically using the helicity approximation,
\begin{equation}
\label{eq:PWvcnumphiHel}
    v_c^H(\phi) =  \frac{\Omega_{H}\big(\boldsymbol{k}_{01}+\abs{\boldsymbol{q}}(\cos(\phi), \sin(\phi))\big)}{\abs{\boldsymbol{q}}},
\end{equation}
we expect the result will be similar. This is checked in figure \ref{fig:PWsuperphiHel}, and it is clear the two approaches give approximately the same results. We expect the helicity approximation to become better the stronger the SOC is, and the figure indicates that the two approaches give more similar results as $\lambda_R$ is increased. 

\begin{figure}
    \centering
    \includegraphics[width=0.7\linewidth]{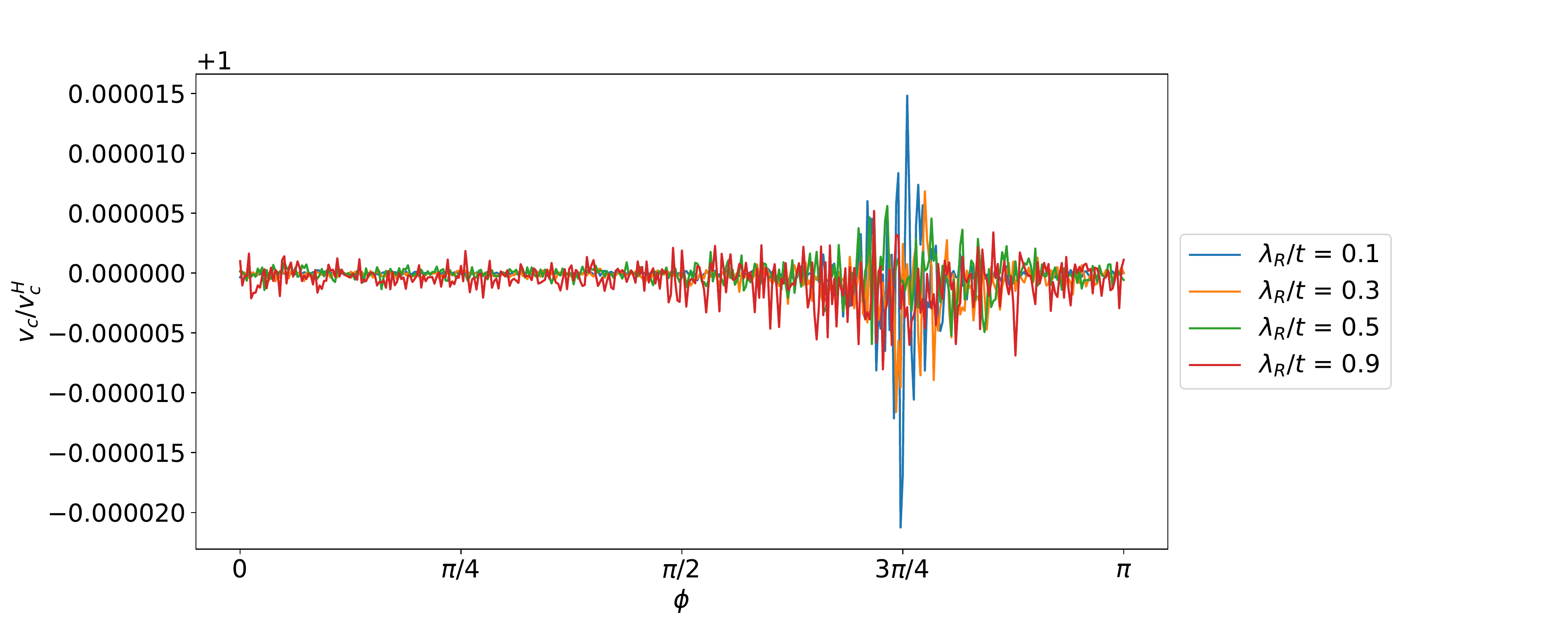}
    \caption{The ratio of the numeric results $v_c$ and $v_c^{H}$ is plotted as a function of the angle with the $k_x$-axis, $\phi$, with $U_s/t = 0.05$ and $\alpha = 0.9$. The value of $k_0 = k_{0m}$ was updated as $\lambda_R$ was changed.}
    \label{fig:PWsuperphiHel}
\end{figure}

Finally, in figures \ref{fig:PWsuperphiDirectionvanv} and \ref{fig:PWsuperphiDirectionvanvHel} we compare the numerically calculated $v_c(\phi)$ and $v_c^{H}(\phi)$ to the analytic expression $v_c^{\textrm{an}}(\phi)$. The ratios are close to $1$, though we see the numeric calculation is less accurate close to the minimum at $\phi = 3\pi/4$. The differences are nevertheless so small that it would not give visible changes in figures \ref{fig:PWsuperal}, \ref{fig:PWsuperlamnew} and \ref{fig:PWsuperphi} if they were produced using either of the numeric methods.

\begin{figure}
    \centering
    \includegraphics[width=0.7\linewidth]{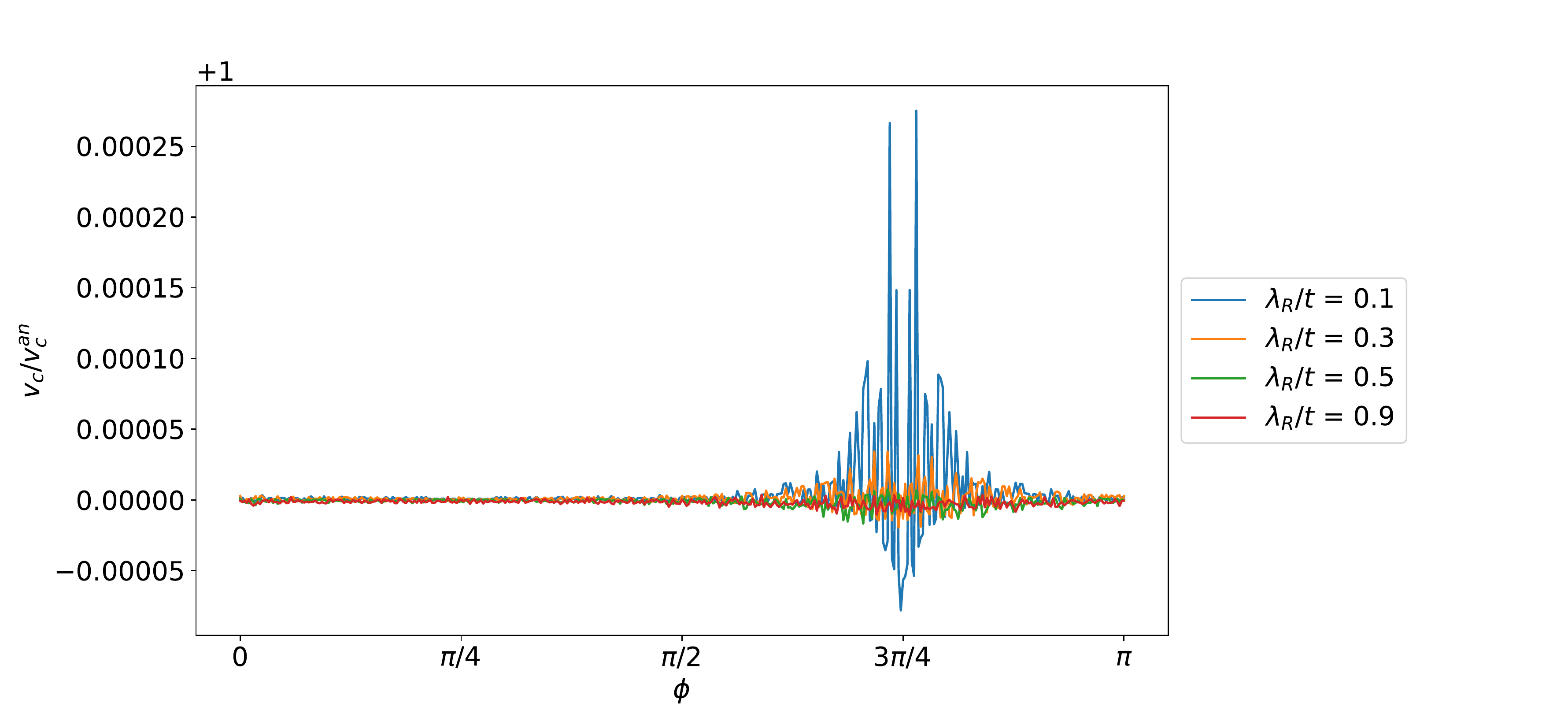}
    \caption{The ratio of the numerically calculated $v_c(\phi)$ from \eqref{eq:PWvcnumphi} and the analytic expression $v_c^{\textrm{an}}(\phi)$ in \eqref{eq:PWvcan} is plotted for various $\lambda_R/t$ with $U_s/t = 0.05$ and $\alpha = 0.9$. The value of $k_0 = k_{0m}$ was updated as $\lambda_R$ was changed. The ratio is close to one, though it appears the numeric approximation is worse closer to the minimal value at $\phi = 3\pi/4$, in particular for small $\lambda_R/t$.}
    \label{fig:PWsuperphiDirectionvanv}
\end{figure}

\begin{figure}
    \centering
    \includegraphics[width=0.7\linewidth]{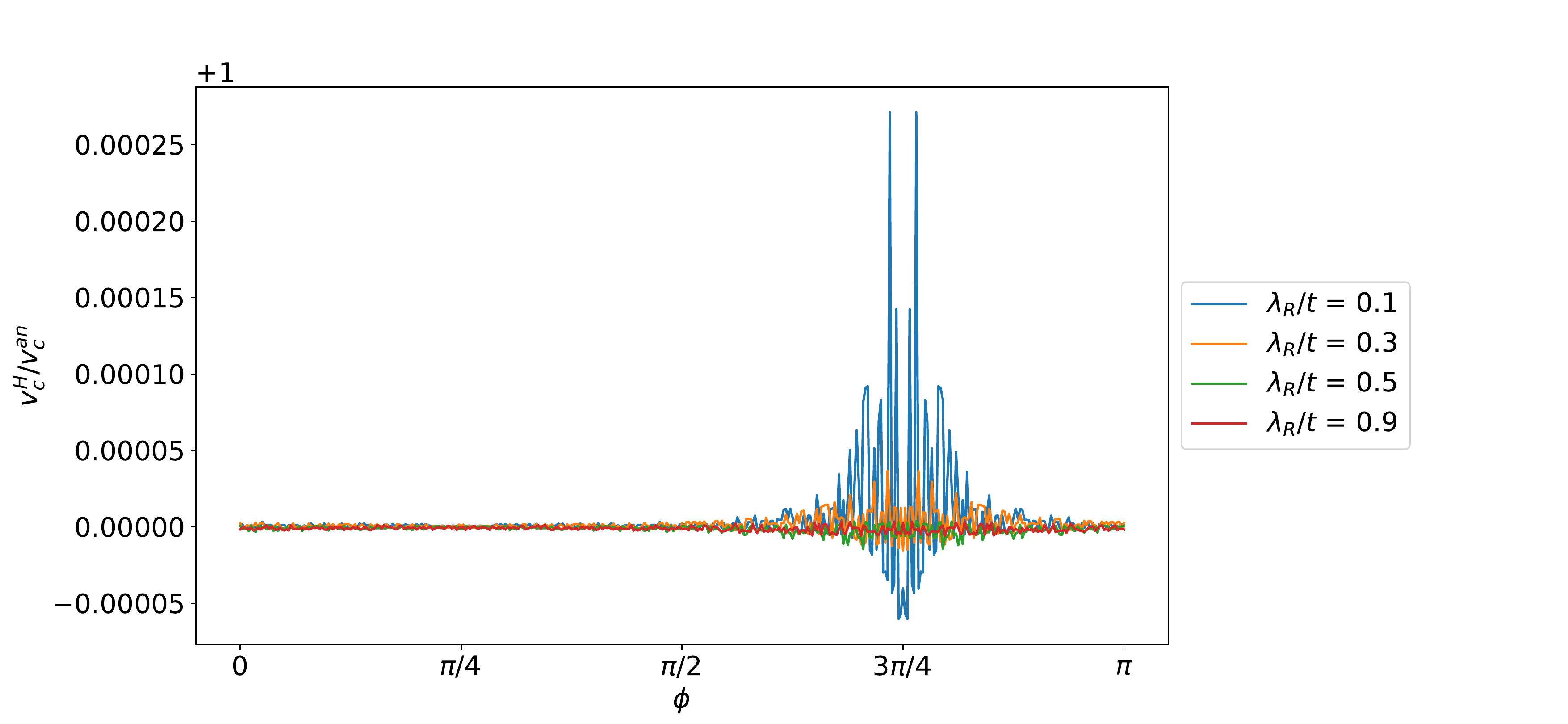}
    \caption{The ratio of the numerically calculated $v_c^{H}(\phi)$ from \eqref{eq:PWvcnumphiHel} and the analytic expression $v_c^{\textrm{an}}(\phi)$ in \eqref{eq:PWvcan} is plotted for various $\lambda_R/t$ with $U_s/t = 0.05$ and $\alpha = 0.9$. The value of $k_0 = k_{0m}$ was updated as $\lambda_R$ was changed. The ratio is close to one, though it appears the numeric approximation is worse closer to the minimal value at $\phi = 3\pi/4$, in particular for small $\lambda_R/t$.}
    \label{fig:PWsuperphiDirectionvanvHel}
\end{figure}

In conclusion we believe the PW phase is stable for $\alpha < 1$ and any nonzero $\lambda_R/t$. Meanwhile, $\lambda_R=0$ leads to $k_0=0$ and thus the NZ phase. There are some indications of small imaginary parts in the eigenvalues even for $\alpha < 1$, but we believe they can be explained as numerical errors rather than indications that the PW phase is unstable. We found an anisotropic critical superfluid velocity as was also reported in \cite{Toniolo}. It is also interesting to see that $v_c^{\textrm{PW}}$ increases as the strength of the spin-orbit coupling is increased, whereas in the PZ phase, the critical superfluid velocity decreases as the strength of SOC is increased, until a point is reached where SOC makes the PZ phase unstable. Such behavior is most likely a result of the fact that nonzero condensate momenta become increasingly favorable as the strength of SOC is increased with all other parameters fixed.

%% file: 5SW.tex
\section{SW Phase}
The SW phase is such that both $\boldsymbol{k}_{01}=(k_0, k_0)$ and $\boldsymbol{k}_{03} = -\boldsymbol{k}_{01}$ are occupied condensate momenta. As a reminder, it was mentioned that it can be thought of as an analogue of Larkin-Ovchinnikov states in superconductors \cite{LO}. 

We assume that $N_{\boldsymbol{k}_{01}}^\uparrow = N_{\boldsymbol{k}_{03}}^\uparrow = N_0^\uparrow/2$ and $N_{\boldsymbol{k}_{01}}^\downarrow =  N_{\boldsymbol{k}_{03}}^\downarrow = N_0^\downarrow/2$. In other words, we assume the condensate is balanced in terms of the condensate momenta. Using \eqref{eq:H0}
\begin{align}
    \begin{split}
        H_0^{''} &= (N_0^\uparrow +N_0^\downarrow)(\epsilon_{\boldsymbol{k}_{01}}+T)  \\
        &+\sqrt{N_0^\uparrow N_0^\downarrow}\abs{s_{\boldsymbol{k}_{01}}}\cos(\gamma_{\boldsymbol{k}_{01}}+\Delta\theta_1) +\sqrt{N_0^\uparrow N_0^\downarrow}\abs{s_{\boldsymbol{k}_{01}}}\cos(\gamma_{\boldsymbol{k}_{03}}+\Delta\theta_3) \\ 
        & +\frac{U}{4N_s}\Big(3(N_0^\uparrow)^2 + 3(N_0^\downarrow)^2 +2\alpha N_0^\uparrow N_0^\downarrow\big(2+\cos(\Delta\theta_1-\Delta\theta_3)\big)\Big) .
    \end{split}
\end{align}
Inserting \eqref{eq:NN0excitup} and \eqref{eq:NN0excitdown} we get 
\begin{align*}
    \begin{split}
        H_0^{''} &=  H_0 -(\epsilon_{\boldsymbol{k}_{01}}+T)\left(\left.\sum_{\boldsymbol{k}}\right.^{'} A_{\boldsymbol{k}}^{\uparrow\dagger}A_{\boldsymbol{k}}^{\uparrow}+\left.\sum_{\boldsymbol{k}}\right.^{'} A_{\boldsymbol{k}}^{\downarrow\dagger}A_{\boldsymbol{k}}^{\downarrow}\right) \\
        &-\frac{\abs{s_{\boldsymbol{k}_{01}}}}{2}\left(\sqrt{\frac{N^\uparrow}{N^\downarrow}}\left.\sum_{\boldsymbol{k}}\right.^{'} A_{\boldsymbol{k}}^{\downarrow\dagger}A_{\boldsymbol{k}}^{\downarrow} + \sqrt{\frac{N^\downarrow}{N^\uparrow}} \left.\sum_{\boldsymbol{k}}\right.^{'} A_{\boldsymbol{k}}^{\uparrow\dagger}A_{\boldsymbol{k}}^{\uparrow}\right)\sum_{i=1,3}\cos(\gamma_{\boldsymbol{k}_{0i}}+\Delta\theta_i) \\
        &-\frac{U}{4N_s}\Bigg( 6N^\uparrow \left.\sum_{\boldsymbol{k}}\right.^{'} A_{\boldsymbol{k}}^{\uparrow\dagger}A_{\boldsymbol{k}}^{\uparrow} + 6N^\downarrow \left.\sum_{\boldsymbol{k}}\right.^{'} A_{\boldsymbol{k}}^{\downarrow\dagger}A_{\boldsymbol{k}}^{\downarrow} \\
        &\mbox{\qquad} + 2\alpha\left(N^\uparrow \left.\sum_{\boldsymbol{k}}\right.^{'} A_{\boldsymbol{k}}^{\downarrow\dagger}A_{\boldsymbol{k}}^{\downarrow}+N^\downarrow \left.\sum_{\boldsymbol{k}}\right.^{'} A_{\boldsymbol{k}}^{\uparrow\dagger}A_{\boldsymbol{k}}^{\uparrow}\right)\big(2+\cos(\Delta\theta_1-\Delta\theta_3)\big)\Bigg).
    \end{split}
\end{align*}
The sum $\left.\sum_{\boldsymbol{k}}\right.^{'}$ excludes the condensate momenta $\pm \boldsymbol{k}_{01}$. Additionally, we defined the operator independent part $H_0$. The remaining part of $H_0^{''}$ is moved to $H_2$ as it is quadratic in excitation operators. 

We choose to fix $N^\uparrow = N^\downarrow = N/2$ and the expression for $H_0$ is then the same as $H_0^{\textrm{SW}}$ given in \eqref{eq:H0SW}.
In $H_2$ we may replace $N_0^\alpha$ by $N^\alpha$ directly to the same order of approximation. The coefficient of $A_{\boldsymbol{k}}^{\uparrow\dagger}A_{\boldsymbol{k}}^{\uparrow}$ is
\begin{align}
    \begin{split}
        M_{1,1}(\boldsymbol{k}) & =\epsilon_{\boldsymbol{k}}-\epsilon_{\boldsymbol{k}_{01}} + \frac{UN}{4N_s}-\frac{UN}{4N_s}\alpha\cos(\Delta\theta_1-\Delta\theta_3) \\
        &-\frac{\abs{s_{\boldsymbol{k}_{01}}}}{2}\left(\cos(\gamma_{\boldsymbol{k}_{01}}+\Delta\theta_1)+\cos(\gamma_{\boldsymbol{k}_{03}}+\Delta\theta_3)\right) \\
        & =\mathcal{E}_{\boldsymbol{k}}+\frac{U_s}{2}\big(1-\alpha\cos(\Delta\theta_1-\Delta\theta_3)\big)+G_{k_0}.
    \end{split}
\end{align}
Here, we defined $G_{k_0}$ as
\begin{align}
    \begin{split}
         G_{k_0} \equiv& \epsilon_{\boldsymbol{0}}-\epsilon_{\boldsymbol{k}_{01}}  - \frac{\abs{s_{\boldsymbol{k}_{01}}}}{2}\big(\cos(\gamma_{\boldsymbol{k}_{01}}+\Delta\theta_1) +\cos(\gamma_{\boldsymbol{k}_{03}}+\Delta\theta_3)\big)  \\
         =& 4t(\cos(k_0 a)-1)\\
         &-\sqrt{2}\lambda_R\abs{\sin(k_0 a)}\big(\cos(\gamma_{\boldsymbol{k}_{01}}+\Delta\theta_1) +\cos(\gamma_{\boldsymbol{k}_{03}}+\Delta\theta_3)\big).
    \end{split}
\end{align}
$H_2$ can now be written
\begin{align}
    \begin{split}
        H_2 = \left.\sum_{\boldsymbol{k}}\right.^{'}& \Bigg\{ M_{1,1}(\boldsymbol{k})\left(A_{\boldsymbol{k}}^{\uparrow\dagger}A_{\boldsymbol{k}}^{\uparrow} + A_{\boldsymbol{k}}^{\downarrow\dagger}A_{\boldsymbol{k}}^{\downarrow}\right) \\
        &+\left(s_{\boldsymbol{k}}+\frac{U_s\alpha}{2}\left( e^{i(\theta_1^{\downarrow}-\theta_{1}^{\uparrow})} + e^{i(\theta_3^\downarrow-\theta_3^\uparrow)}\right)\right)A_{\boldsymbol{k}}^{\uparrow\dagger}A_{\boldsymbol{k}}^{\downarrow}\\
        &+\left(s_{\boldsymbol{k}}^{*}+\frac{U_s\alpha}{2}\left( e^{-i(\theta_{1}^{\downarrow}-\theta_{1}^{\uparrow})}+e^{-i(\theta_3^\downarrow-\theta_3^\uparrow)}\right)\right)A_{\boldsymbol{k}}^{\downarrow\dagger}A_{\boldsymbol{k}}^{\uparrow}\\
        &+ \frac{U_s}{4}\Bigg(\bigg[e^{i2\theta_1^\uparrow}A_{\boldsymbol{k}}^{\uparrow}A_{-\boldsymbol{k}+2\boldsymbol{k}_{01}}^{\uparrow} + 2e^{i(\theta_1^\uparrow+\theta_3^\uparrow)}A_{\boldsymbol{k}}^{\uparrow}A_{-\boldsymbol{k}}^{\uparrow} \\
        &+ e^{i2\theta_3^\uparrow}A_{\boldsymbol{k}}^{\uparrow}A_{-\boldsymbol{k}-2\boldsymbol{k}_{01}}^{\uparrow} +e^{i2\theta_1^\downarrow}A_{\boldsymbol{k}}^{\downarrow}A_{-\boldsymbol{k}+2\boldsymbol{k}_{01}}^{\downarrow} \\
        &+ 2e^{i(\theta_1^\downarrow+\theta_3^\downarrow)}A_{\boldsymbol{k}}^{\downarrow}A_{-\boldsymbol{k}}^{\downarrow} + e^{i2\theta_3^\downarrow}A_{\boldsymbol{k}}^{\downarrow}A_{-\boldsymbol{k}-2\boldsymbol{k}_{01}}^{\downarrow}\\
        &+ \alpha e^{i(\theta_1^\uparrow+\theta_1^\downarrow)}\left(A_{\boldsymbol{k}}^{\downarrow}A_{-\boldsymbol{k}+2\boldsymbol{k}_{01}}^{\uparrow} + A_{\boldsymbol{k}}^{\uparrow}A_{-\boldsymbol{k}+2\boldsymbol{k}_{01}}^{\downarrow}\right) \\
        &+ \alpha \left(e^{i(\theta_1^\downarrow+\theta_3^\uparrow)}+e^{i(\theta_1^\uparrow+\theta_3^\downarrow)}\right)\left(A_{\boldsymbol{k}}^{\downarrow}A_{-\boldsymbol{k}}^{\uparrow} + A_{\boldsymbol{k}}^{\uparrow}A_{-\boldsymbol{k}}^{\downarrow}\right) \\
        &+ \alpha e^{i(\theta_3^\uparrow+\theta_3^\downarrow)}\left(A_{\boldsymbol{k}}^{\downarrow}A_{-\boldsymbol{k}-2\boldsymbol{k}_{01}}^{\uparrow} + A_{\boldsymbol{k}}^{\uparrow}A_{-\boldsymbol{k}-2\boldsymbol{k}_{01}}^{\downarrow}\right) \\
        &+\left(2e^{i(\theta_1^\uparrow-\theta_3^\uparrow)}+\alpha e^{i(\theta_1^\downarrow-\theta_3^\downarrow)}\right)A_{\boldsymbol{k}}^{\uparrow\dagger}A_{\boldsymbol{k}+2\boldsymbol{k}_{01}}^{\uparrow} \\
        &+\left(2e^{-i(\theta_1^\uparrow-\theta_3^\uparrow)}+\alpha e^{-i(\theta_1^\downarrow-\theta_3^\downarrow)}\right)A_{\boldsymbol{k}}^{\uparrow\dagger}A_{\boldsymbol{k}-2\boldsymbol{k}_{01}}^{\uparrow} \\
        &+\left(2e^{i(\theta_1^\downarrow-\theta_3^\downarrow)}+\alpha e^{i(\theta_1^\uparrow-\theta_3^\uparrow)}\right)A_{\boldsymbol{k}}^{\downarrow\dagger}A_{\boldsymbol{k}+2\boldsymbol{k}_{01}}^{\downarrow} \\
        &+\left(2e^{-i(\theta_1^\downarrow-\theta_3^\downarrow)}+\alpha e^{-i(\theta_1^\uparrow-\theta_3^\uparrow)}\right)A_{\boldsymbol{k}}^{\downarrow\dagger}A_{\boldsymbol{k}-2\boldsymbol{k}_{01}}^{\downarrow} \\
        &+\alpha e^{i(\theta_1^\uparrow-\theta_3^\downarrow)}A_{\boldsymbol{k}}^{\downarrow\dagger}A_{\boldsymbol{k}+2\boldsymbol{k}_{01}}^{\uparrow} + \alpha e^{-i(\theta_1^\downarrow-\theta_3^\uparrow)}A_{\boldsymbol{k}}^{\downarrow\dagger}A_{\boldsymbol{k}-2\boldsymbol{k}_{01}}^{\uparrow} \\
        &+\alpha e^{i(\theta_1^\downarrow-\theta_3^\uparrow)}A_{\boldsymbol{k}}^{\uparrow\dagger}A_{\boldsymbol{k}+2\boldsymbol{k}_{01}}^{\downarrow} + \alpha e^{-i(\theta_1^\uparrow-\theta_3^\downarrow)}A_{\boldsymbol{k}}^{\uparrow\dagger}A_{\boldsymbol{k}-2\boldsymbol{k}_{01}}^{\downarrow}\bigg] + \textrm{H.c.} \Bigg)\Bigg\}.
    \end{split}
\end{align}
The sum $\left.\sum_{\boldsymbol{k}}\right.^{'}$ excludes the condensate momenta $\boldsymbol{k} = \pm \boldsymbol{k}_{01}$. Additionally, in the interaction terms there was a restriction in \eqref{eq:H2} that $\boldsymbol{k}'$ should not be equal to a condensate momentum. This means that if for any $\boldsymbol{k}$ a momentum index becomes a condensate momentum, then such a term should be excluded. This can happen for $\boldsymbol{k} = \pm 3\boldsymbol{k}_{01}$ where some of $\pm \boldsymbol{k} \pm 2\boldsymbol{k}_{01}$ become the condensate momenta. Thus, $\boldsymbol{k} = \pm 3\boldsymbol{k}_{01}$ are special momenta we need to treat separately. 

With two condensate momenta, there are now ways to satisfy the Kronecker delta in $H_1$ \eqref{eq:H1} that leaves $\boldsymbol{k}$ as a non-condensate momentum. Specifically, this is the terms of the sum where $i=j \neq i'$. Keeping only terms of order $N/N_s$ and above, we may replace $N_0^\alpha$ by $N^\alpha$ directly in $H_1$. Using \eqref{eq:NN0excitup} and \eqref{eq:NN0excitdown} would yield terms cubic in excitation operators, or equivalently of order $\sqrt{N}/N_s$ which have already been neglected when setting up the Hamiltonian $\eqref{eq:MFTH}$. With $U_s = UN/2N_s$ we get the linear part
\begin{align}
    \begin{split}
        H_1 =  \frac{\sqrt{N}}{4}U_s &\bigg\{\left(e^{i(2\theta_1^\uparrow-\theta_3^\uparrow)}+\alpha e^{i(\theta_1^\uparrow+\theta_1^\downarrow-\theta_3^\downarrow)}\right)A_{3\boldsymbol{k}_{01}}^\uparrow + \textrm{H.c.} \\
        &+\left(e^{i(2\theta_1^\downarrow-\theta_3^\downarrow)}+\alpha e^{i(\theta_1^\uparrow+\theta_1^\downarrow-\theta_3^\uparrow)}\right)A_{3\boldsymbol{k}_{01}}^\downarrow + \textrm{H.c.} \\
        &+\left(e^{-i(\theta_1^\uparrow-2\theta_3^\uparrow)}+\alpha e^{-i(\theta_1^\downarrow-\theta_3^\downarrow-\theta_3^\uparrow)}\right)A_{-3\boldsymbol{k}_{01}}^\uparrow + \textrm{H.c.} \\
        &+\left(e^{-i(\theta_1^\downarrow-2\theta_3^\downarrow)}+\alpha e^{-i(\theta_1^\uparrow-\theta_3^\downarrow-\theta_3^\uparrow)}\right)A_{-3\boldsymbol{k}_{01}}^\downarrow + \textrm{H.c.} \bigg\}.
    \end{split}
\end{align}
Finally, we define the coefficients $c_\sigma^\alpha$ such that
\begin{equation}
\label{eq:SWH1}
    H_1 = \left( c_+^{\uparrow *} A_{3\boldsymbol{k}_{01}}^\uparrow  + c_+^{\downarrow *} A_{3\boldsymbol{k}_{01}}^\downarrow  + c_-^{\uparrow *} A_{-3\boldsymbol{k}_{01}}^\uparrow + c_-^{\downarrow *} A_{-3\boldsymbol{k}_{01}}^\downarrow \right) + \textrm{H.c.}
\end{equation}
An idea to treat these linear terms might be to shift some operators by complex constants, and thus remove the linear terms by completing squares with terms from $H_2$. E.g. if we try something like 
\begin{align*}
    &M_{1,1}(3\boldsymbol{k}_{01})A_{3\boldsymbol{k}_{01}}^{\uparrow\dagger}A_{3\boldsymbol{k}_{01}}^{\uparrow} + c_+^{\uparrow*}A_{3\boldsymbol{k}_{01}}^{\uparrow} +c_+^{\uparrow}A_{3\boldsymbol{k}_{01}}^{\uparrow\dagger} \\
    &= M_{1,1}(3\boldsymbol{k}_{01})\left(A_{3\boldsymbol{k}_{01}}^{\uparrow\dagger}+\frac{c_+^{\uparrow*}}{M_{1,1}(3\boldsymbol{k}_{01})}\right)\left(A_{3\boldsymbol{k}_{01}}^{\uparrow}+\frac{c_+^{\uparrow}}{M_{1,1}(3\boldsymbol{k}_{01})}\right)-\frac{\abs{c_+^{\uparrow}}^2}{M_{1,1}(3\boldsymbol{k}_{01})} \\
    &= M_{1,1}(3\boldsymbol{k}_{01})\Tilde{A}_{3\boldsymbol{k}_{01}}^{\uparrow\dagger}\Tilde{A}_{3\boldsymbol{k}_{01}}^{\uparrow}-\frac{\abs{c_+^{\uparrow}}^2}{M_{1,1}(3\boldsymbol{k}_{01})}.
\end{align*}
All this amounts to, is a shift of $H_0$. The new operators, $\Tilde{A}$, obey the same commutation relations as the old operators, $A$. The problem with this approach, is that both $A_{3\boldsymbol{k}_{01}}^{\uparrow\dagger}$ and $A_{3\boldsymbol{k}_{01}}^{\uparrow}$ appear elsewhere in $H_2$ as well. For the diagonalization procedure, we cannot have different definitions of $A_{3\boldsymbol{k}_{01}}^{\uparrow\dagger}$ and $A_{3\boldsymbol{k}_{01}}^{\uparrow}$ at different places in the Hamiltonian. We have to use either only the old, or only the new, shifted operators. If we want to use the new, shifted operators, we will be forced to add and subtract linear terms to e.g. the term $M_{1,3}A_{3\boldsymbol{k}_{01}}^{\uparrow\dagger}A_{5\boldsymbol{k}_{01}}^{\uparrow}$, leaving us with a multitude of new linear terms. Thus this procedure has not made any progress. 

These problems will not appear in the diagonalized version of $H_2$. Thus, if we can find $H_1$ in terms of the new operators, $\boldsymbol{B}_{\boldsymbol{k}}$, with which $H_2$ is diagonal, it should be possible to use the above method to remove linear terms by completing squares. This will require eigenvectors as well as eigenvalues. As we are unable to obtain analytic eigenvalues and eigenvectors using \textit{Maple}, the transformation of $H_1$ to the new basis, and subsequently the completing of squares will have to be done numerically. In the end, this treatment of $H_1$ will have the effect of changing the free energy $F_\textrm{SW}$, it should not affect the excitation spectrum directly. 

\subsection{Matrix Representation}
We define the operator vector
\begin{align}
    \begin{split}
    \label{eq:SWbra}
        \boldsymbol{A}_{\boldsymbol{k}}^{\dagger} = (&A_{\boldsymbol{k}}^{\uparrow\dagger}, A_{-\boldsymbol{k}}^{\uparrow\dagger}, A_{\boldsymbol{k}+2\boldsymbol{k}_{01}}^{\uparrow\dagger}, A_{-\boldsymbol{k}+2\boldsymbol{k}_{01}}^{\uparrow\dagger}, A_{\boldsymbol{k}-2\boldsymbol{k}_{01}}^{\uparrow\dagger}, A_{-\boldsymbol{k}-2\boldsymbol{k}_{01}}^{\uparrow\dagger}, \\
        &A_{\boldsymbol{k}}^{\downarrow\dagger}, A_{-\boldsymbol{k}}^{\downarrow\dagger}, A_{\boldsymbol{k}+2\boldsymbol{k}_{01}}^{\downarrow\dagger}, A_{-\boldsymbol{k}+2\boldsymbol{k}_{01}}^{\downarrow\dagger}, A_{\boldsymbol{k}-2\boldsymbol{k}_{01}}^{\downarrow\dagger}, A_{-\boldsymbol{k}-2\boldsymbol{k}_{01}}^{\downarrow\dagger}, \\
        &A_{\boldsymbol{k}}^{\uparrow}, A_{-\boldsymbol{k}}^{\uparrow}, A_{\boldsymbol{k}+2\boldsymbol{k}_{01}}^{\uparrow}, A_{-\boldsymbol{k}+2\boldsymbol{k}_{01}}^{\uparrow}, A_{\boldsymbol{k}-2\boldsymbol{k}_{01}}^{\uparrow}, A_{-\boldsymbol{k}-2\boldsymbol{k}_{01}}^{\uparrow}, \\
        &A_{\boldsymbol{k}}^{\downarrow}, A_{-\boldsymbol{k}}^{\downarrow}, A_{\boldsymbol{k}+2\boldsymbol{k}_{01}}^{\downarrow}, A_{-\boldsymbol{k}+2\boldsymbol{k}_{01}}^{\downarrow}, A_{\boldsymbol{k}-2\boldsymbol{k}_{01}}^{\downarrow}, A_{-\boldsymbol{k}-2\boldsymbol{k}_{01}}^{\downarrow}).
    \end{split}
\end{align}
As one can see, $\boldsymbol{k} = \boldsymbol{0}, \pm\boldsymbol{k}_{01}, \pm2\boldsymbol{k}_{01}$ are troublesome, as they leave several elements in $\boldsymbol{A}_{\boldsymbol{k}}^{\dagger}$ equal. In terms of the BV diagonalization procedure, this will lead to a definition of $J$ that does not obey $J^2=I$, and might not even be invertible. Thus, we are forced to treat these parts separately. The condensate momenta are already excluded from the sum in $H_2$. Meanwhile, the special momenta $\boldsymbol{0}, \pm2\boldsymbol{k}_{01}$ and $\pm3\boldsymbol{k}_{01}$ will be treated separately. We write the remaining part of $H_2$ as
\begin{equation}
    H'_2 = \frac{1}{4}\sum_{\boldsymbol{k}\neq \boldsymbol{0}, \pm\boldsymbol{k}_{01}, \pm2\boldsymbol{k}_{01}, \pm3\boldsymbol{k}_{01}}\boldsymbol{A}_{\boldsymbol{k}}^{\dagger} M_{\boldsymbol{k}}\boldsymbol{A}_{\boldsymbol{k}} = \frac{1}{4}\left.\sum_{\boldsymbol{k}}\right.^{'}\boldsymbol{A}_{\boldsymbol{k}}^{\dagger} M_{\boldsymbol{k}}\boldsymbol{A}_{\boldsymbol{k}}.
\end{equation}
We use commutators and we make all $-\boldsymbol{k}$ terms explicit in $H_2$. From the commutators we get a shift in $H_0$,
\begin{align}
    \begin{split}
    \label{eq:SWH0prime}
        H'_0 =& H_0 - \sum_{\boldsymbol{k}\neq \pm\boldsymbol{k}_{01}} \left(\mathcal{E}_{\boldsymbol{k}}+\frac{U_s}{2}\big(1-\alpha\cos(\Delta\theta_1-\Delta\theta_3)\big)+G_{k_0}\right) \\
        =& H_0 -8t\cos(k_0 a) \\
        &-(N_s-2)\left(4t+\frac{U_s}{2}\big(1-\alpha\cos(\Delta\theta_1-\Delta\theta_3)\big)+G_{k_0}\right),
    \end{split}
\end{align}
where $H_0$ is given in \eqref{eq:H0SW}. Note that here, only $\pm\boldsymbol{k}_{01}$ is excluded from the sum, since similar manipulations will be performed at the special momenta. $M_{\boldsymbol{k}}$ is a $24 \cross 24$ matrix on the form
\begin{gather}
\label{eq:SWMat}
    M_{\boldsymbol{k}} = 
    \begin{pmatrix}
    M_1 & M_2 \\
    M_2^* & M_1^* \\
    \end{pmatrix},
\end{gather}
with $M_1 = (M_{1L} | M_{1R})$ and $M_2 = (M_{2L} | M_{2R})$. Here,
\setcounter{MaxMatrixCols}{6}
\begin{gather*}
    M_{1L} = 
    \begin{pmatrix}
    M_{1,1}(\boldsymbol{k}) & 0 & M_{1,3} & 0 & M_{1,3}^* & 0 \\
    0 & M_{1,1}(\boldsymbol{k}) & 0 & M_{1,3} & 0 & M_{1,3}^* \\
    M_{1,3}^{*} & 0 & 0 & 0 & 0 & 0 \\
    0 & M_{1,3}^{*} & 0 & 0 & 0 & 0  \\
    M_{1,3} & 0 & 0 & 0 & 0 & 0 \\
    0 & M_{1,3} & 0 & 0 & 0 & 0  \\
    M_{1,7}^*(\boldsymbol{k}) & 0 & M_{1,11}^* & 0 & M_{1,9}^* & 0 \\
    0 & M_{2,8}^*(\boldsymbol{k}) & 0 & M_{1,11}^* & 0 & M_{1,9}^* \\
    M_{1,9}^{*} & 0 & 0 & 0 & 0 & 0\\
    0 & M_{1,9}^{*} & 0 & 0 & 0 & 0 \\
    M_{1,11}^* & 0 & 0 & 0 & 0 & 0 \\
    0 & M_{1,11}^* & 0 & 0 & 0 & 0 \\
    \end{pmatrix},
\end{gather*}
\begin{gather*}
    M_{1R} = 
    \begin{pmatrix}
    M_{1,7}(\boldsymbol{k}) & 0 & M_{1,9} & 0 & M_{1,11} & 0\\
    0 & M_{2,8}(\boldsymbol{k}) & 0 & M_{1,9} & 0 & M_{1,11}\\
    M_{1,11} & 0 & 0 & 0 & 0 & 0\\
    0 & M_{1,11} & 0 & 0 & 0 & 0 \\
    M_{1,9} & 0 & 0 & 0 & 0 & 0\\
    0 & M_{1,9} & 0 & 0 & 0 & 0 \\
    M_{1,1}(\boldsymbol{k}) & 0 & M_{7,9} & 0 & M_{7,9}^* & 0\\
    0 & M_{1,1}(\boldsymbol{k}) & 0 & M_{7,9} & 0 & M_{7,9}^*\\
    M_{7,9}^* & 0 & 0 & 0 & 0 & 0\\
    0 & M_{7,9}^* & 0 & 0 & 0 & 0 \\
    M_{7,9} & 0 & 0 & 0 & 0 & 0\\
    0 & M_{7,9} & 0 & 0 & 0 & 0 \\
    \end{pmatrix},
\end{gather*}
\begin{gather*}
    M_{2L}^* =
    \begin{pmatrix}
    0 & M_{13,2} & 0 & M_{13,4} & 0 & M_{13,6}\\
    M_{13,2} & 0 & M_{13,4} & 0 & M_{13,6} & 0 \\
    0 & M_{13,4} & 0 & 0 & 0 & 0 \\
    M_{13,4} & 0 & 0 & 0 & 0 & 0 \\
    0 & M_{13,6} & 0 & 0 & 0 & 0 \\
    M_{13,6} & 0 & 0 & 0 & 0 & 0 \\
    0 & M_{13,8} & 0 & M_{13,10} & 0 & M_{13,12} \\
    M_{13,8} & 0 & M_{13,10} & 0 & M_{13,12} & 0 \\
    0 & M_{13,10} & 0 & 0 & 0 & 0 \\
    M_{13,10} & 0 & 0 & 0 & 0 & 0 \\
    0 & M_{13,12} & 0 & 0 & 0 & 0 \\
    M_{13,12} & 0 & 0 & 0 & 0 & 0 \\
    \end{pmatrix}
\end{gather*}
and
\begin{gather*}
    M_{2R}^* = 
    \begin{pmatrix}
    0 & M_{13,8} & 0 & M_{13,10} & 0 &M_{13,12}\\
    M_{13,8} & 0 & M_{13,10} & 0 &M_{13,12} & 0\\
    0 & M_{13,10} & 0 & 0 & 0 & 0 \\
    M_{13,10} & 0 & 0 & 0 & 0 & 0\\
    0 & M_{13,12} & 0 & 0 & 0 & 0 \\
    M_{13,12} & 0 & 0 & 0 & 0 & 0\\
    0 & M_{19,8} & 0 & M_{19,10} & 0 & M_{19,12}\\
    M_{19,8} & 0 & M_{19,10} & 0 & M_{19,12} & 0\\
    0 & M_{19,10} & 0 & 0 & 0 & 0 \\
    M_{19,10} & 0 & 0 & 0 & 0 & 0\\
    0 & M_{19,12} & 0 & 0 & 0 & 0 \\
    M_{19,12} & 0 & 0 & 0 & 0 & 0\\
    \end{pmatrix}.
\end{gather*}
We can confirm that $M_1^\dagger = M_1$ and $M_2^T=M_2$. The matrix elements in $M_1$ are
\begin{align}
    \begin{split}
        \label{eq:SWelem}
        M_{1,1}(\boldsymbol{k}) &= \mathcal{E}_{\boldsymbol{k}}+\frac{U_s}{2}\big(1-\alpha\cos(\Delta\theta_1-\Delta\theta_3)\big)+G_{k_0}, \\
        M_{1,3} &= \frac{U_s}{4}\left(2e^{i(\theta_1^\uparrow-\theta_3^\uparrow)}+\alpha e^{i(\theta_1^\downarrow-\theta_3^\downarrow)}\right), \\
        M_{1,7}(\boldsymbol{k}) &= s_{\boldsymbol{k}}+\frac{U_s\alpha}{2}\left( e^{i(\theta_1^{\downarrow}-\theta_{1}^{\uparrow})} + e^{i(\theta_3^\downarrow-\theta_3^\uparrow)}\right), \\
        M_{2,8}(\boldsymbol{k}) &= -s_{\boldsymbol{k}}+\frac{U_s\alpha}{2}\left( e^{i(\theta_1^{\downarrow}-\theta_{1}^{\uparrow})} + e^{i(\theta_3^\downarrow-\theta_3^\uparrow)}\right), \\
        M_{1,9} &= \frac{U_s\alpha}{4}e^{i(\theta_1^\downarrow-\theta_3^\uparrow)}, \mbox{\qquad\qquad} M_{1,11} =  \frac{U_s\alpha}{4}e^{-i(\theta_1^\uparrow-\theta_3^\downarrow)}, \\ 
        M_{7,9} &= \frac{U_s}{4}\left(2e^{i(\theta_1^\downarrow-\theta_3^\downarrow)}+\alpha e^{i(\theta_1^\uparrow-\theta_3^\uparrow)}\right) , \\ 
    \end{split}
\end{align}
while the elements in $M_2^*$ are
\begin{equation}
\begin{alignedat}{2}
M_{13,2} &= U_s e^{i(\theta_1^\uparrow+\theta_3^\uparrow)},  &\qquad  M_{13,4} &= \frac{U_s}{4}e^{i2\theta_1^\uparrow}, \\
M_{13,6} &= \frac{U_s}{4}e^{i2\theta_3^\uparrow}, & M_{13,8} &= \frac{U_s\alpha}{2}\left(e^{i(\theta_1^\downarrow+\theta_3^\uparrow)}+e^{i(\theta_1^\uparrow+\theta_3^\downarrow)}\right),   \\
 M_{13,10} &= \frac{U_s\alpha}{4}e^{i(\theta_1^\downarrow+\theta_1^\uparrow)},  &   M_{13,12} &= \frac{U_s\alpha}{4}e^{i(\theta_3^\downarrow+\theta_3^\uparrow)},  \\
M_{19,8} &= U_s e^{i(\theta_1^\downarrow+\theta_3^\downarrow)},  & M_{19,10} &= \frac{U_s}{4}e^{i2\theta_1^\downarrow},   \\
 M_{19,12} &= \frac{U_s}{4}e^{i2\theta_3^\downarrow}. & & \\
\end{alignedat}
\end{equation}
Numerically we obtain $24$ eigenvalues of $M_{\boldsymbol{k}}J$, 8 of which are within numerical accuracy $0$ for all $\boldsymbol{k}$. The remaining 16 may be written $\lambda(\boldsymbol{k}) = \pm\Omega_i(\boldsymbol{k}),$ $i=1,2,\dots,8$. The eigenvalues are ordered such that $\Omega_i(\boldsymbol{k}) \geq \Omega_j(\boldsymbol{k})$ if $j>i$. If the angles obey \eqref{eq:gammathetapi} which minimizes $H_0$, these 8 separate eigenvalues reduce to 4 double eigenvalues. Within numerical accuracy, the eigenvalues are inversion symmetric $\Omega_i(-\boldsymbol{k}) = \Omega_i(\boldsymbol{k})$. By calculating the BV norms of the eigenvectors numerically, it is found that for $\pm\Omega_5, \pm\Omega_6, \pm\Omega_7$ and $\pm\Omega_8$ it is the negative eigenvalues that have eigenvectors with positive BV norm. Equivalently then, the positive eigenvalues have eigenvectors with negative BV norm. In the transformation matrix $T_{\boldsymbol{k}}$ the eigenvectors with positive BV norm need to be placed in the left half to satisfy $JT_{\boldsymbol{k}}^\dagger J = T_{\boldsymbol{k}}^{-1}$ or equivalently $T_{\boldsymbol{k}}^\dagger JT_{\boldsymbol{k}} = J$. Thus, the diagonalized matrix $T_{\boldsymbol{k}}^{-1}M_{\boldsymbol{k}}JT_{\boldsymbol{k}} =  D_{\boldsymbol{k}}J$ is 
\begin{align*}
    D_{\boldsymbol{k}}J = \textrm{diag}&\Big(\Omega_1(\boldsymbol{k}), \Omega_2(\boldsymbol{k}), \Omega_3(\boldsymbol{k}), \Omega_4(\boldsymbol{k}), -\Omega_5(\boldsymbol{k}), -\Omega_6(\boldsymbol{k}), \\
    &-\Omega_7(\boldsymbol{k}), -\Omega_8(\boldsymbol{k}), 0, 0, 0, 0, \\
    &-\Omega_1(\boldsymbol{k}), -\Omega_2(\boldsymbol{k}), -\Omega_3(\boldsymbol{k}), -\Omega_4(\boldsymbol{k}), \Omega_5(\boldsymbol{k}), \Omega_6(\boldsymbol{k}), \\
    &\Omega_7(\boldsymbol{k}), \Omega_8(\boldsymbol{k}), 0, 0, 0, 0 \Big).
\end{align*}
To obtain the diagonal matrix $D_{\boldsymbol{k}}$ that enters the Hamiltonian we multiply from the right by $J$, and find
\begin{align*}
    D_{\boldsymbol{k}} = \textrm{diag}&\Big(\Omega_1(\boldsymbol{k}), \Omega_2(\boldsymbol{k}), \Omega_3(\boldsymbol{k}), \Omega_4(\boldsymbol{k}), -\Omega_5(\boldsymbol{k}), -\Omega_6(\boldsymbol{k}), \\
    &-\Omega_7(\boldsymbol{k}), -\Omega_8(\boldsymbol{k}), 0, 0, 0, 0, \\
    &\Omega_1(\boldsymbol{k}), \Omega_2(\boldsymbol{k}), \Omega_3(\boldsymbol{k}), \Omega_4(\boldsymbol{k}), -\Omega_5(\boldsymbol{k}), -\Omega_6(\boldsymbol{k}), \\
    &-\Omega_7(\boldsymbol{k}), -\Omega_8(\boldsymbol{k}), 0, 0, 0, 0 \Big).
\end{align*}
Hence, $H'_2 = \left.\sum_{\boldsymbol{k}}\right.^{'}(\boldsymbol{B}_{\boldsymbol{k}}^{\dagger} D_{\boldsymbol{k}}\boldsymbol{B}_{\boldsymbol{k}})/4$ becomes
 \begin{align}
    \begin{split}
        H'_2 = \frac{1}{2}\left.\sum_{\boldsymbol{k}}\right.^{'} &\Bigg( \sum_{\sigma=1}^{4}\Omega_\sigma(\boldsymbol{k}) \left(B_{\boldsymbol{k},\sigma'}^{\dagger}B_{\boldsymbol{k},\sigma'}+\frac12\right) \\
        &-\sum_{\sigma=5}^{8}\Omega_\sigma(\boldsymbol{k}) \left(B_{\boldsymbol{k},\sigma'}^{\dagger}B_{\boldsymbol{k},\sigma'}+\frac12\right) \\
        &+\sum_{\sigma=9}^{12} 0 \left(B_{\boldsymbol{k},\sigma'}^{\dagger}B_{\boldsymbol{k},\sigma'}+\frac12\right) \Bigg).
    \end{split}
\end{align}
This seems to indicate that one can lower the energy of the system by adding more quasiparticles of type $B_{\boldsymbol{k},\sigma'}$ with $\sigma' = 5', 6', 7', 8'$. In fact, a solution like this, where the negative energies have eigenvectors with positive BV norm are by Pethick and Smith \cite{PethickSmith} called anomalous modes. They state that the appearance of anomalous modes suggests there exists solutions of the Gross-Pitaevskii equation with lower energy than their original solution \cite{PethickSmith}.
However, our approach does not involve solving the Gross-Pitaevskii equation. On physical grounds, the Hamiltonian needs to be bounded from below. Hence, the most natural check here, is to investigate if $\langle H \rangle$ is bounded from below. 

To further explore this solution, and to obtain $\langle H \rangle$, we will add and subtract the maximum value of $\Omega_5(\boldsymbol{k})$, which we denote by $\Omega_0$, to all the bands. This can be thought of as a shift of the zero for the energies. One may also view this procedure as a redefinition of the chemical potential controlling the quasiparticles. We move the chemical potential to just below the lowest energy, such that according to Bose-Einstein statistics all bands have very low filling except for the macroscopic filling in the minima of the lowest band. When these minima occur at the condensate momenta, only the quasiparticle number operators at the condensate momenta have nonzero averages in the limit of zero temperature, and they are not a part of the sum in $H'_2$. Inserting this shift, we get
\begin{align}
    \begin{split}
        H'_2 = -\Omega_0 N_q + \frac{1}{2}\left.\sum_{\boldsymbol{k}}\right.^{'} \Bigg( &\sum_{\sigma=1}^{4}(\Omega_0+\Omega_\sigma(\boldsymbol{k})) \left(B_{\boldsymbol{k},\sigma'}^{\dagger}B_{\boldsymbol{k},\sigma'}+\frac12\right) \\
        &+\sum_{\sigma=5}^{8}(\Omega_0-\Omega_\sigma(\boldsymbol{k})) \left(B_{\boldsymbol{k},\sigma'}^{\dagger}B_{\boldsymbol{k},\sigma'}+\frac12\right) \\
        &+\sum_{\sigma=9}^{12} \Omega_0 \left(B_{\boldsymbol{k},\sigma'}^{\dagger}B_{\boldsymbol{k},\sigma'}+\frac12\right) \Bigg),
    \end{split}
\end{align}
where the quantity 
$$N_q \equiv \frac{1}{2}\left.\sum_{\boldsymbol{k}}\right.^{'}\sum_{\sigma=1}^{12}\left(B_{\boldsymbol{k},\sigma'}^{\dagger}B_{\boldsymbol{k},\sigma'}+\frac12\right)$$
was defined to simplify the expression. Also defining $\Delta\Omega_i = \Omega_{i}+\Omega_0$ for $i=1,2,3,4$,  $\Delta\Omega_i = \Omega_0$ for $i=5,6,7,8$, $\Delta\Omega_i = \Omega_0-\Delta\Omega_{i'} $ for $i = 9,10,11,12$ and $i' = 8,7,6,5$ and renumbering the operators correspondingly we get
\begin{equation}
     H'_2 = -\Omega_0 N_q +  \frac{1}{2}\left.\sum_{\boldsymbol{k}}\right.^{'}\sum_{\sigma=1}^{12} \Delta\Omega_\sigma(\boldsymbol{k})\left(B_{\boldsymbol{k},\sigma}^{\dagger}B_{\boldsymbol{k},\sigma}+\frac12\right).
\end{equation}
According to the definitions above, $\Delta\Omega_i(\boldsymbol{k}) \geq \Delta\Omega_j(\boldsymbol{k})$ if $j>i$, and hence, $\Delta\Omega_{12}(\boldsymbol{k}) = \Omega_0 - \Omega_5(\boldsymbol{k})$ is the lowest energy band.

The special momenta are treated in appendix \ref{app:SW}. The result is that the special treatment of $\boldsymbol{k} = \boldsymbol{0}$ may be incorporated in $H'_2$ by removing the restriction $\boldsymbol{k} \neq \boldsymbol{0}$. On the other hand, the special treatments of $\pm 2\boldsymbol{k}_{01}$ and $\pm3\boldsymbol{k}_{01}$ yielded eigenvalues that did not exactly correspond to the general excitation spectrum $\Omega_\sigma(\boldsymbol{k})$. Therefore, they are kept separate to be sure the treatment is mathematically sound. The physical significance of these deviations at specific values of $\boldsymbol{k}$ is however unclear. In general one would expect the excitation spectrum to be continuous as a function of $\boldsymbol{k}$, while the special values at $\pm 2\boldsymbol{k}_{01}$ and $\pm3\boldsymbol{k}_{01}$ indicate discontinuities.  

Before we minimize the free energy to find the variational parameters, we make some comments on the general excitation spectrum. It appears the eigenvalues remain real only in the vicinity of $k_0 = k_{0m}$. There is also a limit to how far the angles can deviate from \eqref{eq:gammathetapi} before the spectrum becomes complex. These cases are dynamical instabilities \cite{PethickSmith}. We will refer to energetic stability of the SW phase as when there are only two global minima of the energy spectrum which are placed at the condensate momenta $\pm\boldsymbol{k}_{01}$. Thus, investigating the excitation spectrum we can only claim the SW phase is energetically stable if $k_0$ is very close to $k_{0m}$. In figure \ref{fig:Swbandoff} we show the bands for angles that deviate from \eqref{eq:gammathetapi} and a $k_0$ different from $k_{0m}$ to visualize the comments made about the excitation spectrum thus far. Notice that these are not the choices of the variational parameters that minimize the free energy. 
\begin{figure}
    \centering
    \includegraphics[width=0.6\linewidth]{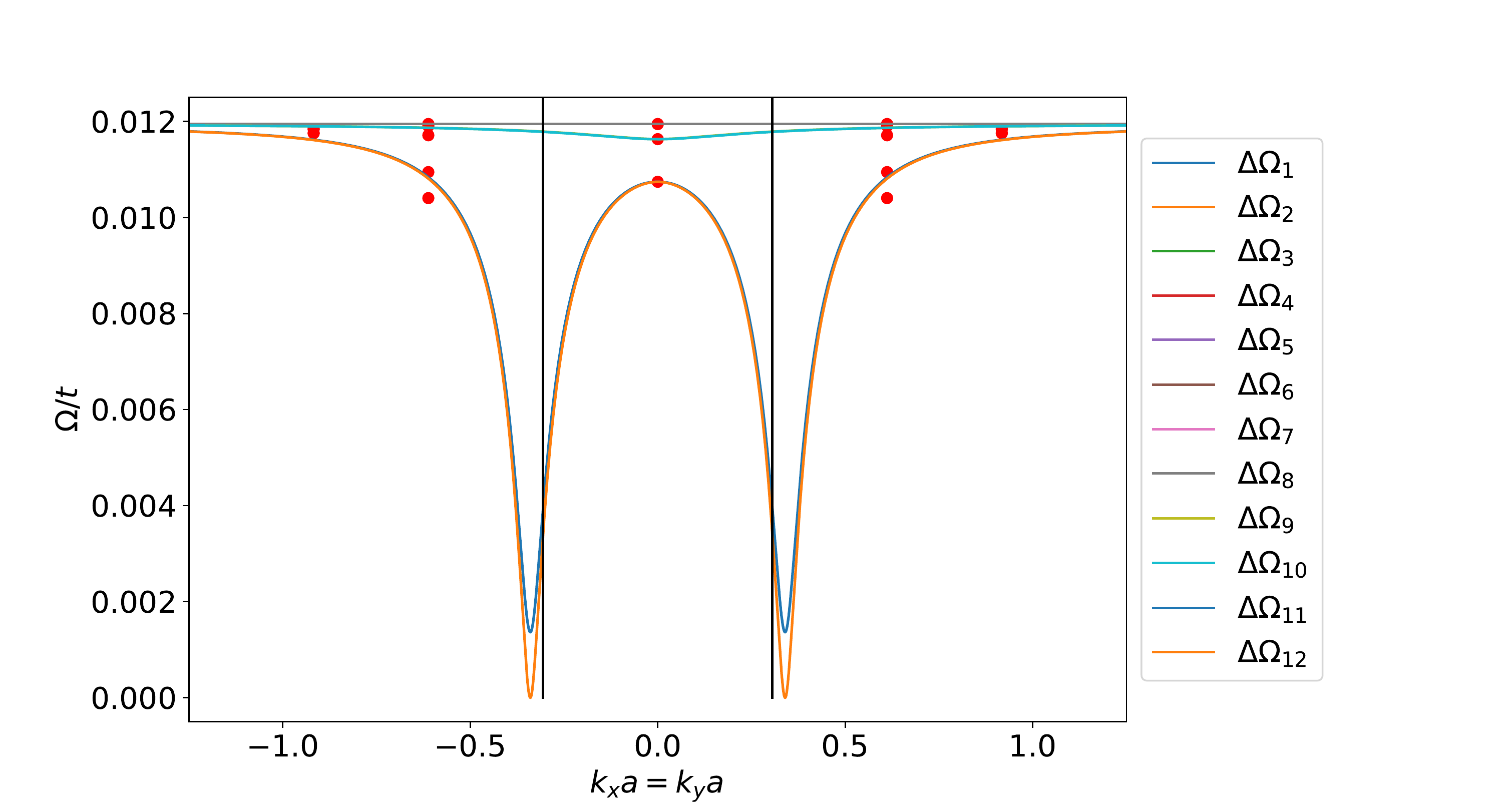}
    \caption{Shows the 8 lowest bands $\Delta\Omega_{\sigma}$ along the $k_x=k_y$ direction for $U_s/t = 0.05$, $\alpha = 1.5$ and $\lambda_R/t = 0.5$. The vertical lines show the position of $k_x=k_y=\pm k_0$. The red points show the special energies found at the special momenta. We set $k_0 = 0.9 k_{0m}$ and one can see that this choice means $\pm\boldsymbol{k}_{01}$ are not the global minima. Furthermore, we used $\Delta\theta_1 = 0.96\pi/4$ and $\Delta\theta_3 = 1.02\cdot5\pi/4$. Note that these are not the choices that minimize the free energy, and so the figure is purely illustrative. \label{fig:Swbandoff}}
\end{figure}


\subsection{Free Energy} \label{sec:SWF}
The Hamiltonian is now $H = H'_0 + H_1 + H'_2 + 2H_2(2\boldsymbol{k}_{01}) + 2H_2(3\boldsymbol{k}_{01})$. Here, $2H_2(2\boldsymbol{k}_{01})$ and  $2H_2(3\boldsymbol{k}_{01})$ are presented in equations \eqref{eq:SWspecial2k01} and \eqref{eq:SWspecial3k01} in appendix \ref{app:SW}. The first obstacle in calculating the free energy is how to treat the linear terms in $H_1$ \eqref{eq:SWH1}. As hinted at earlier, we will numerically transform $H_1$ to the new basis in which $H_2$ is diagonal. The reason we do this, is that we can then simply remove the linear terms by completing squares in $H_2$. This will lead to some operators being shifted by complex numbers, which does not alter the commutation relations, and therefore does not alter the physics described by these operators. The end result is that $H'_0$ is shifted by some real constants. We choose to use the $2H_2(3\boldsymbol{k}_{01})$ part of $H_2$ to perform this removal of linear terms in excitation operators. The transformation matrix for $\boldsymbol{k} = 3\boldsymbol{k}_{01}$ is named $T_3$, and the operator vectors $\boldsymbol{A}_{3k_0}$ and $\boldsymbol{B}_{3k_0}$. The definition of $\boldsymbol{A}_{3k_0}$ is
\begin{align}
    \begin{split}
        \boldsymbol{A}_{3k_0} = (&A_{3\boldsymbol{k}_{01}}^{\uparrow}, A_{-3\boldsymbol{k}_{01}}^{\uparrow}, A_{5\boldsymbol{k}_{01}}^{\uparrow}, A_{-5\boldsymbol{k}_{01}}^{\uparrow}, A_{3\boldsymbol{k}_{01}}^{\downarrow}, A_{-3\boldsymbol{k}_{01}}^{\downarrow}, A_{5\boldsymbol{k}_{01}}^{\downarrow}, A_{-5\boldsymbol{k}_{01}}^{\downarrow}, \\
        &A_{3\boldsymbol{k}_{01}}^{\uparrow\dagger}, A_{-3\boldsymbol{k}_{01}}^{\uparrow\dagger}, A_{5\boldsymbol{k}_{01}}^{\uparrow\dagger}, A_{-5\boldsymbol{k}_{01}}^{\uparrow\dagger}, A_{3\boldsymbol{k}_{01}}^{\downarrow\dagger}, A_{-3\boldsymbol{k}_{01}}^{\downarrow\dagger}, A_{5\boldsymbol{k}_{01}}^{\downarrow\dagger}, A_{-5\boldsymbol{k}_{01}}^{\downarrow\dagger})^T.
    \end{split}
\end{align}
We may revert to the primed numbering on the operators, such that
\begin{align}
    \begin{split}
        2H_2(3\boldsymbol{k}_{01}) &= \sum_{\sigma=1}^4 \omega_{3k_0, \sigma} \left( B_{3\boldsymbol{k}_{01},\sigma'}^{\dagger}B_{3\boldsymbol{k}_{01}, \sigma'} +\frac12  \right)\\
        &-\sum_{\sigma=5}^8 \omega_{3k_0, \sigma} \left( B_{3\boldsymbol{k}_{01},\sigma'}^{\dagger}B_{3\boldsymbol{k}_{01}, \sigma'} +\frac12  \right)
    \end{split}
\end{align}
and
\begin{align}
    \begin{split}
        \boldsymbol{B}_{3k_0} = (&B_{3\boldsymbol{k}_{01}, 1'}, B_{3\boldsymbol{k}_{01}, 2'}, B_{3\boldsymbol{k}_{01}, 3'}, B_{3\boldsymbol{k}_{01},4'}, B_{3\boldsymbol{k}_{01}, 5'}, B_{3\boldsymbol{k}_{01}, 6'}, B_{3\boldsymbol{k}_{01}, 7'}, B_{3\boldsymbol{k}_{01},8'}, \\
        &B_{3\boldsymbol{k}_{01}, 1'}^{\dagger}, B_{3\boldsymbol{k}_{01}, 2'}^{\dagger}, B_{3\boldsymbol{k}_{01}, 3'}^{\dagger}, B_{3\boldsymbol{k}_{01},4'}^{\dagger}, B_{3\boldsymbol{k}_{01}, 5'}^{\dagger}, B_{3\boldsymbol{k}_{01}, 6'}^{\dagger}, B_{3\boldsymbol{k}_{01}, 7'}^{\dagger}, B_{3\boldsymbol{k}_{01},8'}^{\dagger})^T.
    \end{split}
\end{align}
As an example, let us see how a treatment of the part $c B_{3\boldsymbol{k}_{01},1'} + c^* B_{3\boldsymbol{k}_{01},1'}^\dagger$ works.
\begin{align*}
    \begin{split}
        &\omega_{3k_0,1}B_{3\boldsymbol{k}_{01},1'}^{\dagger}B_{3\boldsymbol{k}_{01},1'} + c B_{3\boldsymbol{k}_{01},1'} +c^* B_{3\boldsymbol{k}_{01},1'}^{\dagger} \\
        &=\omega_{3k_0,1}\left(B_{3\boldsymbol{k}_{01},1'}^{\dagger}+\frac{c}{\omega_{3k_0,1}}\right)\left(B_{3\boldsymbol{k}_{01},  1'}+\frac{c^*}{\omega_{3k_0,1}}\right)-\frac{\abs{c}^2}{\omega_{3k_0,1}} \\ 
        &= \omega_{3k_0,1}\Tilde{B}_{3\boldsymbol{k}_{01},1'}^{\dagger}\Tilde{B}_{3\boldsymbol{k}_{01},1'}-\frac{\abs{c}^2}{\omega_{3k_0,1}}.
    \end{split}
\end{align*}
Finally, because the new operators obey the same commutation relations as the old, we remove the tilde on these. In conclusion, we need to find all the terms like $\abs{c}^2/\omega_{3k_0,1}$ and subtract them from $H'_0$ to get $\Tilde{H}_0$. We have that $\boldsymbol{B}_{3k_0}=T_3^\dagger \boldsymbol{A}_{3k_0}$, or conversely $\boldsymbol{A}_{3k_0} = JT_3 J \boldsymbol{B}_{3k_0}$, using that $T^{-1} = JT^\dagger J$ $\iff$ $(T^\dagger)^{-1} = JTJ$ by inversion. Thus we can see that e.g. $A_{3\boldsymbol{k}_{01}}^{\uparrow} = \sum_i (JT_3 J)_{1,i} (\boldsymbol{B}_{3k_0})_i$. All in all we find that
\begin{align}
    \begin{split}
        H_1 &= \sum_{i=1}^{16} \bigg( c_+^{\uparrow*}(JT_3 J)_{1,i} + c_+^{\uparrow}(JT_3 J)_{9,i} +c_+^{\downarrow*} (JT_3 J)_{5,i} +  c_+^{\downarrow}(JT_3 J)_{13,i} \\
        &+c_-^{\uparrow*} (JT_3 J)_{2,i} + c_-^{\uparrow} (JT_3 J)_{10,i} + c_-^{\downarrow*} (JT_3 J)_{6,i} + c_-^{\downarrow} (JT_3 J)_{14,i}\bigg) (\boldsymbol{B}_{3k_0})_i.
    \end{split}
\end{align}
We write this as $H_1 = \sum_i c_i (\boldsymbol{B}_{3k_0})_i$, and note that because of the form of $T$, $c_{i+8} = c_i^*$, meaning it is enough to consider the first 8 values of $i$. Thus, the final equation needed to find $\Tilde{H}_0$ is
\begin{equation}
\label{eq:H13k0noOm0}
    \Tilde{H}_0 = H'_0 -\sum_{i=1}^{4} \frac{\abs{c_i}^2}{\omega_{3k_0, i}} + \sum_{i=5}^{8} \frac{\abs{c_i}^2}{\omega_{3k_0, i}},
\end{equation}
where $H'_0$ is given in \eqref{eq:SWH0prime}. The plus sign in the second sum is because the energies $\omega_{3k_0, i},$ $i=5,6,7,8$ enter the diagonalized version with a negative sign. We note it was only possible to use the primed numbering on the operators because all the energies $\omega_{3k_0, i}$ are nonzero.

The Hamiltonian is now $H = \Tilde{H}_0 + H'_2 + 2H_2(2\boldsymbol{k}_{01}) + 2H_2(3\boldsymbol{k}_{01})$. A remaining question is how we should treat the terms
\begin{equation}
    -\Omega_0 N_q = -\Omega_0 \frac{1}{2}\sum_{\boldsymbol{k}\neq \pm\boldsymbol{k}_{01}, \pm2\boldsymbol{k}_{01}, \pm3\boldsymbol{k}_{01}}\sum_{\sigma=1}^{12} \left( B_{\boldsymbol{k},\sigma}^{\dagger}B_{\boldsymbol{k},\sigma} +\frac12 \right),
\end{equation}
\begin{equation}
    -\Omega_0 N_{q,2k_0} = -\Omega_0 \sum_{\sigma=1}^{10} \left(B_{2\boldsymbol{k}_{01},\sigma}^{\dagger}B_{2\boldsymbol{k}_{01}, \sigma} + \frac12\right)
\end{equation}
and
\begin{equation}
    -\Omega_0 N_{q,3k_0} = -\Omega_0 \sum_{\sigma=1}^{8} \left( B_{3\boldsymbol{k}_{01},\sigma}^{\dagger}B_{3\boldsymbol{k}_{01}, \sigma} +\frac12 \right).
\end{equation}
The negative prefactor means we can not treat them in the usual way we treat number operators when calculating the free energy. They must instead be treated as numbers, and moved into the operator independent part of the Hamiltonian. These numbers can be calculated using the Bose-Einstein distribution with zero chemical potential because the quasiparticles are non-interacting and thus behave like an ideal Bose gas as explained in chapter 4.3 of \cite{Pitaevskii}. Technically, by using $\Delta\Omega_i(\boldsymbol{k})$ rather than $\Omega_i(\boldsymbol{k}), i=1,2,3,4, -\Omega_{i}(\boldsymbol{k}), i=5,6,7,8$ and $0$ we have shifted the chemical potential from $-\Omega_0$ to $0$. We choose to think of $\Delta\Omega_i(\boldsymbol{k})$ as the excitation energies, and the chemical potential as zero. Hence,
\begin{equation}
    \langle B_{\boldsymbol{k},\sigma}^{\dagger}B_{\boldsymbol{k},\sigma} \rangle = \frac{1}{e^{\beta\Delta\Omega_\sigma(\boldsymbol{k})}-1}.
\end{equation}
\begin{equation}
    \langle B_{2\boldsymbol{k}_{01},\sigma}^{\dagger}B_{2\boldsymbol{k}_{01},\sigma} \rangle = \frac{1}{e^{\beta\Delta\omega_{2k_0,\sigma}}-1},
\end{equation}
and
\begin{equation}
    \langle B_{3\boldsymbol{k}_{01},\sigma}^{\dagger}B_{3\boldsymbol{k}_{01},\sigma} \rangle = \frac{1}{e^{\beta\Delta\omega_{3k_0,\sigma}}-1}.
\end{equation}
Given that for $\boldsymbol{k} \neq \pm \boldsymbol{k}_{01}$ these energies are all nonzero, in the limit of $\beta \to \infty$ these expectation values are all zero. In total, the terms originating with redefinition of zero for energies then contribute
\begin{align}
    \begin{split}
        &-\frac{1}{4}\sum_{\boldsymbol{k}\neq \pm\boldsymbol{k}_{01}, \pm2\boldsymbol{k}_{01}, \pm3\boldsymbol{k}_{01}}\sum_{\sigma=1}^{12} \Omega_0 - \frac12  \sum_{\sigma=1}^{10} \Omega_0 - \frac{1}{2} \sum_{\sigma=1}^{8}  \Omega_0 \\
        &= -3\Omega_0(N_s-6) -5\Omega_0 - 4\Omega_0 = -3\Omega_0 (N_s-3).
    \end{split}
\end{align}
We define $\Tilde{H}_{0}^{'}$ as the operator independent part of the Hamiltonian including quantum correction, shift from incorporating $H_1$ and a shift due to the redefinition of zero for the energies. All in all, the Hamiltonian is
\begin{align}
    \begin{split}
    \label{eq:SWfullH}
        H &= \Tilde{H}_{0}^{'} + \frac{1}{2}\left.\sum_{\boldsymbol{k}}\right.^{'}\sum_{\sigma=1}^{12} \Delta\Omega_\sigma(\boldsymbol{k})\left(B_{\boldsymbol{k},\sigma}^{\dagger}B_{\boldsymbol{k},\sigma}+\frac12\right) \\
        &+\sum_{\sigma=1}^{10} \Delta\omega_{2k_0, \sigma} \left( B_{2\boldsymbol{k}_{01},\sigma}^{\dagger}B_{2\boldsymbol{k}_{01}, \sigma} + \frac12  \right) \\
        &+\sum_{\sigma=1}^8 \Delta\omega_{3k_0, \sigma} \left( B_{3\boldsymbol{k}_{01},\sigma}^{\dagger}B_{3\boldsymbol{k}_{01}, \sigma}  + \frac{1}{2} \right),
    \end{split}
\end{align}
where 
\begin{align}
    \begin{split}
    \label{eq:H0tildeprime}
        \Tilde{H}_{0}^{'} &= N(\epsilon_{\boldsymbol{k}_{01}}+T)+ \frac{UN^2}{8N_s}\Big(3+\alpha\big(2+\cos(\Delta\theta_1-\Delta\theta_3)\big)\Big)\\
        & + \frac{N}{2}\abs{s_{\boldsymbol{k}_{01}}}\big(\cos(\gamma_{\boldsymbol{k}_{01}}+\Delta\theta_1) +\cos(\gamma_{\boldsymbol{k}_{03}}+\Delta\theta_3)\big) \\
        &  -8t\cos(k_0 a) -(N_s-2)\left(4t+\frac{U_s}{2}\big(1-\alpha\cos(\Delta\theta_1-\Delta\theta_3)\big)+G_{k_0}\right) \\
        &- \sum_{i=1}^{4} \frac{\abs{c_i}^2}{\omega_{3k_0, i}} + \sum_{i=5}^{8} \frac{\abs{c_i}^2}{\omega_{3k_0, i}}  -3\Omega_0 (N_s-3).
    \end{split}
\end{align}
Once again we focus on $\beta \to \infty$ and find
\begin{align}
    \begin{split}
        F_{\textrm{SW}} = \langle H_{\textrm{SW}} \rangle = &\Tilde{H}_{0}^{'} + \frac{1}{4}\left.\sum_{\boldsymbol{k}}\right.^{'}\sum_{\sigma=1}^{12} \Delta\Omega_\sigma(\boldsymbol{k}) \\
        &+ \frac{1}{2}\sum_{\sigma=1}^{10} \Delta\omega_{2k_0, \sigma} + \frac{1}{2}\sum_{\sigma=1}^8 \Delta\omega_{3k_0, \sigma},
    \end{split}
\end{align}
where the sum  $\left.\sum_{\boldsymbol{k}}\right.^{'}$ excludes $\pm \boldsymbol{k}_{01}, \pm 2\boldsymbol{k}_{01}$ and $\pm 3\boldsymbol{k}_{01}$. The idea for minimization of $F_{\textrm{SW}}$ is similar to the approach in the PW phase. We keep $N/N_s = 1$ and $U/t = 0.1$ fixed such that $U_s/t = 0.05$. We also fix $\alpha$ and $\lambda_R$ to appropriate values. Then we vary $k_0$ to find the value of $k_0$ that minimizes $F_\textrm{SW}$ which will be named $k_{0\textrm{min}}$. As mentioned when investigating the excitation spectrum, the SW phase is energetically stable only if it is $k_{0m}$ that minimizes $F_{\textrm{SW}}$. We also notice from investigating the excitation spectrum that it remains real only close to $k_{0m}$. Our calculation of $F_{\textrm{SW}}$ only makes sense when the energies are real and we can therefore only investigate the set of $k_0$ values that render the excitation spectrum real. The minimum of $F_{\textrm{SW}}$ within this set will be used, unless it is at the boundary. 

\begin{figure}
    \centering
    \begin{subfigure}{.49\textwidth}
      \includegraphics[width=\linewidth]{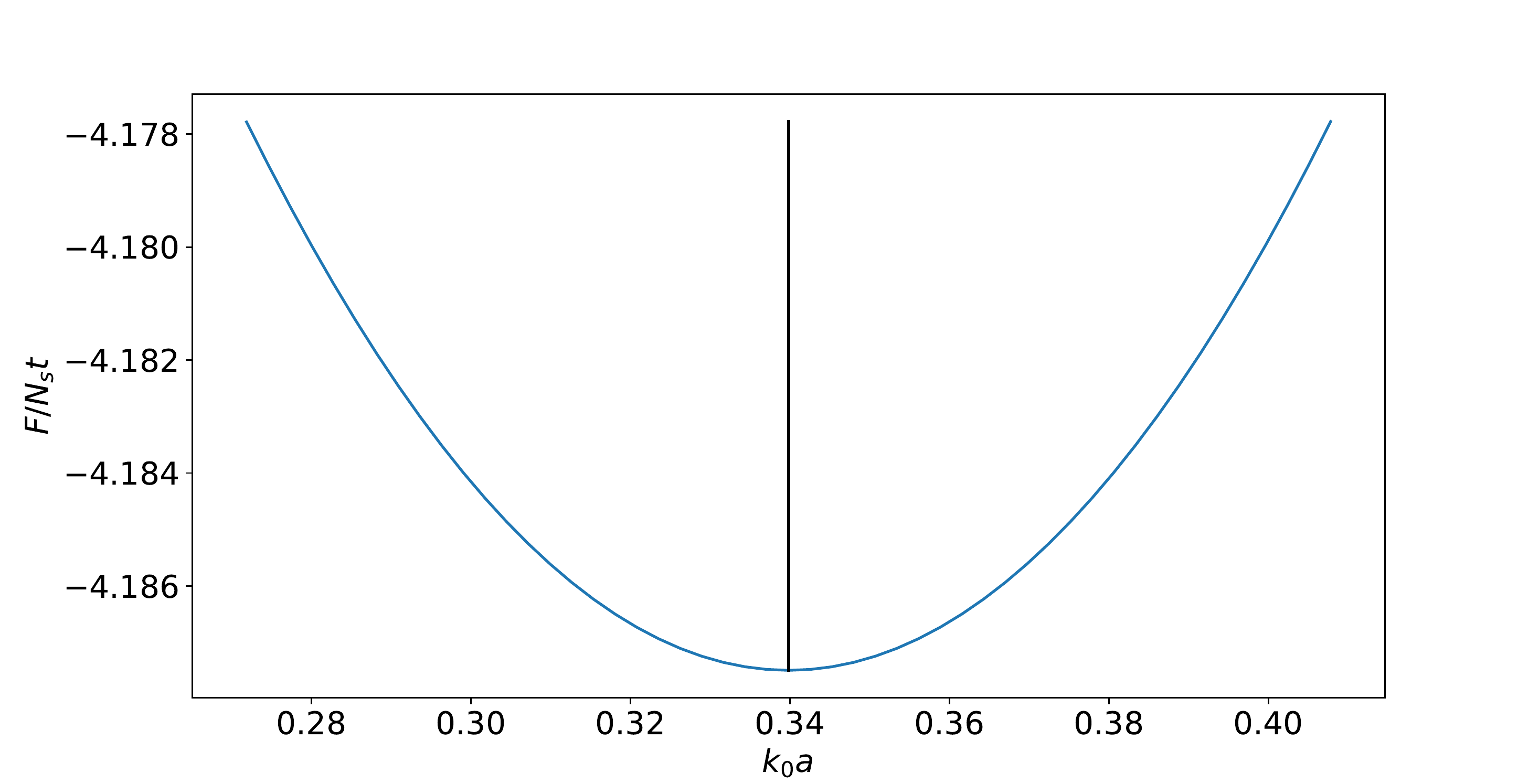}
      \caption{ }
      \label{fig:SWFzoom}
    \end{subfigure}%
    \begin{subfigure}{.49\textwidth}
      \includegraphics[width=\linewidth]{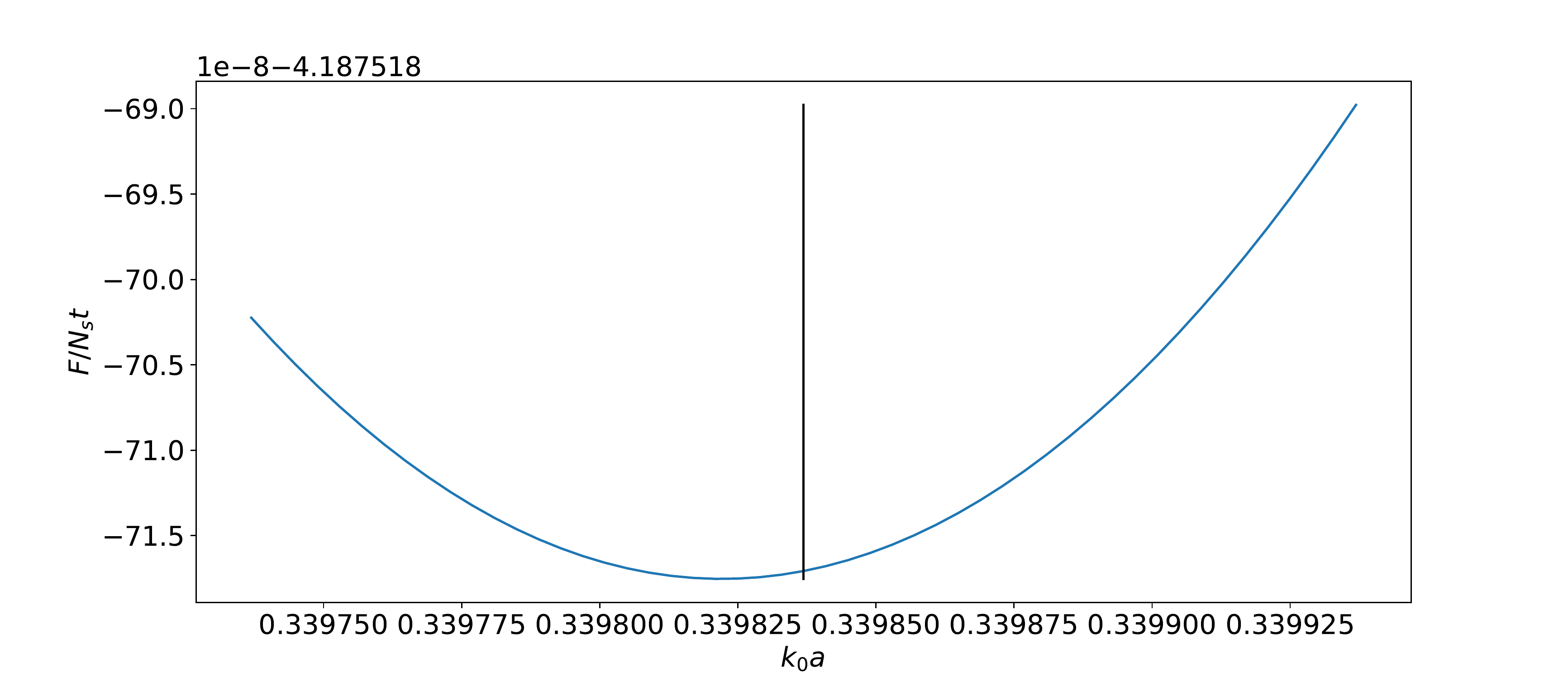}
      \caption{ }
      \label{fig:SWFzoom2}
    \end{subfigure}
    \caption{A plot of $F_{\textrm{SW}}$ as a function of $k_0$. The black vertical line shows the position of $k_0=k_{0m}$. $51$ values of $k_0$ were considered in a set slightly larger than the set for which the energies are real in (a) (only the real parts were used in calculations). The parameters were $T=0$, $U_s/t = 0.05$, $\alpha = 1.5$ and $\lambda_R/t = 0.5$. In (b) we focus on the minimum, and a lattice size of $4\cdot10^{4}$ was used. The minimum moves closer to $k_{0m}$ as the lattice size is increased. \label{fig:SWF}}
\end{figure}

Just as we observed in the PW phase, $k_{0\textrm{min}}$ approaches $k_{0m}$ as $N_s$ increases. For a lattice size of $4\cdot 10^{4}$ we find $k_{0\textrm{min}} a \approx 0.339821$ while $k_{0m} a \approx 0.339837$. The relative error is of order $\order{10^{-5}}$ and should approach zero as the lattice size is increased. The corresponding relative difference in $F_{\textrm{SW}}$ is of order $\order{10^{-10}}$. We therefore state that $k_0 = k_{0\textrm{min}} = k_{0m}$ minimizes $F_{\textrm{SW}}$. 

We also find that satisfying \eqref{eq:gammathetapi} minimizes $F_{\textrm{SW}}$ in terms of the differences $\Delta \theta_1$ and $\Delta\theta_3$. This was found by first assuming $\Delta\theta_3 = 5\pi/4$ and $\theta_1^\downarrow = \pi/4$. Then $F_{\textrm{SW}}$ was calculated for different $\theta_1^\uparrow$, and within numerical accuracy, $\theta_1^\uparrow = 0$ was found to minimize $F_{\textrm{SW}}$. Next, $\Delta\theta_1 = \pi/4$ and $\theta_3^\downarrow = 5\pi/4$ was assumed. It was then found that $\theta_3^\uparrow = 0$ minimizes the free energy. Hence, with $\theta_1^\downarrow = \theta_1^\uparrow + \pi/4$ and $\theta_3^\downarrow = \theta_3^\uparrow + 5\pi/4$ determined, the remaining  free parameters are $\theta_1^\uparrow$ and $\theta_3^\uparrow$. 

We start by investigating minimization in terms of $\theta_1^\uparrow-\theta_3^\uparrow$. It is found that $\theta_1^\uparrow-\theta_3^\uparrow = -\pi/4$ is optimal. Hence, $\theta_3^\uparrow = \theta_1^\uparrow + \pi/4$ and the only remaining angle to vary is $\theta_1^\uparrow$. The relative variations in $F_{\textrm{SW}}$ in terms of this final angle are negligible (of order $\order{10^{-14}}$), and we conclude that $\theta_1^\uparrow$ is free. Choosing a value for $\theta_1^\uparrow$ the remaining angles should be set to
\begin{equation}
\label{eq:SWangles}
    \theta_1^\downarrow = \theta_3^\uparrow = \theta_1^\uparrow + \frac{\pi}{4} \mbox{\qquad\qquad and \qquad\qquad} \theta_3^\downarrow = \theta_3^\uparrow + \frac{5\pi}{4} = \theta_1^\uparrow + \frac{3\pi}{2}.
\end{equation}
It was also found that the dependence of $F_{\textrm{SW}}$ on $\theta_1^\uparrow-\theta_3^\uparrow$ came solely from the contribution of the excitation spectrum.


\subsection{Spin Basis Excitation Spectrum}
Based on minimization of $F_{\textrm{SW}} = \langle H_{\textrm{SW}} \rangle $ at zero temperature, we know that $k_0 = k_{0m}$ minimizes the free energy. We also know that once $\theta_1^\uparrow$ is set, the other angles follow \eqref{eq:SWangles}. Given that the angles obey \eqref{eq:gammathetapi} the number of bands are reduced to 5 separate bands. This enables us to make several simplifications. For the matrix elements, we make the following identifications
\begin{equation}
\begin{alignedat}{3}
M_{1,11} &   \stackrel{(\ref{eq:gammathetapi})}{=} -iM_{1,9}^*,  &\qquad  M_{7,9}  &\stackrel{(\ref{eq:gammathetapi})}{=} -M_{1,3}, &\qquad  M_{13,8} &\stackrel{(\ref{eq:gammathetapi})}{=} 0,\\
 M_{19,8} &\stackrel{(\ref{eq:gammathetapi})}{=} -iM_{13,2},  &  M_{19,10}  &\stackrel{(\ref{eq:gammathetapi})}{=} iM_{13,4},  &  M_{19,12}  &\stackrel{(\ref{eq:gammathetapi})}{=} iM_{13,6},\\
 M_{1,7}(\boldsymbol{k}) & \stackrel{(\ref{eq:gammathetapi})}{=} s_{\boldsymbol{k}},  & M_{2,8}(\boldsymbol{k}) &\stackrel{(\ref{eq:gammathetapi})}{=} -s_{\boldsymbol{k}}.  &  &\\
\end{alignedat}
\end{equation}
The 16 nonzero eigenvalues may now be written $\lambda(\boldsymbol{k}) = \pm\Omega_i(\boldsymbol{k}),$ $i=1,2,3,4$ all of which double eigenvalues. 
These are ordered such that $\Omega_1(\boldsymbol{k}) \geq \Omega_2(\boldsymbol{k}) \geq \Omega_3(\boldsymbol{k}) \geq \Omega_4(\boldsymbol{k})$. By similar arguments as given in the PZ phase, making sure the $M_{\boldsymbol{k}}$ matrix of the SW phase will give similar results, we assume the new operators corresponding to the same eigenvalues can be related by $B_{i+1, -\boldsymbol{k}} = B_{i, \boldsymbol{k}}$. Given that the eigenvalues are inversion symmetric, $\Omega_i(-\boldsymbol{k}) = \Omega_i(\boldsymbol{k})$, we can then limit ourselves to 6 new number operators, one for each nonzero band and 2 for the zero mode. 

\begin{figure}
    \centering
    \begin{subfigure}{.45\textwidth}
      \includegraphics[width=\linewidth]{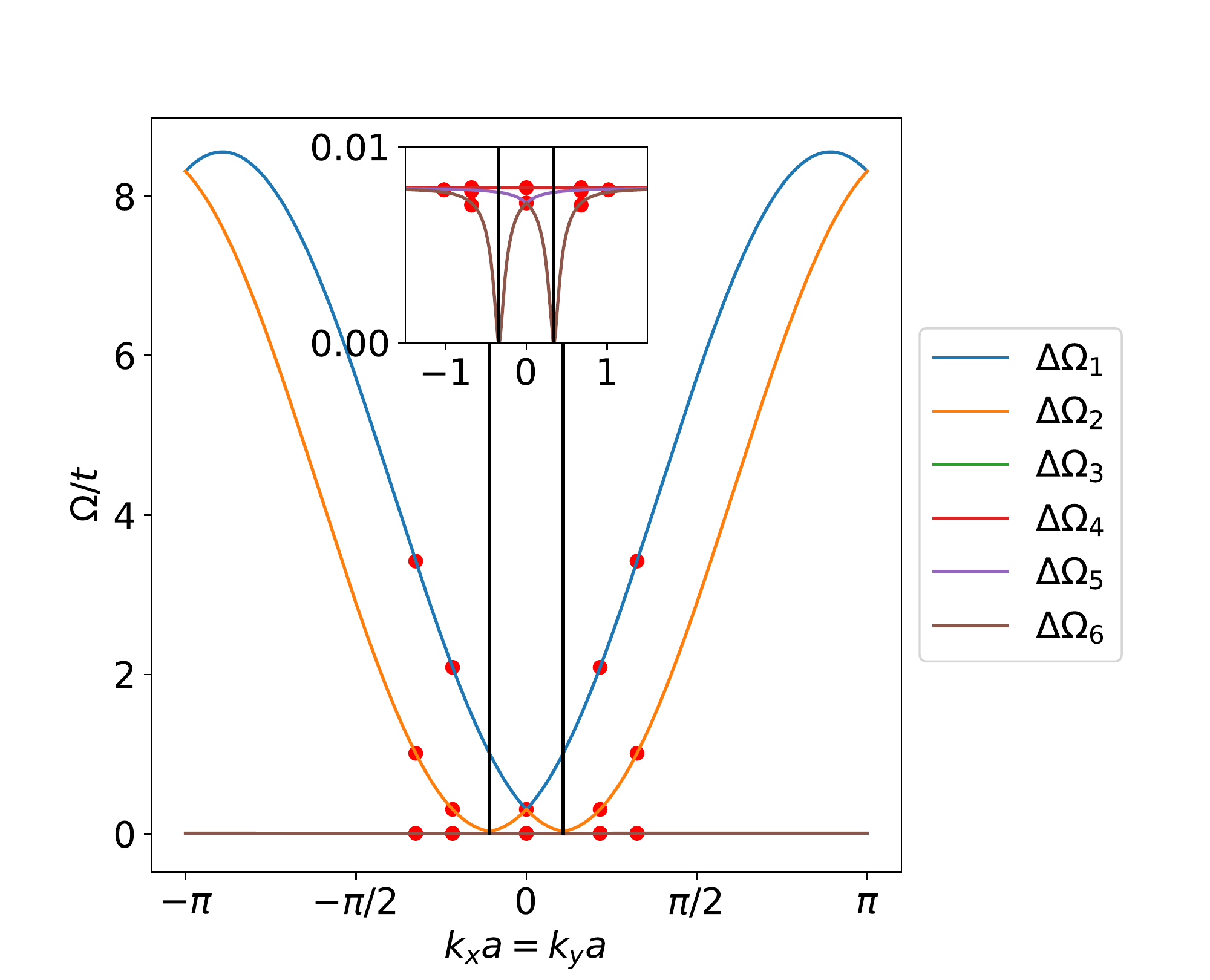}
      \caption{}
    \end{subfigure}%
    \begin{subfigure}{.45\textwidth}
      \includegraphics[width=\linewidth]{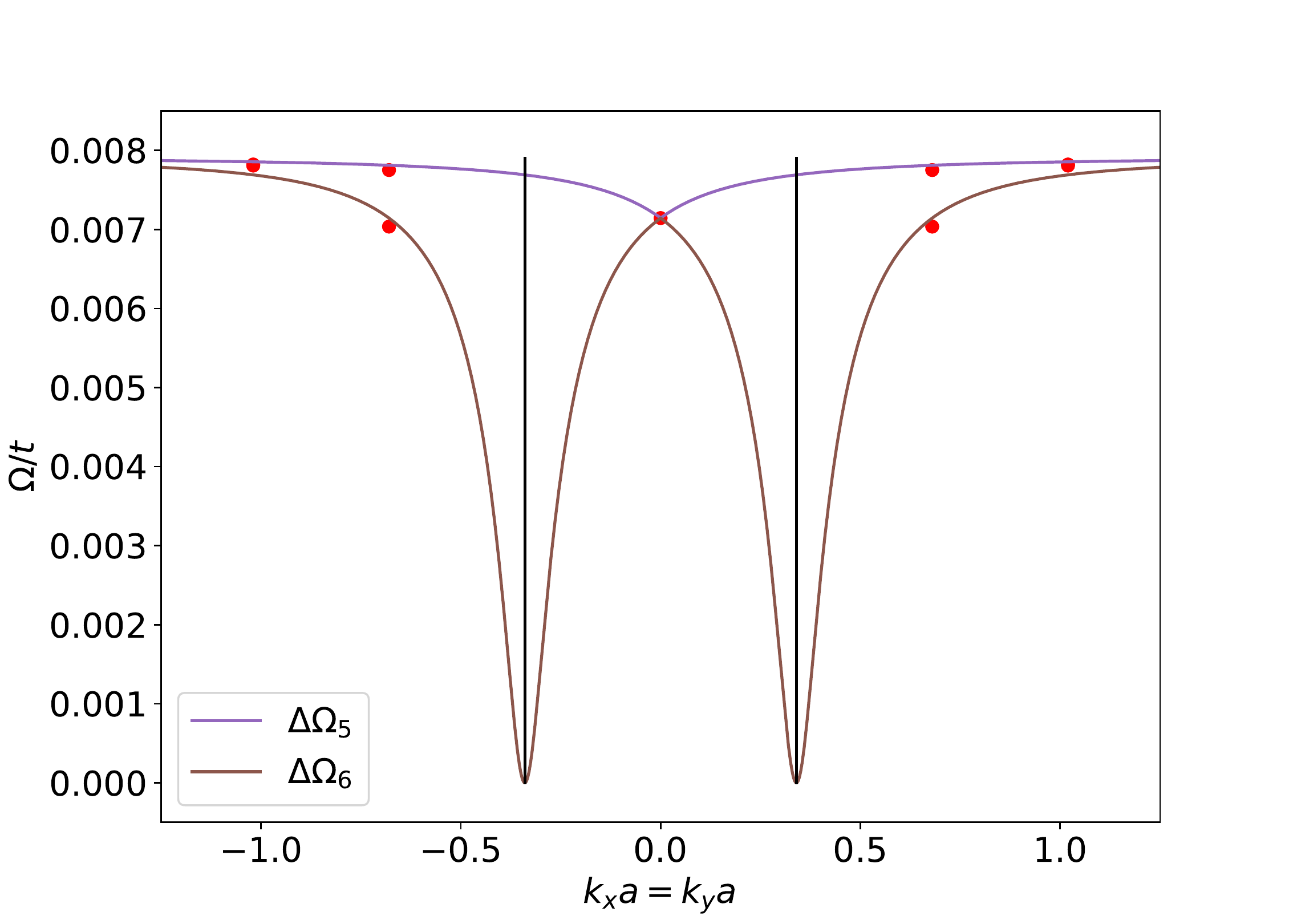}
      \caption{}
    \end{subfigure}%
    \caption{Shows the bands $\Delta\Omega_{\sigma}(\boldsymbol{k})$ along the $k_x=k_y$ direction for $U_s/t = 0.05$, $\alpha = 1.5$ and $\lambda_R/t = 0.5$. The black vertical lines show the position of $k_x=k_y=\pm k_{0m}$, while the red points show the special energies found at the special momenta. The inset in (a) shows the four lowest bands, while in (b) we focus on the two lowest bands.  \label{fig:SWband}}
\end{figure}

Proceeding similarly to the case where the angles were undetermined, we define 
$$N_q \equiv \sum_{\boldsymbol{k} \neq \pm\boldsymbol{k}_{01}, \pm2\boldsymbol{k}_{01}, \pm3\boldsymbol{k}_{01}}\sum_{\sigma=1}^6 B_{\boldsymbol{k},\sigma}^{\dagger}B_{\boldsymbol{k},\sigma}$$
and $\Omega_0$ as the maximum value of $\Omega_3(\boldsymbol{k})$. Also defining $\Delta\Omega_1 = \Omega_1+\Omega_0$, $\Delta\Omega_2 = \Omega_2+\Omega_0$, $\Delta\Omega_3 = \Delta\Omega_4 = \Omega_0$, $\Delta\Omega_5 = \Omega_0-\Omega_4 $ and $\Delta\Omega_{6} = \Omega_0-\Omega_3$, we get
\begin{equation}
     H'_2 = -\Omega_0 N_q + \sum_{\boldsymbol{k}\neq \pm\boldsymbol{k}_{01}, \pm2\boldsymbol{k}_{01}, \pm3\boldsymbol{k}_{01}}\sum_{\sigma=1}^6 \Delta\Omega_\sigma(\boldsymbol{k})\left(B_{\boldsymbol{k},\sigma}^{\dagger}B_{\boldsymbol{k},\sigma}+\frac12\right).
\end{equation}
We again used that $\boldsymbol{k} = \boldsymbol{0}$ can be incorporated in $H'_2$. Notice that both the energy bands and the operators have been given a new numbering. Similarly, the eigenvalues at $\pm2\boldsymbol{k}_{01}$ and $\pm3\boldsymbol{k}_{01}$ become double when the angels satisfy \eqref{eq:gammathetapi}. However the expressions for $H_2(\pm2\boldsymbol{k}_{01})$ and $H_2(\pm3\boldsymbol{k}_{01})$ do not become much simpler.

Figures \ref{fig:SWband} and \ref{fig:SWeigen} show the energy spectrum. The figures show that the SW phase is energetically stable at the chosen parameters. In addition to the gapless roton minima at $\pm\boldsymbol{k}_{01}$ there are gapped roton minima close to $\pm\boldsymbol{k}_{02}$. For $\alpha<1$ some eigenvalues become complex, indicating a dynamical instability. As long as $\alpha>1$, the eigenvalues, including those at the special momenta, remain real. Since we originally had two degrees of freedom, pseudospin up and down, we believe only the two lowest bands $\Delta\Omega_5(\boldsymbol{k})$ and $\Delta\Omega_6(\boldsymbol{k})$ are significant in the sense that the other bands are never occupied. The lowest band is clearly non-linear even close to the minima, and the critical superfluid velocity therefore seems to be zero. In the next section we will use the helicity approximation and find a spectrum which is linear close to the minimum. 

\begin{figure}
    \centering
    \includegraphics[width=0.9\linewidth]{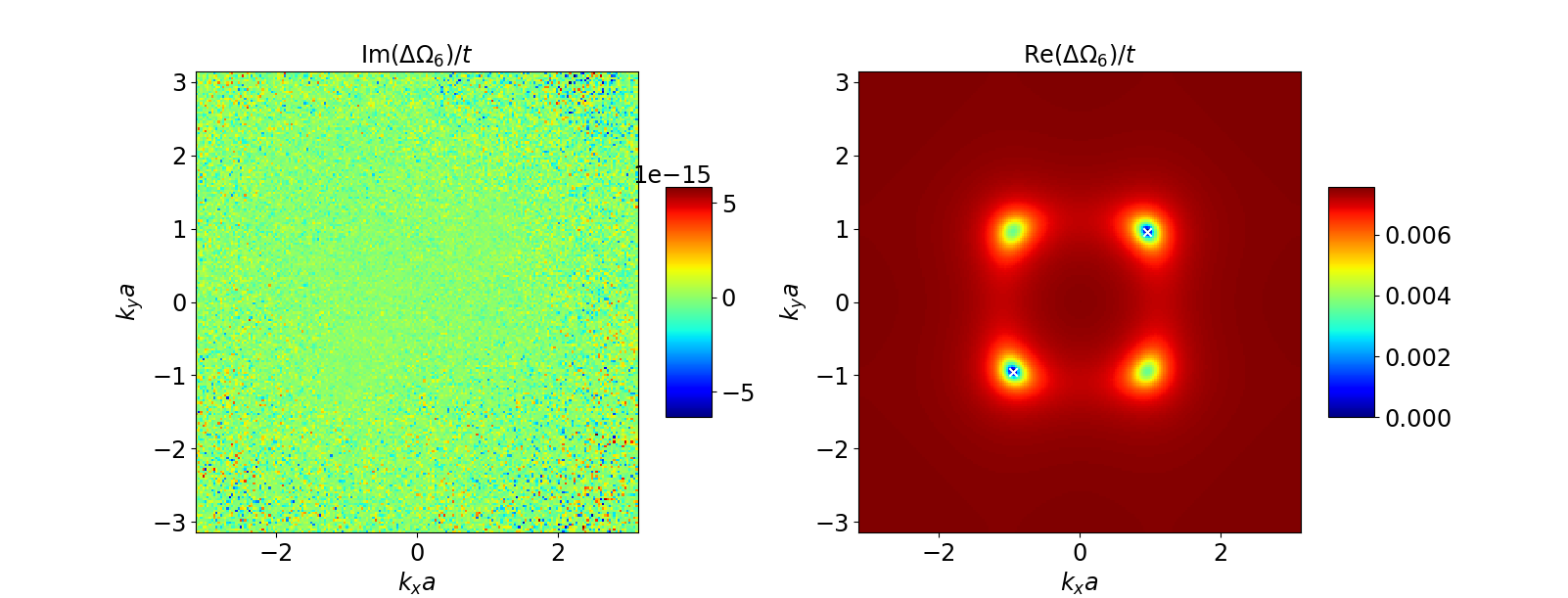}
    \caption{Shows the lowest energy $\Delta\Omega_6(\boldsymbol{k})$ in the first Brillouin zone. The white crosses show the position of $\boldsymbol{k} = \pm\boldsymbol{k}_{01}$, which are clearly the only global minima of the spectrum. We also observe gapped roton minima at $\pm\boldsymbol{k}_{02}$. The parameters were $U_s/t = 0.05$, $\alpha = 1.5$, $\lambda_R/t = 2.0$ and $N_s = 4\cdot 10^{4}$.}
    \label{fig:SWeigen}
\end{figure}


Now assume we are in a parameter regime where the eigenvalues are real. In our calculation of $\langle H_{\textrm{SW}} \rangle $ at zero temperature there was no indication that it is not bounded from below. The occurrence of anomalous modes may indicate an energetic instability in the context of solving the Gross-Pitaevskii equation \cite{PethickSmith}. In the approach we have used here, which involves transforming the description to a new basis wherein the system behaves like an ideal Bose gas of quasiparticles, the Hamiltonian was found to be bounded from below, at least in the sense that $\langle H_{\textrm{SW}} \rangle $ at zero temperature has a minimal value which is finite, i.e. not $-\infty$. In terms of the quasiparticle description, there is no lower energy state than the one where all quasiparticles occupy the lowest energy at the condensate momenta $\pm\boldsymbol{k}_{01}$. This is both a BEC and describes the SW phase, suggesting it is stable.



\subsection{Lowest Energy using Helicity Basis}
Similarly to what was done in the PW phase to obtain analytic eigenvalues we attempt to transform the problem to the helicity basis \eqref{eq:Helicitybasis} and then, since we are focused on BEC, we keep only the lowest band \eqref{eq:helicitybands}. We believe it was natural to first go through the calculation in the original spin basis because all bands are relevant to the calculation of the free energy at zero temperature, $\langle H \rangle$. We thus use the above results for the variational parameters in the following, i.e. $k_0 = k_{0m}$ and \eqref{eq:SWangles} for the angles. This calculation should be well suited to investigate the lowest band, which is the most interesting band in the context of BEC and to obtain the critical superfluid velocity.

Before we turn on interactions, the vast majority of the helicity quasiparticles should reside in the four minima of $\lambda_{\boldsymbol{k}}^-$. Our intuition is that the weak interactions should pick out a certain ground state, and that the energies close to the condensate momenta obtain a Bogoliubov effect such that they become phonon minima. This is what happened for the weakly interacting Bose gas, and in the phases PZ, NZ and PW. Using
\begin{equation}
    \begin{pmatrix} A_{\boldsymbol{k}}^\uparrow \\ A_{\boldsymbol{k}}^\downarrow \end{pmatrix} \approx \frac{1}{\sqrt{2}} \begin{pmatrix} -e^{-i\gamma_{\boldsymbol{k}}} C_{\boldsymbol{k}}\\ C_{\boldsymbol{k}}  \end{pmatrix}
\end{equation}
$H_2$ becomes
\begin{align}
    \begin{split}
        H_2 = \left.\sum_{\boldsymbol{k}}\right.^{''}& \Bigg\{ N_{11}(\boldsymbol{k})C_{\boldsymbol{k}}^{\dagger}C_{\boldsymbol{k}} +\Bigg(\bigg[N_{13}(\boldsymbol{k})C_{\boldsymbol{k}}^{\dagger}C_{\boldsymbol{k}+2\boldsymbol{k}_{01}} +N_{15}(\boldsymbol{k})C_{\boldsymbol{k}}^{\dagger}C_{\boldsymbol{k}-2\boldsymbol{k}_{01}}\\
        &+ (N_{72}/2)(\boldsymbol{k})C_{\boldsymbol{k}}C_{-\boldsymbol{k}} + N_{74}(\boldsymbol{k})C_{\boldsymbol{k}}C_{-\boldsymbol{k}+2\boldsymbol{k}_{01}}\\
        &+ N_{76}(\boldsymbol{k})C_{\boldsymbol{k}}C_{-\boldsymbol{k}-2\boldsymbol{k}_{01}}   \bigg] + \textrm{H.c.} \Bigg)\Bigg\}.
    \end{split}
\end{align}
Here, 
\begin{align}
    \begin{split}
        N_{11}(\boldsymbol{k}) &= M_{1,1}(\boldsymbol{k}) - \abs{s_{\boldsymbol{k}}} -\frac{U_s\alpha}{2}\left(\cos(\gamma_{\boldsymbol{k}}+\theta_1^\downarrow -\theta_1^\uparrow) + \cos(\gamma_{\boldsymbol{k}}+\theta_3^\downarrow -\theta_3^\uparrow) \right), \\
        N_{13}(\boldsymbol{k}) &= \frac{M_{1,3}e^{i(\gamma_{\boldsymbol{k}}-\gamma_{\boldsymbol{k}+2\boldsymbol{k}_{01}})}}{2} + \frac{M_{7,9}}{2} - \frac{M_{1,11}^* e^{-i\gamma_{\boldsymbol{k}+2\boldsymbol{k}_{01}}}}{2}-\frac{M_{1,9}e^{i\gamma_{\boldsymbol{k}}}}{2}, \\
        N_{15}(\boldsymbol{k}) &=\frac{M_{1,3}^*e^{i(\gamma_{\boldsymbol{k}}-\gamma_{\boldsymbol{k}-2\boldsymbol{k}_{01}})}}{2} + \frac{M_{7,9}^*}{2} - \frac{M_{1,9}^* e^{-i\gamma_{\boldsymbol{k}-2\boldsymbol{k}_{01}}}}{2}-\frac{M_{1,11}e^{i\gamma_{\boldsymbol{k}}}}{2}, \\
        N_{72}(\boldsymbol{k}) &= -\frac{M_{13,2}e^{-i2\gamma_{\boldsymbol{k}}}}{2}+\frac{M_{19,8}}{2}, \\
        N_{74}(\boldsymbol{k}) &= \frac{M_{13,4}e^{-i(\gamma_{\boldsymbol{k}}+\gamma_{-\boldsymbol{k}+2\boldsymbol{k}_{01}})}}{2} + \frac{M_{19,10}}{2} - \frac{M_{13,10}}{2}\left(e^{-i\gamma_{-\boldsymbol{k}+2\boldsymbol{k}_{01}}} + e^{-i\gamma_{\boldsymbol{k}}}\right), \\
        N_{76}(\boldsymbol{k}) &= \frac{M_{13,6}e^{-i(\gamma_{\boldsymbol{k}}+\gamma_{-\boldsymbol{k}-2\boldsymbol{k}_{01}})}}{2} + \frac{M_{19,12}}{2} - \frac{M_{13,12}}{2}\left(e^{-i\gamma_{-\boldsymbol{k}-2\boldsymbol{k}_{01}}} + e^{-i\gamma_{\boldsymbol{k}}}\right).
    \end{split}
\end{align}
We used that $e^{-i\gamma_{-\boldsymbol{k}}} = - e^{-i\gamma_{\boldsymbol{k}}}$ because $s_{-\boldsymbol{k}} = -s_{\boldsymbol{k}}$. Using commutators and making $-\boldsymbol{k}$-terms explicit we find
\begin{equation}
    H_2 = \frac{1}{4} \left.\sum_{\boldsymbol{k}}\right.^{'} \boldsymbol{C}_{\boldsymbol{k}}^\dagger N_{\boldsymbol{k}}\boldsymbol{C}_{\boldsymbol{k}}
\end{equation}
The operator vector is
\begin{align}
    \begin{split}
        \boldsymbol{C}_{\boldsymbol{k}} = (&C_{\boldsymbol{k}}, C_{-\boldsymbol{k}}, C_{\boldsymbol{k}+2\boldsymbol{k}_{01}}, C_{-\boldsymbol{k}+2\boldsymbol{k}_{01}}, C_{\boldsymbol{k}-2\boldsymbol{k}_{01}}, C_{-\boldsymbol{k}-2\boldsymbol{k}_{01}}, \\
        &C_{\boldsymbol{k}}^{\dagger}, C_{-\boldsymbol{k}}^{\dagger}, C_{\boldsymbol{k}+2\boldsymbol{k}_{01}}^{\dagger}, C_{-\boldsymbol{k}+2\boldsymbol{k}_{01}}^{\dagger}, C_{\boldsymbol{k}-2\boldsymbol{k}_{01}}^{\dagger}, C_{-\boldsymbol{k}-2\boldsymbol{k}_{01}}^{\dagger} )^T .
    \end{split}
\end{align}
The matrix $N_{\boldsymbol{k}}$ takes the form
\begin{gather}
    N_{\boldsymbol{k}} = 
    \begin{pmatrix}
    N_1 & N_2 \\
    N_2^* & N_1^* \\
    \end{pmatrix},
\end{gather}
with
\setcounter{MaxMatrixCols}{6}
\begin{gather*}
    N_1 = 
    \begin{pmatrix}
    N_{11}(\boldsymbol{k}) & 0 & N_{13}(\boldsymbol{k}) & 0 & N_{15}(\boldsymbol{k}) & 0 \\
    0 & N_{11}(-\boldsymbol{k}) & 0 & N_{13}(-\boldsymbol{k}) & 0 & N_{15}(-\boldsymbol{k}) \\
    N_{13}^*(\boldsymbol{k}) & 0 & 0 & 0 & 0 & 0 \\
    0 & N_{13}^*(-\boldsymbol{k}) & 0 & 0 & 0 & 0 \\
    N_{15}^*(\boldsymbol{k}) & 0 & 0 & 0 & 0 & 0 \\
    0 & N_{15}^*(-\boldsymbol{k}) & 0 & 0 & 0 & 0 \\
    \end{pmatrix}
\end{gather*}
and
\begin{gather*}
    N_2^* = 
    \begin{pmatrix}
    0 & N_{72}(\boldsymbol{k}) & 0 & N_{74}(\boldsymbol{k}) & 0 & N_{76}(\boldsymbol{k}) \\
    N_{72}(\boldsymbol{k}) & 0 & N_{74}(-\boldsymbol{k}) & 0 & N_{76}(-\boldsymbol{k}) & 0 \\
    0 & N_{74}(-\boldsymbol{k}) & 0 & 0 & 0 & 0 \\
    N_{74}(\boldsymbol{k}) & 0 & 0 & 0 & 0 & 0 \\
    0 & N_{76}(-\boldsymbol{k}) & 0 & 0 & 0 & 0 \\
    N_{76}(\boldsymbol{k}) & 0 & 0 & 0 & 0 & 0 \\
    \end{pmatrix}.
\end{gather*}
The eigenvalues of $N_{\boldsymbol{k}}J$ are found numerically, as \textit{Maple} did not provide analytic eigenvalues. We focus only on the lowest band, as that is the one that is relevant for BEC. This turns out to be a double eigenvalue with an anomalous mode. We name the original positive energy $\Omega_H(\boldsymbol{k})$ and then the true lowest band $\Delta\Omega_H(\boldsymbol{k}) \equiv \max_{\boldsymbol{k}}\Omega_H(\boldsymbol{k}) - \Omega_{H}(\boldsymbol{k})$. The helicity basis is undefined at $\boldsymbol{k} = \boldsymbol{0}$. Therefore the current treatment does not cover the points $\boldsymbol{0}$ or $\pm2\boldsymbol{k}_{01}$. These will need to be treated in the original spin basis as has been done previously. The treatment of the special momenta $\pm3\boldsymbol{k}_{01}$ can be done in the helicity basis by the same procedure as in the original spin basis. 

\begin{figure}
    \centering
    \includegraphics[width=0.7\linewidth]{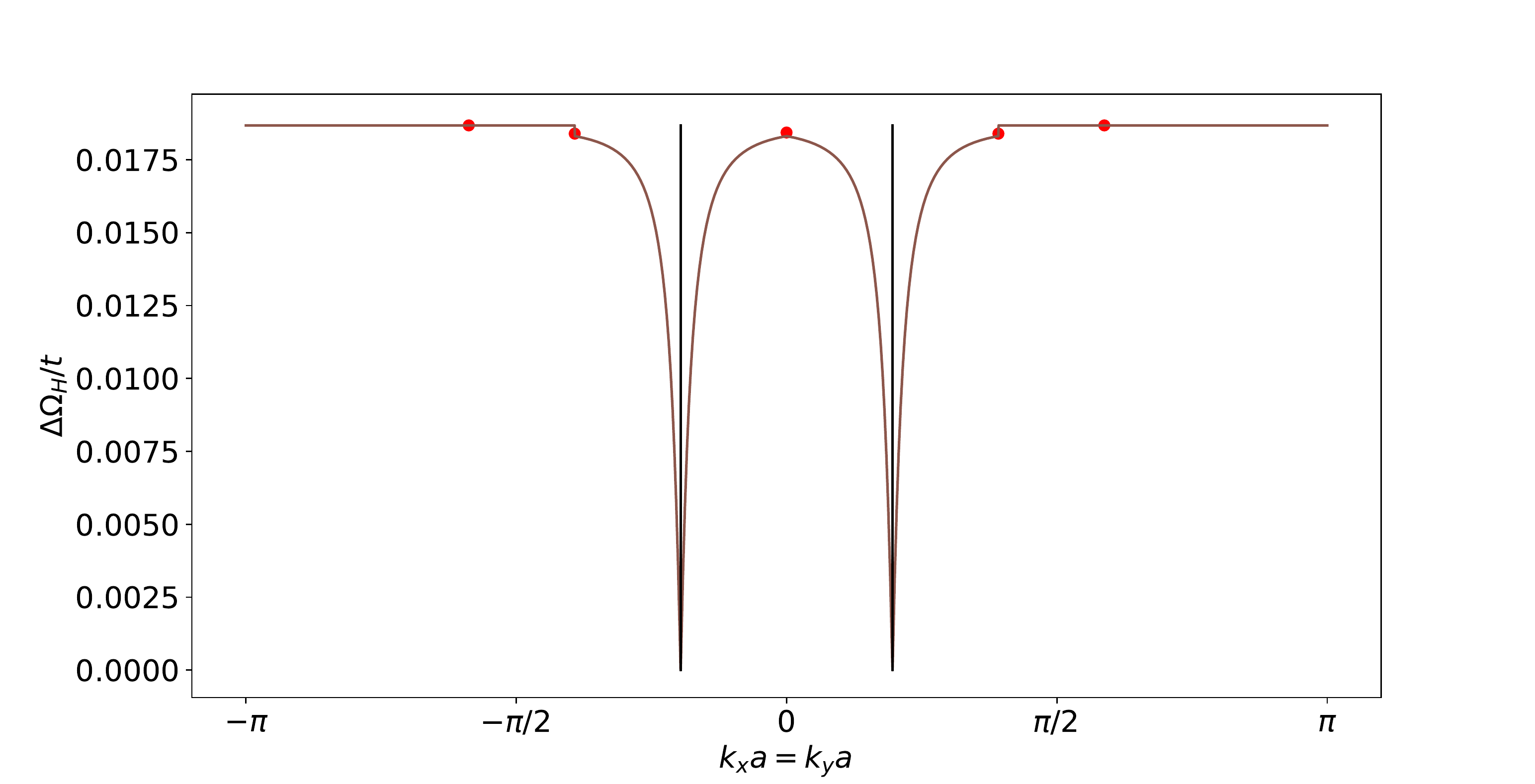}
    \caption{Shows the band $\Delta\Omega_{H}(\boldsymbol{k})$ along the $k_x=k_y$ direction for $U_s/t = 0.05$, $\alpha = 1.5$ and $\lambda_R/t = 1.0$. The black vertical lines show the position of $k_x=k_y=\pm k_{0m}$. The red points show the special energies found at the special momenta.}
    \label{fig:SWHelBand}
\end{figure}

\begin{figure}
    \centering
    \begin{subfigure}{.45\textwidth}
      \includegraphics[width=\linewidth]{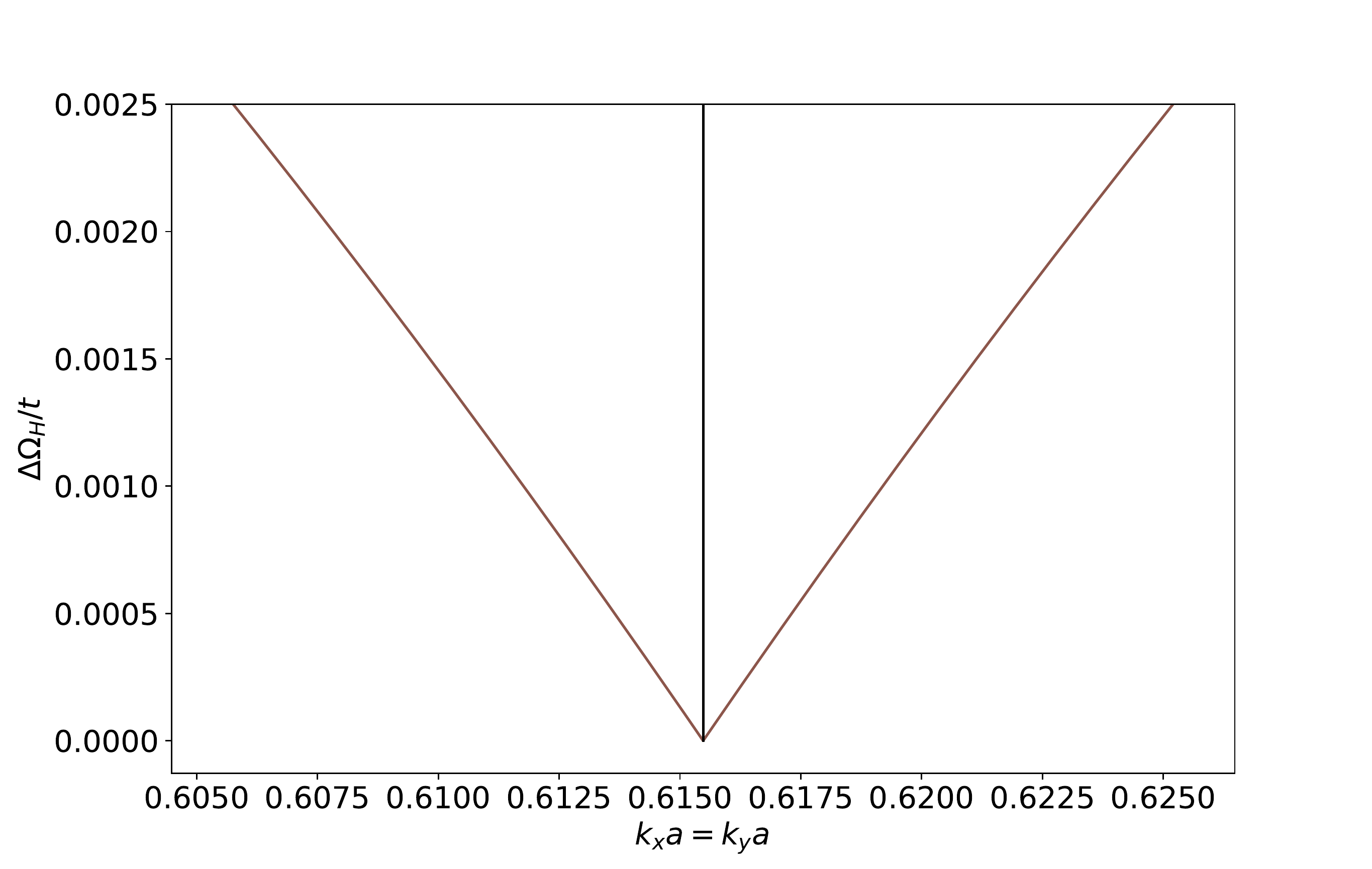}
      \caption{\label{fig:SWHelBandzoom}}
    \end{subfigure}%
    \begin{subfigure}{.45\textwidth}
      \includegraphics[width=\linewidth]{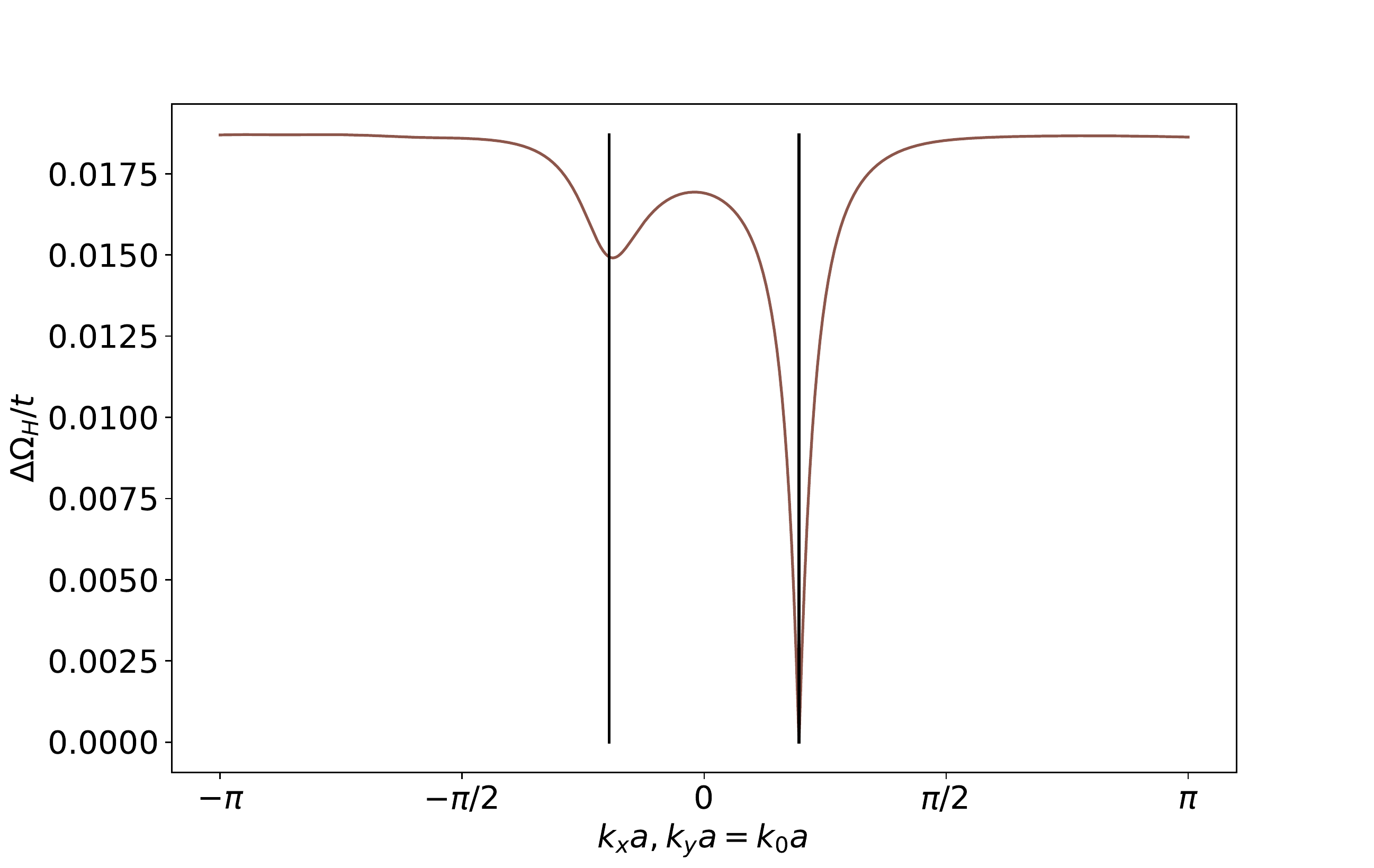}
      \caption{\label{fig:SWHelBandkyk0}}
    \end{subfigure}%
    \caption{(a) shows the linear behavior of $\Delta\Omega_{H}(\boldsymbol{k})$ close to the minimum at $\boldsymbol{k}_{01}$ along the $k_x=k_y$ direction. In (b) we show the band along $k_x$ when $k_y = k_0 = k_{0m}$. The parameters are $U_s/t = 0.05$, $\alpha = 1.5$ and $\lambda_R/t = 1.0$. The black vertical lines show the position of $k_x=k_y= k_{0m}$ (a) and $k_x = \pm k_{0m}$ (b).}
\end{figure}

\begin{figure}
    \centering
    \includegraphics[width=0.9\linewidth]{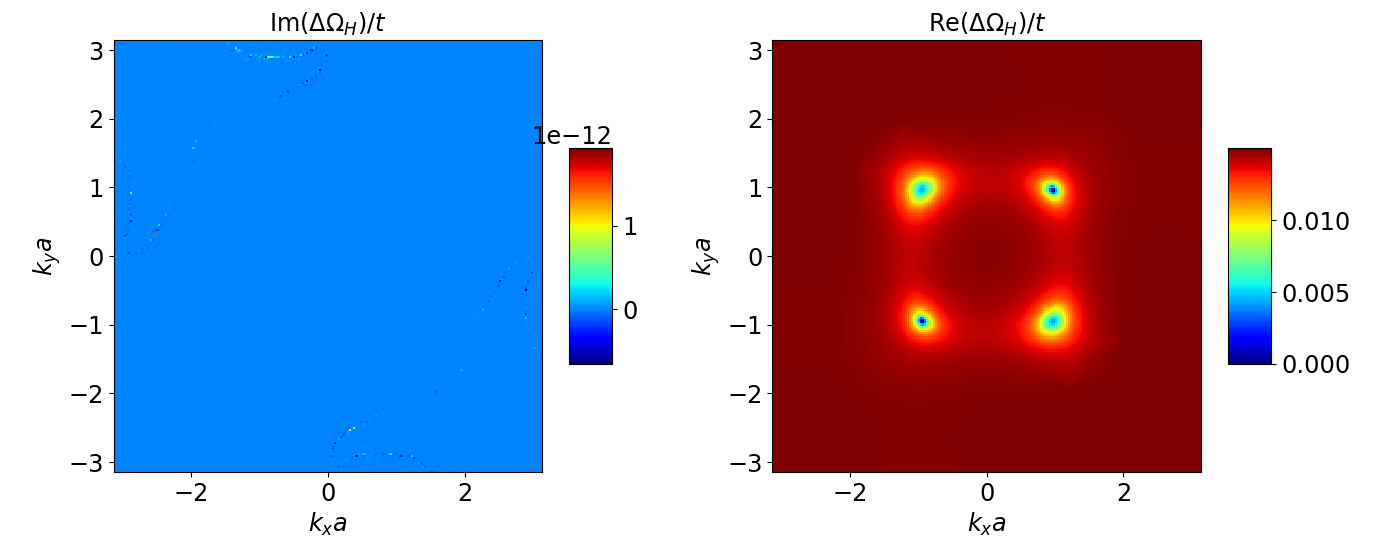}
    \caption{Shows real and imaginary parts of the lowest energy $\Delta\Omega_H(\boldsymbol{k})$ in the first Brillouin zone. The parameters were $U_s/t = 0.05$, $\alpha = 1.5$, $\lambda_R/t = 2.0$ and $N_s = 4\cdot 10^{4}$.}
    \label{fig:SWHeleigen}
\end{figure}

The band  $\Delta\Omega_H(\boldsymbol{k})$ is shown in figure \ref{fig:SWHelBand} along the direction $k_x = k_y$ and in figure \ref{fig:SWHelBandzoom} we focus on the linear behavior close to the minimum at $\boldsymbol{k}_{01}$, suggesting nonzero critical superfluid velocity. The discontinuities at $\boldsymbol{0}$ and $\pm2\boldsymbol{k}_{01}$ are because the helicity basis is undefined for some of the operators here. The lowest special value at $\pm3\boldsymbol{k}_{01}$ coincides with $\Delta\Omega_H(\pm3\boldsymbol{k}_{01})$. In figure \ref{fig:SWHelBandkyk0} we show the band along $k_x$ when $k_y = k_0 = k_{0m}$.  The lowest energy $\Delta\Omega_H(\boldsymbol{k})$ is shown in the 1BZ in figure \ref{fig:SWHeleigen}. We find two global phonon minima at the condensate momenta, and two gapped roton minima placed approximately at $\pm\boldsymbol{k}_{02}$. Accepting that the eigenvalue problem is more prone to numerical errors in the helicitiy basis, the imaginary parts are small enough too claim that the energy is real.

\subsubsection{Critical Superfluid Velocity}
We find an anisotropic critical superfluid velocity. As we did in the PW phase we will give plots of $v_c^{\textrm{min}}$, $v_c^{\textrm{max}}$ and $v_c(\phi)$. We use the minimum at $\boldsymbol{k}_{01}$ to obtain these, and the formulae are
\begin{equation}
    \label{eq:SWvcnum}
    v_c^{\textrm{min}} = \frac{\min_{\boldsymbol{q}}[\Delta\Omega_{H}(\boldsymbol{k}_{01}+\boldsymbol{q})]}{\abs{\boldsymbol{q}}}, \mbox{\qquad\qquad} v_c^{\textrm{max}} = \frac{\max_{\boldsymbol{q}}[\Delta\Omega_{H}(\boldsymbol{k}_{01}+\boldsymbol{q})]}{\abs{\boldsymbol{q}}}
\end{equation}
and
\begin{equation}
    v_c(\phi) =  \frac{\Delta\Omega_{H}\big(\boldsymbol{k}_{01}+\abs{\boldsymbol{q}}(\cos(\phi), \sin(\phi))\big)}{\abs{\boldsymbol{q}}}.
\end{equation}

\begin{figure}
    \centering
    \includegraphics[width=0.7\linewidth]{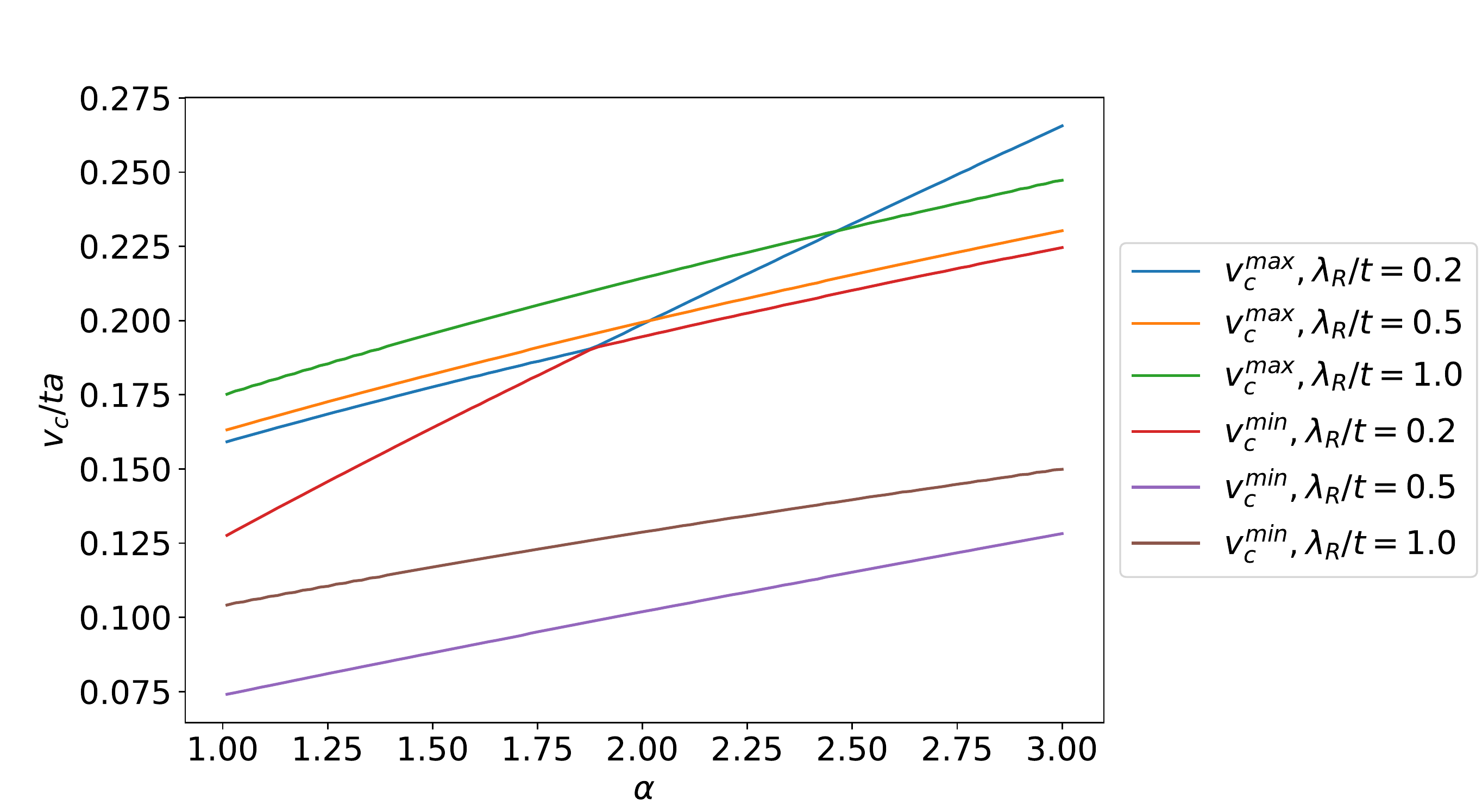}
    \caption{Maximum and minimum values of $v_c^{\textrm{SW}}$ are plotted against $\alpha$ for various $\lambda_R$ with $U_s/t = 0.05$. The value of $k_0 = k_{0m}$ was updated as $\lambda_R$ was changed.}
    \label{fig:SWsuperal}
\end{figure}

\begin{figure}
    \centering
    \includegraphics[width=0.7\linewidth]{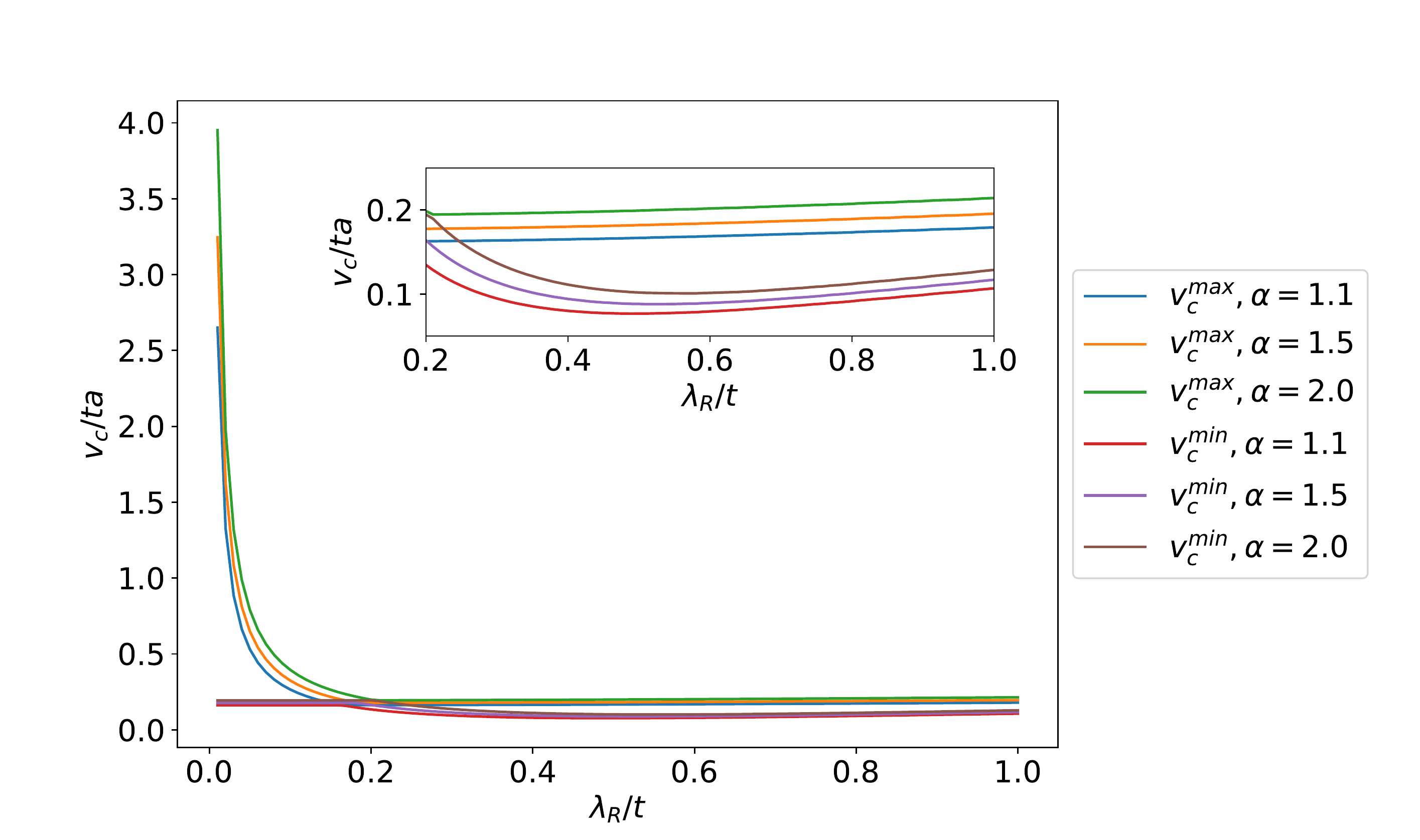}
    \caption{Maximum and minimum values of $v_c^{\textrm{SW}}$ are plotted against $\lambda_R$ for various $\alpha$, with $U_s/t = 0.05$. The value of $k_0 = k_{0m}$ was updated at each $\lambda_R$. Note that $\lambda_R/t \geq 0.01$ was used. The inset focuses on the behavior when $\lambda_R/t \geq 0.2$.}
    \label{fig:SWsuperlam}
\end{figure}

\begin{figure}
    \centering
    \includegraphics[width=0.7\linewidth]{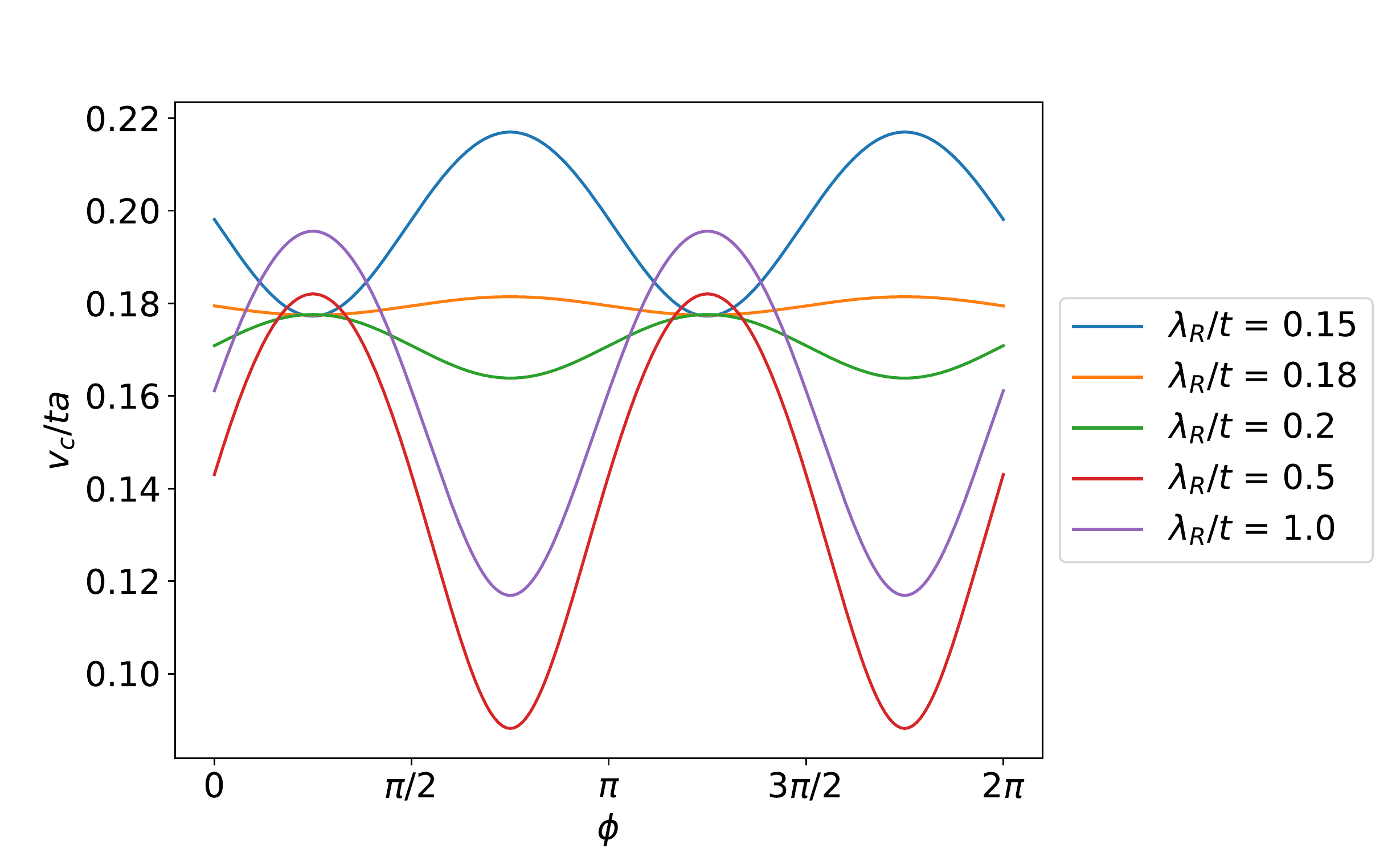}
    \caption{$v_c^{\textrm{SW}}$ is plotted against the angle with the $k_x$-axis, $\phi$, with $U_s/t = 0.05$ and $\alpha = 1.5$. The value of $k_0 = k_{0m}$ was updated as $\lambda_R$ was changed. $v_c^{\textrm{SW}}$ is approximately $\pi$-periodic.}
    \label{fig:SWsuperphi}
\end{figure}

Like in the PW phase, we use $|\boldsymbol{q}|a = 10^{-5}$ in producing the figures. For the SW phase, the critical superfluid velocity is shown in figure \ref{fig:SWsuperal} as a function of $\alpha$ for various values of $\lambda_R$. In figure \ref{fig:SWsuperlam} we plot it as a function of $\lambda_R$ for various values of $\alpha$. Both figures suggest $v_c^{\textrm{SW}}$ increases with increasing $\alpha$. The behavior with $\lambda_R$ is more exotic, and can be understood from figure \ref{fig:SWsuperphi}, showing $v_c(\phi)$ at several $\lambda_R/t$. The direction in which the critical superfluid velocity is greatest appears to change via an isotropic case at a certain $\lambda_R/t$ that depends on $\alpha$. For stronger SOC the behavior is as in the PW phase. $v_c$ is greatest along $\boldsymbol{k}_{01}$ and smallest perpendicular to it. Meanwhile, for weaker SOC the opposite is true. One can understand why the critical superfluid velocity along $\boldsymbol{k}_{01}$ decreases when a weak $\lambda_R/t$ decreases further. The value of $k_0$ will also decrease, and so the two global phonon minima move closer and closer. Hence, there is a limit to how large the energy can become between the minima. On the other hand, the direction perpendicular to $\boldsymbol{k}_{01}$ has no such limitation and the slope there increases. Once we have traversed the strange effect at weak SOC, both $v_c^{\textrm{min}}$ and $v_c^{\textrm{max}}$ begin to increase with increasing $\lambda_R$.

Remember that the energy spectrum obtained from the original spin basis suggested zero critical superfluid velocity. As mentioned, we would intuitively expect the excitation spectrum to be linear close to the condensate momenta due to a Bogoliubov effect from the interactions. A natural question is why the treatment in the original spin basis did not catch the superfluid behavior. An attempt to investigate this is presented in appendix \ref{app:SWdifference}, but no significant insights were gained.

\subsection{Comparison of Spin and Helicity Basis Results}
In the PW phase the results for the critical superfluid velocity were the same whether we used the lower helicity band, or the original spin basis. The global behavior of the energy bands were also similar. The helicity result $\Omega_H(\boldsymbol{k})$ compared favourably to the lowest spin result $\Omega_2(\boldsymbol{k})$, while the upper spin result $\Omega_1(\boldsymbol{k})$ resembled the upper helicity band \eqref{eq:helicitybands}, $\lambda_{\boldsymbol{k}}^+$, that was neglected in the helicity approximation. Globally, the same is true in the SW phase. The bands obtained in the helicity basis are similar to the bands $\Delta\Omega_2(\boldsymbol{k}), \Delta\Omega_4(\boldsymbol{k})$ and $\Delta\Omega_6(\boldsymbol{k})$ in the spin basis. There is however a major difference between $\Delta\Omega_H(\boldsymbol{k})$ and $\Delta\Omega_6(\boldsymbol{k})$ close to the minima at the condensate momenta $\pm\boldsymbol{k}_{01}$. While the helicity approximation gives a linear behavior, the result in the original spin basis gave an approximately quadratic behavior. 

At $\alpha=1$, i.e. equal strength of inter- and intracomponent interactions the spin basis result also displays linear behavior in its lowest nonzero band. However, here the anomalous modes are zero, and so the lowest band is technically the eigenvalues that are zero for all $\boldsymbol{k}$. The stability of the SW phase is at best questionable in such a case. A possible explanation for the non-linearity of the SW phase in the spin basis at $\alpha>1$ is found by considering the PW phase. When $\alpha <1$ the PW phase has a linear minimum at $\boldsymbol{k}_{01}$ and a gapped roton minimum at $-\boldsymbol{k}_{01}$. At $\alpha = 1$ the roton minimum becomes ungapped. For $\alpha>1$ the roton minimum becomes negative and is hence lower than the linear minimum at $\boldsymbol{k}_{01}$ suggesting the PW phase is unstable. Now imagine a superposition of two PW phases, one at $\boldsymbol{k}_{01}$ and one at $-\boldsymbol{k}_{01}$. For $\alpha>1$ the roton minimum at $-\boldsymbol{k}_{01}$ due to a PW phase at $\boldsymbol{k}_{01}$ becomes lower than the linear minimum due to the PW phase at $-\boldsymbol{k}_{01}$ and vice versa. Hence, the SW phase is the result, two negative, approximately quadratic, global minima at $\pm \boldsymbol{k}_{01}$.


The SW phase was detected experimentally in \cite{SWdetectedExp} though for a slightly different system than what is studied here. In \cite{SWdetectedExp} a continuum BEC is loaded into a 1D optical superlattice. Additionally a different SOC scheme is used, realizing a model similar to what is described in \cite{SWdetectedTheory} and the 1D Raman induced SOC that was first implemented in \cite{lin2011expSOC}. Both \cite{SWdetectedExp, SWdetectedTheory} claim the SW phase shows superfluid behaviour, though the experimental evidence of superfluidity appears to be based solely on the fact that a sharp momentum distribution is observed in time-of-flight \cite{SWdetectedExp}. This should however also be true for a BEC that is not superfluid. In \cite{SWdetectedTheory} the drag force is calculated, and it is shown that the time-scale over which dissipation occurs is larger that the duration of the experiment, and so the motion of an impurity can be considered as dissipationless. A linear dispersion and nonzero critical superfluid velocity was reported for the SW phase in a similar system in  \cite{li2013superstripes}. While this does not prove the SW phase should have a nonzero critical superfluid velocity in the prescense of a square 2D optical lattice and Rashba SOC, it is an indication that the results obtained in the helicity approximation are sensible. 

There is also another difference between the spin and helicity basis results. The maximum value $\Omega_0$ of $\Omega_3(\boldsymbol{k})$ becomes zero at $\alpha = 3$ and is hence small close to $\alpha =3$. The maximum value of $\Omega_H(\boldsymbol{k})$ is in general larger than $\Omega_0$ and does not have a zero for $\alpha>1$. In the original spin basis this means that close to $\alpha=3$ the lowest special energy at $\pm 2\boldsymbol{k}_{01}$,  $\Delta \omega_{2k_0, 10}$ may become negative. At $\lambda_R/t = 0.5$ this is contained within $\alpha = (2.72,3.53)$, at $\lambda_R/t = 1.0$ it is contained within $\alpha = (2.90, 3.12)$ and at $\lambda_R/t = 2.0$ it is contained within $\alpha = (2.96, 3.05)$. The lowest special energy at $\pm 3\boldsymbol{k}_{01}$,  $\Delta \omega_{3k_0, 8}$ may also become negative, but that happens only when $\Delta \omega_{2k_0, 10}$ has already become negative. 

At face value, this appears to be an energetic instability, in the sense that the global minima of the excitation spectrum are no longer at $\boldsymbol{k} = \pm\boldsymbol{k}_{01}$. A similar energetic instability does not occur for the helicity approximation. Nevertheless, we suggest treating this as a mathematical curiosity in the spin basis rather that an indication of instability in the SW phase. The reason being that we view the necessity of treating the special momenta separately as mathematical artifacts pertaining to the BV diagonalization procedure. The most natural result physically is a continuous excitation spectrum, in which case $\Delta\Omega_\sigma(\pm2\boldsymbol{k}_{01})$ is considered to be the energy at $\boldsymbol{k} = \pm2\boldsymbol{k}_{01}$ rather than the special values $\Delta\omega_{2k_0,\sigma}$.  

In conclusion, the results for the lowest band in the helicity approximation are more in accordance with our intuition and published literature \cite{SWdetectedExp, SWdetectedTheory, li2013superstripes}. There is also an argument that this method is best suited to investigate the behavior close to the minimum of the spectrum, since we before introducing interactions focused solely on the lowest helicity energy band. We therefore suggest the presence of a nonzero, anisotropic critical superfluid velocity in the SW phase based on these results. In addition there are no indications of either dynamic or energetic instabilities of the SW phase at $\alpha>1$ using the results from the helicity approximation even when treating the special momenta in a mathematically sound way.

\textit{Note: We later realized that it is the results from the spin basis that must be trusted. The helicity approximation fails to describe the SW phase, as discussed in our paper \cite{PhysRevA.102.053318}. Apparently, neglecting the upper helicity band is not a good approximation in the SW phase. This band affects the lowest band of the excitation spectrum, and the SW phase shows zero sound velocity of the excitations.}

%% file: 7LWShort.tex
\section{LW Phase} \label{sec:LWphase}
The LW phase is such that $\boldsymbol{k}_{01}=(k_0, k_0)$, $\boldsymbol{k}_{02}=(-k_0, k_0)$, $\boldsymbol{k}_{03} = -\boldsymbol{k}_{01}$ and $\boldsymbol{k}_{04} = -\boldsymbol{k}_{02}$ are occupied condensate momenta. We will find that the LW phase is not present in the phase diagram in chapter \ref{chap:PDdiscuss}, but nevertheless believe a treatment of the LW phase is relevant, for the purposes of proving just that. 
The general approach and results bear many similarities with the SW phase, and for the sake of brevity we postpone the calculations to appendix \ref{app:LW}.
The final result for the excitation spectrum is presented below.

With the values for the variational parameters found in appendix \ref{app:LW}, the 8 nonzero bands become 4 double nonzero bands. We may then describe the system as having 18 bands. However, we started out with only two degrees of freedom, pseudospin up and down. Hence, all bands apart from the lowest two will be assumed to be frozen out, i.e. to have occupation numbers zero. If we redefine the 4 positive eigenvalues as $\Omega_\sigma(\boldsymbol{k})$ for $\sigma = 1,2,3,4$, the shifted energies are $\Delta\Omega_\sigma(\boldsymbol{k}) = \Omega_0 + \Omega_\sigma(\boldsymbol{k})$ for $\sigma = 1,2$, $\Delta\Omega_\sigma(\boldsymbol{k}) = \Omega_0$ for $\sigma = 3,\dots, 16$ and $\Delta\Omega_\sigma(\boldsymbol{k}) = \Omega_0 - \Omega_{\sigma'}(\boldsymbol{k})$ for $\sigma = 17,18$ and $\sigma' = 4,3$. Hence, the diagonalized Hamiltonian can be rewritten
\begin{align}
    \begin{split}
        H_2 = &-\Omega_0 N_q + \frac{1}{2}\left.\sum_{\boldsymbol{k}}\right.^{'} \sum_{\sigma=1}^{16} \Delta\Omega_\sigma(\boldsymbol{k}) \\
        &+\left.\sum_{\boldsymbol{k}}\right.^{'} \sum_{\sigma=17}^{18} \Delta\Omega_\sigma(\boldsymbol{k}) \left( B_{\boldsymbol{k},\sigma}^\dagger B_{\boldsymbol{k},\sigma} +\frac12 \right).
    \end{split}
\end{align}
Now, $\Omega_0$ is the maximum value of $\Omega_3(\boldsymbol{k})$, and the definition of $N_q$ follows the usual procedure.

\begin{figure}[ht]
    \centering
    \begin{subfigure}{.45\textwidth}
      \includegraphics[width=\linewidth]{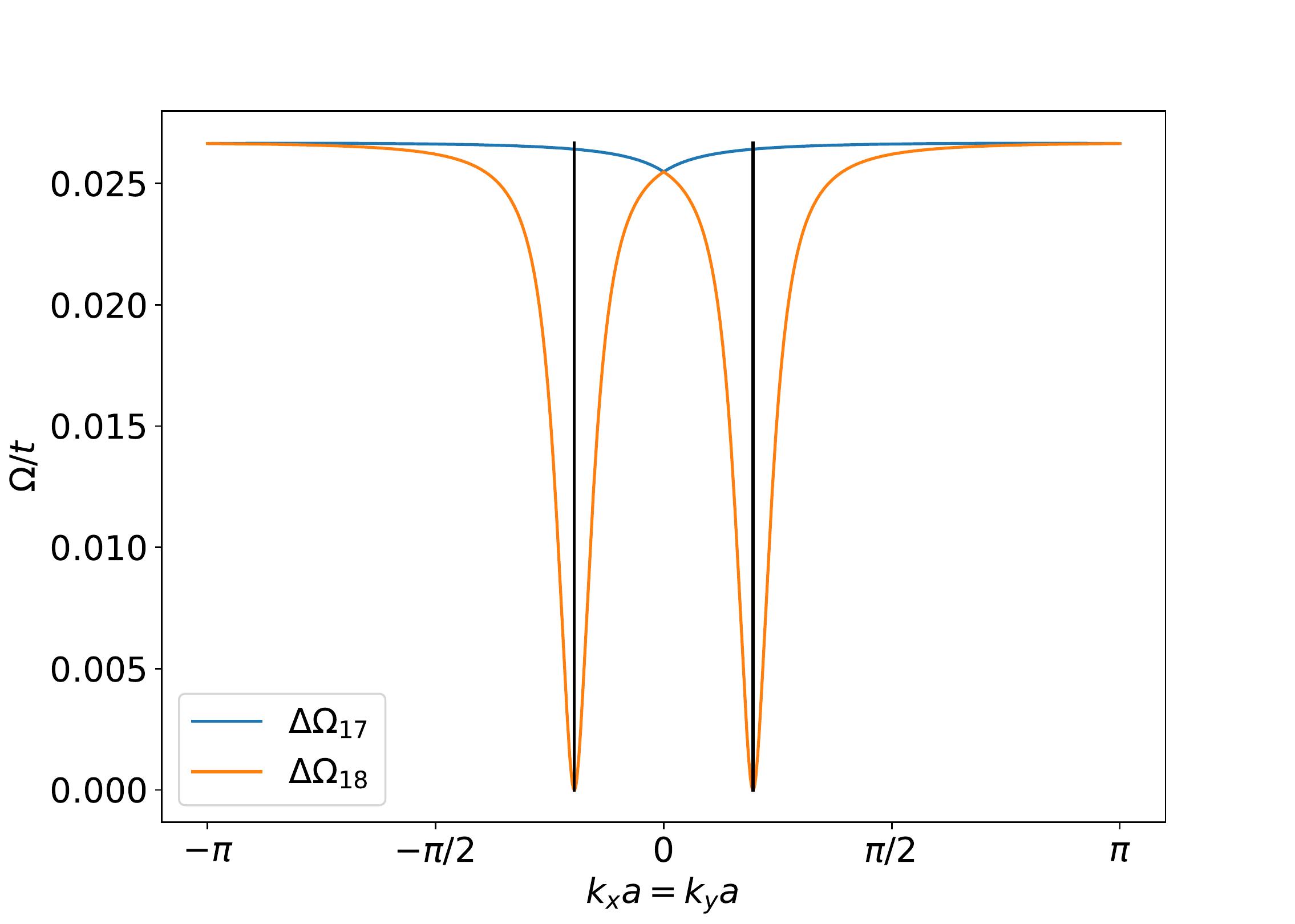}
      \caption{}
    \end{subfigure}%
    \begin{subfigure}{.45\textwidth}
      \includegraphics[width=\linewidth]{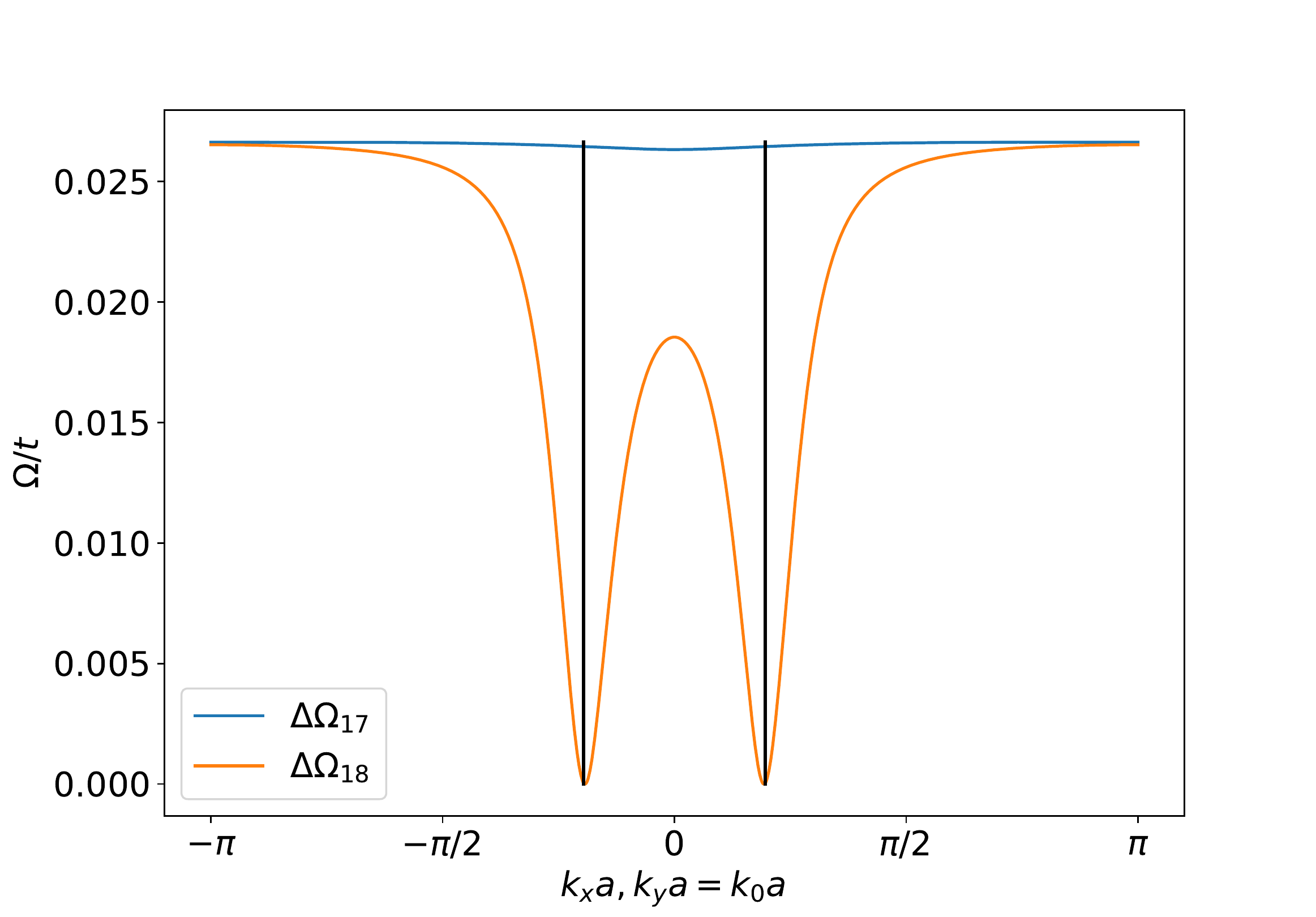}
      \caption{}
    \end{subfigure}%
    \caption{The two lowest bands $\Delta\Omega_{17}(\boldsymbol{k})$ and $\Delta\Omega_{18}(\boldsymbol{k})$ are plotted along $k_x=k_y$ in (a) and along $k_x$ for $k_y = k_0 = k_{0m}$ in (b). The minima are nonlinear. The parameters are $U_s/t = 0.05$, $\alpha = 1.5$ and $\lambda_R/t = 1.0$. The black vertical lines show the position of $k_x=k_y= \pm k_{0m}$ (a) and $k_x = \pm k_{0m}$ (b). \label{fig:LWband}}
\end{figure}

\begin{figure}[ht]
    \centering
    \includegraphics[width=0.9\linewidth]{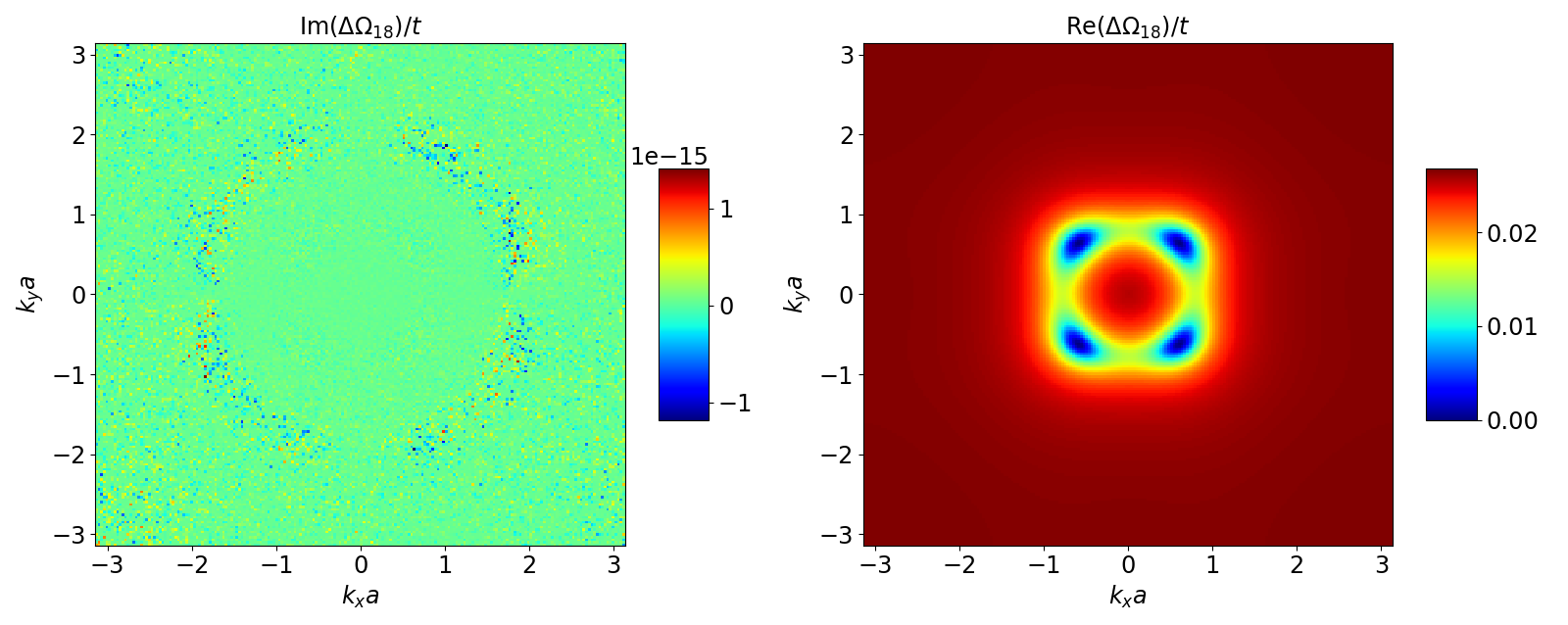}
    \caption{Shows the lowest energy $\Delta\Omega_{18}(\boldsymbol{k})$ for $U_s/t = 0.05$, $\alpha = 1.5$, $\lambda_R/t = 1.0$ and $N_s = 4\cdot 10^4$. The four global minima are located at the four condensate momenta in the LW phase. \label{fig:LWeigen}}
\end{figure}

The two lowest bands are plotted in figure \ref{fig:LWband} while the lowest energy is shown in the 1BZ in figure \ref{fig:LWeigen}. The four global minima occur at the condensate momenta, and they are found to be roton minima. Unlike the SW phase, the helicity approximation does not give phonon minima. The global minima are still approximately quadratic close to the minima. As there is no new physical insight gained from the helicity basis it is omitted here. In conclusion, the critical superfluid velocity is zero in the LW phase.

%% file: chapter5.tex

\chapter{Phase Diagram and Discussion} \label{chap:PDdiscuss}
\section{Phase Diagram Based on Free Energy}
As mentioned previously, the PZ phase is special because it has $N_0^\downarrow = 0$. Given that $(N-N_0)/N \ll 1$ the choice $N^\uparrow = N^\downarrow$ for the input parameters is not possible. It should be possible to engineer this phase for any $\alpha$ or $\lambda_R$ given that $N^\uparrow \approx N \gg N^\downarrow$ and that the energy offset $T^\downarrow \neq T^\uparrow = T$ is chosen such that $2\Delta = T^\downarrow-T+ 2U_s(\alpha-1)$ obeys $\Delta > U_s$ and $\Delta \geq 2\lambda_R^2/t$.

Next, we choose $N^\uparrow = N^\downarrow$ and $T^\uparrow = T^\downarrow = T$. The phases under consideration are the NZ, PW, SW and LW phases. For no SOC, only NZ is possible since nonzero condensate momenta requires SOC, and it was found that the NZ phase is only stable for $\alpha \leq 1$. For $\alpha>1$ there is no stable state when $\lambda_R = 0$, $N^\uparrow = N^\downarrow$ and $T^\downarrow = T^\uparrow = T$. From now on we focus on nonzero SOC, and the phases PW, SW and LW. It was found that the PW phase can only be stable for $\alpha<1$, while the SW phase is stable for $\alpha>1$. Meanwhile, the LW phase is stable for $\lambda_R/t \gtrsim 0.52+0.22\alpha$ and $\alpha$ greater than a lower limit that approaches $1$ from above as the strength of SOC is increased. 

Investigations of the free energy at zero temperature show that $\langle H_{\textrm{LW}} \rangle > \langle H_{\textrm{SW}} \rangle$ for $\alpha>1$ and $\lambda_R/t \gtrsim 0.52+0.22\alpha$, and so the SW phase will be preferred here. For $\alpha<1$ and $\lambda_R/t >0$ the only candidate is the PW phase. The result is presented in figure \ref{fig:PDH}. The main difference from the results using $H_0$ is that at $\lambda_R = 0$ there is no stable state when $\alpha>1$. This is because the NZ phase was found to be unstable here. Since the ground state energy, $\langle H \rangle$, of the LW phase is higher than that of the SW phase, there is reason to assume similar results would have been obtained for the ignored phases C1 and C2. Just as the LW phase, they did not enter the phase diagram of figure \ref{fig:PDH0new} when neglecting excitations. A final point regarding stability is appropriate here. We found the criteria for dynamic and energetic stability of the LW phase assuming it existed. These calculations nevertheless show it does not exist, and so the phase is not stable, in the sense that the SW phase will be preferred at all input parameters where the LW phase is a candidate.


\begin{figure}
    \centering
    \includegraphics[width=0.7\linewidth]{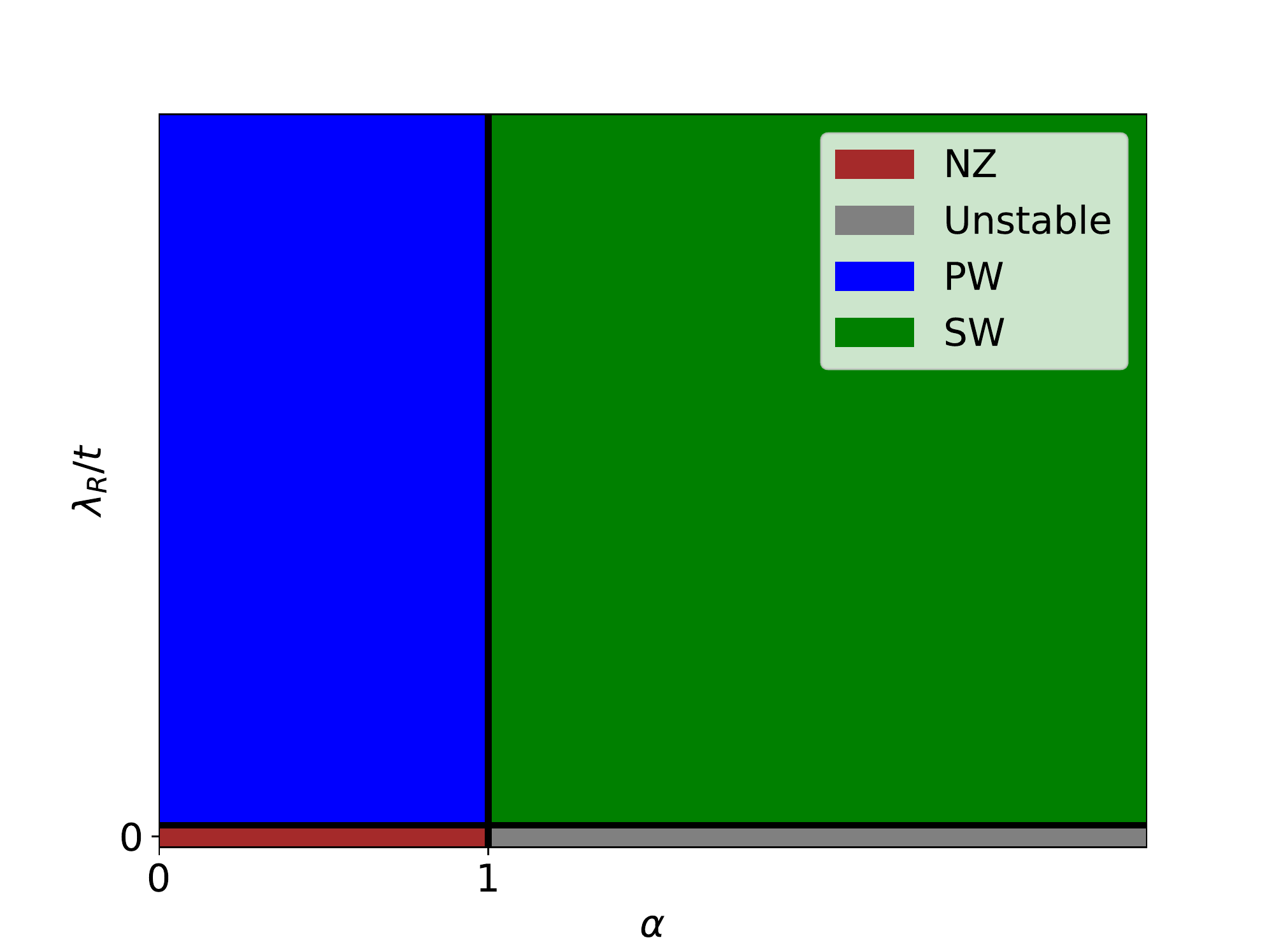}
    \caption[Phase diagram]{Phase diagram when $N^\uparrow = N^\downarrow$ and $T^\uparrow = T^\downarrow = T$. The effects of the elementary excitations have been included. The region of $\lambda_R = 0$ has been exaggerated for better visibility.}
    \label{fig:PDH}
\end{figure}

\section{Ground State Depletion}
The validity of the mean field theory approach requires that the ground state depletion is low. The calculation of the ground state depletion follows the same procedure that was used for the one-component, weakly interacting Bose gas in chapter \ref{sec:Weaklyinteracting}. We start with
\begin{equation}
    \frac{N-N_0}{N} = \frac{1}{N}\left.\sum_{\boldsymbol{k}}\right.^{'} \sum_\alpha \langle A_{\boldsymbol{k}}^{\alpha\dagger} A_{\boldsymbol{k}}^\alpha \rangle.
\end{equation}
The mean value is then transformed to the diagonal basis, in which we can use Bose-Einstein statistics. We focus on zero temperature, such that in the end all the mean values of operators will vanish. We are left with the terms analogous to $|v_{\boldsymbol{k}}|^2$ in \eqref{eq:GSdepnoSoc}, which originated from a commutator. The interesting terms are then the ones relating $A_{\boldsymbol{k}}^\alpha$ to $B_{\boldsymbol{k},\sigma}^\dagger$. The absolute squares of these coefficients are the analogues of $|v_{\boldsymbol{k}}|^2$. These can be obtained from the numerically constructed transformation matrices $T_{\boldsymbol{k}}$, using that $\boldsymbol{A}_{\boldsymbol{k}} = JT_{\boldsymbol{k}}J\boldsymbol{B}_{\boldsymbol{k}}$. This was checked numerically in the phases PZ, NZ, PW, SW and LW, and the ground state depletion was always less than $1\%$ when $U/t = 1/10$. This confirms the validity of the mean field theory performed in chapter \ref{chap:MFT}. Additionally, the ground state depletion became worse at higher $U/t$, as expected. For instance it was $5\%$ in the PW phase with $U=t$ and $\lambda_R = t$. It was also found that making $\lambda_R \gg t$ allows for stronger interactions while keeping the ground state depletion low, which can be viewed in conjunction with a discussion in \cite{Toniolo}.

\section{Discussion}
Most results have been presented together with a discussion. This section is devoted to general discussion of the overall results.

In this thesis we allowed for a complex phase factor in the term
\begin{equation}
    A_{\boldsymbol{k}_{0i}}^\alpha \to \sqrt{N_{\boldsymbol{k}_{0i}}^\alpha}e^{-i\theta_{\boldsymbol{k}_{0i}}^\alpha}.
\end{equation}
The angles $\theta_{\boldsymbol{k}_{0i}}^\alpha$ were shown to be arbitrary in the zero momentum phases PZ and NZ. However, in the SOC induced nonzero condensate momentum phases PW, SW and LW the angles proved to be important. In particular, satisfying
\begin{equation}
    \gamma_{\boldsymbol{k}_{0i}} + \theta_i^\downarrow -\theta_i^\uparrow = \pi
\end{equation}
was important. For the PW phase, this was vital both for the stability and the nonzero critical superfluid velocity. Similarly, if all the angles were set to $0$ by assumption, one would have found that the SW and LW phases were unstable. An interesting point is that the angles were set to zero by assumption in \cite{Toniolo}. It appears the importance of the angles is less pronounced in the helicity basis. Since the phase of the SOC term, $\gamma_{\boldsymbol{k}}$, is involved in the transformation to the helicity basis, it appears the information carried by $\theta_1^\uparrow$ and $\theta_1^\downarrow$ in the spin basis is contained in $\gamma_{\boldsymbol{k}_{01}}$ in the helicity basis. 

The obtained transition line of $\alpha=1$ between the PW and SW phases is the same as that reported in \cite{wangPWSWTransition} by numerically obtaining the wave function that minimizes the Gross-Pitaevskii energy. The states PW and SW were also reported as the only possible states. As further elaborated in \cite{zhaiSWRashbaRev} the wave function in the PW phase gives a uniform total density. In the SW phase both components have a modulated density showing a periodic striped structure. When the density of one component is highest, the density of the other component is at its lowest. Thus, the overlap of the two components is minimized, explaining why the SW phase is energetically favorable when $\alpha>1$, i.e. $U^{\uparrow\downarrow} = U^{\downarrow\uparrow} > U^{\uparrow\uparrow} = U^{\downarrow\downarrow}$. Similar arguments are given in \cite{SOCOLRev}. There, it is also argued that the LW phase, which does not have a uniform density, is less effective than the SW phase at minimizing the overlap of the  two pseudospin components. These phase diagram results have however not taken the elementary excitations into account. This thesis has done so, and confirms the results reported in \cite{wangPWSWTransition, SOCOLRev}.

\cleardoublepage

%% file: chapter6.tex

\chapter{Conclusion and Outlook} \label{chap:ConcOutlook}

An analytic framework for a theoretical treatment of a two-component, weakly interacting, spin-orbit coupled Bose gas bound to a Bravais lattice has been developed. This was largely based on the framework developed by Janssønn \cite{master}, though it proved convenient to adapt it to the canonical ensemble, in which the total number of particles is an input parameter, rather than the grand canonical ensemble employed by \cite{master} where the number of condensate particles are controlled by the chemical potential. Using mean field theory, the Hamiltonian was presented on a form that was at most quadratic in excitation operators, paving the way for exploration of the quasiparticle excitation spectrum using the BV diagonalization procedure. 

We specialized to Rashba SOC and a 2D square optical lattice, and the framework was subsequently applied to a zero-momentum phase with no SOC, called the NZ phase. The results for the excitation spectrum and the critical superfluid velocity were in accordance with \cite{LS} who had previously studied the same phase. Further it was used to describe a SOC induced phase with a single nonzero condensate momentum, which had previously been studied in \cite{Toniolo} and was named the plane wave (PW) phase. The results regarding excitation spectrum and critical superfluid velocity were found to be a special case of the results reported in \cite{Toniolo}, and both were studied in greater detail for the case of no Zeeman field. There is therefore ample reason to assume the framework developed in this thesis is valid. We also considered a polarized zero-momentum phase that can exist in the presence of SOC. This bears similarities to the results in \cite{Toniolo} with strong Zeeman splitting.

Another SOC induced state called the stripe wave (SW) phase was also studied, and like the PW phase it is a bosonic analogue of Fulde-Ferrell-Larkin-Ovchinnikov states in superconductors \cite{FF, LO, FFLO}. For the case of Rashba SOC and the presence of an optical lattice, the author has not found its excitation spectrum reported in the literature. For the continuum BEC with Raman induced SOC along one direction, its excitation spectrum was reported in \cite{li2013superstripes}, wherein a nonzero critical superfluid velocity was reported. Using the helicity approximation performed in \cite{Toniolo} to the PW phase, we also found a nonzero, anisotropic critical superfluid velocity in the SW phase. Unlike the PW phase, the results in the helicity approximation and using the original spin basis gave different results regarding superfluidity in the SW phase. The reason for this remains unclear to the author. 

Furthermore, we studied the excitation spectrum of one of the phases widely believed not to exist \cite{SOCOLRev,wangPWSWTransition, zhaiSWRashbaRev}, namely the LW phase where four nonzero condensate momenta are occupied. Most phase diagrams reported in the literature have been made by neglecting excitations. In this thesis the effects of the excitations have been considered, and the free energy at zero temperature, i.e. the ground state energy, has been used to determine the phase diagram. 

In this process, we have treated terms in the Hamiltonian that are linear in excitation operators, which have not previously been explored in the literature \cite{master}. This was performed by transforming the linear part to the basis in which the quadratic part of the Hamiltonian was diagonal. In this way, the linear terms could be removed by completing squares, and they gave a shift of the free energy. The end result was in agreement with the general results reported in \cite{SOCOLRev,wangPWSWTransition, zhaiSWRashbaRev}. When intercomponent interactions are weaker than intracomponent interactions, the PW, Fulde-Ferrell analogous phase is preferred, while the SW, Larkin-Ovchinnikov analogous phase is preferred when intercomponent interactions are strongest. 

A natural generalization of the results in this thesis would be to introduce an external Zeeman field as was done in \cite{Toniolo}. The effect of the Zeeman field will be to introduce pseudospin imbalance, and, for strong enough Zeeman field, the nonzero condensate momenta will all converge to zero. Regarding the SW phase, treating a Zeeman field may also make the origin of the distinction in critical superfluid velocity between the two methods used to obtain the excitation spectrum more clear. The reason being that the transformation to the helicity basis does not suffer discontinuities in the presence of a Zeeman field \cite{Toniolo}. It may also be of interest to expand the treatment to other lattice configurations or other SOC schemes. Considering a linear combination of Dresselhaus and Rashba SOC may reveal interesting physics, and possibly make the theoretical results more in accordance with experimentally realizable SOC schemes. 

Another interesting quantity in a two-component, superfluid BEC is the superfluid drag density. This was calculated in \cite{LS} by a method that requires Galilean invariance based on \cite{fildrag}. A method that does not rely on Galilean invariance presented in \cite{weichman} was used by Hartmann \cite{Hartmann}, whose calculations were unsuccessful in the presence of SOC. It may be possible to revisit this by drawing inspiration from this thesis, together with Hartmann's thesis \cite{Hartmann}.




\cleardoublepage

%% file: references.tex
\pagestyle{fancy}
\fancyhf{}
\renewcommand{\chaptermark}[1]{\markboth{\chaptername\ \thechapter.\ #1}{}}
\renewcommand{\sectionmark}[1]{\markright{\thesection\ #1}}
\renewcommand{\headrulewidth}{0.1ex}
\renewcommand{\footrulewidth}{0.1ex}
\fancyfoot[LE,RO]{\thepage}
\fancypagestyle{plain}{\fancyhf{}\fancyfoot[LE,RO]{\thepage}\renewcommand{\headrulewidth}{0ex}}

\phantomsection 
\addcontentsline{toc}{chapter}{Bibliography} 
\bibliographystyle{IEEEtran}
\bibliography{main.bbl}

\cleardoublepage

%% file: appendix.tex



\begin{appendix}

\chapter{Further Details in the SW Phase} \label{app:SW}
\section{The Special Momenta} 
We start by looking at the special momentum $\boldsymbol{k}=\boldsymbol{0}$. We name this part of the Hamiltonian $H_2(\boldsymbol{0})$, and write it as
\begin{equation}
    H_2(\boldsymbol{0}) = \frac12 \boldsymbol{A}_0^{\dagger}M_0\boldsymbol{A}_0,
\end{equation}
where we define a new basis
\begin{align}
    \begin{split}
        \boldsymbol{A}_0^{\dagger} = (&A_{\boldsymbol{0}}^{\uparrow\dagger}, A_{2\boldsymbol{k}_{01}}^{\uparrow\dagger}, A_{-2\boldsymbol{k}_{01}}^{\uparrow\dagger}, A_{\boldsymbol{0}}^{\downarrow\dagger}, A_{2\boldsymbol{k}_{01}}^{\downarrow\dagger}, A_{-2\boldsymbol{k}_{01}}^{\downarrow\dagger}, \\
        &A_{\boldsymbol{0}}^{\uparrow}, A_{2\boldsymbol{k}_{01}}^{\uparrow}, A_{-2\boldsymbol{k}_{01}}^{\uparrow}, A_{\boldsymbol{0}}^{\downarrow}, A_{2\boldsymbol{k}_{01}}^{\downarrow}, A_{-2\boldsymbol{k}_{01}}^{\downarrow}),
    \end{split}
\end{align}
wherein no operators are repeated. $M_0$ has the form
\begin{gather}
    M_{0} = 
    \begin{pmatrix}
    M_{0, 1} & M_{0,2} \\
    M_{0, 2}^* & M_{0, 1}^* \\
    \end{pmatrix},
\end{gather}
with
\setcounter{MaxMatrixCols}{6}
\begin{gather*}
    M_{0, 1} = 
    \begin{pmatrix}
    M_{1,1}(\boldsymbol{0}) & M_{1,3} & M_{1,3}^{*} & M_{1,7}(\boldsymbol{0}) & M_{1,9} & M_{1,11} \\
    M_{1,3}^* & 0 & 0 & M_{1,11} & 0 & 0 \\
    M_{1,3} & 0 & 0 & M_{1,9} & 0 & 0  \\
    M_{1,7}^*(\boldsymbol{0}) & M_{1,11}^* & M_{1,9}^* & M_{1,1}(\boldsymbol{0}) & M_{7,9} & M_{7,9}^* \\
    M_{1,9}^{*} & 0 & 0 & M_{7,9}^* & 0 & 0 \\
    M_{1,11}^* & 0 & 0 & M_{7,9} & 0 & 0 \\
    \end{pmatrix}
\end{gather*}
and
\begin{gather*}
    M_{0, 2}^* = 
    \begin{pmatrix}
    M_{13,2} & M_{13,4} & M_{13,6} & M_{13,8} & M_{13,10} & M_{13,12}\\
    M_{13,4} & 0 & 0 & M_{13,10} & 0 & 0 \\
    M_{13,6} & 0 & 0 & M_{13,12} & 0 & 0 \\
    M_{13,8} & M_{13,10} & M_{13,12} & M_{19,8} & M_{19,10} & M_{19,12}\\
    M_{13,10} & 0 & 0 & M_{19,10} & 0 & 0 \\
    M_{13,12} & 0 & 0 & M_{19,12} & 0 & 0 \\
    \end{pmatrix}.
\end{gather*}
Numerically, we find 8 nonzero eigenvalues $\lambda = \pm \omega_{0,i}$, $i=1,2,3,4$ while 4 eigenvalues are within numerical accuracy $0$. Numerical investigations of the transformation matrix $T_0$ show that the two smallest nonzero eigenvalues appear with a minus sign in the diagonalized version. Defining $\Delta\omega_{0,1} = \omega_{0,1}+\Omega_0$, $\Delta\omega_{0,2} = \omega_{0,2}+\Omega_0$, $\Delta\omega_{0,5} = \Omega_0-\omega_{0,4}$ and $\Delta\omega_{0,6} = \Omega_0-\omega_{0,3}$, $H_2(\boldsymbol{0})$ can be written
\begin{align}
    \begin{split}
        H_2(\boldsymbol{0}) = &-\Omega_0 \sum_{\sigma=1}^6 \left(B_{\boldsymbol{0},\sigma}^{\dagger}B_{\boldsymbol{0},\sigma}+\frac12\right) + \Delta\omega_{0, 1}\left(B_{\boldsymbol{0},1}^{\dagger}B_{\boldsymbol{0},1}+\frac12\right) \\
        &+ \Delta\omega_{0, 2}\left(B_{\boldsymbol{0},2}^{\dagger}B_{\boldsymbol{0},2}+\frac12\right) +\Omega_0\left(B_{\boldsymbol{0},3}^{\dagger}B_{\boldsymbol{0},3} +B_{\boldsymbol{0},4}^{\dagger}B_{\boldsymbol{0},4} +1\right)  \\
        &+\Delta\omega_{0, 5}\left(B_{\boldsymbol{0},5}^{\dagger}B_{\boldsymbol{0},5} +\frac12\right) + \Delta\omega_{0, 6}\left(B_{\boldsymbol{0},6}^{\dagger}B_{\boldsymbol{0},6} +\frac12\right).
    \end{split}
\end{align}
Within numerical accuracy, $\Delta\omega_{0,1} = \Delta\Omega_1(\boldsymbol{0}) = \Delta\Omega_2(\boldsymbol{0})$, $\Delta\omega_{0, 2} = \Delta\Omega_3(\boldsymbol{0}) = \Delta\Omega_4(\boldsymbol{0})$, $\Delta\omega_{0,5} = \Delta\Omega_9(\boldsymbol{0}) = \Delta\Omega_{10}(\boldsymbol{0})$ and $\Delta\omega_{0,6} = \Delta\Omega_{11}(\boldsymbol{0}) = \Delta\Omega_{12}(\boldsymbol{0})$. Thus $H_2(\boldsymbol{0})$ can be incorporated in $H'_2$ if we drop the $\boldsymbol{k}\neq \boldsymbol{0}$ limitation in the sums in the definition of $N_q$ and $H'_2$. In $H'_2$ the operators are not yet defined at $\boldsymbol{k}=\boldsymbol{0}$. One may simply define them in such a way that this inclusion of $\boldsymbol{k} = \boldsymbol{0}$ makes sense.

We use commutators and treat $H_2(2\boldsymbol{k}_{01})$ and $H_2(-2\boldsymbol{k}_{01})$ simultaneously by having first made $-\boldsymbol{k}$-term explicit in the sum in $H_2$. A $20\cross20$ matrix $M_{2k_0}$ is found. If we define
\begin{footnotesize}
\begin{align}
    \begin{split}
        \boldsymbol{A}_{2k_0}^{\dagger} = (&A_{2\boldsymbol{k}_{01}}^{\uparrow\dagger}, A_{-2\boldsymbol{k}_{01}}^{\uparrow\dagger}, A_{4\boldsymbol{k}_{01}}^{\uparrow\dagger}, A_{\boldsymbol{0}}^{\uparrow\dagger}, A_{-4\boldsymbol{k}_{01}}^{\uparrow\dagger}, A_{2\boldsymbol{k}_{01}}^{\downarrow\dagger}, A_{-2\boldsymbol{k}_{01}}^{\downarrow\dagger}, A_{4\boldsymbol{k}_{01}}^{\downarrow\dagger}, A_{\boldsymbol{0}}^{\downarrow\dagger}, A_{-4\boldsymbol{k}_{01}}^{\downarrow\dagger}, \\
        &A_{2\boldsymbol{k}_{01}}^{\uparrow}, A_{-2\boldsymbol{k}_{01}}^{\uparrow}, A_{4\boldsymbol{k}_{01}}^{\uparrow}, A_{\boldsymbol{0}}^{\uparrow}, A_{-4\boldsymbol{k}_{01}}^{\uparrow}, A_{2\boldsymbol{k}_{01}}^{\downarrow}, A_{-2\boldsymbol{k}_{01}}^{\downarrow}, A_{4\boldsymbol{k}_{01}}^{\downarrow}, A_{\boldsymbol{0}}^{\downarrow}, A_{-4\boldsymbol{k}_{01}}^{\downarrow}),
    \end{split}
\end{align}
\end{footnotesize}
then $H_2(-2\boldsymbol{k}_{01}) = H_2(2\boldsymbol{k}_{01}) = (\boldsymbol{A}_{2k_0}^{\dagger}M_{2k_0}\boldsymbol{A}_{2k_0})/4$. The origin of our problems is that elements $4+6i$ and $5+6i$, $i=0,1,2,3$, of $\boldsymbol{A}_{2\boldsymbol{k}_{01}}$ are equal. This is why $M_{2k_0}$ has a size of 4 columns and rows less than $M_{\boldsymbol{k}}$. We can use this to construct $M_{2k_0}$ from $M_{\boldsymbol{k}}$. One simply combines rows and columns that corresponds to the elements that are equal in the original basis. The first steps are to add columns $5+6i$ to columns $4+6i$ for all $i=0,1,2,3$. Then one adds rows $5+6i$ to rows $4+6i$ for all $i=0,1,2,3$. Finally, rows and columns $5+6i$ for all $i=0,1,2,3$ are removed to obtain the $20\cross 20$ matrix $M_{2k_0}$. As check, we obtained $M_{2k_0}$ by writing out the Hamiltonian at $\boldsymbol{k} = 2\boldsymbol{k}_{01}$ confirming the above procedure is valid. A similar procedure could also have been used to obtain the matrix $M_0$. For brevity, we do not give an explicit expression for $M_{2k_0}$ since it is obtainable from $M_{\boldsymbol{k}}$.

We obtain 16 nonzero eigenvalues of $M_{2k_0}J$ that can be written $\lambda = \pm\omega_{2k_0, i}$, $i=1, 2, \dots 8$, while four eigenvalues are $0$. Numerical investigations of the transformation matrix $T_{2k_0}$ show that the four smallest nonzero eigenvalues appear with a minus sign in the diagonalized version. On diagonal form we write $H_2(2\boldsymbol{k}_{01}) = H_2(-2\boldsymbol{k}_{01}) =$
\begin{align}
    \begin{split}
    \frac{1}{2}&\Big\{\sum_{\sigma = 1}^{4}\omega_{2k_0, \sigma}\left( B_{2\boldsymbol{k}_{01},\sigma^{'}}^{\dagger}B_{2\boldsymbol{k}_{01}, \sigma^{'}}+\frac12\right) \\
    & -\sum_{\sigma = 5}^{8}\omega_{2k_0, \sigma}\left( B_{2\boldsymbol{k}_{01},\sigma^{'}}^{\dagger}B_{2\boldsymbol{k}_{01}, \sigma^{'}}+\frac12\right) \\
    & + \sum_{\sigma = 9}^{10} 0 \left( B_{2\boldsymbol{k}_{01},\sigma^{'}}^{\dagger}B_{2\boldsymbol{k}_{01}, \sigma^{'}}+\frac12\right)\Big\}.
    \end{split}
\end{align}
We shift the zero of energy by $\Omega_0$ for these eigenvalues as well. Defining $\Delta\omega_{2k_{01},i} = \Omega_0+\omega_{2k_0, i}$ for $i=1,2,3,4$, $\Delta\omega_{2k_{01},i} = \Omega_0$ for $i=5,6$, $\Delta\omega_{2k_{01},i} = \Omega_0-\omega_{2k_0, i'}$  for $i=7,8,9,10$ and $i'=8,7,6,5$ and renumbering the operators we arrive at $H_2(2\boldsymbol{k}_{01}) + H_2(-2\boldsymbol{k}_{01}) = 2H_2(2\boldsymbol{k}_{01})$,
\begin{align}
    \begin{split}
    \label{eq:SWspecial2k01}
        2H_2(2\boldsymbol{k}_{01}) &= -\Omega_0 N_{q,2k_0}+\sum_{\sigma=1}^{10} \Delta\omega_{2k_0, \sigma} \left( B_{2\boldsymbol{k}_{01},\sigma}^{\dagger}B_{2\boldsymbol{k}_{01}, \sigma} +\frac12  \right).
    \end{split}
\end{align}
To simplify the expression, we defined 
$$N_{q,2k_0} \equiv \sum_{\sigma=1}^{10} \left(B_{2\boldsymbol{k}_{01},\sigma}^{\dagger}B_{2\boldsymbol{k}_{01}, \sigma}+\frac12\right).$$
The energies $\Delta\omega_{2k_{01},\sigma}$ do not agree completely with the energy spectrum $\Delta\Omega_\sigma(\boldsymbol{k})$ at $\pm2\boldsymbol{k}_{01}$ and therefore we will keep the treatment of $\boldsymbol{k}=\pm2\boldsymbol{k}_{01}$ separate. 

Moving on to the special treatment of $\pm3\boldsymbol{k}_{01}$ we note that it is not the occurrence of equal operators in the basis that is our problem. Rather it is the occurrence of condensate operators which have already been treated as complex numbers. Terms with these condensate operators are excluded from the sum in $H_2$. The basis \eqref{eq:SWbra} at $3\boldsymbol{k}_{01}$ contains condensate operators in elements $4+6i$ and $5+6i$ for $i=0,1,2,3$. Removing these, we define a new basis
\begin{align}
    \begin{split}
        \boldsymbol{A}_{3k_0}^{\dagger} = (&A_{3\boldsymbol{k}_{01}}^{\uparrow\dagger}, A_{-3\boldsymbol{k}_{01}}^{\uparrow\dagger}, A_{5\boldsymbol{k}_{01}}^{\uparrow\dagger}, A_{-5\boldsymbol{k}_{01}}^{\uparrow\dagger}, A_{3\boldsymbol{k}_{01}}^{\downarrow\dagger}, A_{-3\boldsymbol{k}_{01}}^{\downarrow\dagger}, A_{5\boldsymbol{k}_{01}}^{\downarrow\dagger}, A_{-5\boldsymbol{k}_{01}}^{\downarrow\dagger}, \\
        &A_{3\boldsymbol{k}_{01}}^{\uparrow}, A_{-3\boldsymbol{k}_{01}}^{\uparrow}, A_{5\boldsymbol{k}_{01}}^{\uparrow}, A_{-5\boldsymbol{k}_{01}}^{\uparrow}, A_{3\boldsymbol{k}_{01}}^{\downarrow}, A_{-3\boldsymbol{k}_{01}}^{\downarrow}, A_{5\boldsymbol{k}_{01}}^{\downarrow}, A_{-5\boldsymbol{k}_{01}}^{\downarrow}).
    \end{split}
\end{align}
Then, we can write
\begin{equation}
    H_2(3\boldsymbol{k}_{01}) = H_2(-3\boldsymbol{k}_{01}) = \frac{1}{4} \boldsymbol{A}_{3k_0}^{\dagger} M_{3k_0}  \boldsymbol{A}_{3k_0}.
\end{equation}
The $16\cross 16$ matrix $M_{3k_0}$ can be obtained from the $24 \cross 24$ matrix $M_{\boldsymbol{k}}$ in \eqref{eq:SWMat} at $3\boldsymbol{k}_{01}$ by removing rows and columns $4+6i$ and $5+6i$, $i=0,1,2,3$, as the entries in these rows and columns correspond to the terms we should remove. The eigenvalues of $M_{3k_0}J$ found numerically can be written $\pm\omega_{3k_0,i}$ for $i=1,2, \dots 8$, i.e. $8$ positive and $8$ negative eigenvalues. The four smallest positive eigenvalues have eigenvectors with negative BV norm and thus enter the diagonalized form with a negative sign. We find $H_2(3\boldsymbol{k}_{01}) = H_2(-3\boldsymbol{k}_{01}) =$
\begin{align}
    \begin{split}
    \frac{1}{2}&\Big\{\sum_{\sigma = 1}^{4}\omega_{3k_0, \sigma}\left( B_{3\boldsymbol{k}_{01},\sigma^{'}}^{\dagger}B_{3\boldsymbol{k}_{01}, \sigma^{'}}+\frac12\right) \\
    & -\sum_{\sigma = 5}^{8}\omega_{3k_0, \sigma}\left( B_{3\boldsymbol{k}_{01},\sigma^{'}}^{\dagger}B_{3\boldsymbol{k}_{01}, \sigma^{'}}+\frac12\right) \Big\}.
    \end{split}
\end{align}
We shift the zero of energy by $\Omega_0$ for these eigenvalues as well. Defining $\Delta\omega_{3k_0,i} = \Omega_0+\omega_{3k_0, i}$ for $i = 1,2,3,4$, $\Delta\omega_{3k_0,i} = \Omega_0-\omega_{3k_0, i'}$ for $i=5,6,7,8$ and $i' = 8,7,6,5$ and renumbering the operators we arrive at $H_2(3\boldsymbol{k}_{01}) + H_2(-3\boldsymbol{k}_{01}) = 2H_2(3\boldsymbol{k}_{01})$,
\begin{align}
    \begin{split}
    \label{eq:SWspecial3k01}
        2H_2(3\boldsymbol{k}_{01}) &= -\Omega_0 N_{q,3k_0} +\sum_{\sigma=1}^8 \Delta\omega_{3k_0, \sigma} \left( B_{3\boldsymbol{k}_{01},\sigma}^{\dagger}B_{3\boldsymbol{k}_{01}, \sigma} +\frac12  \right).
    \end{split}
\end{align}
To simplify the expression, we defined
$$N_{q,3k_0} \equiv \sum_{\sigma=1}^{8} \left(B_{3\boldsymbol{k}_{01},\sigma}^{\dagger}B_{3\boldsymbol{k}_{01}, \sigma}+\frac12 \right).$$
The energies $\Delta\omega_{3k_0,\sigma}$ do not agree completely with the energy spectrum $\Delta\Omega_\sigma(\boldsymbol{k})$ at $\pm3\boldsymbol{k}_{01}$ and therefore we will keep the treatment of $\boldsymbol{k}=\pm3\boldsymbol{k}_{01}$ separate.

\section{Differences Between Spin and Helicity Basis Results} \label{app:SWdifference}
We should first mention that the differences between using the spin basis and using the helicity basis are not due to a fundamental difference between the bases. In fact, it can be shown that if we used the full helicity basis \eqref{eq:Helicitybasis} we would obtain the same eigenvalues as in the original spin basis. The two methods become different once we neglect the upper helicity band \eqref{eq:helicitybands}, $\lambda_{\boldsymbol{k}}^+$, an approximation for which there is no equivalent in the spin basis. 


In the PW phase, the two methods gave the same results regarding the critical superfluid velocity. The biggest difference between the PW and SW phases, is that the SW phase contains interactions that mix different condensate momenta. We introduce coefficients $\Gamma_\pm^{\alpha\beta} = \{0,1\}$ to terms like $\Gamma_\pm^{\alpha\beta} e^{i(\theta_1^\alpha \pm \theta_3^\beta)}$ and its H.c. as these originate from such interactions. The objective is that if we can track down which terms give rise to the linear behavior in the helicity approximation, we may understand the origin of the distinction between the results in the two approaches. Interactions that mix different condensate momenta are also present in $H_0^{''}$. However, we choose to let the condensate remain unchanged and focus on the excitations, i.e. use the same $H_0^{''}$ and hence the same $M_{1,1}(\boldsymbol{k})$. The matrix elements that are changed become 
\begin{align}
    \begin{split}
        \label{eq:SWelemGamma}
        M_{1,3}(\boldsymbol{k}) &= \frac{U_s}{4}\left(2\Gamma_-^{\uparrow\uparrow}e^{i(\theta_1^\uparrow-\theta_3^\uparrow)}+\alpha\Gamma_-^{\downarrow\downarrow} e^{i(\theta_1^\downarrow-\theta_3^\downarrow)}\right), \\
        M_{1,9} &= \frac{U_s\alpha}{4}\Gamma_-^{\downarrow\uparrow}e^{i(\theta_1^\downarrow-\theta_3^\uparrow)}, \\
        M_{1,11} &=  \frac{U_s\alpha}{4}\Gamma_-^{\uparrow\downarrow}e^{-i(\theta_1^\uparrow-\theta_3^\downarrow)}, \\ 
        M_{7,9} &= \frac{U_s}{4}\left(2\Gamma_-^{\downarrow\downarrow}e^{i(\theta_1^\downarrow-\theta_3^\downarrow)}+\alpha \Gamma_-^{\uparrow\uparrow}e^{i(\theta_1^\uparrow-\theta_3^\uparrow)}\right) , \\ 
        M_{13,2} &= U_s \Gamma_+^{\uparrow\uparrow}e^{i(\theta_1^\uparrow+\theta_3^\uparrow)}, \\
        M_{13,8} &= \frac{U_s\alpha}{2}\left(\Gamma_+^{\downarrow\uparrow}e^{i(\theta_1^\downarrow+\theta_3^\uparrow)}+\Gamma_+^{\uparrow\downarrow}e^{i(\theta_1^\uparrow+\theta_3^\downarrow)}\right), \\
        M_{19,8} &= U_s \Gamma_+^{\downarrow\downarrow}e^{i(\theta_1^\downarrow+\theta_3^\downarrow)}. 
    \end{split}
\end{align}
There are in total $2^8 = 256$ possible choices for the set $\Gamma_\sigma^{\alpha\beta}$. All possibilities have not been explored, however the numerous choices that were gave no significant insights. All versions where only one or only two $\Gamma_\sigma^{\alpha\beta} = 1$ were attempted. Additionally all versions were only one, two or three $\Gamma_\sigma^{\alpha\beta} = 0$ were attempted, along with several other cases.  For instance, with all $\Gamma_\sigma^{\alpha\beta} = 1$ except one, the result was mostly that both approaches gave real eigenvalues with non-linear behavior close to the minimum, or that both approaches gave complex eigenvalues. Therefore, it appears all the terms that mix condensate momenta are needed to obtain the distinction between the two approaches.


It was also attempted to change $H_0^{''}$, and hence change $M_{1,1}(\boldsymbol{k})$, by taking into account the interactions in the condensate that mix condensate momenta. This also gave no significant insights. Finally, it was attempted to remove all interactions involving $\boldsymbol{k}_{03}$ from the excitations, i.e. all $\Gamma_\sigma^{\alpha\beta}=0$ in addition to $M_{13,6} = M_{13,12} = M_{19,12} = 0$. In this case, both approaches yielded a spectrum with two phonon minima at $\pm\boldsymbol{k}_{01}$. Hence, it seems it is the presence of two condensate momenta that removes the linearity in the spin basis. Nevertheless, it remains unclear why the same is not true in the helicity approximation.


\input{9LWAppLong}

\end{appendix}

%% file: 9LWAppLong.tex

\chapter{LW Phase Calculations} \label{app:LW}
The LW phase is such that $\boldsymbol{k}_{01}=(k_0, k_0)$, $\boldsymbol{k}_{02}=(-k_0, k_0)$, $\boldsymbol{k}_{03} = -\boldsymbol{k}_{01}$ and $\boldsymbol{k}_{04} = -\boldsymbol{k}_{02}$ are occupied condensate momenta. We assume that $N_{\boldsymbol{k}_{0i}}^\alpha = N_0^\alpha/4$ i.e. a balanced condensate in terms of the momenta. The expression for $H_0^{''}$ derived from \eqref{eq:H0} is then
\begin{align}
    \begin{split}
        H_0^{''} &= (\epsilon_{\boldsymbol{k}_{01}}+T)N_0 + \frac{1}{2}\sqrt{N_0^\uparrow N_0^\downarrow}\abs{s_{\boldsymbol{k}_{01}}}\sum_{i=1}^4 \cos(\gamma_{\boldsymbol{k}_{0i}}+\Delta\theta_i)\\
        +&\frac{U}{8N_s}\Bigg((N_0^\uparrow)^2\bigg(7+2\cos(\theta_1^\uparrow+\theta_3^\uparrow-\theta_2^\uparrow-\theta_4^\uparrow)\bigg)\\
        &\mbox{\qquad}+ (N_0^\downarrow)^2\bigg(7+2\cos(\theta_1^\downarrow+\theta_3^\downarrow-\theta_2^\downarrow-\theta_4^\downarrow)\bigg) \\
        &+N_0^\uparrow N_0^\downarrow \alpha \bigg[8+\cos(\Delta\theta_1-\Delta\theta_2) + \cos(\Delta\theta_1-\Delta\theta_3) + \cos(\Delta\theta_1-\Delta\theta_4) \\
        &\mbox{\qquad\qquad}+ \cos(\Delta\theta_2-\Delta\theta_3) + \cos(\Delta\theta_2-\Delta\theta_4) + \cos(\Delta\theta_3-\Delta\theta_4)\\
        &\mbox{\qquad\qquad} + \cos(\theta_1^\uparrow+\theta_3^\downarrow-\theta_2^\downarrow-\theta_4^\uparrow) + \cos(\theta_1^\downarrow+\theta_3^\uparrow-\theta_2^\uparrow-\theta_4^\downarrow) \\
        &\mbox{\qquad\qquad} + \cos(\theta_1^\uparrow+\theta_3^\downarrow-\theta_2^\uparrow-\theta_4^\downarrow) + \cos(\theta_1^\downarrow+\theta_3^\uparrow-\theta_2^\downarrow-\theta_4^\uparrow) \bigg]\Bigg).
    \end{split}
\end{align}
We insert \eqref{eq:NN0excitup} and \eqref{eq:NN0excitdown} and obtain
\begin{align}
    \begin{split}
        &H_0^{''} = H_0 -(\epsilon_{\boldsymbol{k}_{01}}+T)\left.\sum_{\boldsymbol{k}}\right.^{'}\sum_\alpha A_{\boldsymbol{k}}^{\alpha\dagger}A_{\boldsymbol{k}}^{\alpha} \\
        &-\frac{\abs{s_{\boldsymbol{k}_{01}}}}{4}\left(\sqrt{\frac{N^\uparrow}{N^\downarrow}}\left.\sum_{\boldsymbol{k}}\right.^{'} A_{\boldsymbol{k}}^{\downarrow\dagger}A_{\boldsymbol{k}}^{\downarrow} + \sqrt{\frac{N^\downarrow}{N^\uparrow}} \left.\sum_{\boldsymbol{k}}\right.^{'} A_{\boldsymbol{k}}^{\uparrow\dagger}A_{\boldsymbol{k}}^{\uparrow}\right)\sum_{i=1}^4 \cos(\gamma_{\boldsymbol{k}_{0i}}+\Delta\theta_i) \\
        &-\frac{U}{8N_s}\Bigg(2N^\uparrow \bigg(7+2\cos(\theta_1^\uparrow+\theta_3^\uparrow-\theta_2^\uparrow-\theta_4^\uparrow)\bigg)\left.\sum_{\boldsymbol{k}}\right.^{'} A_{\boldsymbol{k}}^{\uparrow\dagger}A_{\boldsymbol{k}}^{\uparrow} \\
        &\mbox{\qquad}+ 2N^\downarrow \bigg(7+2\cos(\theta_1^\downarrow+\theta_3^\downarrow-\theta_2^\downarrow-\theta_4^\downarrow)\bigg)\left.\sum_{\boldsymbol{k}}\right.^{'} A_{\boldsymbol{k}}^{\downarrow\dagger}A_{\boldsymbol{k}}^{\downarrow} \\
        &\mbox{\qquad} +\alpha \bigg[8+\cos(\Delta\theta_1-\Delta\theta_2) + \cos(\Delta\theta_1-\Delta\theta_3) + \cos(\Delta\theta_1-\Delta\theta_4) \\
        &\mbox{\qquad\qquad}+ \cos(\Delta\theta_2-\Delta\theta_3) + \cos(\Delta\theta_2-\Delta\theta_4) + \cos(\Delta\theta_3-\Delta\theta_4)\\
        &\mbox{\qquad\qquad} + \cos(\theta_1^\uparrow+\theta_3^\downarrow-\theta_2^\downarrow-\theta_4^\uparrow) + \cos(\theta_1^\downarrow+\theta_3^\uparrow-\theta_2^\uparrow-\theta_4^\downarrow) \\
        &\mbox{\qquad\qquad} + \cos(\theta_1^\uparrow+\theta_3^\downarrow-\theta_2^\uparrow-\theta_4^\downarrow) + \cos(\theta_1^\downarrow+\theta_3^\uparrow-\theta_2^\downarrow-\theta_4^\uparrow) \bigg] \\
        &\mbox{\qquad\qquad\qquad\qquad} \cdot\left(N^\uparrow \left.\sum_{\boldsymbol{k}}\right.^{'} A_{\boldsymbol{k}}^{\downarrow\dagger}A_{\boldsymbol{k}}^{\downarrow}+N^\downarrow \left.\sum_{\boldsymbol{k}}\right.^{'} A_{\boldsymbol{k}}^{\uparrow\dagger}A_{\boldsymbol{k}}^{\uparrow}\right) \Bigg).
    \end{split}
\end{align}
Here we defined the new $H_0$ in terms of $N^\uparrow$ and $N^\downarrow$. From now on, we insert the choice $N^\uparrow = N^\downarrow = N/2$ for the input parameters. Then, $H_0 = H_0^{\textrm{LW}}$ given in \eqref{eq:H0LW}.
The rest of $H_0^{''}$ is moved to $H_2$ as it is quadratic in excitation operators. 


When setting up the phase diagram in figure \ref{fig:PDH0new}, we neglected elementary excitations, and found that the system will not enter the LW phase. The objective in this appendix is to include elementary excitations to see if their effects change this conclusion. The linear part of the Hamiltonian \eqref{eq:H1} now contains a multitude of terms, as the Kronecker delta renders $\boldsymbol{k}$ a non-condensate momentum for all choices with $i' \neq i,j$ except when $i \neq j$ and they take the values 1 and 3, or 2 and 4. We may replace $N_0^\alpha$ by $N^\alpha$ directly to the same order of approximation. With the choice $N^\uparrow = N^\downarrow = N/2$, $H_1$ becomes
\begin{align}
    \begin{split}
        H_1 = \frac{\sqrt{N}U_s}{8\sqrt{2}}\sum_{\alpha}\bigg( &c_{+1}^{\alpha} A_{3\boldsymbol{k}_{01}}^\alpha +  c_{-1}^{\alpha} A_{-3\boldsymbol{k}_{01}}^\alpha + c_{+2}^{\alpha} A_{3\boldsymbol{k}_{02}}^\alpha + c_{-2}^{\alpha} A_{-3\boldsymbol{k}_{02}}^\alpha \\
        &+c_{1-2}^{\alpha}A_{2\boldsymbol{k}_{01}-\boldsymbol{k}_{02}}^\alpha + c_{-1+2}^{\alpha}A_{-2\boldsymbol{k}_{01}+\boldsymbol{k}_{02}}^\alpha \\
        &+ c_{1+2}^{\alpha}A_{2\boldsymbol{k}_{01}+\boldsymbol{k}_{02}}^\alpha + c_{-1-2}^{\alpha}A_{-2\boldsymbol{k}_{01}-\boldsymbol{k}_{02}}^\alpha  \\
        &+c_{2-1}^{\alpha}A_{2\boldsymbol{k}_{02}-\boldsymbol{k}_{01}}^\alpha + c_{-2+1}^{\alpha}A_{-2\boldsymbol{k}_{02}+\boldsymbol{k}_{01}}^\alpha \\
        &+ c_{2+1}^{\alpha}A_{2\boldsymbol{k}_{02}+\boldsymbol{k}_{01}}^\alpha + c_{-2-1}^{\alpha}A_{-2\boldsymbol{k}_{02}-\boldsymbol{k}_{01}}^\alpha \bigg) + \textrm{H.c.},
    \end{split}
\end{align}
where we defined
\begin{align}
    \begin{split}
    \label{eq:LWH1coeff}
        c_{+1}^{\alpha} &= e^{i(2\theta_1^{\alpha}-\theta_3^{\alpha})}+\alpha e^{i(\theta_1^{\alpha} + \theta_1^{\Bar{\alpha}}-\theta_3^{\Bar{\alpha}})}, \\
        c_{-1}^{\alpha} &= e^{i(2\theta_3^{\alpha}-\theta_1^{\alpha})}+\alpha e^{i(\theta_3^{\alpha} + \theta_3^{\Bar{\alpha}}-\theta_1^{\Bar{\alpha}})}, \\
        c_{+2}^{\alpha} &= e^{i(2\theta_2^{\alpha}-\theta_4^{\alpha})}+\alpha e^{i(\theta_2^{\alpha} + \theta_2^{\Bar{\alpha}}-\theta_4^{\Bar{\alpha}})} ,\\
        c_{-2}^{\alpha} &= e^{i(2\theta_4^{\alpha}-\theta_2^{\alpha})}+\alpha e^{i(\theta_4^{\alpha} + \theta_4^{\Bar{\alpha}}-\theta_2^{\Bar{\alpha}})} ,\\
        c_{1-2}^{\alpha} &= e^{i(2\theta_1^{\alpha}-\theta_2^{\alpha})}+\alpha e^{i(\theta_1^{\alpha} + \theta_1^{\Bar{\alpha}}-\theta_2^{\Bar{\alpha}})}+e^{i(\theta_1^{\alpha}+\theta_4^{\alpha}-\theta_3^{\alpha})} \\
        & \mbox{\qquad} +\alpha e^{i(\theta_1^{\alpha} + \theta_4^{\Bar{\alpha}}-\theta_3^{\Bar{\alpha}})} +e^{i(\theta_4^{\alpha}+\theta_1^{\alpha}-\theta_3^{\alpha})}+\alpha e^{i(\theta_4^{\alpha} + \theta_1^{\Bar{\alpha}}-\theta_3^{\Bar{\alpha}})} ,\\
        c_{-1+2}^{\alpha} &= e^{i(2\theta_3^{\alpha}-\theta_4^{\alpha})}+\alpha e^{i(\theta_3^{\alpha} + \theta_3^{\Bar{\alpha}}-\theta_4^{\Bar{\alpha}})}+e^{i(\theta_2^{\alpha}+\theta_3^{\alpha}-\theta_1^{\alpha})} \\
        & \mbox{\qquad} +\alpha e^{i(\theta_2^{\alpha} + \theta_3^{\Bar{\alpha}}-\theta_1^{\Bar{\alpha}})}  +e^{i(\theta_3^{\alpha}+\theta_2^{\alpha}-\theta_1^{\alpha})}+\alpha e^{i(\theta_3^{\alpha} + \theta_2^{\Bar{\alpha}}-\theta_1^{\Bar{\alpha}})} ,\\
        c_{1+2}^{\alpha} &= e^{i(2\theta_1^{\alpha}-\theta_4^{\alpha})}+\alpha e^{i(\theta_1^{\alpha} + \theta_1^{\Bar{\alpha}}-\theta_4^{\Bar{\alpha}})}+e^{i(\theta_1^{\alpha}+\theta_2^{\alpha}-\theta_3^{\alpha})} \\
        & \mbox{\qquad} +\alpha e^{i(\theta_1^{\alpha} + \theta_2^{\Bar{\alpha}}-\theta_3^{\Bar{\alpha}})} +e^{i(\theta_2^{\alpha}+\theta_1^{\alpha}-\theta_3^{\alpha})}+\alpha e^{i(\theta_2^{\alpha} + \theta_1^{\Bar{\alpha}}-\theta_3^{\Bar{\alpha}})} ,\\
        c_{-1-2}^{\alpha} &= e^{i(2\theta_3^{\alpha}-\theta_2^{\alpha})}+\alpha e^{i(\theta_3^{\alpha} + \theta_3^{\Bar{\alpha}}-\theta_2^{\Bar{\alpha}})}+e^{i(\theta_3^{\alpha}+\theta_4^{\alpha}-\theta_1^{\alpha})}\\
        & \mbox{\qquad} +\alpha e^{i(\theta_3^{\alpha} + \theta_4^{\Bar{\alpha}}-\theta_1^{\Bar{\alpha}})} +e^{i(\theta_4^{\alpha}+\theta_3^{\alpha}-\theta_1^{\alpha})}+\alpha e^{i(\theta_4^{\alpha} + \theta_3^{\Bar{\alpha}}-\theta_1^{\Bar{\alpha}})} ,\\
        c_{2-1}^{\alpha} &= e^{i(2\theta_2^{\alpha}-\theta_1^{\alpha})}+\alpha e^{i(\theta_2^{\alpha} + \theta_2^{\Bar{\alpha}}-\theta_1^{\Bar{\alpha}})}+e^{i(\theta_2^{\alpha}+\theta_3^{\alpha}-\theta_4^{\alpha})}\\
        & \mbox{\qquad} +\alpha e^{i(\theta_2^{\alpha} + \theta_3^{\Bar{\alpha}}-\theta_4^{\Bar{\alpha}})}  +e^{i(\theta_3^{\alpha}+\theta_2^{\alpha}-\theta_4^{\alpha})}+\alpha e^{i(\theta_3^{\alpha} + \theta_2^{\Bar{\alpha}}-\theta_4^{\Bar{\alpha}})} ,\\
        c_{-2+1}^{\alpha} &= e^{i(2\theta_4^{\alpha}-\theta_3^{\alpha})}+\alpha e^{i(\theta_4^{\alpha} + \theta_4^{\Bar{\alpha}}-\theta_3^{\Bar{\alpha}})}+e^{i(\theta_1^{\alpha}+\theta_4^{\alpha}-\theta_2^{\alpha})} \\
        & \mbox{\qquad}  +\alpha e^{i(\theta_1^{\alpha} + \theta_4^{\Bar{\alpha}}-\theta_2^{\Bar{\alpha}})} +e^{i(\theta_4^{\alpha}+\theta_1^{\alpha}-\theta_2^{\alpha})}+\alpha e^{i(\theta_4^{\alpha} + \theta_1^{\Bar{\alpha}}-\theta_2^{\Bar{\alpha}})} ,\\
        c_{2+1}^{\alpha} &= e^{i(2\theta_2^{\alpha}-\theta_3^{\alpha})}+\alpha e^{i(\theta_2^{\alpha} + \theta_2^{\Bar{\alpha}}-\theta_3^{\Bar{\alpha}})}+e^{i(\theta_2^{\alpha}+\theta_1^{\alpha}-\theta_4^{\alpha})} \\
        & \mbox{\qquad} +\alpha e^{i(\theta_2^{\alpha} + \theta_1^{\Bar{\alpha}}-\theta_4^{\Bar{\alpha}})} +e^{i(\theta_1^{\alpha}+\theta_2^{\alpha}-\theta_4^{\alpha})}+\alpha e^{i(\theta_1^{\alpha} + \theta_2^{\Bar{\alpha}}-\theta_4^{\Bar{\alpha}})} ,\\
        c_{-2-1}^{\alpha} &= e^{i(2\theta_4^{\alpha}-\theta_1^{\alpha})}+\alpha e^{i(\theta_4^{\alpha} + \theta_4^{\Bar{\alpha}}-\theta_1^{\Bar{\alpha}})}+e^{i(\theta_3^{\alpha}+\theta_4^{\alpha}-\theta_2^{\alpha})} \\
        & \mbox{\qquad} +\alpha e^{i(\theta_3^{\alpha} + \theta_4^{\Bar{\alpha}}-\theta_2^{\Bar{\alpha}})} +e^{i(\theta_4^{\alpha}+\theta_3^{\alpha}-\theta_2^{\alpha})}+\alpha e^{i(\theta_4^{\alpha} + \theta_3^{\Bar{\alpha}}-\theta_2^{\Bar{\alpha}})} .
    \end{split}
\end{align}
Here, $\Bar{\alpha} = \downarrow$ if $\alpha = \uparrow$ and vice versa. We stress the difference between the superscript $\alpha$ representing psuedospin states, and the coefficient $\alpha = U^{\uparrow\downarrow}/U^{\uparrow\uparrow}$. The subscripts of the coefficients are related to the presence and sign of the momenta $\boldsymbol{k}_{01}, \boldsymbol{k}_{02}$ in the momentum indices of the corresponding operators. Our strategy for including $H_1$ will be similar to the SW phase. We will try to transform $H_1$ numerically to the new basis in which $H_2$ is diagonal, and then remove all linear terms by completing squares.

Next, we turn to $H_2$ in \eqref{eq:H2}. Writing out the sums over momentum indices and over $\boldsymbol{k}'$ we find that $\boldsymbol{k}'$ can take 18 separate values. Two of them are $\pm\boldsymbol{k}$, the rest we name $\boldsymbol{p}_i$ and $\boldsymbol{q}_i$ with $i = 1, \dots, 8$ and we define them in table \ref{tab:LWmomenta}. Notice that $\boldsymbol{q}_i(\boldsymbol{k}) = \boldsymbol{p}_i(-\boldsymbol{k})$.
\begin{table}[hb]
    \centering
    \caption{A set of momenta that appear as indices in the Hamiltonian.}
    \begin{tabular}{|c|c|c|}
        \hline
        $i$ & $\boldsymbol{p}_i$ & $\boldsymbol{q}_i$ \\
        \hline
        1 & $\boldsymbol{k}+2\boldsymbol{k}_{01}$ & $-\boldsymbol{k}+2\boldsymbol{k}_{01}$ \\
        \hline
        2 & $\boldsymbol{k}-2\boldsymbol{k}_{01}$ & $-\boldsymbol{k}-2\boldsymbol{k}_{01}$ \\
        \hline
        3 & $\boldsymbol{k}+2\boldsymbol{k}_{02}$ & $-\boldsymbol{k}+2\boldsymbol{k}_{02}$ \\
        \hline
        4 & $\boldsymbol{k}-2\boldsymbol{k}_{02}$ & $-\boldsymbol{k}-2\boldsymbol{k}_{02}$ \\
        \hline
        5 & $\boldsymbol{k}+\boldsymbol{k}_{01}+\boldsymbol{k}_{02}$ & $-\boldsymbol{k}+\boldsymbol{k}_{01}+\boldsymbol{k}_{02}$ \\
        \hline
        6 & $\boldsymbol{k}+\boldsymbol{k}_{01}-\boldsymbol{k}_{02}$ & $-\boldsymbol{k}+\boldsymbol{k}_{01}-\boldsymbol{k}_{02}$ \\
        \hline
        7 & $\boldsymbol{k}-\boldsymbol{k}_{01}+\boldsymbol{k}_{02}$ & $-\boldsymbol{k}-\boldsymbol{k}_{01}+\boldsymbol{k}_{02}$ \\
        \hline
        8 & $\boldsymbol{k}-\boldsymbol{k}_{01}-\boldsymbol{k}_{02}$ & $-\boldsymbol{k}-\boldsymbol{k}_{01}-\boldsymbol{k}_{02}$ \\
        \hline
    \end{tabular}
    \label{tab:LWmomenta}
\end{table}
After writing out these sums, $H_2$ becomes
\begin{align}
    \begin{split}
        H_2 &= \left.\sum_{\boldsymbol{k}}\right.^{'}  \sum_{\alpha\beta} \eta_{\boldsymbol{k}}^{\alpha\beta} A_{\boldsymbol{k}}^{\alpha\dagger}A_{\boldsymbol{k}}^{\beta}+\frac{N}{16N_s}\left.\sum_{\boldsymbol{k}}\right.^{''} \sum_{\alpha\beta} U^{\alpha\beta} \\
        \cdot&\Bigg(\bigg[\left(e^{i(\theta_1^\alpha+\theta_3^\beta)}+e^{i(\theta_3^\alpha+\theta_1^\beta)}+e^{i(\theta_2^\alpha+\theta_4^\beta)}+e^{i(\theta_4^\alpha+\theta_2^\beta)}\right)A_{\boldsymbol{k}}^\beta A_{-\boldsymbol{k}}^\alpha \\
        &+ e^{i(\theta_1^\alpha+\theta_1^\beta)}A_{\boldsymbol{k}}^\beta A_{\boldsymbol{q}_1}^\alpha + e^{i(\theta_3^\alpha+\theta_3^\beta)}A_{\boldsymbol{k}}^\beta A_{\boldsymbol{q}_2}^\alpha \\
        &+e^{i(\theta_2^\alpha+\theta_2^\beta)}A_{\boldsymbol{k}}^\beta A_{\boldsymbol{q}_3}^\alpha + e^{i(\theta_4^\alpha+\theta_4^\beta)}A_{\boldsymbol{k}}^\beta A_{\boldsymbol{q}_4}^\alpha \\
        &+\left(e^{i(\theta_1^\alpha+\theta_2^\beta)}+e^{i(\theta_2^\alpha+\theta_1^\beta)}\right)A_{\boldsymbol{k}}^\beta A_{\boldsymbol{q}_5}^\alpha + \left(e^{i(\theta_1^\alpha+\theta_4^\beta)}+e^{i(\theta_4^\alpha+\theta_1^\beta)}\right)A_{\boldsymbol{k}}^\beta A_{\boldsymbol{q}_6}^\alpha \\
        &+\left(e^{i(\theta_2^\alpha+\theta_3^\beta)}+e^{i(\theta_3^\alpha+\theta_2^\beta)}\right)A_{\boldsymbol{k}}^\beta A_{\boldsymbol{q}_7}^\alpha + \left(e^{i(\theta_3^\alpha+\theta_4^\beta)}+e^{i(\theta_4^\alpha+\theta_3^\beta)}\right)A_{\boldsymbol{k}}^\beta A_{\boldsymbol{q}_8}^\alpha \\
        &+ \left(e^{i(\theta_1^\alpha-\theta_1^\beta)}+e^{i(\theta_2^\alpha-\theta_2^\beta)}+e^{i(\theta_3^\alpha-\theta_3^\beta)}+e^{i(\theta_4^\alpha-\theta_4^\beta)}\right)A_{\boldsymbol{k}}^{\beta\dagger} A_{\boldsymbol{k}}^\alpha + 4A_{\boldsymbol{k}}^{\beta\dagger} A_{\boldsymbol{k}}^\beta \\
        &+e^{i(\theta_1^\alpha-\theta_3^\beta)}A_{\boldsymbol{k}}^{\beta\dagger} A_{\boldsymbol{p}_1}^\alpha + e^{i(\theta_1^\alpha-\theta_3^\alpha)}A_{\boldsymbol{k}}^{\beta\dagger} A_{\boldsymbol{p}_1}^\beta \\
        &+ e^{i(\theta_3^\alpha-\theta_1^\beta)}A_{\boldsymbol{k}}^{\beta\dagger} A_{\boldsymbol{p}_2}^\alpha + e^{i(\theta_3^\alpha-\theta_1^\alpha)}A_{\boldsymbol{k}}^{\beta\dagger} A_{\boldsymbol{p}_2}^\beta \\
        &+e^{i(\theta_2^\alpha-\theta_4^\beta)}A_{\boldsymbol{k}}^{\beta\dagger} A_{\boldsymbol{p}_3}^\alpha + e^{i(\theta_2^\alpha-\theta_4^\alpha)}A_{\boldsymbol{k}}^{\beta\dagger} A_{\boldsymbol{p}_3}^\beta \\
        &+ e^{i(\theta_4^\alpha-\theta_2^\beta)}A_{\boldsymbol{k}}^{\beta\dagger} A_{\boldsymbol{p}_4}^\alpha + e^{i(\theta_4^\alpha-\theta_2^\alpha)}A_{\boldsymbol{k}}^{\beta\dagger} A_{\boldsymbol{p}_4}^\beta \\
        &+ \left(e^{i(\theta_1^\alpha-\theta_4^\beta)}+e^{i(\theta_2^\alpha-\theta_3^\beta)}\right)A_{\boldsymbol{k}}^{\beta\dagger} A_{\boldsymbol{p}_5}^\alpha + \left(e^{i(\theta_1^\alpha-\theta_4^\alpha)}+e^{i(\theta_2^\alpha-\theta_3^\alpha)}\right)A_{\boldsymbol{k}}^{\beta\dagger} A_{\boldsymbol{p}_5}^\beta \\
        &+ \left(e^{i(\theta_1^\alpha-\theta_2^\beta)}+e^{i(\theta_4^\alpha-\theta_3^\beta)}\right)A_{\boldsymbol{k}}^{\beta\dagger} A_{\boldsymbol{p}_6}^\alpha + \left(e^{i(\theta_1^\alpha-\theta_2^\alpha)}+e^{i(\theta_4^\alpha-\theta_3^\alpha)}\right)A_{\boldsymbol{k}}^{\beta\dagger} A_{\boldsymbol{p}_6}^\beta \\
        &+ \left(e^{i(\theta_2^\alpha-\theta_1^\beta)}+e^{i(\theta_3^\alpha-\theta_4^\beta)}\right)A_{\boldsymbol{k}}^{\beta\dagger} A_{\boldsymbol{p}_7}^\alpha + \left(e^{i(\theta_2^\alpha-\theta_1^\alpha)}+e^{i(\theta_3^\alpha-\theta_4^\alpha)}\right)A_{\boldsymbol{k}}^{\beta\dagger} A_{\boldsymbol{p}_7}^\beta \\
        &+ \left(e^{i(\theta_3^\alpha-\theta_2^\beta)}+e^{i(\theta_4^\alpha-\theta_1^\beta)}\right)A_{\boldsymbol{k}}^{\beta\dagger} A_{\boldsymbol{p}_8}^\alpha \\
        &\mbox{\qquad\qquad\qquad}+ \left(e^{i(\theta_3^\alpha-\theta_2^\alpha)}+e^{i(\theta_4^\alpha-\theta_1^\alpha)}\right)A_{\boldsymbol{k}}^{\beta\dagger} A_{\boldsymbol{p}_8}^\beta\bigg] + \textrm{H.c}\Bigg). \\
    \end{split}
\end{align}
The sum $\left.\sum_{\boldsymbol{k}}\right.^{'}$ excludes the condensate momenta $\pm\boldsymbol{k}_{01}$ and $\pm\boldsymbol{k}_{02}$. The same goes for the sum $\left.\sum_{\boldsymbol{k}}\right.^{''}$. However, the double prime on this sum also indicates that any term containing a condensate momentum as an index is excluded. For instance, for $\boldsymbol{k} = 3\boldsymbol{k}_{01}$ the term $A_{\boldsymbol{k}}^{\beta\dagger} A_{\boldsymbol{p}_2}^\alpha$ becomes $A_{3\boldsymbol{k}_{01}}^{\beta\dagger} A_{\boldsymbol{k}_{01}}^\alpha$. Since the condensate operators have already been treated as complex numbers, these terms should be excluded from the sum. 

\section{Matrix Representation}
Including the terms from $H_0^{''}$ the coefficient of $A_{\boldsymbol{k}}^{\uparrow\dagger}A_{\boldsymbol{k}}^{\uparrow}$ is 
\begin{align}
    \begin{split}
        &\epsilon_{\boldsymbol{k}}-\epsilon_{\boldsymbol{0}} + \epsilon_{\boldsymbol{0}} - \epsilon_{\boldsymbol{k}_{01}}  -\frac{\abs{s_{\boldsymbol{k}_{01}}}}{4}\sum_{i=1}^4 \cos(\gamma_{\boldsymbol{k}_{0i}}+\Delta\theta_i) \\
        &+\frac{U_s}{8} \Bigg(2-4\cos(\theta_1^\uparrow+\theta_3^\uparrow-\theta_2^\uparrow-\theta_4^\uparrow)\\
        &\mbox{\qquad} -\alpha \bigg[\cos(\Delta\theta_1-\Delta\theta_2) + \cos(\Delta\theta_1-\Delta\theta_3) + \cos(\Delta\theta_1-\Delta\theta_4) \\
        &\mbox{\qquad\qquad}+ \cos(\Delta\theta_2-\Delta\theta_3) + \cos(\Delta\theta_2-\Delta\theta_4) + \cos(\Delta\theta_3-\Delta\theta_4)\\
        &\mbox{\qquad\qquad} + \cos(\theta_1^\uparrow+\theta_3^\downarrow-\theta_2^\downarrow-\theta_4^\uparrow) + \cos(\theta_1^\downarrow+\theta_3^\uparrow-\theta_2^\uparrow-\theta_4^\downarrow) \\
        &\mbox{\qquad\qquad} + \cos(\theta_1^\uparrow+\theta_3^\downarrow-\theta_2^\uparrow-\theta_4^\downarrow) + \cos(\theta_1^\downarrow+\theta_3^\uparrow-\theta_2^\downarrow-\theta_4^\uparrow) \bigg]\Bigg) \\
        &\equiv \mathcal{E}_{\boldsymbol{k}} + E_{k_0}^\uparrow.
    \end{split}
\end{align}
The coefficient of $A_{\boldsymbol{k}}^{\downarrow\dagger}A_{\boldsymbol{k}}^{\downarrow}$ is $\mathcal{E}_{\boldsymbol{k}} + E_{k_0}^\downarrow$ where the only significant change is that 
\begin{equation*}
    \cos(\theta_1^\uparrow+\theta_3^\uparrow-\theta_2^\uparrow-\theta_4^\uparrow) \mbox{\qquad is replaced by \qquad} \cos(\theta_1^\downarrow+\theta_3^\downarrow-\theta_2^\downarrow-\theta_4^\downarrow).
\end{equation*}

Our basis is now of length 72 due to the 18 separate momenta, two pseoduspin indices and the presence of terms that individually do not conserve particle numbers. The first 18 elements of $\boldsymbol{A}_{\boldsymbol{k}}$ are
\begin{align}
    \begin{split}
       &A_{\boldsymbol{k}}^\uparrow, A_{-\boldsymbol{k}}^\uparrow, A_{\boldsymbol{p}_1}^\uparrow, A_{\boldsymbol{q}_1}^\uparrow, A_{\boldsymbol{p}_2}^\uparrow, A_{\boldsymbol{q}_2}^\uparrow, A_{\boldsymbol{p}_3}^\uparrow, A_{\boldsymbol{q}_3}^\uparrow, A_{\boldsymbol{p}_4}^\uparrow, \\
       &A_{\boldsymbol{q}_4}^\uparrow, A_{\boldsymbol{p}_5}^\uparrow, A_{\boldsymbol{q}_5}^\uparrow, A_{\boldsymbol{p}_6}^\uparrow, A_{\boldsymbol{q}_6}^\uparrow, A_{\boldsymbol{p}_7}^\uparrow, A_{\boldsymbol{q}_7}^\uparrow, A_{\boldsymbol{p}_8}^\uparrow, A_{\boldsymbol{q}_8}^\uparrow. 
    \end{split}
\end{align}
The next 18 elements are the same only with pseudospin down, while the last 36 are the adjoints of the first 36. To obtain a matrix representation of the problem, we use commutators and make $-\boldsymbol{k}$-term explicit. As in the other phases, there are some momenta at which our basis contains copies of the same operators. For the LW phase we have 25 special momenta; $\boldsymbol{0}, \pm\boldsymbol{k}_{01}, \pm\boldsymbol{k}_{02}, \pm2\boldsymbol{k}_{01}, \pm2\boldsymbol{k}_{02},$ $(\boldsymbol{k}_{01}\pm\boldsymbol{k}_{02})/2, (-\boldsymbol{k}_{01}\pm\boldsymbol{k}_{02})/2,$ $\boldsymbol{k}_{01}\pm\boldsymbol{k}_{02}, -\boldsymbol{k}_{01}\pm\boldsymbol{k}_{02},$ $(3\boldsymbol{k}_{01}\pm\boldsymbol{k}_{02})/2, (-3\boldsymbol{k}_{01}\pm\boldsymbol{k}_{02})/2, (\boldsymbol{k}_{01}\pm3\boldsymbol{k}_{02})/2$ and $(-\boldsymbol{k}_{01}\pm3\boldsymbol{k}_{02})/2$. Except for the condensate momenta, these momenta are all part of the sum in $H_2$ given that the condensate momenta are lattice points in momentum space. 

In addition there are some special momenta where the occurrence of condensate operators means a special treatment is required. There are 12 special momenta of this kind $\pm3\boldsymbol{k}_{01}, \pm3\boldsymbol{k}_{02}, 2\boldsymbol{k}_{01}\pm\boldsymbol{k}_{02}, -2\boldsymbol{k}_{01}\pm\boldsymbol{k}_{02},$ $2\boldsymbol{k}_{02}\pm\boldsymbol{k}_{01}$ and $-2\boldsymbol{k}_{02}\pm\boldsymbol{k}_{01}$. Notice that these are the same momenta that appear as indices in $H_1$. These points will therefore be used to remove the linear terms and will for that purpose be calculated correctly. Regarding the quadratic part $H_2$, we assume the correction due to treating all special momenta in a correct way, compared to ignoring the problems are negligible. We thus write the quadratic part of the Hamiltonian as
\begin{equation}
    H_2 = \frac{1}{4}\left.\sum_{\boldsymbol{k}}\right.^{'} \boldsymbol{A}_{\boldsymbol{k}}^{\dagger}M_{\boldsymbol{k}}\boldsymbol{A}_{\boldsymbol{k}},
\end{equation}
where the prime on the sum indicates that we exclude the condensate momenta. The factor $1/4$ is because we have used commutators and made $-\boldsymbol{k}$-terms explicit to rewrite $H_2$. Our matrix $M_{\boldsymbol{k}}$ is a $72\cross 72$ matrix, and thus too large to conveniently show here. However, we note that the matrix is very sparse, as there are in total $16$ blocks of $16 \cross 16$ zero matrices because operators with momentum indices $\boldsymbol{p}_i$ and $\boldsymbol{q}_i$ do not mix with each other. Using the fact that $M_{\boldsymbol{k}}$ is of the form
\begin{gather}
\label{eq:LWM}
    M_{\boldsymbol{k}} = 
    \begin{pmatrix}
    M_1 & M_2 \\
    M_2^* & M_1^* \\
    \end{pmatrix},
\end{gather}
with $M_1^\dagger = M_1$ and $M_2^T=M_2$ it is in fact enough to specify rows 1, 2, 19 and 20 of $M_1$ and $M_2^*$. The rest of the matrix $M_{\boldsymbol{k}}$ can then be filled, and the remaining unspecified entries are 0. For instance, column 1 of $M_1$ will be the complex conjugate of row 1 of $M_1$. These 8 rows are
\begin{align}
    \begin{split}
        M_{1, \textrm{row }1} = (&M_{1,1}(\boldsymbol{k}), 0, M_{1,3}, 0, M_{1,3}^*, 0, M_{1,7}, 0, M_{1,7}^*, 0, M_{1,11}, 0, \\
        &M_{1,13}, 0, M_{1,13}^*, 0, M_{1,11}^*, 0, M_{1,19}(\boldsymbol{k}), 0, M_{1,21}, 0, M_{1,23}, 0, \\
        &M_{1,25}, 0, M_{1,27}, 0, M_{1,29}, 0, M_{1,31}, 0, M_{1,33}, 0, M_{1,35}, 0), \\
        M_{1, \textrm{row }2} = (&0, M_{1,1}(\boldsymbol{k}), 0, M_{1,3}, 0, M_{1,3}^*, 0, M_{1,7}, 0, M_{1,7}^*, 0, M_{1,11}, \\
        &0, M_{1,13}, 0, M_{1,13}^*, 0, M_{1,11}^*, 0, M_{1,19}(-\boldsymbol{k}), 0, M_{1,21}, 0, M_{1,23}, \\
        &0, M_{1,25}, 0, M_{1,27}, 0, M_{1,29}, 0, M_{1,31}, 0, M_{1,33}, 0, M_{1,35}), \\
        M_{1, \textrm{row }19} = (&M_{1,19}^* (\boldsymbol{k}), 0, M_{1,23}^*, 0, M_{1,21}^*, 0, M_{1,27}^*, 0, M_{1,25}^*, 0, M_{1,35}^*, 0\\
        &M_{1,33}^*, 0, M_{1,31}^*, 0, M_{1,29}^*, 0, M_{19,19}(\boldsymbol{k}), 0, M_{19,21}, 0, M_{19,21}^*, 0, \\
        & M_{19,25}, 0, M_{19,25}^*, 0, M_{19,29}, 0, M_{19,31}, 0, M_{19,31}^*, 0, M_{19,29}^*, 0), \\
        M_{1, \textrm{row }20} = (&0, M_{1,19}^* (-\boldsymbol{k}), 0, M_{1,23}^*, 0, M_{1,21}^*, 0, M_{1,27}^*, 0, M_{1,25}^*, 0, M_{1,35}^*,\\
        &0, M_{1,33}^*, 0, M_{1,31}^*, 0, M_{1,29}^*, 0, M_{19,19}(\boldsymbol{k}), 0, M_{19,21}, 0, M_{19,21}^*, \\
        & 0, M_{19,25}, 0, M_{19,25}^*, 0, M_{19,29}, 0, M_{19,31}, 0, M_{19,31}^*, 0, M_{19,29}^*), \\
    \end{split}
\end{align}
and
\begin{align}
    \begin{split}
        M_{2, \textrm{row }1}^* = (&0, M_{37,2}, 0, M_{37,4}, 0, M_{37,6}, 0, M_{37,8}, 0, M_{37,10}, 0, M_{37,12}, \\
        &0, M_{37,14}, 0, M_{37,16}, 0, M_{37,18}, 0, M_{37,20}, 0, M_{37,22}, 0, M_{37,24},\\
        &0, M_{37,26}, 0, M_{37,28}, 0, M_{37,30}, 0, M_{37,32}, 0, M_{37,34}, 0, M_{37,36}),\\
        M_{2, \textrm{row }2}^* = (&M_{37,2}, 0, M_{37,4}, 0, M_{37,6}, 0, M_{37,8}, 0, M_{37,10}, 0, M_{37,12}, 0, \\
        &M_{37,14}, 0, M_{37,16}, 0, M_{37,18}, 0, M_{37,20}, 0, M_{37,22}, 0, M_{37,24}, 0,\\
        &M_{37,26}, 0, M_{37,28}, 0, M_{37,30}, 0, M_{37,32}, 0, M_{37,34}, 0, M_{37,36}, 0),\\
        M_{2, \textrm{row }19}^* = (&0, M_{37,20}, 0, M_{37,22}, 0, M_{37,24}, 0, M_{37,26}, 0, M_{37,28}, 0, M_{37,30},\\
        &0, M_{37,32}, 0, M_{37,34}, 0, M_{37,36}, 0, M_{55,20}, 0, M_{55,22}, 0, M_{55,24}, \\
        &0, M_{55,26}, 0, M_{55,28}, 0, M_{55,30}, 0, M_{55,32}, 0, M_{55,34}, 0, M_{55,36}),\\
        M_{2, \textrm{row }20}^* = (&M_{37,20}, 0, M_{37,22}, 0, M_{37,24}, 0, M_{37,26}, 0, M_{37,28}, 0, M_{37,30}, 0,\\
        &M_{37,32}, 0, M_{37,34}, 0, M_{37,36}, 0, M_{55,20}, 0, M_{55,22}, 0, M_{55,24}, 0,\\
        &M_{55,26}, 0, M_{55,28}, 0, M_{55,30}, 0, M_{55,32}, 0, M_{55,34}, 0, M_{55,36}, 0).
    \end{split}
\end{align}
To visualize this matrix, imagine an extension of the matrix in the SW phase given in \eqref{eq:SWMat} to a $72 \cross 72$ matrix with the same pattern. The elements in the first row of $M_1$ are
\begin{align}
    \begin{split}
        M_{1,1}(\boldsymbol{k}) &= \mathcal{E}_{\boldsymbol{k}} + E_{k_0}^\uparrow,\\
        M_{1,3} &= \frac{U_s}{8}\left(2e^{i(\theta_1^\uparrow-\theta_3^\uparrow)}+\alpha e^{i(\theta_1^\downarrow-\theta_3^\downarrow)} \right) , \\ 
        M_{1,7} &= \frac{U_s}{8}\left(2e^{i(\theta_2^\uparrow-\theta_4^\uparrow)}+\alpha e^{i(\theta_2^\downarrow-\theta_4^\downarrow)} \right) , \\ 
        M_{1,11} &=  \frac{U_s}{8}\left(2e^{i(\theta_1^\uparrow-\theta_4^\uparrow)}+2e^{i(\theta_2^\uparrow-\theta_3^\uparrow)}+\alpha e^{i(\theta_1^\downarrow-\theta_4^\downarrow)}+\alpha e^{i(\theta_2^\downarrow-\theta_3^\downarrow)} \right), \\
        M_{1,13} &= \frac{U_s}{8}\left(2e^{i(\theta_1^\uparrow-\theta_2^\uparrow)}+2e^{i(\theta_4^\uparrow-\theta_3^\uparrow)}+\alpha e^{i(\theta_1^\downarrow-\theta_2^\downarrow)}+\alpha e^{i(\theta_4^\downarrow-\theta_3^\downarrow)} \right), \\
        M_{1,19} &= s_{\boldsymbol{k}} + \frac{U_s\alpha}{4}\left(e^{i(\theta_1^\downarrow-\theta_1^\uparrow)}+e^{i(\theta_2^\downarrow-\theta_2^\uparrow)}+e^{i(\theta_3^\downarrow-\theta_3^\uparrow)}+e^{i(\theta_4^\downarrow-\theta_4^\uparrow)}\right), \\
        M_{1,21} &= \frac{U_s\alpha}{8}e^{i(\theta_1^\downarrow-\theta_3^\uparrow)}, \mbox{\qquad\qquad} M_{1,23} = \frac{U_s\alpha}{8}e^{i(\theta_3^\downarrow-\theta_1^\uparrow)},\\
        M_{1,25} &= \frac{U_s\alpha}{8}e^{i(\theta_2^\downarrow-\theta_4^\uparrow)}, \mbox{\qquad\qquad} M_{1,27} = \frac{U_s\alpha}{8}e^{i(\theta_4^\downarrow-\theta_2^\uparrow)},\\
        M_{1,29} &= \frac{U_s\alpha}{8}\left(e^{i(\theta_1^\downarrow-\theta_4^\uparrow)}+e^{i(\theta_2^\downarrow-\theta_3^\uparrow)}\right),\\
        M_{1,31} &= \frac{U_s\alpha}{8}\left(e^{i(\theta_1^\downarrow-\theta_2^\uparrow)}+e^{i(\theta_4^\downarrow-\theta_3^\uparrow)}\right),\\
        M_{1,33} &= \frac{U_s\alpha}{8}\left(e^{i(\theta_2^\downarrow-\theta_1^\uparrow)}+e^{i(\theta_3^\downarrow-\theta_4^\uparrow)}\right),\\
        M_{1,35} &= \frac{U_s\alpha}{8}\left(e^{i(\theta_3^\downarrow-\theta_2^\uparrow)}+e^{i(\theta_4^\downarrow-\theta_1^\uparrow)}\right).\\
    \end{split}
\end{align}
The new elements appearing in row 19 of $M_1$ are
\begin{align}
    \begin{split}
        M_{19,19}(\boldsymbol{k}) &= \mathcal{E}_{\boldsymbol{k}} + E_{k_0}^\downarrow,\\
        M_{19,21} &= \frac{U_s}{8}\left(2e^{i(\theta_1^\downarrow-\theta_3^\downarrow)}+\alpha e^{i(\theta_1^\uparrow-\theta_3^\uparrow)} \right)  ,\\
        M_{19,25} &=  \frac{U_s}{8}\left(2e^{i(\theta_2^\downarrow-\theta_4^\downarrow)}+\alpha e^{i(\theta_2^\uparrow-\theta_4^\uparrow)} \right) ,\\
        M_{19,29} &= \frac{U_s}{8}\left(2e^{i(\theta_1^\downarrow-\theta_4^\downarrow)}+2e^{i(\theta_2^\downarrow-\theta_3^\downarrow)}+\alpha e^{i(\theta_1^\uparrow-\theta_4^\uparrow)}+\alpha e^{i(\theta_2^\uparrow-\theta_3^\uparrow)} \right), \\
        M_{19,31} &= \frac{U_s}{8}\left(2e^{i(\theta_1^\downarrow-\theta_2^\downarrow)}+2e^{i(\theta_4^\downarrow-\theta_3^\downarrow)}+\alpha e^{i(\theta_1^\uparrow-\theta_2^\uparrow)}+\alpha e^{i(\theta_4^\uparrow-\theta_3^\uparrow)} \right). 
    \end{split}
\end{align}
Turning to $M_2^*$, the elements in its first row are
\begin{align}
    \begin{split}
        M_{37,2} &= \frac{U_s}{2}\left(e^{i(\theta_1^\uparrow+\theta_3^\uparrow)}+e^{i(\theta_2^\uparrow+\theta_4^\uparrow)}  \right), \\
        M_{37,4} &= \frac{U_s}{8}e^{i2\theta_1^\uparrow}, \mbox{\qquad\qquad\qquad} M_{37,6} = \frac{U_s}{8}e^{i2\theta_3^\uparrow}, \\
        M_{37,8} &= \frac{U_s}{8}e^{i2\theta_2^\uparrow}, \mbox{\qquad\qquad\qquad} M_{37,10} = \frac{U_s}{8}e^{i2\theta_4^\uparrow}, \\
        M_{37,12} &= \frac{U_s}{4}e^{i(\theta_1^\uparrow+\theta_2^\uparrow)}, \mbox{\qquad\qquad\quad} M_{37,14} = \frac{U_s}{4}e^{i(\theta_1^\uparrow+\theta_4^\uparrow)}, \\
        M_{37,16} &= \frac{U_s}{4}e^{i(\theta_2^\uparrow+\theta_3^\uparrow)}, \mbox{\qquad\qquad\quad} M_{37,18} = \frac{U_s}{4}e^{i(\theta_3^\uparrow+\theta_4^\uparrow)}, \\
        M_{37,20} &= \frac{U_s\alpha}{4}\left(e^{i(\theta_1^\uparrow+\theta_3^\downarrow)}+e^{i(\theta_1^\downarrow+\theta_3^\uparrow)}+e^{i(\theta_2^\uparrow+\theta_4^\downarrow)}+e^{i(\theta_2^\downarrow+\theta_4^\uparrow)}\right), \\
        M_{37,22} &= \frac{U_s\alpha}{8}e^{i(\theta_1^\uparrow+\theta_1^\downarrow)}, \mbox{\qquad\qquad} M_{37,24} = \frac{U_s\alpha}{8}e^{i(\theta_3^\uparrow+\theta_3^\downarrow)}, \\
        M_{37,26} &= \frac{U_s\alpha}{8}e^{i(\theta_2^\uparrow+\theta_2^\downarrow)}, \mbox{\qquad\qquad} M_{37,28} = \frac{U_s\alpha}{8}e^{i(\theta_4^\uparrow+\theta_4^\downarrow)}, \\
        M_{37,30} &= \frac{U_s\alpha}{8}\left(e^{i(\theta_1^\downarrow+\theta_2^\uparrow)}+e^{i(\theta_2^\downarrow+\theta_1^\uparrow)}\right), \\
        M_{37,32} &= \frac{U_s\alpha}{8}\left(e^{i(\theta_1^\downarrow+\theta_4^\uparrow)}+e^{i(\theta_4^\downarrow+\theta_1^\uparrow)}\right), \\
        M_{37,34} &= \frac{U_s\alpha}{8}\left(e^{i(\theta_2^\downarrow+\theta_3^\uparrow)}+e^{i(\theta_3^\downarrow+\theta_2^\uparrow)}\right), \\
        M_{37,36} &= \frac{U_s\alpha}{8}\left(e^{i(\theta_3^\downarrow+\theta_4^\uparrow)}+e^{i(\theta_4^\downarrow+\theta_3^\uparrow)}\right).
    \end{split}
\end{align}
Finally, the new elements appearing in row 19 of $M_2^*$ are
\begin{align}
    \begin{split}
        M_{55,20} &= \frac{U_s}{2}\left(e^{i(\theta_1^\downarrow+\theta_3^\downarrow)}+e^{i(\theta_2^\downarrow+\theta_4^\downarrow)}  \right), \\
        M_{55,22} &= \frac{U_s}{8}e^{i2\theta_1^\downarrow}  , \mbox{\qquad\qquad\quad} M_{55,24} = \frac{U_s}{8}e^{i2\theta_3^\downarrow}  , \\ 
        M_{55,26} &= \frac{U_s}{8}e^{i2\theta_2^\downarrow}  , \mbox{\qquad\qquad\quad}M_{55,28} = \frac{U_s}{8}e^{i2\theta_4^\downarrow}  , \\  
        M_{55,30} &= \frac{U_s}{4}e^{i(\theta_1^\downarrow+\theta_2^\downarrow)}  , \mbox{\qquad\qquad} M_{55,32} = \frac{U_s}{4}e^{i(\theta_1^\downarrow+\theta_4^\downarrow)}  , \\ 
        M_{55,34} &= \frac{U_s}{4}e^{i(\theta_2^\downarrow+\theta_3^\downarrow)}  , \mbox{\qquad\qquad} M_{55,36} = \frac{U_s}{4}e^{i(\theta_3^\downarrow+\theta_4^\downarrow)} . 
    \end{split}
\end{align}

Due to the terms with number operators in $H_2$ the use of commutators yields a shift
\begin{align}
    \begin{split}
        H'_0 &= H_0 -\frac{1}{2}\left.\sum_{\boldsymbol{k}}\right.^{'}\left(M_{1,1}(\boldsymbol{k}) + M_{19,19}(\boldsymbol{k})\right) \\
        &= H_0 - 16t\cos(k_0 a) - (N_s-4)\left(4t+\frac{E_{k_0}^\uparrow+E_{k_0}^\downarrow}{2}\right),
    \end{split}
\end{align}
in the operator independent part of the Hamiltonian. 

The excitation spectrum is the eigenvalues of $M_{\boldsymbol{k}}J$. The bands bear resemblance to the bands calculated in the SW phase. We have 8 nonzero positive bands, where the four smallest have negative BV norm eigenvectors. There are also 8 nonzero negative bands that are the negatives of the 8 positive bands. The four negative eigenvalues with smallest absolute values have positive BV norm eigenvectors. Hence, it is the negative bands that enter the diagonalized Hamiltonian. Meanwhile, there are a total of 56 eigenvalues that are within numerical accuracy zero. The eigenvalues can be represented by $\lambda(\boldsymbol{k}) = \pm\Omega_i(\boldsymbol{k}),$ $i=1,2,\dots,8$. The eigenvalues are ordered such that $\Omega_i(\boldsymbol{k}) \geq \Omega_j(\boldsymbol{k})$ if $j>i$. 
We can write $H_2$ as
\begin{align}
    \begin{split}
        H_2 = \frac{1}{2}\left.\sum_{\boldsymbol{k}}\right.^{'} \Bigg(& \sum_{\sigma=1}^4 \Omega_\sigma(\boldsymbol{k}) \left( B_{\boldsymbol{k},\sigma'}^\dagger B_{\boldsymbol{k},\sigma'} +\frac12 \right) \\
        &-\sum_{\sigma = 5}^8 \Omega_\sigma(\boldsymbol{k}) \left( B_{\boldsymbol{k},\sigma'}^\dagger B_{\boldsymbol{k},\sigma'} +\frac12 \right) \\
        &+\sum_{\sigma = 9}^{36} 0 \left( B_{\boldsymbol{k},\sigma'}^\dagger B_{\boldsymbol{k},\sigma'} +\frac12 \right)\Bigg).
    \end{split}
\end{align}
Just as we did in the SW phase, we will shift the zero of the energies by adding and subtracting the maximum value of $\Omega_5(\boldsymbol{k})$ which we name $\Omega_0$. Defining $\Delta\Omega_\sigma \equiv \Omega_0 +\Omega_\sigma$ for $\sigma = 1,2,3,4$, $\Delta\Omega_\sigma \equiv \Omega_0$ for $\sigma = 5, 6, \dots, 33$, $\Delta\Omega_\sigma = \Omega_0 - \Omega_{\sigma'}$ for $\sigma = 33,34,35,36$ and $\sigma' = 8,7,6,5$ and renumbering the operators correspondingly, we get
\begin{align}
    \begin{split}
        H_2 = & -\Omega_0 N_q + \frac{1}{2}\left.\sum_{\boldsymbol{k}}\right.^{'} \sum_{\sigma=1}^{36} \Delta\Omega_\sigma(\boldsymbol{k}) \left( B_{\boldsymbol{k},\sigma}^\dagger B_{\boldsymbol{k},\sigma} +\frac12 \right),
    \end{split}
\end{align}
where we defined
\begin{equation}
    N_q \equiv \frac{1}{2}\left.\sum_{\boldsymbol{k}}\right.^{'} \sum_{\sigma=1}^{36} \left( B_{\boldsymbol{k},\sigma}^\dagger B_{\boldsymbol{k},\sigma} +\frac12 \right).
\end{equation}

\section{The Special Momenta}
In appendix \ref{app:SW} we presented a general procedure to treat the special momenta due to repeated entries in the basis. In the LW phase we ignore the effects of these special momenta, but if one were to check them, the general procedure would be an effective way of doing so.


The special momenta related to the occurrence of condensate momenta in the basis $\boldsymbol{A}_{\boldsymbol{k}}$ will be treated correctly in order to remove the linear terms in $H_1$. We however neglect the difference such a treatment causes in the quadratic part $H_2$. The procedure to treat these terms will be shown using the example $\boldsymbol{k} = 3\boldsymbol{k}_{01}$ and is similar to the treatment in the SW phase. We always treat two special momenta simultaneously, and the general structure of the results are the same for all the special momenta of this type.

Having made $-\boldsymbol{k}$-terms explicit one can see that $H_2(3\boldsymbol{k}_{01}) = H_2(-3\boldsymbol{k}_{01})$ and they can be treated simultaneously. At $\boldsymbol{k} = 3\boldsymbol{k}_{01}$ terms $4+18i$ and $5+18i$ for $i=0,1,2,3$ in the basis are condensate momenta. Removing these terms from the basis, and also removing the corresponding rows and columns from $M_{3\boldsymbol{k}_{01}}$ we define a new operator vector $\boldsymbol{A}_{3k_{01}}$ of length $64$ and a new $64 \cross 64$ matrix $M_{3k_{01}}$. The lack of bold font on $k_{01}$ serves to indicate we have reduced the size of the matrix, it is not an indication that $\boldsymbol{k}_{01}$ is no longer a vector. The matrix $M_{3k_{01}}J$ has 8 positive eigenvalues, 8 negative eigenvalues and a total of 48 eigenvalues that within numerical accuracy are zero. The four lowest positive eigenvalues have anomalous modes, and hence it is their negatives that enter the diagonalized Hamiltonian.  We find $H_2(3\boldsymbol{k}_{01}) = H_2(-3\boldsymbol{k}_{01}) =$
\begin{align}
    \begin{split}
    \frac{1}{2}&\Big\{\sum_{\sigma = 1}^{4}\omega_{3k_{01}, \sigma}\left( B_{3\boldsymbol{k}_{01},\sigma^{'}}^{\dagger}B_{3\boldsymbol{k}_{01}, \sigma^{'}}+\frac12\right) \\
    & -\sum_{\sigma = 5}^{8}\omega_{3k_{01}, \sigma}\left( B_{3\boldsymbol{k}_{01},\sigma^{'}}^{\dagger}B_{3\boldsymbol{k}_{01}, \sigma^{'}}+\frac12\right) \\
    & +\sum_{\sigma = 9}^{32}0\left( B_{3\boldsymbol{k}_{01},\sigma^{'}}^{\dagger}B_{3\boldsymbol{k}_{01}, \sigma^{'}}+\frac12\right) \Big\}.
    \end{split}
\end{align}
We shift the zero of energy by $\Omega_0$ for these eigenvalues as well. Defining $\Delta\omega_{3k_{01},i} = \Omega_0+\omega_{3k_{01}, i}$ for $i = 1,2,3,4$, $\Delta\omega_{3k_{01},i} = \Omega_0$ for $i = 5, \dots, 28$ and $\Delta\omega_{3k_{01},i} = \Omega_0-\omega_{3k_0, i'}$ for $i=29,30,31,32$ and $i' = 8,7,6,5$ and renumbering the operators we arrive at $H_2(3\boldsymbol{k}_{01}) + H_2(-3\boldsymbol{k}_{01}) = 2H_2(3\boldsymbol{k}_{01})$,
\begin{align}
    \begin{split}
        2H_2(3\boldsymbol{k}_{01}) &= -\Omega_0 N_{q,3k_{01}}+\sum_{\sigma=1}^{32} \Delta\omega_{3k_{01}, \sigma} \left( B_{3\boldsymbol{k}_{01},\sigma}^{\dagger}B_{3\boldsymbol{k}_{01}, \sigma} +\frac12  \right).
    \end{split}
\end{align}
To simplify the expression, we defined
$$N_{q,3k_{01}} \equiv \sum_{\sigma=1}^{32} \left(B_{3\boldsymbol{k}_{01},\sigma}^{\dagger}B_{3\boldsymbol{k}_{01}, \sigma}+\frac12 \right).$$
The energies $\Delta\omega_{3k_{01},i}$ do not agree completely with the energy spectrum $\Delta\Omega_\sigma(\boldsymbol{k})$ at $\pm3\boldsymbol{k}_{01}$ as was also the case in the SW phase.

\section{Free Energy}
The treatment of $H_1$ follows the same idea used in the SW phase. For the terms $\sim A_{3\boldsymbol{k}_{01}}^\alpha$, $A_{3\boldsymbol{k}_{01}}^{\alpha\dagger}$,  $A_{-3\boldsymbol{k}_{01}}^\alpha$ and $A_{-3\boldsymbol{k}_{01}}^{\alpha\dagger}$ we use the special treatment of $\boldsymbol{k} = \pm3\boldsymbol{k}_{01}$. Using that $\boldsymbol{A}_{3k_{01}} = JT_{3k_{01}}J \boldsymbol{B}_{3k_{01}}$ we can transform $H_1$ to the basis in which $H_2$ is diagonal. For instance, $A_{3\boldsymbol{k}_{01}}^\uparrow = \sum_i (JT_{3k_{01}}J)_{1,i} (\boldsymbol{B}_{3k_{01}})_i$. All in all, we find that
\begin{align}
    \begin{split}
        H_1 = \frac{\sqrt{N}U_s}{8\sqrt{2}}\sum_{i=1}^{64}\left\{\sum_{j\in\mathcal{J}}\sum_{k \in \mathcal{K}} \right\} \Big[ & c_{j}^{\uparrow} (JT_{k}J)_{1,i} + c_{j}^{\downarrow} (JT_{k}J)_{17,i} \\
        &+ c_{j}^{\uparrow*} (JT_{k}J)_{33,i} + c_{j}^{\downarrow*} (JT_{k}J)_{49,i} \\
        &+c_{-j}^{\uparrow} (JT_{k}J)_{2,i} + c_{-j}^{\downarrow} (JT_{k}J)_{18,i}\\
        &+ c_{-j}^{\uparrow*} (JT_{k}J)_{34,i} + c_{-j}^{\downarrow*} (JT_{k}J)_{50,i}\Big](\boldsymbol{B}_{k})_i ,
    \end{split}
\end{align}
where $\mathcal{J} = \{+1, 1-2, 1+2, +2, 2-1, 2+1\}$, $\mathcal{K} = \{3k_{01}, 2k_{01}-k_{02}, 2k_{01}+k_{02}, 3k_{02}, 2k_{02}-k_{01}, 2k_{02}+k_{01}\}$ and the coefficients $c_{\pm j}^\alpha$ are given in \eqref{eq:LWH1coeff}. When $j$ is element $i$ of $\mathcal{J}$, $k$ is element $i$ of $\mathcal{K}$.
We write this as
\begin{align}
    \begin{split}
        H_1 = \sum_{i=1}^{64} \Big[ &c_{1,i}(\boldsymbol{B}_{3k_{01}})_i + c_{1-2,i}(\boldsymbol{B}_{2k_{01}-k_{02}})_i + c_{1+2,i}(\boldsymbol{B}_{2k_{01}+k_{02}})_i \\
        &+c_{2,i} (\boldsymbol{B}_{3k_{02}})_i + c_{2-1}(\boldsymbol{B}_{2k_{02}-k_{01}})_i + c_{2+1}(\boldsymbol{B}_{2k_{02}+k_{01}})_i \Big],
    \end{split}
\end{align}
and note that for all these coefficients, $c_{i+32} = c_{i}^*$, meaning it is enough to consider the first $32$. We define the energies $E_{1,i} = \Delta\omega_{3k_{01}, i}$ for $i=1,2,3,4$, $E_{1,i} = \Delta\omega_{3k_{01}, i'}$ for $i=5,6,7,8$ and $i' = 32,31,30,29$ and $E_{1,i} = \Delta\omega_{3k_{01}, i'}$ for $i=9, \dots,32$ and $i' = 5,\dots,28$. Similar definition are made at the other momenta. Finally then, we may remove $H_1$ by completing squares with terms like
\begin{align}
    \begin{split}
        \sum_{\sigma=1}^{32} \Delta\omega_{3k_{01}, \sigma} \left( B_{3\boldsymbol{k}_{01},\sigma}^{\dagger}B_{3\boldsymbol{k}_{01}, \sigma} +\frac12  \right).
    \end{split}
\end{align}
This leads to a shift of the operator independent part. We find
\begin{align}
    \begin{split}
        \Tilde{H}_0 = H'_0 - \sum_{i=1}^{32} \bigg( &\frac{\abs{c_{1,i}}^2}{E_{1,i}} +  \frac{\abs{c_{1-2,i}}^2}{E_{1-2,i}} +  \frac{\abs{c_{1+2,i}}^2}{E_{1+2,i}} \\
        &+ \frac{\abs{c_{2,i}}^2}{E_{2,i}} + \frac{\abs{c_{2-1,i}}^2}{E_{2-1,i}} + \frac{\abs{c_{2+1,i}}^2}{E_{2+1,i}} \bigg).
    \end{split}
\end{align}
Notice that e.g. $2\boldsymbol{k}_{01}-\boldsymbol{k}_{02}$ is also a part of the basis at $\boldsymbol{k} = 3\boldsymbol{k}_{01}$. In fact, it would have been enough to use the transformation matrices at $\boldsymbol{k} = 3\boldsymbol{k}_{01}$ and $\boldsymbol{k} = 3\boldsymbol{k}_{02}$ to transform $H_1$ to the diagonal basis. However, that yielded unsatisfactory results. Therefore, a procedure where the $A$-operators were always among the first two operators in the basis was used.

At zero temperature the term
\begin{equation}
    -\Omega_0 N_q = -\frac{\Omega_0}{2}\left.\sum_{\boldsymbol{k}}\right.^{'} \sum_{\sigma=1}^{36} \left( B_{\boldsymbol{k},\sigma}^\dagger B_{\boldsymbol{k},\sigma} +\frac12 \right)
\end{equation}
reduces to $-9\Omega_0(N_s-4)$. We define $\Tilde{H}'_0 = \Tilde{H}_0 -9\Omega_0(N_s-4)$ as the final operator independent part. At zero temperature the free energy  is the same as $\langle H \rangle$ and reads
\begin{align}
    \begin{split}
        F_{\textrm{LW}} = \Tilde{H}'_0 +\frac{1}{4}\left.\sum_{\boldsymbol{k}}\right.^{'} \sum_{\sigma=1}^{36} \Delta\Omega_\sigma(\boldsymbol{k}).
    \end{split}
\end{align}

The result in the other phases were that minimization of the free energy gave $k_{0\textrm{min}} = k_{0m}$ when $N_s$ becomes large and that the angles obey \eqref{eq:gammathetapi}. We also find that the excitation spectrum becomes complex for $k_0$ too far away from $k_{0m}$ or for angles that vary too much from \eqref{eq:gammathetapi}. Hence, let us first assume $k_{0\textrm{min}} = k_{0m}$ and that \eqref{eq:gammathetapi} holds. For the LW phase, \eqref{eq:gammathetapi} implies
\begin{equation}
\label{eq:LWgammathetapi}
    \theta_1^\downarrow-\theta_1^\uparrow = \frac{\pi}{4}, \mbox{\quad} \theta_2^\downarrow-\theta_2^\uparrow = \frac{7\pi}{4}, \mbox{\quad} \theta_3^\downarrow-\theta_3^\uparrow = \frac{5\pi}{4}, \mbox{\quad} \theta_4^\downarrow-\theta_4^\uparrow = \frac{3\pi}{4}.
\end{equation}
It seems natural to investigate what value of the sum $ \theta_1^\uparrow + \theta_3^\uparrow - \theta_2^\uparrow -\theta_4^\uparrow$ minimizes $F_{\textrm{LW}}$ since this sum of the angels appears several places in $H'_0$. The first indication is that $-\pi/2$ is the optimal value of the sum. We therefore set $\theta_4^\uparrow = \theta +\pi/2+\theta_1^\uparrow + \theta_3^\uparrow - \theta_2^\uparrow$, and find that $\theta = 0$ is optimal. Thus, we conclude that 
\begin{equation}
\label{eq:LWanglesup}
    \theta_1^\uparrow + \theta_3^\uparrow - \theta_2^\uparrow -\theta_4^\uparrow = -\frac{\pi}{2},
\end{equation}
When this is true, together with \eqref{eq:LWgammathetapi}, we find that $M_{1,1}(\boldsymbol{k}) = M_{19,19}(\boldsymbol{k})$ which seems natural in the absence of a Zeeman field. Note that this result is different from a similar relation found in \cite{master} in which it was found that one of these is zero, while the other is $\pi$. This might explain the problems encountered regarding the assumptions $N_0^\uparrow = N_0^\downarrow$ and $\mu^\uparrow = \mu^\downarrow$ in \cite{master}.

In order to investigate the differences $\Delta\theta_i = \theta_i^\downarrow -\theta_i^\uparrow$, we assume \eqref{eq:LWanglesup} holds and that $k_{0\textrm{min}} = k_{0m}$. We set $\theta_4^\uparrow = \pi/2 + \theta_1^\uparrow + \theta_3^\uparrow-\theta_2^\uparrow$ and $\theta_4^\downarrow = 3\pi/2 + \theta_1^\downarrow + \theta_3^\downarrow - \theta_2^\downarrow$ which means we will vary both $\Delta\theta_i$ for $i=1,2$ or $3$ and $\Delta\theta_4$ at the same time. First, we set $\theta_1^\downarrow = \pi/4$ and find that $\theta_1^\uparrow = 0$ is optimal. Next, we set $\theta_2^\downarrow = 7\pi/4$ and find that $\theta_2^\uparrow = 0$ is optimal. Finally we set $\theta_3^\downarrow = 5\pi/4$ and find that $\theta_3^\uparrow = 0$ is optimal. In all these cases the two $\Delta\theta_i$ not under consideration are set to the values in \eqref{eq:LWgammathetapi}. 

The final free angles are $\theta_1^\uparrow, \theta_2^\uparrow$ and $\theta_3^\uparrow$. We find that
\begin{equation}
    \label{eq:LWangles23up}
    \theta_2^\uparrow = \theta_3^\uparrow = \theta_1^\uparrow + \pi
\end{equation}
is optimal. With all angles determined by $\theta_1^\uparrow$, the variations of $F_{\textrm{LW}}$ are negligible when varying $\theta_1^\uparrow$. In conclusion, we believe \eqref{eq:LWgammathetapi}, \eqref{eq:LWanglesup} and \eqref{eq:LWangles23up} determine the angles when a value of $\theta_1^\uparrow$ is chosen.


\begin{figure}
    \centering
    \begin{subfigure}{.49\textwidth}
      \includegraphics[width=\linewidth]{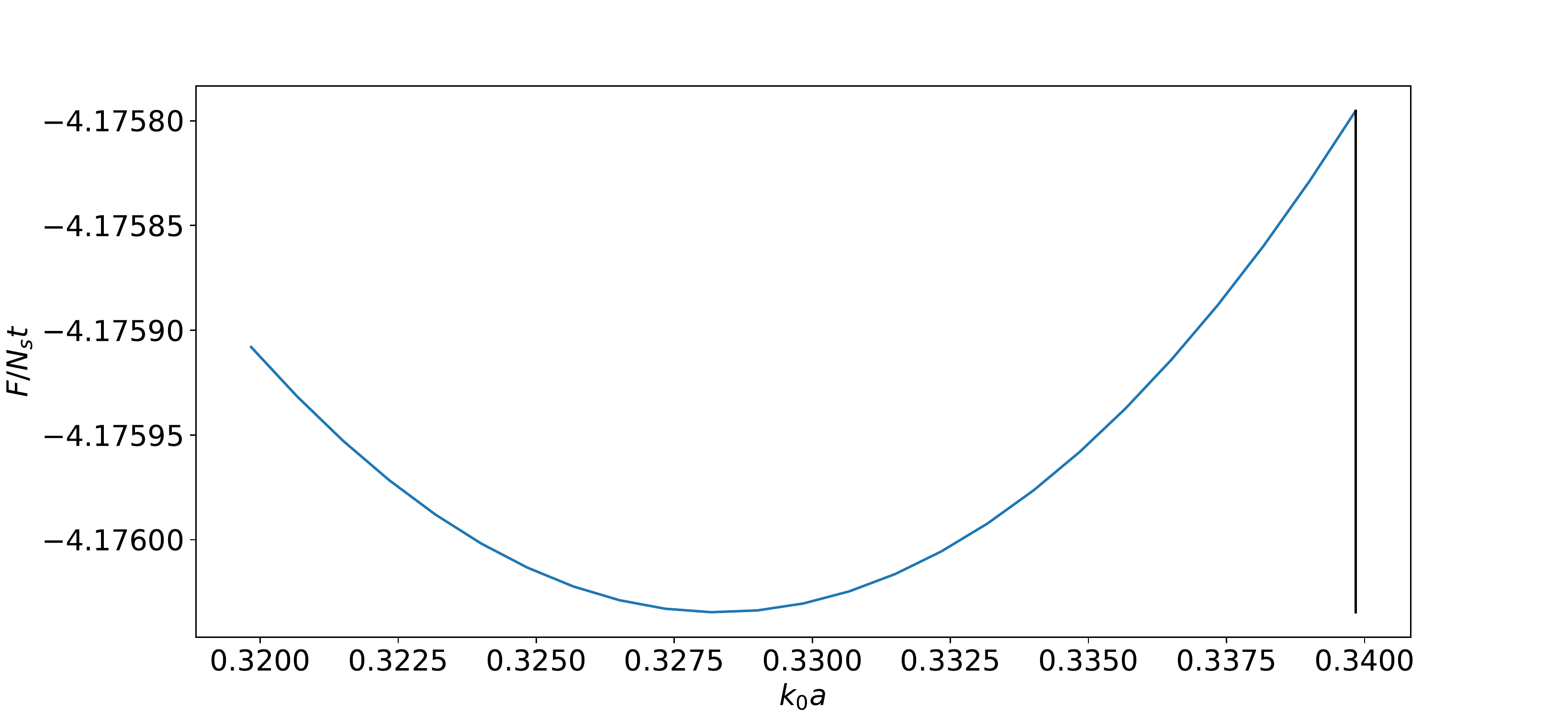}
      \caption{}
    \end{subfigure}%
    \begin{subfigure}{.49\textwidth}
      \includegraphics[width=\linewidth]{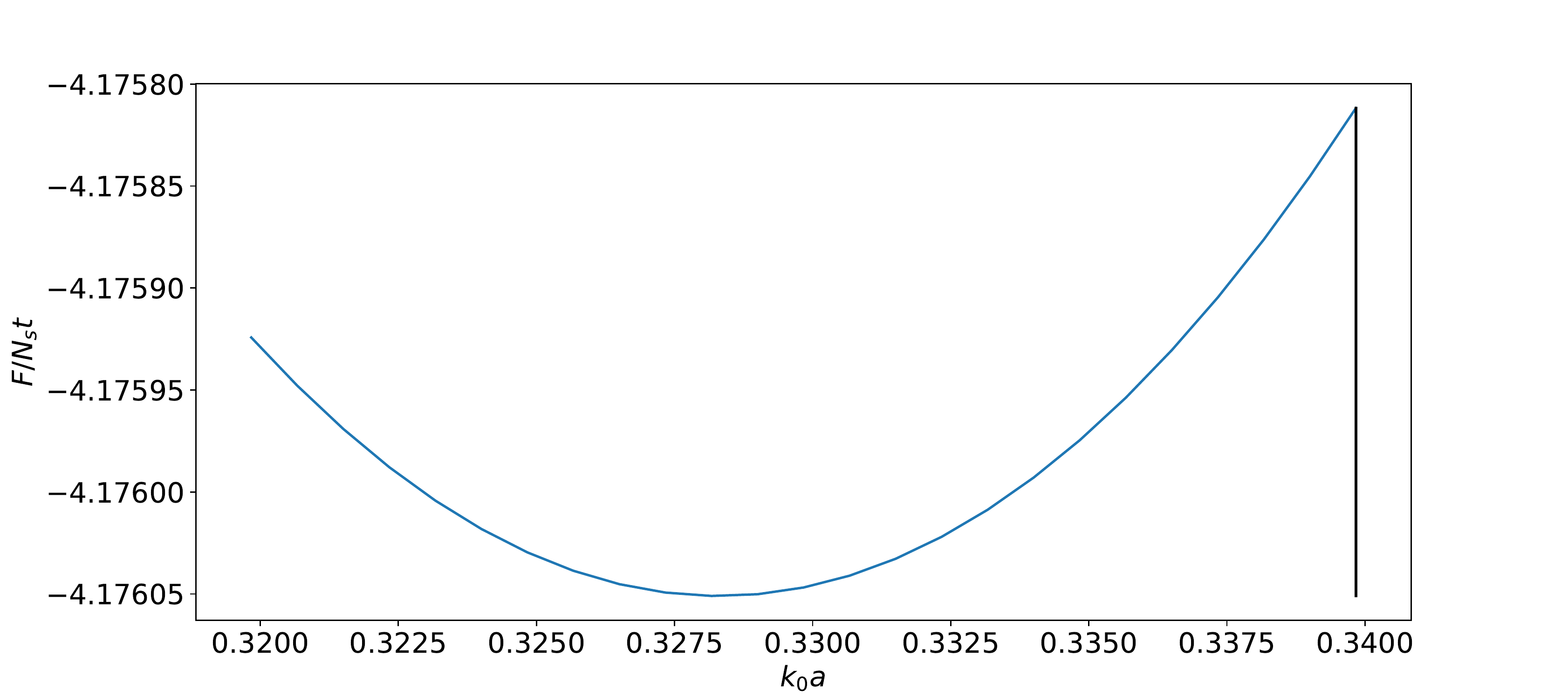}
      \caption{}
    \end{subfigure}%
    \caption{The free energy as a function of $k_0$. The black vertical line shows the position of $k_{0m}$. The angles are set to the values found to minimize $F_{\textrm{LW}}$, while the other parameters are $U_s/t = 0.05$, $\alpha = 1.5$ and $\lambda_R/t = 0.5$. In (a) the lattice size was $N_s 
    = 1600$, while in (b) it was $N_s = 4 \cdot 10^{4}$. The two figures are almost indistinguishable and lead to the same result for $k_{0\textrm{min}}$. \label{fig:LWFk0al15}}
\end{figure}

\begin{figure}
    \centering
    \begin{subfigure}{.49\textwidth}
      \includegraphics[width=\linewidth]{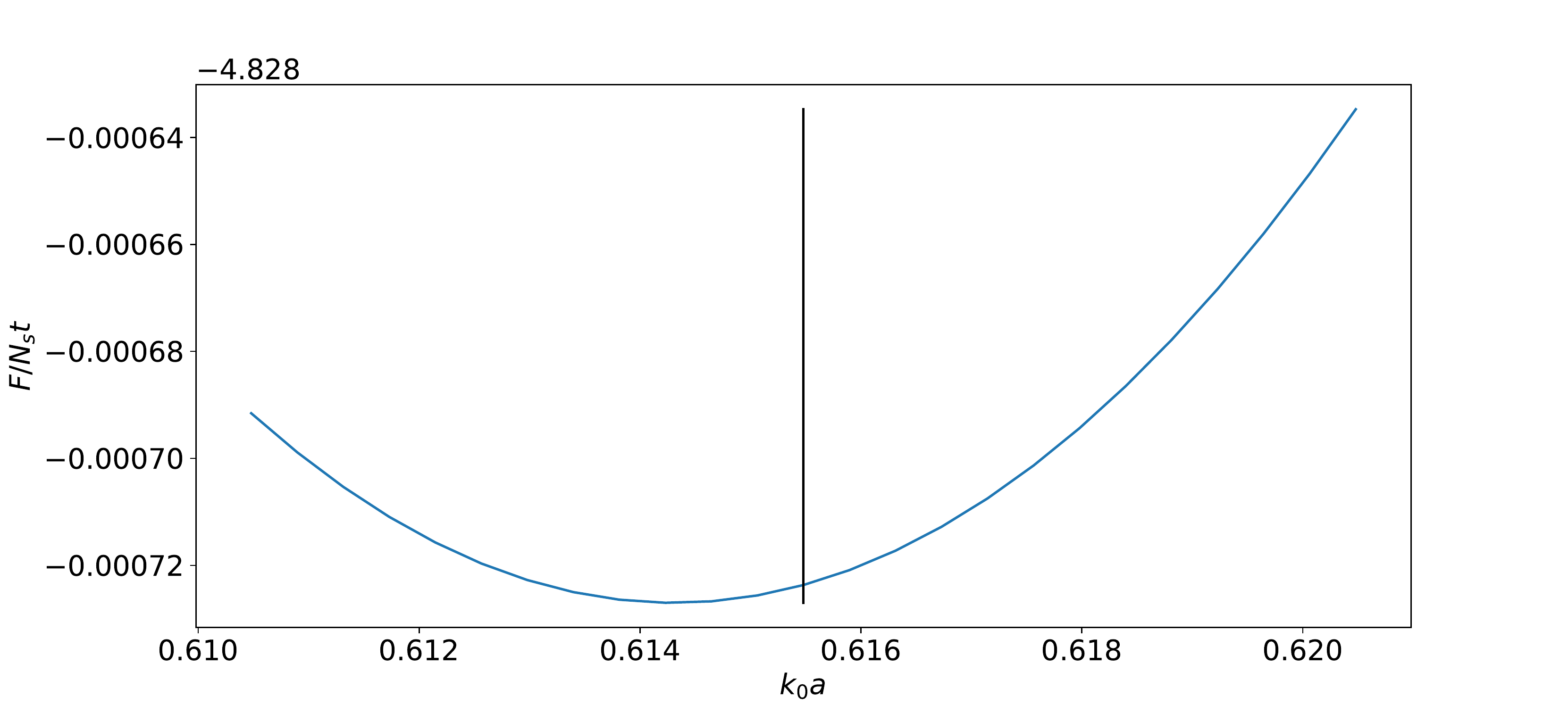}
      \caption{}
    \end{subfigure}%
    \begin{subfigure}{.49\textwidth}
      \includegraphics[width=\linewidth]{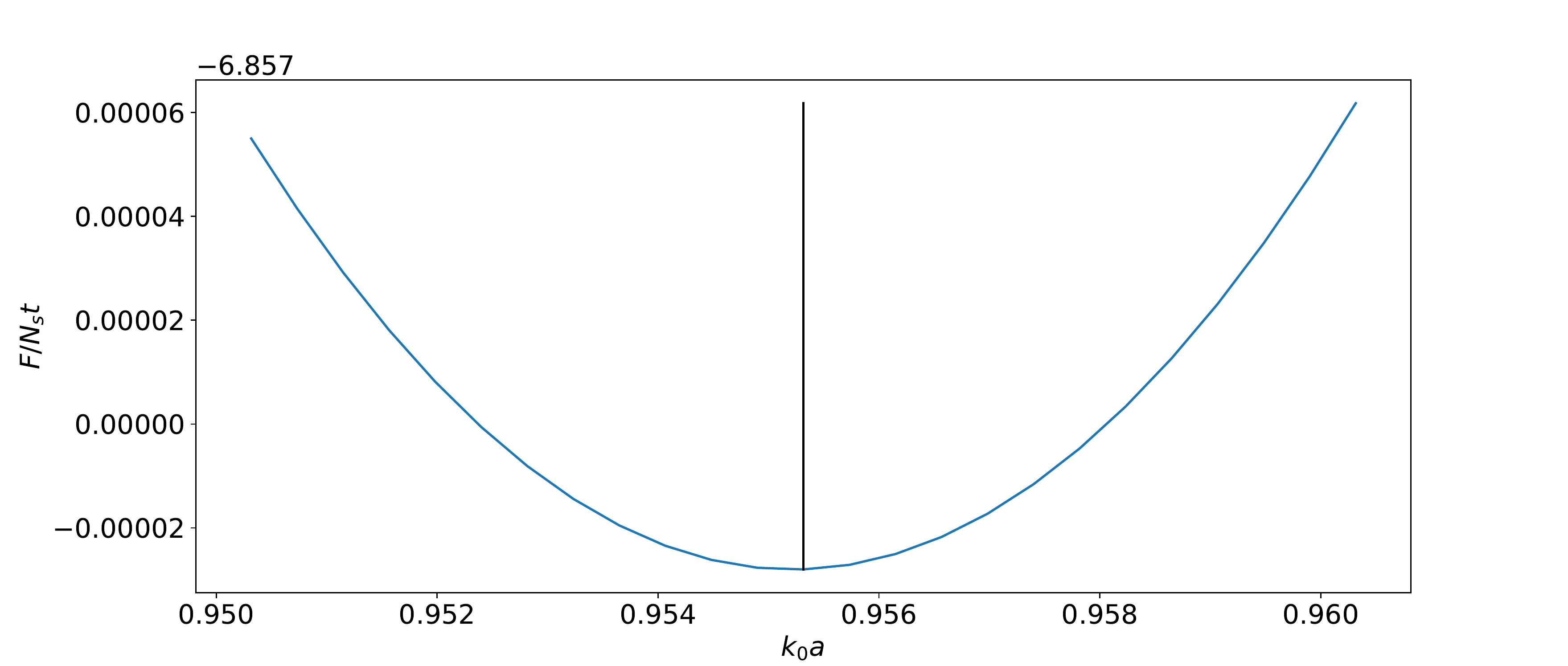}
      \caption{}
    \end{subfigure}%
    \caption{The free energy as a function of $k_0$. The black vertical line shows the position of $k_{0m}$. The angles are set to the values found to minimize $F_{\textrm{LW}}$, while the other parameters are $U_s/t = 0.05$, $\alpha = 1.5$, $\lambda_R/t = 1.0$ (a) and $\lambda_R/t = 2.0$ (b). The lattice size is $N_s 
    = 1600$. It is clear that $k_{0\textrm{min}}$ moves closer to $k_{0m}$ as the strength of SOC is increased. \label{fig:LWFk0lam}}
\end{figure}

Using these values of the angles, we now investigate minimization of $F_{\textrm{LW}}$ with respect to $k_0$. Unlike in the PW and SW phase, $k_{0\textrm{min}}$ appears to be independent of the lattice size. As shown in figure \ref{fig:LWFk0al15} $N_s=1600$ and $N_s = 4 \cdot 10^{4}$ give the same $k_{0\textrm{min}}$, and therefore, $N_s = 1600$ is used in the following. However, it is found that $k_{0\textrm{min}}$ approaches $k_{0m}$ as $\lambda_R$ is increased which is shown in figure \ref{fig:LWFk0lam}. 

With the choices of the angles found to minimize $F_{\textrm{LW}}$ the excitation spectrum is real when $\alpha$ is greater than a lower limit that is greater than $1$ and approaches $1$ from above and the strength of SOC is increased. For $\lambda_R/t = 0.75$ the limit is $\alpha \gtrsim 1.03$ while for $\lambda_R/t = 3.0$ the limit is $\alpha \gtrsim 1.004$. 

We also find that there is a lower limit on the SOC strength $\lambda_R/t$ to ensure that $k_0 = k_{0\textrm{min}}$ gives a spectrum with its global minima at the condensate momenta $\boldsymbol{k}_{0i} = (\pm k_0, \pm k_0)$. It is not completely clear how to quantitatively describe this limit on $\lambda_R/t$. However, we once again point out that $k_0$ is a discrete quantity, as it must be equal to an integer number of lattice spacings in momentum space. Let us say our lattice size is approximately $4 \cdot 10^{5}$ which is slightly larger than in typical experiments \cite{expBEC3dfillingKetterle,expBEC2dfillingSpielman}. The lattice spacing in momentum space is then $\approx 0.01/a$. The value of $k_x = k_y >0$ that corresponds to a global minimum of the excitation spectrum is named $k_g$. We require that $k_x=k_y = k_{0\textrm{min}}$ and $k_x = k_y = k_g$ correspond to the same lattice site in the discrete case. The difference between them should then be significantly less than the lattice spacing $\approx 0.01/a$. One half seems to lenient, while one tenth is probably too strict. We hence arrive at the somewhat arbitrary, though nevertheless reasonable, requirement $|k_g- k_{0\textrm{min}}|a < 0.002$. 

Since this classification is rather heuristic we only find approximate results for the $\alpha$-dependent lower limit on $\lambda_R/t$. We investigate the limit for three values of $\alpha$, $\alpha = 1.05, 1.5, 2.9$ and use these to extrapolate the approximate behavior for all $\alpha$. For $\alpha = 1.5$ and $\lambda_R = 0.8$ we find $|k_g- k_{0\textrm{min}}|a \approx 0.0024$, while for $\lambda_R = 0.85$ we find  $|k_g- k_{0\textrm{min}}|a \approx 0.0020$. Hence we say the limit is $\lambda_R \gtrsim 0.85$. At $\alpha = 1.05$ the limit is $\lambda_R/t \gtrsim 0.75$ while at $\alpha = 2.9$ we find $\lambda_R/t \gtrsim 1.16$. This fits rather well with the linear relation $\lambda_R/t = 0.52 + 0.22\alpha$ which we assume is approximately valid. 
We also find that $|k_{0m}- k_{0\textrm{min}}|a < 0.002$ at these limiting $\lambda_R/t$ values, and that the difference becomes smaller for stronger SOC or lower $\alpha$. Using $k_{0\textrm{min}} = k_{0m}$ therefore seems like a safe approximation.



In conclusion we believe the LW phase is dynamically stable for $\alpha$ greater than a lower limit above $1$ that moves close to $1$ as $\lambda_R$ is increased. In addition energetic stability sets in for $\lambda_R/t \gtrsim 0.52 + 0.22\alpha$. The choices \eqref{eq:LWgammathetapi}, \eqref{eq:LWanglesup} and \eqref{eq:LWangles23up} for the angles and $k_0 = k_{0m}$ minimizes $F_{\textrm{LW}}$. 

Note that this has been a heuristic minimization of $F_{\textrm{LW}}$ given the number of variational parameters. We can certainly claim to have found a local minimum of $F_{\textrm{LW}}$ within the set of values that render the LW phase stable. There is however no guarantee we have found the global minimum of $F_{\textrm{LW}}$. A more rigorous method to determine the variational parameters would be to use simulated annealing \cite{simulatedannealing}. Calculating the free energy $F_{\textrm{LW}}$ was however such a computationally heavy procedure that the more heuristic approach with educated guesses was used. We also note that if the global minimum of $F_{\textrm{LW}}$ lies outside the set of values that renders the LW phase stable one at least has to traverse an energy barrier to move from the minimum we have found to such a global minimum. In addition, the calculation of $F_{\textrm{LW}}$ is unclear at dynamic instabilities.